\documentclass[lettersize,journal]{IEEEtran}
\usepackage{graphicx}

\usepackage{xr}
\usepackage{mathrsfs}
\usepackage{amsthm}
\usepackage{amsmath}
\usepackage{amsfonts}
\usepackage{amssymb}
\usepackage{commath}
\usepackage[ansinew]{inputenc}
\usepackage{color}
\usepackage{url}
\usepackage{latexsym,enumerate}
\usepackage{booktabs}
\usepackage{framed}
\usepackage{subfigure}
\usepackage{multirow} 
\usepackage{color}
\usepackage{colortbl}
\usepackage{yfonts}
\usepackage{dsfont}
\usepackage{booktabs}
\usepackage{multirow}
\usepackage{ifthen}

\usepackage{algorithm}
\usepackage{algorithmic}

\definecolor{gray1}{rgb}{0.8,0.8,0.8}
\definecolor{gray2}{rgb}{0.95,0.95,0.95}

\newcommand{\argmin}{\mathop{\rm argmin}}

\newcommand{\cF}{{\mathcal F}}

\newcommand{\RQUnetVAE}{RQUNet-VAE~}

\newtheorem{thm}{Theorem}[section]
\newtheorem{prop}[thm]{Proposition}

\makeatletter
\newcommand*{\addFileDependency}[1]{% argument=file name and extension
  \typeout{(#1)}
  \@addtofilelist{#1}
  \IfFileExists{#1}{}{\typeout{No file #1.}}
}
\makeatother

\def\hidefigures{0}

\begin{document}

% \title{Bayesian Poisson Semi-supervised Learning for Hyperspectral Satellite Video Segmentation}

\title{Riesz-Quincunx-Unet Variational Auto-Encoder for Satellite Image Denoising}

\author{Duy~H.~Thai$^1$,
        Xiqi~Fei$^1$,
        Minh~Tri~Le$^1$,
        Andreas~Z\"ufle$^2$,
        Konrad Wessels$^1$\\

        $^1$George Mason University, Department of Geography and Geoinformation Science, USA
        
        $^2$Emory University, Department of Computer Science, USA

\thanks{}% <-this % stops a space
\thanks{}% <-this % stops a space
\thanks{Manuscript received ??, 2022; revised ??.}}

\maketitle

%%
%\newpage
%\tableofcontents
%
%%
%\newpage

\begin{abstract} 
Multiresolution deep learning approaches, such as the U-Net architecture, have achieved high performance in classifying and segmenting images. However, these approaches do not provide a latent image representation and cannot be used to decompose, denoise, and reconstruct image data. The U-Net and other convolutional neural network (CNNs) architectures commonly use pooling to enlarge the receptive field, which usually results in irreversible information loss.
This study proposes to include a Riesz-Quincunx (RQ) wavelet transform, which combines 1) higher-order Riesz wavelet transform and 2) orthogonal Quincunx wavelets (which have both been used to reduce blur in medical images) inside the U-net architecture, to reduce noise in satellite images and their time-series. 
%The proposed deterministic wavelet expansion and neural network in multi-scale analysis. This expansion defines a generalized Besov space with learnable frames. This dictionary approach enhances sparsity of a signal in this space.
%
In the transformed feature space, we propose a variational approach to understand how random perturbations of the features affect the image to further reduce noise.
Combining both approaches, we introduce a hybrid \RQUnetVAE scheme for image and time series decomposition used to reduce noise in satellite imagery. 
We present qualitative and quantitative experimental results  that demonstrate that our proposed \RQUnetVAE
was more effective at reducing noise in satellite imagery compared to other state-of-the-art methods. We also apply our scheme to several applications for multi-band satellite images, including: image denoising, image and time-series decomposition by diffusion and image segmentation.

\end{abstract}


% -----------------

%
\section{Introduction}  \label{sec:motivation}

The temporal frequency of medium resolution, optical satellite imagery, such as Landsat 8 and 9  and Sentinel2A\& B, has increase significant in the past four years from one observation every 16 days with Landsat 8 to an observation every 2.9 days on average \cite{li2017global,roy2019landsat}. Such time series of multi-spectral satellite data enables many novel applications such as global  agriculture monitoring and land cover change at the appropriate resolution of land management impacts that has not been possible with low resolution (250-500m GSD) time-series such as MODIS \cite{kleynhans2012land,salmon2013land,van2012hitempo}. Moreover, the harmonized Landsat Sentinel2 products open up new avenues for real-time monitoring of a wider variety of change phenomenon \cite{claverie2018harmonized}.
There is a wide variety of time series analysis methods for change detection, which focus primarily on the temporal domain while the spatial domain has been largely neglected (for review see \cite{zhu2017change}). The availability of hyper-temporal medium resolution imagery allows new application in the spatial domain, such as semantic segmentation with UNET \cite{ronneberger2015u} to track objects of change through time. 

Change detection methods are hampered by significant noise in the time series that remains despite various processing efforts to reduce the noise \cite{zhu2017change}. This noise is firstly caused by clouds, and cloud shadows that result in data gaps even when they are correctly detected and serious change artifacts are caused when some clouds or their edges remain undetected \cite{roy2005prototyping}. Second, atmospheric variability, notably water vapor and aerosols that cause variability in top of atmosphere reflectance, despite best efforts at atmospheric correction \cite{vermote2016preliminary,zhang2018characterization}. Third, BRDF variation due to variation in sun-sensor geometry and non-Lambertian reflectance properties of the target. Fourth, variations in spectral reflectance due to seasonal vegetation phenology, which become hard to model when clouds result in missing data that makes the time-series irregular and unpredictable. These issues are often addressed by multi-date composites or temporal interpolation to fill in clouds, shadows and other missing data. 
Despite efforts to reduce the impact of spatio-temporal noise in the satellite time series, it often limits the accuracy of timely land cover change detection over very large areas using either conventional time-series analysis \cite{woodcock2020transitioning}, as well as the new generation of machine learning methods \cite{zhu2017change}.

%
 %Image decomposition stays at the heart of approximation theory to decompose an image into several components. 
 
 Various cutting-edge techniques have been developed for image decomposition that can be used to remove noise from conventional images or videos, for example, wavelet smoothing techniques \cite{UnserVandeville2010, FeilnerVilleUnser2005, CaiChanShen2008} and regularization methods \cite{RudinOsherFatemi1992}, but these methods require parameter selection that vary greatly between datasets and cannot optimally represent varying signals over space and time commonly found in earth observation images. On the other hand, a neural network (UNet) can automatically learn optimal local representation of signal and noise \cite{ronneberger2015u, tian2020deep}, but incurs high computational cost and requires prohibitively large sets of domain-specific training data.
The approach proposed in this paper combines the strengths of convolutional neural networks and conventional smoothing techniques.
For this purpose we introduce a novel non-subsampled high order Riesz-Quincunx wavelet with variational auto-encoder UNet (RQUNet-VAE) as a hybrid combination of deterministic wavelet expansion (two-dimensional Riesz transformation~\cite{UnserVandeville2010}, Quincunx wavelet) and a variational version of a convolutional neural network \cite{pu2016variational}. 
The proposed RQUNet-VAE approach uses framelet decomposition \cite{YeHanCha2018} to map an image into a sparse feature space and leverages UNet~\cite{ronneberger2015u} to enable learning of the frames of the feature space.
The rationale is that such a hybrid method should provide better feature representation and artifact reduction compared to conventional approaches. Therefore, while previous studies used UNet-VAE for image segmentation \cite{kohl2018probabilistic}, we use it to mimic the properties of latent factor model in our RQUNet-VAE decomposition.

To implement this concept, we first introduce a non-subsampled version of high-order Riesz wavelet expansion~\cite{UnserVandeville2010, UnserChenouardVandeville2011, UnserSageVandeville2009} with Quincunx sampling~\cite{FeilnerVilleUnser2005}. Quincunx sampling is used to reduce redundancy of an expansion while high order Riesz transform is used to increase the directional property of wavelet expansion. 
The hybrid model reduces computational cost with deterministic bases as predefined parameters instead of letting all bases be learnable parameters. 

Next, this Riesz-Quincunx wavelet is integrated into the skip-connections of the deep neural network UNet-VAE~\cite{kingma2013auto, kohl2018probabilistic, esser2018variational} for learning new bases from the training dataset.
Our rationale of using wavelet expansion is that signals extracted from the UNet-VAE encoder, at the skip-connection level, contains both the main signal and details (as well as noise) of an input image which are separated into scaling and wavelet coefficients. Truncation of these coefficients eliminates small wavelet coefficients which contain noise and detailed texture. By decoding the remaining coefficients back into image space, we obtain a denoised version of the original image. 

%This dictionary learning approach provides flexible expansion whose encoder and decoder play as analysis and synthesis operators in wavelet expansion and also satisfy the unity condition by forcing a mean square error term in the loss function for learning model parameters. The variational term here is to learn variance of an image's mean by a latent variable sampled from a normal distribution conditioned on an input signal.

The theoretical framework of RQUNet-VAE is based on Hankel matrix algebra~\cite{partington1988introduction}, framelet decomposition \cite{YeHanCha2018, YinGaoLuDaubechies2017} and proximal operators \cite{PolsonScottWillard2015}. 
Framelet decomposition uses isotropic family-matrix convolution to combine all channels of a multi-band image in a convolutional operation with learnable frames. 
Furthermore, there is a connection between framelet decomposition and sampling with a finite rate of innovation \cite{VetterliMarzilianoBlu2002} via Hankel matrix theory and annihilating filter. 
We prove that RQUNet-VAE also relates to latent factor model whose loss function is defined by Kullback-Leibler divergence~\cite{van2014renyi}.

Finally we demonstrate how to apply our proposed \RQUnetVAE to satellite image time series, that is, sequences of images of the same area. Reducing noise satellite image time series is challenging due to the severe background noise resulting from spectral variability caused by changing environmental conditions due to atmospheric and seasonal variability, remaining small clouds, as well as variable sun-sensor geometry~\cite{roy2019landsat,zhang2018characterization}. The level of noise is very different from that of conventional videos which contains slow motion changes between subsequent images that can be removed with time delay embedding, for example. In contrast, the satellite time series are constituted by discrete frames that are independently capture several days apart, causing large variability in reflectance properties of the background and objects. 

To test the effectiveness of our proposed concept, the RQUNet-VAE was applied to image decomposition, image denoising and segmentation of satellite images and their time-series. Our experimental results show that our hybrid method provides better feature representation and artifact reduction than traditional approaches.
The objectives of this paper are: \vspace{-0.15cm}
\begin{enumerate}

 \item introduce \RQUnetVAE as a generalized wavelet expansion approach;

 \item extend RQUnet-VAE to a diffusion process \cite{RichterThaiHuckemann2020} and enable spectral decomposition \cite{Gilboa2014}; 

 \item apply it to image denoising and segmentation in noisy environment for multi-band satellite images and their time-series.
% and compare them to traditional methods?? 
\end{enumerate}

%
% Moreover, we also apply our technique in multi-band image segmentation and provide a superior result in segmentation of large scale objects in comparison with a typical Unet. Note that a typical Unet is proposed for medical image segmentation whose noise is often assumed to be i.i.d. Gaussian and this noise is different from a noise in satellite video.
% Note that Smoothing to remove noise [or stabilize the environment changes] is the 1st step for all later methods. 

% We mainly apply this idea to multi-band image decomposition by spectral analysis to separate features into different scales; then, we extend this decomposition for video processing by using wavelet expansion in time, e.g. Haar wavelet.

%The RQUnet VAE was demonstrated on (i) 1m aerial photographs (NAIP) with added noise, Sentinel2 (S2) images and S2 time series visualised as videos.
Organization of the paper is:
Section \ref{sec:RieszQuincunxUnetVAE} provides the RQUNet-VAE expansion, including mathematical properties and image or time-series decomposition; Section \ref{sec:ExperimentalResults} gives numerical examples and comparisons of image denoising and segmentation in noisy environments for multi-band satellites images. Section \ref{sec:Conclusion} gives the conclusions, Mathematical background, proofs, and additional experiments are are provided in our supplemental material (SM).
% We need to include above the experiments we used to test the benefit of RQ-Unet 

% -------------------------------
\vspace{-0.4cm}
\section{Riesz-Quincunx-UNet variational auto-encoder (\RQUnetVAE)}  \vspace{-0.2cm}
\label{sec:RieszQuincunxUnetVAE}

\subsection{Notations and Definitions:} \vspace{-0.2cm}
\label{sec:NotationsDefinitions}
The following list provides an overview of notations used throughout this work. \vspace{-0.1cm}
\begin{itemize}
 \item continuous coordinate in the spatial domain: $x = (x_1 \,, x_2) \in \mathbb R^2$,

 \item discrete coordinate in the spatial domain: $k = (k_1 \,, k_2) \,, m = (m_1 \,, m_2) \in \mathbb Z^2$,

 \item Fourier coordinate: $\omega = (\omega_1 \,, \omega_2) \in [-\pi \,, \pi]^2$,

 \item Complex number: $j = \sqrt{-1}$, 

 \item Image domain: $\Omega = \{ 1 \,, \ldots \,, n_1 \} \times \{ 1 \,, \ldots \,, n_2 \}$, $\abs{\Omega} = n_1 \times n_2$,

 \item Images and stacks of images:

  \begin{itemize}
   \item a gray-scale image $\underline{f} = \begin{pmatrix} f_1 & \hdots &  f_{n_2} \end{pmatrix} \in \mathbb R^{\abs{\Omega}} \,, f_i \in \mathbb R^{n_1}$, 

   \item a multi-channel image $\underline{\underline{f}} = \left\{ \underline{f_1} \,, \hdots \,, \underline{f_{P}} \right\} \in \mathbb R^{\abs{\Omega} \times P} \,, \underline{f_i} \in \mathbb R^{\abs{\Omega}}$,

   \item a set of observed images $\textfrak{F} := \underline{\underline{\underline{f}}} = \left\{ \underline{\underline{f_i}} \right\}_{i=1}^T \in \mathbb R^{\abs{\Omega} \times P \times T}$,
  \end{itemize}

 \item Distribution and Lebesgue density: $\underline{f} \stackrel{\text{i.i.d.}}{\sim} \mathbb{P}\left( \underline{f} \right) := \mathbb{P}\left( \text{d} \underline{f} \right) = \mathtt{p}\left( \underline{f} \right) \text{d}\underline{f}$,

 \item Continuous Fourier transform of a continuous function $a \in L_2(\mathbb R^2)$ is:
\begin{align*}
 a(x) &\stackrel{\cF}{\longleftrightarrow} 
 \widehat{a}(\omega) = \int_{\mathbb R^2} a(x) e^{-j \langle x \,, \omega \rangle_{\ell_2} } \text{d} x \,,
\end{align*} 
and its discrete version is computed via Poisson summation formulae:
\begin{align*}
 a[k] &:= a(x) \mid_{x = k \in \mathbb Z^2} 
 \\  % ----
 \stackrel{\cF}{\longleftrightarrow} 
 \widehat{A} \left( e^{j \omega} \right) &= \sum_{k \in \mathbb Z^2} a[k] e^{-j \langle k \,, \omega \rangle_{\ell_2} } 
 = \sum_{k \in \mathbb Z^2} \widehat{a}(2 \pi k + \omega) \,.
\end{align*}

\end{itemize}
Additional definitions can be found in the supplemental material.

\begin{figure*}[t]
\begin{center}  

\includegraphics[width=0.99\textwidth]{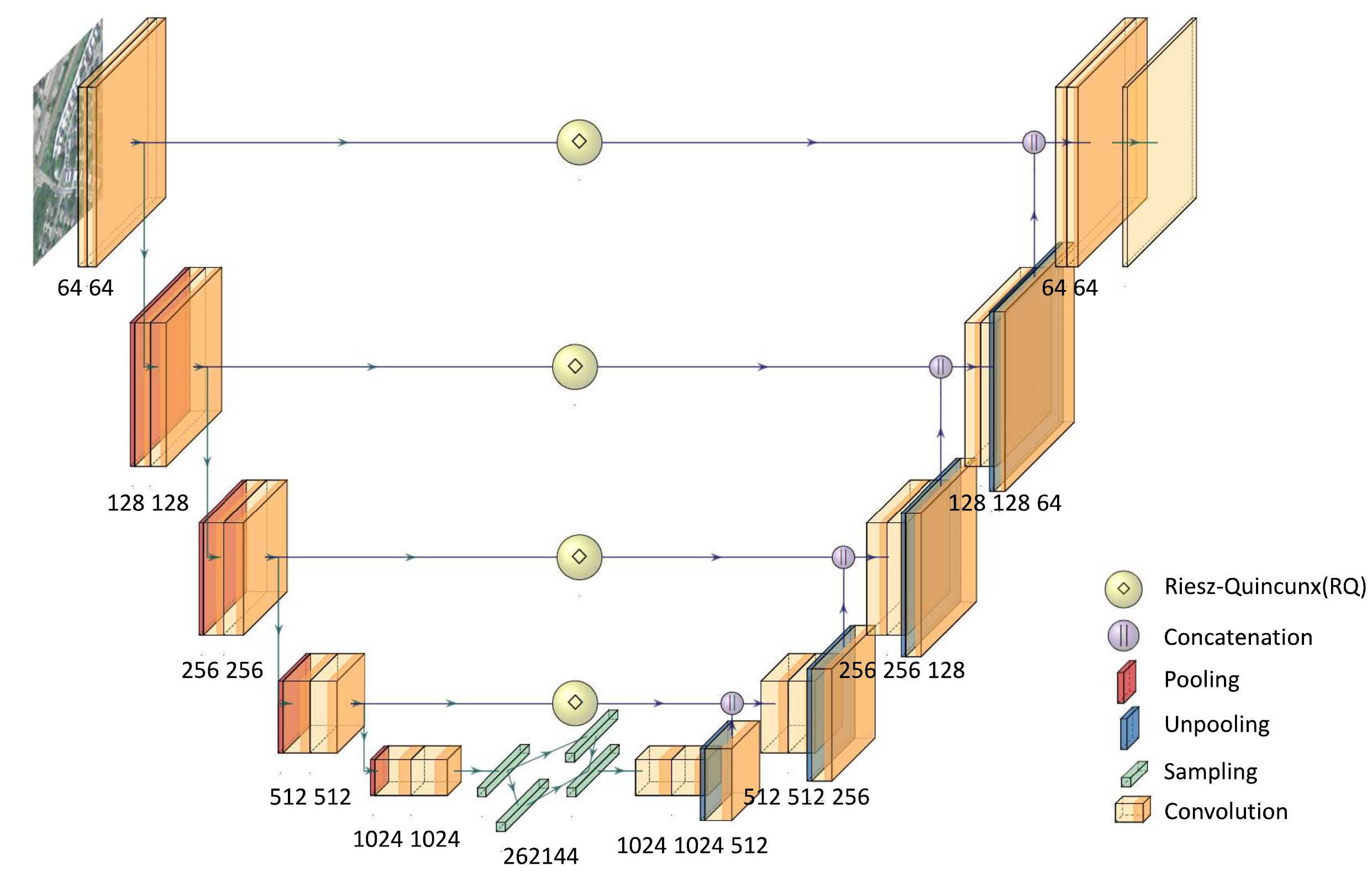}

\caption{\RQUnetVAE architecture. \label{fig:RQUnetVAEarchitecture}}\vspace{-0.4cm}
\end{center}
\end{figure*}

\subsection{\RQUnetVAE Architecture Overview}\label{sec:architecture_overview}
The network architecture of  our proposed \RQUnetVAE is illustrated in  Figure~\ref{fig:RQUnetVAEarchitecture}. The network utilizes a UNet~\cite{ronneberger2015u} as our primary backbone architecture with the modified skip-connection signals, between the encoder and decoder paths, and bottom layer. The network introduces the Variational Autoencoder (VAE) in the bottom layer to learn the signal distribution in the latent space during training on large datasets. The VAE layer combined with the UNet backbone allows the model to learn and retain much of the input information, similar to generative models, to perform image reconstruction. The reconstructed image then can go through a convolution layer, or a classifier, to produce the final classification or segmentation output. The primary objective of our proposed network is to remove certain levels of noise in the input images in order to produce improved, quality reconstructed images using the Riesz-Quincunx (RQ) scheme in the skip-connections, before performing image classification. Since the RQ scheme is computational heavy, we only applied the RQ computation on a pre-trained network to produce predictions on input images.

\subsection{$N$-th order Riesz Quincunx non-subsampled wavelet}\label{sec:RQ}

Before proposing our RQUNet-VAE expansion, we firstly introduce framelet decomposition \cite{YeHanCha2018} in the following proposition. 
%Notation and background used in the following can be found in Appendix~\ref{app:background}.

\begin{prop}  \label{prop:FrameletDecomposition}
Given an image $\underline{f} \in \mathbb R^{n_1 \times n_2}$, all wavelet filter banks $\underline{\Phi} \,, \underline{\tilde{\Phi}} \in \mathbb R^{n_2 d_1 \times d_2}$ (see definitions of wavelet filter banks in Section~4 of the supplemental material)
and local basis $\underline{\Xi} \,, \underline{\tilde{\Xi}} \in \mathbb R^{n_1 \times d}$ (see Equation~5 in the supplemental material for details) 
satisfy the unity conditions
\begin{align}
 \label{eq:OrthonormalFrameConditions}
 \underline{\tilde{\Xi}} \, \underline{\Xi}^\text{T} &= \sum_{i=1}^d \tilde{\xi}_i \xi_i^\text{T} = \text{Id}_{n_1 \times n_1} \,,~
 \underline{\Phi} \, \underline{\tilde{\Phi}}^\text{T} = \sum_{i=1}^{d_2} \tilde{\phi}_i \phi_i^\text{T} = \text{Id}_{n_2 d_1 \times n_2 d_1} %\,;
\end{align}
then, a framelet decomposition is:
\begin{align}  \label{eq:FrameletDecomposition:prop}
 \underline{f} &= \mathscr{H}^\dagger_{d_1 \mid n_2} \left( \underline{\tilde{\Xi}} \underline{c_f} \tilde{\Phi}^\text{T} \right)% \,,~
 \notag
 \\
 &= \frac{1}{d_1} \sum_{s=1}^{d_2} \sum_{l=1}^d \sum_{i=1}^{n_2}
 \begin{pmatrix} \left\langle f_k \,, \mathfrak{C}_{\phi_s^i} \left( \xi_l \right) \right\rangle_{\ell_2} 
 \mathfrak{C}_{\tilde{\phi}^{1}_s} \left( \tilde{\xi}_l \right)
 & 
 \\
 \cdot  \hdots \cdot
 \\
 \left\langle f_k \,, \mathfrak{C}_{\phi_s^i} \left( \xi_l \right) \right\rangle_{\ell_2} \mathfrak{C}_{\tilde{\phi}^{n_2}_s} \left( \tilde{\xi}_l \right) \end{pmatrix}
 \notag
 \\ % -----
 &= \frac{1}{d_1} \sum_{s=1}^{d_2} \begin{pmatrix}  
 \mathfrak{C}_{\tilde{\phi}^{1}_s} \left( \underline{\tilde{\Xi}} c_{f,s} \right)
 & \hdots & 
 \mathfrak{C}_{\tilde{\phi}^{n_2}_s} \left( \underline{\tilde{\Xi}} c_{f,s} \right) \end{pmatrix}% \,,
\end{align}
where $\mathscr{H}^\dagger_{d_1 \mid n_2}$ is the extended Hankel matrix of image \underline{f} as described in the SM. Framelet coefficients are:
\begin{align}  \label{eq:FrameletCoefficients}
 \underline{c_f} &:= \begin{pmatrix} c_{f,1} & \hdots & c_{f,d_2} \end{pmatrix}
 = \underline{\Xi}^\text{T} \mathscr{H}_{d_1 \mid n_2} \left( \underline{f} \right) \underline{\Phi} \,.
\end{align}

\end{prop}
\begin{proof}
For self-containment, we provide a proof of Proposition~\ref{prop:FrameletDecomposition} in Section 5.1. in the SM.
\end{proof}
Inspired by Proposition \ref{prop:FrameletDecomposition}, we introduce the proposed RQUNet-VAE expansion in the following.

\paragraph{Riesz-Quincunx wavelet expansion:}
To have wavelet expansion, firstly we provide a definition of frame by isotropic polyharmonic Bspline and $N$-th order Riesz transform, see~\cite{UnserVandeville2010}. In particular, an isotropic polyharmonic B-splines basis is defined in the Fourier domain as: $x \in \mathbb R^2$, 
\begin{align}  \label{eq:B-spline}
 \beta_\gamma(x) \stackrel{\cF}{\longleftrightarrow}
 \widehat \beta_\gamma(\omega) 
 = \frac{ \widehat{V}^{iso}(e^{j\omega})^\frac{\gamma}{2} }{ \norm{\omega}^\gamma_{\ell_2} } 
\end{align}
with a 2D localization operator:
\begin{align*}
 &\widehat{V}^\text{iso}(e^{j\omega}) = \frac{10}{3} 
 \\&
 - \frac{1}{3} \left[ 4 \cos \omega_1 + 4 \cos \omega_2 + \cos(\omega_1 + \omega_2) + \cos(\omega_1 - \omega_2) \right] \,.
\end{align*}  
Its dual function is defined via an auto-correlation function: $k \in \mathbb Z^2$,
\begin{align}  \label{eq:dualBspline_autocorrelation}
 \tilde \beta_\gamma(x)
 &\stackrel{\cF}{\longleftrightarrow}
 \widehat{\tilde \beta}_\gamma(\omega) = \frac{ \widehat \beta_\gamma(\omega) }{ \widehat{A}(e^{j\omega}) } \,,~
 \notag
 \\
 a(k) &\stackrel{\cF}{\longleftrightarrow}
 \widehat{A}(e^{j\omega}) = \sum_{m \in \mathbb Z^2} \widehat \beta_{2\gamma}(2\pi m + \omega) \,. 
\end{align}
An impulse response of the $L$-th order Riesz transform is: 
\begin{align*}
 \mathscr{R}^l \{ \delta \}(x) ~\stackrel{\cF}{\longleftrightarrow}~
 \widehat{\mathscr{R}^l}(\omega) = (-j)^L \sqrt{\frac{L!}{n!(L-l)!}} \frac{\omega_1^l \omega_2^{L-l}}{\left( \omega_1^2 + \omega_2^2 \right)^\frac{L}{2}} \,,
\end{align*}
for $l = 0 \,, \ldots \,, L$ and where $\delta(\cdot)$ is the Dirac delta function.
Secondly, wavelet expansion is defined as follows: Given a 2D dyadic sampling matrix
$\underline{\text{D}} = \begin{pmatrix} 1 & 1 \\ 1 & -1 \end{pmatrix}$, we have 
$\widehat{\rho}_0 = [0 \,, 0]^\text{T} \,, \widehat{\rho}_1 = [1 \,, 0]^\text{T}$
and
\begin{align*}
 &\underline{\text{D}}^i = \begin{cases} 2^\frac{i}{2} \text{Id}_2   \,, & i \text{ is even}
  \\ 2^\frac{i - 1}{2} \underline{\text{D}} \,, & i \text{ is odd}
 \end{cases} \,,~
 \abs{\text{det } \underline{\text{D}}^i} = 2^i \,.
\end{align*}
Non-subsampled scaling and wavelet spaces
$\mathcal{V}_i^\text{ns} = \text{span}_{m \in \mathbb Z^2} \left\{ \varphi_i \left( \cdot - m \right)
\right\}$
and
$\mathcal{W}_{il}^\text{ns} = \text{span}_{m \in \mathbb Z^2} \left\{ \psi_{il} \left( \cdot - m \right) \right\}$
(for $i = 0 \,, \ldots \,, I \,; l = 0 \,, \ldots \,, L$)
satisfy a multiscale decomposition
$\mathcal{V}_{i-1}^\text{ns} = \mathcal{V}_{i}^\text{ns} \oplus \sum_{l=0}^L \mathcal{W}_{il}^\text{ns} \,, i = 1 \,, \ldots \,, I
$
where primal/dual scaling and wavelet functions are:
\begin{align*}
 \varphi_i (x) &= 2^{-\frac{i}{2}} \beta_\gamma \left( \underline{\text{D}}^{-i} x \right) \,,~
 \tilde{\varphi}_i (x) = 2^{-\frac{i}{2}} \tilde{\beta}_\gamma \left( \underline{\text{D}}^{-i} x \right) \,,
 \\  % ----
 \psi_{il}(x) &= 2^{-\frac{i}{2}} \mathscr{R}^l \{ \psi \}\left( \underline{\text{D}}^{-i} x \right) \,,~
 \tilde{\psi}_{il}(x) = 2^{-\frac{i}{2}} \mathscr{R}^l \{ \tilde{\psi} \}\left( \underline{\text{D}}^{-i} x \right) \,.
\end{align*}
Their Fourier transforms are
\begin{align*}
 \widehat{\varphi}_i (\omega) &= 2^\frac{i}{2} \widehat{\beta_\gamma} \left( \underline{\text{D}}^{i \text{T}} \omega \right) \,,~
 \widehat{\tilde \varphi}_i (\omega) = 2^\frac{i}{2} \widehat{\tilde \beta_\gamma} \left( \underline{\text{D}}^{i \text{T}} \omega \right) \,,
 \\  % ----
 \widehat{\psi_{il}}(\omega) &= 2^\frac{i}{2} \widehat{\mathscr{R}^l} \left( \underline{\text{D}}^{i \text{T}} \omega \right) \widehat{\psi} \left( \underline{\text{D}}^{i \text{T}} \omega \right) \,,~
 % ----
 \\
 \widehat{\tilde{\psi}_{il}}(\omega) &= 2^\frac{i}{2} \widehat{\mathscr{R}^l} \left( \underline{\text{D}}^{i \text{T}} \omega \right) \widehat{\tilde{\psi}} \left( \underline{\text{D}}^{i \text{T}} \omega \right) \,.
\end{align*}
Due to a discrete Fourier transform of continuous functions, by Poisson summation we have the following proposition:
\begin{prop}  \label{prop:UnityCondition}
 The scaling and wavelet functions satisfy the unity condition in the Fourier domain:
 \begin{align}  \label{eq:identity:prop}
  \widehat{\tilde{\varphi}}^*_I (\omega) \widehat{\varphi}_I (\omega)
  + \sum_{i=0}^I \sum_{l=0}^L 
  \widehat{\tilde{\psi}}^*_{il} (\omega) \widehat{\psi}_{il}(\omega)
  + \widehat{e}(\omega) = 1
 \end{align}
 up to a discretization error:
 {\small
 \begin{align}  \label{eq:ErrorUnity}
  &\widehat{e}(\omega) = \widehat{\tilde{\varphi}}^*_I (\omega) \sum_{k \in \mathbb Z^2 \backslash \{ 0 \}} \widehat{\varphi}_I \left( 2\pi k + \omega \right) 
  \notag
  \\ %
  &+ \sum_{m \in \mathbb Z^2 \backslash \{0\}} \widehat{\tilde{\varphi}}^*_I \left( 2\pi m + \omega \right) \left( \widehat{\varphi}_I (\omega)
  + \sum_{k \in \mathbb Z^2 \backslash \{ 0 \}} \widehat{\varphi}_I \left( 2\pi k + \omega \right) \right)
  \notag
  \\  % ----
  &+ \sum_{i=0}^I \sum_{l=0}^L \Bigg[ \widehat{\tilde{\psi}}^*_{il} (\omega) \sum_{k \in \mathbb Z^2 \backslash \{ 0 \}} \widehat{\psi}_{il} \left( 2\pi k + \omega \right) 
  \notag
  \\ %
  &+ \sum_{m \in \mathbb Z^2 \backslash \{0\}} \widehat{\tilde{\psi}}^*_{il} \left( 2\pi m + \omega \right)
  \left( \widehat{\psi}_{il}(\omega)
  + \sum_{k \in \mathbb Z^2 \backslash \{ 0 \}} \widehat{\psi}_{il} \left( 2\pi k + \omega \right) \right)
  \Bigg] \,.
 \end{align}
 }
\end{prop}
\begin{proof}
We provide a proof of Proposition~\ref{prop:UnityCondition} in Section~5.2 in SM.
\end{proof}
To compensate the error in Equation~\ref{eq:ErrorUnity} in the unity condition (Equation~\ref{eq:identity:prop}), wavelet function at scale $i=0$ is defined as:
\begin{align*}
 \widehat{\psi}_{0}(\omega)
 = \frac{1}{\widehat{\tilde{\psi}}^*_{0} (\omega)}
 \left( 1 - \widehat{\tilde{\varphi}}^*_I (\omega) \widehat{\varphi}_I (\omega)
 - \sum_{i=1}^I \widehat{\tilde{\psi}}^*_{i} (\omega) \widehat{\psi}_{i}(\omega) \right) \,;
\end{align*}
then, we have a wavelet expansion for an image $f \in \ell_2(\mathbb Z^2)$:
\begin{align}  \label{eq:WaveletExpansion}
 f[k] &= \sum_{m \in \mathbb Z^2}
 \left\langle f \,, \tilde \varphi_I (\cdot - m) \right\rangle_{\ell_2} \varphi_I (k - m)
 \notag
 \\ %
 &+ \sum_{i=0}^I \sum_{l=0}^L \sum_{m \in \mathbb Z^2}
 \left\langle f \,, \tilde \psi_{il} (\cdot - m) \right\rangle_{\ell_2} \psi_{il} (k - m) \,.
\end{align}
Following \cite{UnserChenouardVandeville2011}, primal and dual wavelet functions are defined as:
\begin{align} \label{eq:waveletfunc}
 \widehat \psi(\omega) &= 2^{-\frac{1}{2}}
 \widehat{G}(e^{j\underline{\text{D}}^{-\text{T}} \omega}) \widehat \beta_\gamma(\underline{\text{D}}^{-\text{T}} \omega) \,,
 \notag
 \\ %
 \widehat{\tilde \psi}(\omega) &= 2^{-\frac{1}{2}} \widehat{\tilde{G}} \left( e^{j \underline{\text{D}}^{-\text{T}} \omega} \right) \widehat{\tilde \beta}_\gamma(\underline{\text{D}}^{-\text{T}} \omega) \,,
\end{align}
where the refinement and highpass filters are
\begin{align*} 
 % \label{eq:refinementfilter}
 h[k] &\stackrel{\cF}{\longleftrightarrow}
 \widehat{H}(e^{j \omega}) = 2^\frac{1}{2}  \frac{ \widehat \beta_\gamma(\underline{\text{D}}^\text{T} \omega) }{ \widehat \beta_\gamma(\omega) } \,,
 \notag
 \\
 \tilde h[k] &\stackrel{\cF}{\longleftrightarrow}
 \widehat{\tilde{H}}(e^{j \omega}) = \frac{ \widehat{A}(e^{j \omega}) }{ \widehat{A}(e^{j \underline{\text{D}}^\text{T} \omega}) } \widehat{L}(e^{j \omega}) \,,
 \notag
 \\
 % \label{eq:waveletfilter:VDVilleUnserPolyharmonicBspline}
 g[k] &\stackrel{\cF}{\longleftrightarrow}
 \widehat{G}(e^{j \omega}) = - e^{- j \omega_1} \widehat{H}( e^{- j (\omega + \pi)}) 
 \widehat{A}( e^{j (\omega + \pi)} ) \,,
 \notag
 \\ %
 \tilde{g}[k] &\stackrel{\cF}{\longleftrightarrow}
 \widehat{\tilde{G}}(e^{j \omega}) = - e^{- j \omega_1} \frac{ \widehat{H}( e^{- j (\omega + \pi)} ) }{ \widehat{A}(e^{j \underline{\text{D}}^\text{T} \omega }) } \,,  
\end{align*}
and a scaled auto-correlation function is:
{\small
\begin{align}   
 \label{eq:autocorrelationfunc:2scale}
 \widehat{A}(e^{j \text{D}^\text{T} \omega}) &=  \frac{1}{2} \abs{ \widehat{H}(e^{j \omega}) }^2 \widehat{A}(e^{j \omega}) + \frac{1}{2} \abs{ \widehat{H}( e^{j (\omega+\pi)} ) }^2 \widehat{A}( e^{j (\omega + \pi)} ) \,.
\end{align}
}

% ===========

%
\paragraph{Non-subsampled Riesz-Quincunx wavelet smoothing and Hankel matrix}
Denote $\underline{\mathfrak{C}}_{\underline{\phi}} := \underline{\mathfrak{C}}^*_{\underline{\varphi_I}} \underline{\mathfrak{C}}_{\underline{\tilde\varphi_I}}$
where a matrix kernel $\underline{\phi} \in \mathbb R^{n_1 \times n_2}$ and its matrix form $\underline{\Phi} := \underline{\tilde\Phi_I} \, \underline{\check\Phi_I^\text{T}}$ 
are defined in Equation~3 in the SM where $\underline{\tilde\Phi_I}$ and $\underline{\Phi_I}$ are matrix forms of $\underline{\tilde\varphi_I}$ and $\underline{\varphi_I}$, respectively.
Similarly, we denote wavelet kernel tensors
$\underline{\underline{\psi}} = \left\{ \psi_p \right\}_{p=1}^P := \left\{ \psi_{il} \right\}_{i=0,\ldots,I}^{l=0,\ldots,L}$
and 
$\underline{\underline{\tilde \psi}} = \left\{ \tilde{\psi}_p \right\}_{p=1}^P := \left\{ \tilde{\psi}_{il} \right\}_{i=0,\ldots,I}^{l=0,\ldots,L}$
with $\underline{\psi_p} \,, \underline{\tilde\psi_p} \in \mathbb R^{\abs{\Omega}} \, (p = 1 \,, \ldots \,, P)$ 
and their matrix form are
$\underline{\check{\Psi}} = \begin{pmatrix} \underline{\check{\Psi}_1} & \hdots & \underline{\check{\Psi}_P} \end{pmatrix}$ and 
$\underline{\tilde{\Psi}} = \begin{pmatrix} \underline{\tilde\Psi_1} & \hdots & \underline{\tilde\Psi_P} \end{pmatrix}$
where block element matrices 
$\underline{\Psi_p} \,, \underline{\tilde\Psi_p} \in \mathbb R^{n_2 n_1 \times n_2}$ are also are defined in Equation~3 in the supplemental material.
For an image $\underline{f} \in \mathbb R^{n_1 \times n_2}$, from proposition \ref{prop:FrameletDecomposition} we have the following proposition:
\begin{prop}  \label{prop:RieszQuincunxWavelet:FrameletDecomposition}
 A non-subsampled Riesz Quincunx wavelet  has a form of framelet decomposition (\ref{eq:FrameletDecomposition:prop}):
 \begin{align}  \label{eq:RieszQuincunxWavelet}
  \underline{f} &= \underline{\mathfrak{C}}_{\underline{\phi}}\left(\underline{f}\right)
  + \underline{\underline{\mathfrak{C}}}^*_{\underline{\underline{\psi}}}
  \circ \text{prox}_{\mu\mathscr{P}} \circ
  \underline{\underline{\mathfrak{C}}}_{\underline{\underline{\tilde\psi}}}\left(\underline{f}\right) 
  \notag
  \\ % ----
  &= n_1 \mathscr{H}^\dagger_{n_1 \mid n_2} \left( \mathscr{H}_{n_1 \mid n_2}\left(\underline{f}\right)
  \left( \underline{\Phi}
  + \underline{\tilde{\Psi}} \, \underline{\check{\Psi}^\text{T}} \right) \right) \,.
 \end{align}
 Scaling and wavelet filter bank matrices satisfy the unity condition:
 \begin{align}  \label{eq:RieszQuincunxUnityCondition}
  \underline{\Phi} + \underline{\tilde{\Psi}} \, \underline{\check{\Psi}^\text{T}} 
  = \frac{1}{n_1} \text{Id}_{n_1 n_2 \times n_1 n_2} \,.
\end{align}
 
\end{prop}
\begin{proof}
We provide a proof of Proposition~\ref{prop:RieszQuincunxWavelet:FrameletDecomposition} in Section~5.3 in SM.
\end{proof}
A smoothing version of a framelet decomposition (Equation~\ref{eq:RieszQuincunxWavelet}) is defined with proximity operators $\text{prox}_{\mu\mathscr{P}} \, (\text{for } \mu > 0)$ (as defined in Equation~13 in the supplemental material):
{\small
\begin{align*}  
 \underline{\tilde f} &= n_1 \mathscr{H}^\dagger_{n_1 \mid n_2}  
 \left( \mathscr{H}_{n_1 \mid n_2}\left(\underline{f}\right) \underline{\Phi}
 + \text{prox}_{\mu\mathscr{P}} \left\{ \mathscr{H}_{n_1 \mid n_2}\left(\underline{f}\right) \underline{\tilde{\Psi}} \right\}
 \underline{\check{\Psi}^\text{T}} \right) \,.
\end{align*}
}
Its iterative scheme, called generalized intersection algorithm with fixpoints is described in \cite{RichterThaiHuckemann2020}.

% ============

%
\subsection{RQUNet-VAE expansion} 
Given a multi-channel image $\underline{\underline{f}} = \left\{ \underline{f_1} \,, \hdots \,, \underline{f_{P}} \right\}
\in \mathbb R^{\abs{\Omega} \times P}$, we introduce our RQUNet-VAE expansion via mappings for skip-connecting signal, latent variables, and a reconstructed signal. Then, we propose RQUNet-VAE expansion and its functional space for regularization. 

\paragraph{Mappings:}
We introduce filter banks in an encoder
$ \underline{\underline{\underline{\theta^{1(i)}}}} \in \mathbb R^{d_1 \times d_2 \times P \times 2^iL}$
and 
$\underline{\underline{\underline{\theta^{2(i)}}}} 
 \in \mathbb R^{d_1 \times n_2 \times P \times 2^iL}$
whose matrix forms are 
$\underline{\underline{\Theta^{1(i)}}} 
\in \mathbb R^{d_1 d_2 P \times d_2 \times 2^iL} \,,
\underline{\underline{\Theta^{1(i)}}} 
\in \mathbb R^{d_1 n_2 P \times n_2 \times 2^iL}$ 
(as defined in Equation~6 in the supplemental material).
Similar for filter banks $\left\{ \left( \underline{\underline{\underline{\tilde{\theta}^{1(i)}}}} \,, \underline{\underline{\underline{\tilde{\theta}^{2(i)}}}} \right) \right\}_{i-0}^{I-1}$ in a decoder.
For batch-normalization and dropout layers, we refer the readers to \cite{IoffeSzegedy2015, SrivastavaHintonKrizhevskySutskeverSalakhutdinov2014}.

{\bfseries a.1. Skip-connecting signal (encoder):}
Given a local basis $\underline{\Xi^{(i)}}$
and an analysis operator at scale $i = 1 \,, \ldots \,, I$: 
\begin{align} \label{eq:UnetEncoderOperator:prop}
 \mathscr{T}^{(i)} &:= \mathscr{R}_p \circ \mathscr{B}
 \circ \text{prox}_\text{ReLU} \circ \underline{\underline{\underline{\mathfrak{C}}}}^\text{iso}_{\underline{\underline{\underline{\theta^{2(i)}}}}}
 \circ \text{prox}_\text{ReLU} \circ  \underline{\underline{\underline{\mathfrak{C}}}}^\text{iso}_{\underline{\underline{\underline{\theta^{1(i)}}}}} 
 \notag
 \\& %
 \,:\, \mathbb R^{n_1 \times n_2 \times P} \rightarrow
 \mathbb R^{2^{-i} n_1 \times 2^{-i} n_2 \times 2^i L} \,,
\end{align}
we define an iterated mapping:
\begin{align*}
 \mathscr{C}^{(i)} &= \mathscr{T}^{(i)} \circ \underline{\Xi^{(i-1),\text{T}}} \mathscr{T}^{(i-1)} 
 \circ \ldots \circ 
 \underline{\Xi^{(1),\text{T}}} \mathscr{T}^{(1)} 
 \notag
 \\
 &\,:\, \mathbb R^{n_1 \times n_2 \times P} \rightarrow \mathbb R^{2^{-(i-1)} n_1 \times 2^{-(i-1)} n_2 \times 2^{(i-1)} L} \,.
\end{align*}
Then, we have the following proposition describing a mapping for the skip-connecting signal:
\begin{prop}  \label{prop:SkipConnectingMapping}
 A mapping for skip-connecting signal in RQUnet-VAE is:
 \begin{align}  \label{eq:SkipConnectingMapping}
  &\mathscr{C}_{\gamma_c} \,:\, \mathbb R^{n_1 \times n_2 \times P} \rightarrow \left\{ \mathbb{R}^{2^{-(i-1)} n_1 \times 2^{-(i-1)} n_2 \times 2^{(i-1)} L} \right\}_{i=1}^I \,;
  \notag
  \\ % ---
  &\underline{\underline{\underline{c}}} = \mathscr{C}_{\gamma_c} \left(\underline{\underline{f}} \right) 
  % ---
  = \left\{ \underline{\underline{c^{(i)}}} = \mathscr{C}^{(i)} \left( \underline{\underline{f}} \right) \right\}_{i=0}^{I-1} 
  % ---
  = \underline{\mathfrak{C}}_{\underline{\phi}} \left( \underline{\underline{\underline{c}}} \right)  
  + \underline{\underline{\mathfrak{C}}}^*_{\underline{\underline{\psi}}}
  \underline{\underline{\mathfrak{C}}}_{\underline{\underline{\tilde\psi}}} \left( \underline{\underline{\underline{c}}} \right) 
 \end{align}
 which is equivalent to:
 \begin{align*}
  \underline{\underline{c^{(i)}}}
  &= \left\{ \underline{c^{(i)}_{l}} \right\}_{l=0}^{2^{i-1} L-1} \,,~
  \underline{c^{(i)}_{l}} = \underline{\mathfrak{C}}_{\underline{\phi}} \left( \underline{c^{(i)}_{l}} \right)  
  + \underline{\underline{\mathfrak{C}}}^*_{\underline{\underline{\psi}}}
  \underline{\underline{\mathfrak{C}}}_{\underline{\underline{\tilde\psi}}} \left( \underline{c^{(i)}_{l}} \right)
 \end{align*}
 for $i = 0 \,, \ldots \,, I-1$
 where unknown filter banks are 
 $\gamma_c := \left\{ \left( \underline{\underline{\underline{\theta^{1(i)}}}} \,, \underline{\underline{\underline{\theta^{2(i)}}}} \right) \right\}_{i=1}^I$.

\end{prop}
\begin{proof}
We provide a proof of Proposition~\ref{prop:SkipConnectingMapping} in Section~5.4 in SM.
\end{proof}
%

% ----------------------------

%
{\bfseries a.2. Variational term (encoder):}
We firstly define linear mappings as a perceptron network for latent variable: 
\begin{align}  \label{eq:MappingLatentVariable}
 \mathscr{F}^\mu(\cdot) &= \underline{W^\mu} \text{vec}(\cdot) + b^\mu \,,
 \notag
 \\
 \mathscr{F}^\sigma(\cdot) &= \underline{W^\sigma} \text{vec}(\cdot) + b^\sigma \,:\, \mathbb R^{2^{-i-1} n_1 \times 2^{-i-1} n_2 \times 2^i L} \rightarrow \mathbb R^d \,;
\end{align}
where 
$\underline{W^\mu} \,, \underline{W^\sigma} \in \mathbb R^{d \times 2^{-(I+1)} n_1 n_2 L}$ and 
$b^\mu \,, b^\sigma \in \mathbb R^{d}$
and $d$ is latent dimension and $\text{vec}(\cdot)$ is a vectorize operation. 
Denote Hadamard product $\odot$ and model's parameters as:
\begin{align*}
 \gamma_c \cup \gamma_s \cup \gamma_m
 := \left\{ \underline{W^\mu} \,, b^\mu \,, \underline{W^\sigma} \,, b^\sigma \,, \left\{ \left( \underline{\underline{\underline{\theta^{1(i)}}}} \,, \underline{\underline{\underline{\theta^{2(i)}}}} \right) \right\}_{i=1}^I
 \right\} \,;
\end{align*}
then, we have the following proposition for a variational term:
\begin{prop}  \label{prop:RQUnetVAE:VariationalTerm}
 A latent variable in RQUnet-VAE is sampled from a distribution:
 \begin{align}  \label{eq:latentscaling:VUAE:1}
  z &= \mathscr{M}_{\gamma_m} \left( \underline{\underline{f}} \right) 
  + \mathscr{S}_{\gamma_s}^\frac{1}{2} \left( \underline{\underline{f}} \right) \odot \epsilon
  \,,~ \epsilon \stackrel{\text{i.i.d.}}{\sim} \mathcal{N}_d \left( \mathbf{0}_d \,, \text{Id}_d \right) 
 \end{align} 
 where maps for the mean and variance for the latent variable are:
 \begin{align}  \label{eq:Variation:mean}
  \mathscr{M}_{\gamma_m} := \mathscr{F}^\mu \circ \mathscr{R}_p \circ \underline{\Xi^{(I),\text{T}}} \mathscr{T}^{(I)} \circ \ldots \circ \underline{\Xi^{(1),\text{T}}} \mathscr{T}^{(1)} \,,
  % --
  \\  \label{eq:Variation:variance}
  \mathscr{S}_{\gamma_s}^\frac{1}{2} := \mathscr{F}^\sigma \circ \mathscr{R}_p \circ \underline{\Xi^{(I),\text{T}}} \mathscr{T}^{(I)} \circ \ldots \circ \underline{\Xi^{(1),\text{T}}} \mathscr{T}^{(1)} \,.
 \end{align}
 
\end{prop}
\begin{proof}
We provide a proof of Proposition~\ref{prop:RQUnetVAE:VariationalTerm} in Section~5.5 in SM.
\end{proof}
%

% ===========

%
{\bfseries a.3. Decoder:}
Firstly, we define a mapping in the decoder: 
{\small
\begin{align*}
 \tilde{\mathscr{T}}^{(i)} &= \left(
 \mathcal{Z}_1^{-1} \,,
 \text{prox}_\text{ReLU} \circ \underline{\underline{\underline{\mathfrak{C}}}}^\text{iso}_{\underline{\underline{\underline{\tilde{\theta}^{1(i+1)}}}}} \circ \text{prox}_\text{ReLU} \circ \underline{\underline{\underline{\mathfrak{C}^\text{iso}}}}_{\underline{\underline{\underline{\tilde{\theta}^{2(i+1)}}}}} \circ \mathscr{P}_{\underline{\tilde{\Xi}^{(i)}_{\text{aug}}}} \right),
 \\  % ---
 \tilde{\mathscr{T}}^{(0)} &= \text{prox}_\text{ReLU} \circ \underline{\underline{\underline{\mathfrak{C}}}}^\text{iso}_{\underline{\underline{\underline{\tilde{\theta}^{1(1)}}}}} \circ \text{prox}_\text{ReLU} \circ \underline{\underline{\underline{\mathfrak{C}^\text{iso}}}}_{\underline{\underline{\underline{\tilde{\theta}^{2(1)}}}}} \circ \mathscr{P}_{\underline{\tilde{\Xi}^{(0)}_{\text{aug}}}}  \,,
\end{align*}
}
where $\mathcal{Z}_1^{-1}$ is scale-delayed 1 step back of the 1st argument and a concatenate layer is defined as an operation
$\mathscr{P}_{\underline{\tilde{\Xi}^{(i)}_{\text{aug}}}} \left( \underline{\underline{\widehat{c}^{(i)}_{\text{aug}}}} \right)
= \begin{pmatrix} \underline{\underline{c^{(i)}}} & \mathscr{B} \circ \underline{\tilde{\Xi}^{(i)}} \, \underline{\underline{\widehat{s}^{(i+1)}}} \end{pmatrix}$,
i.e. input signal at concatenation layer includes bypass signal $c^{(i)}$ and unpooled-batchnormed lowpass signal $\mathscr{B} \circ \underline{\tilde{\Xi}^{(i)}} \, \underline{\underline{\widehat{s}^{(i+1)}}}$.

Next, we define an iterated mapping for a skip-connecting signal $\underline{\underline{\underline{c}}}$ and a lowpass signal $\underline{\underline{s^{(I)}}}$ at scale $I$ as:
\begin{align*}
 \tilde{\mathscr{T}}_I \left( \underline{\underline{\underline{c}}} \,, \underline{\underline{s^{(I)}}} \right) = \tilde{\mathscr{T}}^{(0)} \circ \tilde{\mathscr{T}}^{(1)} \circ \ldots \circ \tilde{\mathscr{T}}^{(I-1)} \begin{pmatrix} \underline{\underline{c^{(I-1)}}} & \underline{\underline{s^{(I)}}} \end{pmatrix} \,.
\end{align*}
Note that $\underline{\underline{s^{(I)}}}$ is reconstructed from a latent variable $z$.
Denote unknown parameters in a decoder as:
\begin{align*}
 \alpha := \left\{ \underline{W^s} \,, b^s \,, \underline{\underline{\underline{\theta^0}}} \,, \left\{ \left( \underline{\underline{\underline{\tilde\theta^{1(i)}}}} \,, \underline{\underline{\underline{\tilde\theta^{2(i)}}}} \right) \right\}_{i=1}^I \right\} \,.
\end{align*}
Then, we have the following proposition:
\begin{prop}  \label{prop:RQUnetVAE:decoder}
 Given a skip-connecting signal $\underline{\underline{\underline{c}}}$ and a latent variable $y$ from an encoder, 
 a decoder mapping in RQUNet-VAE is defined as
 \begin{align}  \label{eq:DecoderMapping}
  \mathscr{D}_\alpha \left( \underline{\underline{\underline{c}}} \,, z \right) &:= \text{prox}_\text{ReLU} \circ \underline{\underline{\underline{\mathfrak C}}}^\text{iso}_{\underline{\underline{\underline{\theta^0}}}} 
  \circ \tilde{\mathscr{T}}_I \left( \underline{\underline{\underline{c}}} \,, \text{uvec} \left( \underline{W^s} \, z + b^s \right) \right) \,.
 \end{align}

\end{prop}
\begin{proof}
We provide a proof of Proposition~\ref{prop:RQUnetVAE:decoder} in Section~5.6 in SM.
\end{proof}
%
% ==========

% 
\paragraph{RQUNet-VAE expansion}
An auto-encoder is recast as a latent factor model (a decoder) whose latent variable is computed from the observed data as an encoder:
\begin{align}  
 \label{eq:Encoder:1}
 z &= \mathscr{M}_{\gamma_m} \left(\underline{\underline{f}} \right) + \mathscr{S}_{\gamma_s}^\frac{1}{2} \left( \underline{\underline{f}} \right) \odot \epsilon \,,~
 \epsilon \stackrel{\text{i.i.d.}}{\sim} \mathcal{N}_d( \mathbf{0}_d \,, \text{Id}_d) \,,~
 % ---
 \\ \label{eq:Encoder:4} 
 \underline{\underline{\underline{c}}} &= \mathscr{C}_{\gamma_c} \left(\underline{\underline{f}} \right) \,,
 \\ 
 \label{eq:Decoder:1}
 \underline{\underline{\tilde{f}}} &= \mathscr{D}_\alpha \left( \underline{\underline{\underline{c}}} \,, z \right) + \sigma \underline{\underline{e}} \,,~ \underline{\underline{e}} \stackrel{\text{i.i.d.}}{\sim} \mathcal{N}_{\abs{\Omega} \times P}(0 \,, \text{Id}) \,,~ 
 \\ \label{eq:Encoder:3}
 \underline{\underline{\tilde{f}}} &= \underline{\underline{f}} \in \textfrak{F} \,,
\end{align}
with a known standard deviation $\sigma > 0$ and the data set $\textfrak{F} = \left\{ \underline{\underline{f_i}} \right\}_{i=1}^{T} \subset \mathbb R^{\abs{\Omega} \times P}$.
Combining Equation~\ref{eq:Decoder:1},  Equation~\ref{eq:Encoder:1},  Equation~\ref{eq:Encoder:4}, and  Equation~\ref{eq:Encoder:3}, we obtain a variational auto-encoder:
\begin{align}  \label{eq:autoencodersummary}  
 \underline{\underline{f}} &= \mathscr{D}_\alpha \left( \mathscr{C}_{\gamma_c} \left(\underline{\underline{f}} \right) \,, \mathscr{M}_{\gamma_m} \left(\underline{\underline{f}} \right) + \mathscr{S}_{\gamma_s}^\frac{1}{2} \left( \underline{\underline{f}} \right) \odot \epsilon \right) 
 + \sigma \underline{\underline{e}} \,,
\end{align}
with standard normal random variables $\epsilon \stackrel{\text{i.i.d.}}{\sim} \mathcal{N}_d( \mathbf{0}_d \,, \text{Id}_d)$ and $\underline{\underline{e}} = \left[ e_{l,c} \right]_{l \in \Omega}^{c = 1, \ldots, p} \,, e_{l,c} \stackrel{\text{i.i.d.}}{\sim} \mathcal{N}(0 \,, 1)$.

The auto-encoder system (Equations \ref{eq:Decoder:1}-\ref{eq:Encoder:3}) is recast as Bayesian inference:
\begin{align}  \label{eq:Decoder:Bayes}
 \underline{\underline{\tilde{f}}} \mid z  \stackrel{\text{i.i.d.}}{\sim}& \mathcal{N}_{n_1 \times n_2 \times P} \left( \mathscr{D}_\alpha \left( \underline{\underline{\underline{c}}} \,, z \right) \,, \sigma^2 \text{Id} \right) 
 \notag
 \\
 &= \mathbb{H}_\alpha\left(\underline{\underline{\tilde{f}}} \mid z \right) = \mathtt{h}_\alpha\left(\underline{\underline{\tilde{f}}} \mid z\right) \text{d}\underline{\underline{\tilde{f}}} \,,~
 \\ \label{eq:Encoder:approxBayes}
 z \mid \underline{\underline{f}} \stackrel{\text{i.i.d.}}{\sim}& 
 \mathbb{K}_\alpha\left(z \mid \underline{\underline{f}}\right) = \mathtt{k}_\alpha\left(z \mid \underline{\underline{f}}\right) \text{d}y
 \notag
 \\ %
 &\approx \, \mathcal{N}_d \left( \mathscr{M}_{\gamma_m}\left(\underline{\underline{f}}\right) \,, \text{diag} \left\{ \mathscr{S}_{\gamma_s} \left( \underline{\underline{f}} \right) \right\} \right) \,.
\end{align}
An explanation for the above hierarchical model is: an observed signal $\underline{\underline{\tilde{f}}}$ is assumed to be sampled from a normal distribution whose latent variable $z$ is sampled from an unknown distribution $\mathbb{K}_\alpha$ which is approximated by a normal distribution (Equation~\ref{eq:Encoder:approxBayes})
parameterized by $(\gamma_m, \gamma_s)$.
The key idea for $z \mid \underline{\underline{f}} \sim \mathbb{K}_\alpha$ to depend on a parameter $\alpha$ is because of the auto-encoder (Equations \ref{eq:Encoder:1}-\ref{eq:Encoder:3}), i.e. $\underline{\underline{\tilde{f}}} = \underline{\underline{f}}$. To see this, Bayes' rule for a conditional density of random variable $z \mid \underline{\underline{f}}$ is:
\begin{align*}
 \mathtt{p} \left( z \mid \underline{\underline{f}} \right) &= \frac{\mathtt{p}\left(\underline{\underline{f}} \mid z \right) \mathtt{p}(z)}{\mathtt{p}\left(\underline{\underline{f}}\right)}   
 = \frac{\mathtt{p}\left(\underline{\underline{\tilde{f}}} \mid z\right) \mathtt{p}(z)}{\mathtt{p}\left(\underline{\underline{f}}\right)}
 \\ %
 &\stackrel{(\ref{eq:Decoder:Bayes})}{=}
 \frac{\mathtt{h}_\alpha\left(\underline{\underline{\tilde{f}}} \mid z\right) \mathtt{p}(z)}{\mathtt{p}\left(\underline{\underline{f}}\right)}
 := \mathtt{k}_\alpha\left(z \mid \underline{\underline{f}}\right) \,,
\end{align*}
which depends on a decoder's parameter $\alpha$.
Since the above density is intractible because of the incomputable integral in a marginal distribution $\mathtt{p}\left(\underline{\underline{f}}\right) = \int_{\mathbb R^d} \mathtt{h}_\alpha\left(\underline{\underline{\tilde{f}}} \mid z\right) \mathtt{p}(z) \text{d}z$, 
a distribution $\mathbb{K}_\alpha$ is approximated by a normal distribution:
\begin{align}  \label{eq:BayesRule:2}
 \mathtt{k}_\alpha\left(z \mid \underline{\underline{f}}\right) &= \frac{\mathtt{h}_\alpha\left(\underline{\underline{f}} \mid z\right) \mathtt{p}(z)}{\mathtt{p}\left(\underline{\underline{f}}\right)}
 \notag
 \\
 &\approx \mathcal{N}_d \left( z \,; \mathscr{M}_{\gamma_m}\left(\underline{\underline{f}}\right) \,, \text{diag} \left\{ \mathscr{S}_{\gamma_s} \left( \underline{\underline{f}} \right) \right\} \right) \,.
\end{align}
Denote variance vector as
\begin{align*}
 \mathscr{S}_{\gamma_s} \left( \underline{\underline{f}} \right) &= \mathscr{S}_{\gamma_s}^\frac{1}{2} \left( \underline{\underline{f}} \right)
 \odot \mathscr{S}_{\gamma_s}^\frac{1}{2} \left( \underline{\underline{f}} \right) 
 \notag
 \\
 &= \begin{pmatrix} \mathscr{S}_{\gamma_s} \left( \underline{\underline{f}} \right)_1 & \hdots & \mathscr{S}_{\gamma_s} \left( \underline{\underline{f}} \right)_d \end{pmatrix}^\text{T} \in \mathbb R^d \,.
\end{align*}
Choose a standard normal prior distribution $\mathbb{P}(z) = \mathtt{p}(z) \text{d}z = \mathcal{N}_d \left( \mathbf{0}_d \,, \text{Id}_d \right)$,
we have the following proposition for finding model's parameters: 
\begin{prop}  \label{prop:RQUnetVAEparameters}
 Unknown parameters in RQUnet-VAE are obtained from the following minimization problem:
 \begin{align}  \label{eq:RQUnetVAELossFunc:prop}
  \left( \gamma_c^\dagger \,, \gamma_m^\dagger \,, \gamma_s^\dagger \,, \alpha^\dagger \right) 
  &= \argmin \mathscr{L} \left( \gamma_c \,, \gamma_m, \gamma_s \,, \alpha \right)
 \end{align}
 where:
 {\small
 \begin{align*}
  &\mathscr{L} (\cdot) := \frac{1}{2 \sigma^2} \sum_{i=1}^{T} \mathbb E_{\epsilon \sim \mathcal{N}(0 \,, \text{Id})} \bigg[
  \\&
  \norm{\underline{\underline{f_i}} 
  - \mathscr{D}_\alpha \left( \mathscr{C}_{\gamma_c} \left(\underline{\underline{f_i}} \right) \,, \mathscr{M}_{\gamma_m} \left( \underline{\underline{f_i}} \right) + \mathscr{S}_{\gamma_s}^\frac{1}{2} \left( \underline{\underline{f_i}} \right) \odot \epsilon \right)}^2_{\ell_2} \bigg]
  \\ %
  &+ \sum_{i=1}^{T} \text{KL} \left[ \mathcal{N}_d \left( \mathscr{M}_{\gamma_m}\left(\underline{\underline{f}}\right) \,, \text{diag} \left\{ \mathscr{S}_{\gamma_s} \left( \underline{\underline{f}} \right) \right\} \right) \mid \mid \mathcal{N}_d \left( \mathbf{0}_d \,, \text{Id}_d \right) \right]
 \end{align*}
 }
 and KL-divergence is defined as:
 {\small
 \begin{align}  \label{eq:RQUnetVAE:KLloss2}
 &\text{KL} \left[ \cdot \mid \mid \cdot \right]
 \notag
 \\ %
 &= \frac{1}{2} \left( 
 \norm{\mathscr{M}_{\gamma_m}\left(\underline{\underline{f}}\right)}^2_{\ell_2} - d 
 + \sum_{i=1}^d \left( \mathscr{S}_{\gamma_s} \left( \underline{\underline{f}} \right)_i
 - \log \mathscr{S}_{\gamma_s} \left( \underline{\underline{f}} \right)_i \right)
 \right) \,.
\end{align}
}
\end{prop}
\begin{proof}
We provide a proof of Proposition~\ref{prop:RQUnetVAEparameters} in Section~5.7 in SM.
\end{proof}
It is clear that the only one random variable in the above minimization is $\epsilon$, so the gradient descent method can be applied for model parameters $\left( \gamma_c \,, \gamma_m, \gamma_s \,, \alpha \right)$.
Since high order Riesz-Quincunx wavelet expansion is an identity operator, for training these model parameters, we remove this layer in the training procedure. But, later we use it in its smoothing version with the trained parameters $\left( \gamma_c^\dagger \,, \gamma_m^\dagger \,, \gamma_s^\dagger \,, \alpha^\dagger \right)$ to truncate small wavelet coefficients of signals in skip-connection. A summary of RQUNet-VAE with a training procedure (without Riesz-Quincunx wavelet) is shown in the SM for encoder and decoder, respectively.

% --------------------

% --------------------

%
\paragraph{Generalized Besov space by RQUNet-VAE and proximal operators}
We note that Equation~\ref{eq:RQUnetVAELossFunc:prop} is a non-convex minimization problem with potentially many local minimas. Assume an existence of minima $\left( \alpha^\dagger \,, \gamma_c^\dagger \,, \gamma_m^\dagger \,, \gamma_s^\dagger \right)$ such that a condition of perfect reconstruction occurs, i.e. $\underline{\underline{\tilde f}} = \underline{\underline{f}}$; then, the auto-encoder (Equations \ref{eq:Decoder:1}-\ref{eq:Encoder:3}) plays as a generalized wavelet expansion with learnable parameters:
\begin{align*}  
 &\lim_{\sigma \rightarrow 0}
 \norm{ \underline{\underline{f}} 
 - \mathscr{D}_{\alpha^\dagger} \left( \mathscr{C}_{\gamma_c^\dagger} \left(\underline{\underline{f}} \right) \,, \mathscr{M}_{\gamma_m^\dagger} \left(\underline{\underline{f}} \right) + \mathscr{S}_{\gamma_s^\dagger}^\frac{1}{2} \left( \underline{\underline{f}} \right) \odot \epsilon \right) 
 + \sigma \underline{\underline{e}} }^2_{\ell_2} 
 \\
 &= 0 \,,
\end{align*}
having its deterministic version:
\begin{align*}
 \underline{\underline{f}} = \mathscr{D}_{\alpha^\dagger} \left( \mathscr{C}_{\gamma_c^\dagger} \left(\underline{\underline{f}} \right) \,, \mathscr{M}_{\gamma_m^\dagger} \left(\underline{\underline{f}} \right) + \mathscr{S}_{\gamma_s^\dagger}^\frac{1}{2} \left( \underline{\underline{f}} \right) \odot \epsilon \right) \,,
\end{align*}
where encoder and decoder are forward and backward wavelet transform. Then, it induces a generalized Besov space:
\begin{align}  \label{eq:GeneralizedBesovSpace}
 \textfrak{B} := \left\{ \underline{\underline{f}} \in \mathbb R^{\abs{\Omega} \times P} \,:\, \norm{\underline{\underline{f}}}_{\textfrak{B}} := \mathscr{P} \circ \mathscr{C}_{\gamma_c} \left(\underline{\underline{f}} \right) < \infty \right\} \,,
\end{align}
where function $\mathscr{P}(\cdot)$ acts on the skip-connecting signal $\underline{\underline{\underline{c}}} = \mathscr{C}_{\gamma_c} \left(\underline{\underline{f}} \right)$ for  proximal mapping as described in Equation~13 in the supplemental material.

Then, given an image $\underline{\underline{f}} \in \mathbb R^{\abs{\Omega} \times P}$ with an expansion in a space $\textfrak{B}$ (\ref{eq:GeneralizedBesovSpace}),
we have a regularization in that space $\textfrak{B}$ with the optimally trained parameters as generalized wavelet smoothing via a proximal operator:
\begin{align}  \label{eq:RieszQuincunxWaveletSmoothing}
 \underline{\underline{\underline{\tilde{c}}}} &= \argmin_{\underline{\underline{\underline{w}}}} \left\{ \mathscr{P} \left( \underline{\underline{\underline{w}}} \right) + \frac{1}{2\mu} \norm{ \mathscr{C}_{\gamma_c} \left(\underline{\underline{f}} \right) 
 - \underline{\underline{\underline{w}}}
 }^2_\text{F} \right\}
 \notag
 \\  % ---
 &= \text{prox}_{\mu \mathscr{P}} \circ \mathscr{C}_{\gamma_c} \left(\underline{\underline{f}} \right)  
 \notag
 \\  % ---
 &= \underline{\mathfrak{C}}_{\underline{\phi}} \circ \mathscr{C}_{\gamma_c} \left(\underline{\underline{f}} \right)  
 + \underline{\underline{\mathfrak{C}}}^*_{\underline{\underline{\psi}}}
 \circ \text{prox}_{\mu\mathscr{P}} \circ
 \underline{\underline{\mathfrak{C}}}_{\underline{\underline{\tilde\psi}}} \circ \mathscr{C}_{\gamma_c} \left(\underline{\underline{f}} \right) \,,
\end{align}
where its element form is:
\begin{align*}
 \underline{\underline{\underline{\tilde c}}} 
 &= \left\{ \underline{\underline{\tilde c^{(i)}}} \right\}_{i=0}^{I-1} \,,~
 \underline{\underline{\tilde c^{(i)}}} = \left\{ \underline{\tilde c^{(i)}_{l}} \right\}_{l=0}^{2^{i-1} L-1} \,,
 \\  % ----
 \underline{\tilde c^{(i)}_{l}} &= \underline{\mathfrak{C}}_{\underline{\phi}} \left( \underline{c^{(i)}_{l}} \right)  
 + \underline{\underline{\mathfrak{C}}}^*_{\underline{\underline{\psi}}}
 \circ \text{prox}_{\mu\mathscr{P}} \circ 
 \underline{\underline{\mathfrak{C}}}_{\underline{\underline{\tilde\psi}}} \left( \underline{c^{(i)}_{l}} \right)
\end{align*}
for $l = 0 \,, \ldots \,, 2^{i-1} L-1 \,,~ i = 0 \,, \ldots \,, I-1$.
Note that $\text{prox}_{\mu\mathscr{P}}(\cdot)$ in Equation~\ref{eq:RieszQuincunxWaveletSmoothing} acts on wavelet coefficients only. A reason is scaling coefficients contain main energy of a signal whose values are large. These scaling coefficients should be preserved during a shrinking process while wavelet coefficients mainly contain oscillating signals in different scales, including noise which should be removed.

Then, we have a smoothed image as:
\begin{align*}
 \underline{\underline{\tilde f}} = \mathscr{D}_{\alpha^\dagger} \left( \underline{\underline{\underline{\tilde{c}}}} \,,~ \mathscr{M}_{\gamma_m} \left(\underline{\underline{f}} \right) + \mathscr{S}_{\gamma_s}^\frac{1}{2} \left( \underline{\underline{f}} \right) \odot \epsilon \right) %
 \in \mathbb R^{\abs{\Omega} \times P} \,.
\end{align*}
In summary, we have the following proposition:
\begin{prop}  \label{prop:RQUnetVAEsmoothing}
Given a learned parameter $\left( \alpha^\dagger \,, \gamma_c^\dagger \,, \gamma_m^\dagger \,, \gamma_s^\dagger \right)$ obtain by training RQUNet-VAE with a dataset, we have RQUNet-VAE smoothing (with a parameter $\mu > 0$) for an image $\underline{\underline{\tilde f}}$ as a solution of a regularization in a generalized Besov space $\textfrak{B}$ (\ref{eq:GeneralizedBesovSpace}):
\begin{align}  \label{eq:RQUnetVAEsmoothing}
 \underline{\underline{\tilde f}} = \mathscr{D}_{\alpha^\dagger} \bigg(& \underline{\mathfrak{C}}_{\underline{\phi}}
 \circ \mathscr{C}_{\gamma_c^\dagger} \left(\underline{\underline{f}} \right)
 + \underline{\underline{\mathfrak{C}}}^*_{\underline{\underline{\psi}}}
 \circ \text{prox}_{\mu\mathscr{P}} \circ 
 \underline{\underline{\mathfrak{C}}}_{\underline{\underline{\tilde\psi}}} 
 \circ \mathscr{C}_{\gamma_c^\dagger} \left(\underline{\underline{f}} \right)
 \notag
 \\ %
 &\,, \mathscr{M}_{\gamma_m^\dagger} \left(\underline{\underline{f}} \right) + \mathscr{S}_{\gamma_s^\dagger}^\frac{1}{2} \left( \underline{\underline{f}} \right) \odot \epsilon \bigg) \,.
\end{align}

\end{prop}
\begin{proof}
A proof of Proposition~\ref{prop:RQUnetVAEsmoothing} is directly obtained from the above paragraphs.
\end{proof}
%

% =======

%
\subsection{RQUNet-VAE iterative shrinkage Lagrangian system}\label{theory:shrinkage}
\paragraph{Decomposition for multi-band image}
\RQUnetVAE~ iterative shrinkage algorithm for multi-band image $\underline{\underline{f}} \in \mathbb R^{\abs{\Omega} \times P}$ is described in Algorithm~1 in the SM.
This is equivalent to a nonlinear mapping with unknown model's parameters $\Gamma$:
$\underline{\underline{u}}^{(0)} = \underline{\underline{f}} \,, \underline{\underline{\lambda}}^{(0)} = \mathbf{0}$ for 
$\tau = 1 \,, \ldots \,, N$ and
\begin{align}  
 \left( \underline{\underline{u}}^{(\tau)} \,, \underline{\underline{\lambda}}^{(\tau+1)} \right) &= \mathscr{K}_\mu\left( \underline{\underline{f}} \,, \underline{\underline{u}}^{(\tau-1)} \,, \underline{\underline{\lambda}}^{(\tau)} \,; \Gamma \right) \,.
\end{align}

\paragraph{Diffusion process and spectral decomposition for multiband image}
For a multiband image $\underline{\underline{f}} \in \mathbb R^{\abs{\Omega} \times P}$, a diffusion process by a Lagrangian system (\ref{eq:RQUnetVAELagrangian:regularization}) is:
$\underline{\underline{u}}^{(0)} = \underline{\underline{f}} \,, \tau = 1 \,, \ldots \,, N$,
\begin{align}  \label{eq:RQUnetVAELagrangian:regularization}
 \left( \underline{\underline{u}}^{(\tau)} \,, \underline{\underline{\lambda}}^{(\tau+1)} \right) &= \mathscr{K}_\mu\left( \underline{\underline{u}}^{(\tau-1)} \,, \underline{\underline{u}}^{(\tau-1)} \,, \underline{\underline{\lambda}}^{(\tau)} \,; \Gamma \right) \,.
\end{align}
%
%Then, a spectral image decomposition~\cite{Gilboa2014} is: 
%
Given $\left\{ \underline{\underline{u}}^{(\tau)} \right\}_{\tau = 1}^N$ generated by  diffusion process (\ref{eq:RQUnetVAELagrangian:regularization}) in Algorithm 2, we have a discrete TV-like transform as in~\cite{Gilboa2014}: 
\begin{align*}
 \underline{\underline{\phi}}^{(\tau)} := \frac{\tau}{\beta} \left( \underline{\underline{u}}^{(\tau+1)} - 2 \underline{\underline{u}}^{(\tau)} + \underline{\underline{u}}^{(\tau-1)} \right) \,,
\end{align*}
whose inverse transform is: $\underline{\underline{\tilde{f}}}^{(N)} := \left( 1 + N \right) \underline{\underline{u}}^{(N)} - N \underline{\underline{u}}^{(N + 1)}$,
\begin{align*}
 \underline{\underline{f}} = \underline{\underline{\tilde{f}}}^{(N)} + \beta \sum_{\tau = 1}^N \underline{\underline{\phi}}^{(\tau)} \,.
\end{align*}
Its filtered version 
\begin{align*}
 \underline{\underline{f_{H_N}}} = H^{(N)} \underline{\underline{\tilde{f}}}^{(N)} + \beta \sum_{\tau = 1}^N H^{(\tau)} \underline{\underline{\phi}}^{(\tau)} 
\end{align*}    
is defined, e.g. by the ideal spectral-filter $H^{(\tau)}$:
\begin{align*}
 H_\text{lowpass}^{(\tau)} &= \begin{cases} 0 \,, & \tau \in \{ 0 \,, \ldots \,, \tau_1 \} \,, \\ 1 \,, & \tau \in \{ \tau_1 \,, \ldots \,, N \} \end{cases} \,,
 \\
 H_\text{highpass}^{(\tau)} &= \begin{cases} 1 \,, & \tau \in \{0 \,, \ldots \,, t_1\} \,, \\ 0 \,, & \tau \in \{\tau_1 \,, \ldots \,, N \} \end{cases} \,,~
 \\
 H_\text{bandpass}^{(\tau)} &= \begin{cases} 0 \,, & \tau \in \{0 \,, \ldots \,, \tau_1\} \,, \\ 1 \,, & \tau \in \{\tau_1 \,, \ldots \,,\tau_2\} \,, \\ 0 \,, & \tau \in \{\tau_2 \,, \ldots \,, N\} \end{cases} \,,
 \\
 H_\text{bandstop}^{(\tau)} &= \begin{cases} 1 \,, & \tau \in \{0 \,, \ldots \,, \tau_1\} \,, \\ 0 \,, & \tau \in \{\tau_1 \,, \ldots \,, \tau_2\} \,, \\ 1 \,, & \tau \in \{\tau_2 \,, \ldots \,, N\} \,, \end{cases} 
\end{align*}
where the time threshold $\tau_i$ are selected by the TV-spectrum $S^{(\tau)} := \norm{ \underline{\underline{\phi}}^{(\tau)} }_{\ell_1}$.

% ======

%
\paragraph{Decomposition for multiband time series}\label{sec:decomposition_time_series}
Next, we show how our \RQUnetVAE can be applied to satellite image time series, that is, a sequence of images of the same area.

Given a multi-band video 
$\underline{\underline{\underline{f}}}
= \left\{ \underline{\underline{f_1}} \,, \hdots \,, \underline{\underline{f_{T}}} \right\}
\in \mathbb R^{\abs{\Omega} \times P \times T}$,
$\underline{\underline{f_t}} = \left\{ \underline{f_{t,1}} \,, \hdots \,, \underline{f_{t,P}} 
 \right\}$,
$\underline{f_{t,p}} \in \mathbb R^{\abs{\Omega}}$
and scalling and wavelet bases, e.g. 1D Haar bases:
\begin{align}  \label{eq:1DHaarBases}
 \phi_{I,t,m} &= \begin{cases} 1 \,, & t \geq 2^{I} m ~\&~  t < 2^{I} (1 + m)
 \\ 0 \,, & \text{else} \end{cases} \,,
 \notag
 \\ % ----
 \xi_{i,t,m} &= \begin{cases} 1 \,,& t \geq 2^{i} m ~\&~ t < 2^{i} \left( \frac{1}{2} + m \right)
 \\ -1 \,,& t  \geq 2^{i} \left( \frac{1}{2} + m \right) ~\&~ t < 2^{i} \left( 1 + m \right)
 \\ 0 \,,& \text{else} \end{cases} \,,
\end{align}
a time-wavelet smoothing expansion for a video $\underline{\underline{\underline{f}}} \in \mathbb R^{\abs{\Omega} \times P \times T}$ by proximal operator $\mu \mathscr{P}$ is defined as: $t = 1 \,, \ldots \,, T$,
\begin{align*}  
 \underline{\underline{\tilde f_t}} &= \frac{1}{2^\frac{I}{2}} \sum_{s = 0}^{2^{-I} T - 1}  \frac{1}{2^\frac{I}{2}} \sum_{t'=1}^{T} \underline{\underline{f_{t'}}} \, \phi_{I,t',s} \phi_{I,t,s}
 \\ %
 &+ \sum_{i=1}^I \frac{1}{2^\frac{i}{2}} \sum_{s=0}^{2^{-i} T - 1} \text{prox}_{\mu\mathscr{P}} \bigg\{ \underbrace{ \frac{1}{2^\frac{i}{2}} \sum_{t'=1}^{T} \underline{\underline{f_{t'}}} \, \xi_{i,t',s} }_{ := \underline{\underline{w_{i,m}}} } \bigg\} \xi_{i,t,s} \,,
\end{align*}
which is equivalent to:
\begin{align}  \label{eq:TimeWaveletDecomposition1}
 \underline{\underline{\underline{\tilde f}}}
 &:= \mathscr{G}_\phi \left( \underline{\underline{\underline{f}}} \right)
 + \mathscr{W}^*_{\xi} \circ \text{prox}_{\mu\mathscr{P}} \circ \mathscr{W}_{\xi} \left( \underline{\underline{\underline{f}}} \right) 
 % ----
 = \left\{ \underline{\underline{\tilde f_1}} \,, \hdots \,, \underline{\underline{\tilde f_{T}}} \right\} \,.
\end{align}
Wavelet coefficient in (\ref{eq:TimeWaveletDecomposition1}) is
$\underline{\underline{\underline{w_i}}}
= \left\{ \underline{\underline{w_{i,1}}} \,, \hdots \,, \underline{\underline{w_{i,T}}} \right\}
:= \mathscr{W}_{\xi} \left\{ \underline{\underline{\underline{f}}} \right\}
\in \mathbb R^{\abs{\Omega} \times P \times 2^{-i} T} \,,~
\underline{\underline{w_{i,m}}} = \left\{ \underline{w_{i,m,1}} \,, \hdots \,, \underline{w_{i,m,P}} 
\right\}$
and $\underline{w_{i,m,p}} \in \mathbb R^{\abs{\Omega}}$.
Combined with Equation~(\ref{eq:RQUnetVAELagrangian:regularization}), RQUnet-VAE's iterative shrinkage algorithm for multi-band video $\underline{\underline{\underline{f}}} \in \mathbb R^{n_1 \times n_2 \times P \times T}$ is described in Algorithm~2 in SM, called \RQUnetVAE scheme 2.

% ---------------

%
\subsection{RQUNet-VAE based segmentation}  \label{sec:DenoisingSegmentation:theory}
In this section, we propose a mathematical framework for a segmentation problem with our RQUNet-VAE which serves as a smoothing term.
Given a dataset of original images (without artificially added noise) with ground-truth masks, we first train the UNet-VAE to obtain the optimal model parameters. Next we predict a segmented image from a noisy image using the RQUNet-VAE with a smoothing term as a truncation scheme of the $N$-th order Riesz-Quincunx wavelet expansion on the skip-connecting signals. 
If the training data is $\left( \textfrak{F} \,, \textfrak{F}^\text{gt} \right) := \left\{ \underline{\underline{f_i}} \,, \underline{\underline{f_i^\text{gt}}} \right\}_{i=1}^n \subset \mathbb R^{\abs{\Omega} \times P} \times \mathbb R^{\abs{\Omega} \times K}$ and $K$-dimensional hot key tensors of the ground-truth masks, i.e. pixel intensity and its allocation are: 
\begin{align*}
 \underline{\underline{f_i}} &:= \left[ f_{i,l} \right]_{l \in \Omega} \,,~
 f_{i,l} \in \mathbb R^P \,,~ 
 % ----
 \underline{\underline{f_i^\text{gt}}} := \left[ f_{i,l}^\text{gt} \right]_{l \in \Omega} \,,
\end{align*}
where $\underline{f^\text{gt}_{i,l}} = \left[ f^\text{gt}_{i,l,k} \right]_{k=1}^K \in \mathbb R^K$ is a one-hot-key vector.
Given the loss function: $\underline{\underline{x}} \,, \underline{\underline{y}} \in \mathbb R^{\abs{\Omega} \times K}$,
\begin{align*}
 \mathscr{H} \left( \underline{\underline{x}} \,, \underline{\underline{y}} \right)
 &= \sum_{l \in \Omega} \sum_{k=1}^K x_{l,k} 
 \log \frac{ \exp \left( y_{l,k} \right) }{ \sum_{h=1}^K \exp \left( y_{l,h} \right) } \,;
\end{align*}
similar to proposition \ref{prop:RQUnetVAEparameters}, a loss function of our RQUNet-VAE based segmentation problem is as follows:
\begin{prop}  \label{prop:RQUnetVAEsegmentation}
 Unknown parameters in RQUNet-VAE based segmentation are obtained from the following minimization problem:
 \begin{align}  \label{eq:RQUnetVAELossFunc:segmentation}
  &\left( \underline{\underline{\underline{\theta}}}^\dagger \,, \gamma_c^\dagger \,, \gamma_m^\dagger \,, \gamma_s^\dagger \,, \alpha^\dagger \right)
  = \argmin \sum_{i=1}^{n} \mathscr{L} \left( \underline{\underline{\underline{\theta}}} \,, \gamma_c \,, \gamma_m \,, \gamma_s \,, \alpha \right)
 \end{align}
 where:
 {\scriptsize
 \begin{align*}
  &\mathscr{L} \left( \cdot \right)
  := \text{KL} \left[ \mathcal{N}_d \left( \mathscr{M}_{\gamma_m}\left(\underline{\underline{f_i}}\right) \,, \text{diag} \left\{ \mathscr{S}_{\gamma_s} \left( \underline{\underline{f_i}} \right) \right\} \right) \mid \mid \mathcal{N}_d \left( \mathbf{0}_d \,, \text{Id}_d \right) \right] 
  \\&
  - \frac{1}{2 n \sigma^2} \mathbb E_{\epsilon \sim \mathcal{N}(0 \,, \text{Id})} \Big[ 
  \mathscr{H} \Big( \underline{\underline{f^\text{gt}_i}} \,,~ \underline{\underline{\underline{\mathfrak{C}^\text{iso}}}}_{\underline{\underline{\underline{\theta}}}}
  \circ 
  \mathscr{D}_\alpha \Big( \mathscr{C}_{\gamma_c} \left(\underline{\underline{f_i}} \right) \,, \mathscr{M}_{\gamma_m} \left( \underline{\underline{f_i}} \right) 
  \\& % ---
  + \mathscr{S}_{\gamma_s} \left( \underline{\underline{f_i}} \right)^\frac{1}{2} \odot \epsilon \Big) \Big) \Big] \,,
 \end{align*}
 }
 and the K-L divergence is defined in (\ref{eq:RQUnetVAE:KLloss2}). 
  
\end{prop}
\begin{proof}
We provide a proof of Proposition~\ref{prop:RQUnetVAEsegmentation} in Section~5.8 in SM.
\end{proof}
After training the model parameters by minimizing the loss function (\ref{eq:RQUnetVAELossFunc:segmentation}), the segmentation of a new noisy image is predicted:
\begin{align*}
 f^\text{new} = f_0^\text{new} + \sigma \epsilon \,,
\end{align*}
with standard Gaussian noise $\epsilon = \left[ \epsilon_{l,c} \right]^{c=1,\ldots,p}_{l \in \Omega} \in \mathbb R^{\abs{\Omega} \times P}$ 
and $\epsilon_{l,c} \sim \mathcal{N}(0 \,, 1)$. 
From proposition \ref{prop:RQUnetVAEsmoothing} by adding $N$-th order Riesz Quincunx wavelet truncation with proximal operator $\text{prox}_{\mu\mathscr{P}}(\cdot)$ parameterized by a smoothing parameter $\mu$ as:
{\small
\begin{align*}  
 &u \left( \underline{\underline{f^\text{new}}} \right) 
 = \left\{ u_1 \left( \underline{\underline{f^\text{new}}} \right) \,, \hdots \,, 
 u_K \left( \underline{\underline{f^\text{new}}} \right) \right\}
 \in \mathbb R^{\abs{\Omega} \times K}
 \notag
 \\  % ----
 &= \underline{\underline{\underline{\mathfrak{C}^\text{iso}}}}_{\underline{\underline{\underline{\theta}}}}
 \circ 
 \mathscr{D}_{\alpha^\dagger} \bigg( \underline{\mathfrak{C}}_{\underline{\phi}}
 \circ \mathscr{C}_{\gamma_c^\dagger} \left(\underline{\underline{f^\text{new}}} \right)
 + \underline{\underline{\mathfrak{C}}}^*_{\underline{\underline{\psi}}}
 \circ \text{prox}_{\mu\mathscr{P}} \circ 
 \underline{\underline{\mathfrak{C}}}_{\underline{\underline{\tilde\psi}}} 
 \circ \mathscr{C}_{\gamma_c^\dagger} \left(\underline{\underline{f^\text{new}}} \right)
 \,,~ 
 \\ % ----
 &\qquad \qquad \qquad \quad
 \mathscr{M}_{\gamma_m^\dagger} \left(\underline{\underline{f_i^\text{new}}} \right) + \mathscr{S}_{\gamma_s^\dagger}^\frac{1}{2} \left( \underline{\underline{f_i^\text{new}}} \right) \odot \epsilon \bigg) \,,
 \\  % ---
 &\underline{\underline{y^\text{new}}} = \text{argmax}_{k = 1 \,, \ldots \,, K} u_k \left( \underline{\underline{f^\text{new}}} \right)
 \in \mathbb R^{\abs{\Omega} \times K} \,,~ u_k \left( \underline{\underline{f^\text{new}}} \right) \in \mathbb R^{\abs{\Omega}} \,.
\end{align*}
}
%
%For experimental result, see Section 
%\ref{sec:DenoisingSegmentation}.

% ===================================
% \newpage

%
\section{Experimental Results}  \label{sec:ExperimentalResults}

%
%\subsection{Image Denoising by RQUNet-VAE schemes 1 and 2}  \label{sec:ImageDenoising}
%

%\textcolor{red}{Duy, Move the following sentence to introduction or technical part: - where does it belong?}
%Note that scheme 1 of RQUNet-VAE is a generalized wavelet smoothing technique and scheme 2 is based on GIAF \cite{RichterThaiHuckemann2020}.
% Why is this stateed here and not in introduction or in the methods section? Agreed

\subsection{Experimental Setup}
This section describes the dataset used for our experiments, introduces competitor solutions, and defines evaluation metrics to measure the ability of an algorithm to reduce the noise of an image.

\subsubsection{Datasets}\label{subsec:datasets}

We used Sentinel-2 satellite images (S30:MSI harmonized, V1.5) processed as part of the Harmonized Landsat Sentinel-2 (HLS) dataset obtained from USGS DAAC (https://lpdaac.usgs.gov/data/get-started-data/collection-overview/missions/harmonized-landsat-sentinel-2-hls-overview/) \cite{claverie2018harmonized}. The S30:MSI harmonized surface reflectance product is resampled from the original 10m to 30m resolution and adjusted to Landsat8 spectral response function in order to ultimately create a harmonized time series with a 2-3 day revisit.  Radiometric and geometric corrections are applied to convert data to surface reflectance, adjust for BRDF differences, and spectral bandpass differences. The study was focused on the HLS tile 18STH covering a large part of Northern Virginia in the US, for 2016-2020. From the main HLS tile's time series, 510 images were randomly created with a size of $256 \times 256$ pixels and including only the three visible bands (4R, 3G, 2B). Quality Assessment (QA) layers were used to exclude images  with more than 30 percent cloud shadow, adjacent cloud, cloud and cirrus clouds. The 500 images were used for training of our \RQUnetVAE  and competitor approaches, while ten images were used for testing in the denoising and decomposition experiments. The ten test images are visualized in SM.

For the time series decomposition experiments the above Sentinel-2 data were used to create 80, randomly located time series of length 99 images and image size $40 \times 40$ pixels. Images were normalized from 0 to 1 using image minimum and maximum. The RQUNet-VAE was trained on each image in the dataset to obtain all unknown model parameters.  All the Sentinel-2 datasets are available on Github~\footnote{\url{https://github.com/trile83/RQUNetVAE}}.

For our image segmentation experiments in Section~\ref{sec:DenoisingSegmentation} we used National Agriculture Imagery Program (NAIP) images coinciding with high resolution ground truth land cover data of each pixel in an image \cite{national2012national}.  We acquired NAIP Nature Color imagery of northern Virginia consisting of RGB bands which is similar to some commercial satellite imagery (e.g. BlackSky, Planet). For training and validation (ground truth) we used the 1m resolution land cover dataset for the Chesapeake Bay watershed \cite{allenby2013implementing}. The dataset contains six classes - Water, Tree and shrubs, Herbaceous, Barren, Impervious (roads) and Impervious (other). We combined Impervious (other) and Impervious (roads) together into an integrated Impervious class. This provided us with four classes Water, Tree and Shrubs, Grass, and Impervious. For the segmentation experiments these classes were grouped into three classes, Vegetated (tree, grass, shrubs), Water and Impervious.

\subsubsection{Evaluated Algorithms}\label{subsec:competitors}
Since \RQUnetVAE Scheme 1 is based on harmonic analysis, we compared it to state-of-the-art approaches including wavelet CDF 9/7~\cite{UnserBlu2003, Sweldens1997, daubechies1998factoring}, curvelet~\cite{CandesDonoho2004} and Riesz Dyadic wavelet kernel~\cite{RichterThaiHuckemann2020}, while the iterative Scheme 2 method was compared to state-of-the-art iterative methods such as directional TV-L2~\cite{RudinOsherFatemi1992, GoldsteinOsher2009, ThaiGottschlich2016DG3PD} and GIAF \cite{RichterThaiHuckemann2020}.

% To ensure a fair comparison, all parameters of these competitors are obtained by heuristic search to minimize a mean square error between denoised images and the  original. 

%\textcolor{red}{do we need to include "competitor" solutions - its just established "baseline" methods - not sure. Some of these are from 2020 papers. I would think that these are quite state-of-the-art, not just (naive) baselines. The term baseline is usually used for something very naive/simple. Competitor is usually used for (not naive) state-of-the-art solutions that solve the same problem. - Duy?}

Note that since the competitor methods were designed for gray-scale images, the methods were applied independently to each band of every multi-band image. The RQUNet-VAE, on the other hand, was designed for multi-channels images.

\subsubsection{Evaluation Metrics}\label{subsec:evaluationmetrics}
\paragraph{Image reconstruction}
To evaluate the ability of an algorithm to reduce the noise of an image, this section defines two commonly used measures. Given a clean multi-channel image $\underline{\underline{f}} \in \mathbb R^{\abs{\Omega} \times P}$ as ground-truth and a denoised image $\underline{\underline{f}}^\dagger \in \mathbb R^{\abs{\Omega} \times P}$, we use 
\begin{itemize}
 \item peak-signal-to-noise-ratio (PSNR):
  \begin{align*}
   \text{PSNR} = 10 \log \frac{\max \left( \underline{\underline{f}} \right)}{\text{MSE}} \,,
   \text{MSE} = \frac{1}{n_1 n_2 P} \norm{ \underline{\underline{f}} - \underline{\underline{f}}^\dagger }^2_{\ell_2}
  \end{align*} 

 \item Structural similarity index (SSIM) \cite{WangBovik2009}:
  \begin{align*}
   \text{SSIM} = \frac{ \left( 2 \mu_f \mu_{f^\dagger} + c_1 \right) \left( 2 \sigma_{f f^\dagger} \right) }{ \left( \mu_f^2 + \mu_{f^\dagger}^2 + c_1 \right) \left( \sigma_f^2 + \sigma_{f^\dagger}^2 + c_2 \right) }
  \end{align*}
  where $(\mu_f \,, \sigma_f^2)$ and $(\mu_{f^\dagger} \,, \sigma_{f^\dagger}^{2})$ are mean and variance of ground-truth $\underline{\underline{f}}$ and its denoised image $\underline{\underline{f}}^\dagger$ and $\sigma_{f f^\dagger}$ are their covariance. 
  $c_1 = (0.01 r)^2 \,, c_2 = (0.03 r)^2$ where $r$ is the dynamic range of pixel-values. 
 
\end{itemize}

%\textcolor{red} {Is the "ground truth" image the original image without added noise? I suggest we use "original" image. If we are talking about a denoised image, where do we describe the generation of the noisy image?}

%
\paragraph{Image segmentation}
To evaluate our RQUNet-VAE for segmentation of noisy test data, given that our statistical method provides uncertainty quantification, we propose the following evaluation framework:  

Given a noisy test image $\underline{\underline{f}} \in \mathbb{R}^{\abs{\Omega} \times P}$ as an input of the trained RQUNet-VAE segmentation and its ground-truth $\underline{\underline{f}}^\text{gt} = \left[ f_l^\text{gt} \right]_{l \in \Omega} \in \mathbb R^{n_1 \times n_2 \times K} \,, f_l \in \mathbb R^K$ and the trained RQUNet-VAE by minimizing the loss function (\ref{eq:RQUnetVAELossFunc:segmentation}) from the clean dataset,
we run prediction $n$ time to obtain segmented masks $u^{(i)} \left( \underline{\underline{f}} \right) = \left[ u^{(i)}_l \left( \underline{\underline{f}} \right) \right]_{l \in \Omega}
\in \mathbb R^{\abs{\Omega} \times K} \,, u^{(i)}_l \left( \underline{\underline{f}} \right) \in \mathbb R^K$ for $i = 1 \,, \ldots \,, n$. 
For segmentation problem, class-balanced accuracy is defined for every pixel $l \in \Omega$:
\begin{align}  \label{eq:Accuracy}
 \widehat{p}^n_l = \frac{1}{n} \sum_{i=1}^n \delta_{ \{ u^{(i)}_l( \underline{\underline{f}}) \,, f_l^\text{gt} \} } \,,~
 \delta_{ \{ u^{(i)}_l( \underline{\underline{f}}) \,, f_l^\text{gt} \} } 
 = \begin{cases} 1 \,,& u^{(i)}_l( \underline{\underline{f}}) = f_l^\text{gt} 
 \\ 0 \,,& \text{else} \end{cases} 
\end{align}
which is modeled by a Binomial random variable approximated by Normal distribution (for large $n$):
\begin{align}
 Y_i &:= \delta_{ \{ u^{(i)}_l( \underline{\underline{f}}) \,, f_l^\text{gt} \} } 
 \sim \text{Bernoulli} \left( \widehat{p}^n_l \,, 1 \right) \,,
 \\ % ----
 \label{eq:AccuracyRandom}
 p^n_l &= \frac{1}{n} \sum_{i=1}^n Y_i
 \sim \frac{1}{n} \text{Binomial} \left( \widehat{p}^n_l \,, n \right)
 % ---
 \approx \mathcal{N} \left( \widehat{p}^n_l \,, \frac{\widehat{p}^n_l(1-\widehat{p}^n_l)}{n} \right) \,.
\end{align}

% ----------

%\newpage

%
\subsubsection{Training Procedure}

For comparision with other state-of-the-art-methods, all  parameters have  been  optimized  for  minimal mean-square-error (MSE) via heuristic search for  each  individual  training  image.

Since RQUNet-VAE is a hybrid model of deterministic high-order Riesz-Quincunx wavelet, we apply a training procedure for UNet-VAE  with 100 epochs and a batch size of 16 using the Adam optimization method \cite{kingma2014adam}. The source code of our training in PyTorch can be found at \url{https://github.com/trile83/RQUNetVAE}.

%

% Begin Table
\begin{table}[ht]
\begin{center}
\begin{tabular}{|c|c|c|}
 \hline
 Image Set & &Sentinel-2, std = 0.04
 \\ \hline
 Number of Images & &10
 \\ \hline
 \multirow{2}{*}{RQUNet-VAE scheme 1} 
  &PSNR &$\mathbf{38.693}$ \\ 
  &SSIM &$\mathbf{0.969}$
 \\ \hline
 \multirow{2}{*}{Riesz Dyadic} 
  &PSNR &$37.433$ \\ 
  &SSIM &$0.959$
 \\ \hline
 \multirow{2}{*}{curvelet} 
  &PSNR &$36.314$ \\ 
  &SSIM &$0.942$
 \\ \hline
 \multirow{2}{*}{wavelet CDF 9/7} 
  &PSNR &$36.061$ \\ 
  &SSIM &$0.945$
 \\ \hline %   
 \bottomrule
 \multirow{2}{*}{RQUNet-VAE scheme 2} 
  &PSNR &$\mathbf{39.087}$ \\ 
  &SSIM & $\mathbf{0.971}$
 \\ \hline
 \multirow{2}{*}{GIAF-Riesz Dyadic}
  &PSNR &$38.987$ \\ 
  &SSIM &$0.969$
 \\ \hline
 \multirow{2}{*}{TV-L2 $(L=2)$} 
  &PSNR &$38.522$ \\ 
  &SSIM &$0.968$
 \\ \hline
 \multirow{2}{*}{TV-L2 $(L=9)$} 
  &PSNR &$38.53$ \\ 
  &SSIM &$0.968$ 
  \\ \hline
\end{tabular}
\end{center}
\caption{PSNR and SSIM: Mean over the three image sets (1st nd 2nd best in bold).}
\label{MeanVarianceComparison:ImageDenoising}
\end{table}

% % Begin Table
% \begin{table}[ht]
% \begin{center}
% \begin{tabular}{|c|c|c|}
%  \hline
%  %
%  Image Set & &Sentinel-2, std = 0.04
%  \\ \hline
%  %
%  Number of Images & &10
%  \\ \hline
%  %
%  \multirow{2}{*}{RQUnet-VAE scheme 1} 
%   &PSNR &$\mathbf{38.693 \pm 0.876}$ \\ 
%   &SSIM &$\mathbf{0.969 \pm 0.0046}$
%  \\ \hline
%  %
%  \multirow{2}{*}{Riesz Dyadic} 
%   &PSNR &$37.433 \pm
%  1.402$ \\ 
%   &SSIM &$0.959 \pm
%  0.0098$
%  \\ \hline
%  % 
%  \multirow{2}{*}{curvelet} 
%   &PSNR &$36.314 \pm 1.47$ \\ 
%   &SSIM &$0.942 \pm 0.0244$
%  \\ \hline
%  % 
%  \multirow{2}{*}{wavelet CDF 9/7} 
%   &PSNR &$36.061 \pm 0.814$ \\ 
%   &SSIM &$0.945 \pm
%  0.0077$
%  \\ \hline %   
%  \bottomrule
%  % 
%  \multirow{2}{*}{RQUnet-VAE scheme 3} 
%   &PSNR &$\mathbf{39.087 \pm 0.9264}$ \\ 
%   &SSIM & $\mathbf{0.971 \pm 0.0049}$
%  \\ \hline
%  % 
%  \multirow{2}{*}{GIAF-Riesz Dyadic}
%   &PSNR &$38.987 \pm 1.1592$ \\ 
%   &SSIM &$0.969 \pm 0.0066$
%  \\ \hline
%  %  
%  \multirow{2}{*}{TV-L2 $(L=2)$} 
%   &PSNR &$38.522 \pm 1.21$ \\ 
%   &SSIM &$0.968 \pm 0.006$
%  \\ \hline
%  %   
%  \multirow{2}{*}{TV-L2 $(L=9)$} 
%   &PSNR &$38.53 \pm 1.1345$ \\ 
%   &SSIM &$0.968 \pm 0.0056$ 
%   \\ \hline
% \end{tabular}
% \end{center}
% \caption{PSNR and SSIM: Mean $\pm$ standard deviation over the three image sets (1st nd 2nd best in bold).}
% \label{MeanVarianceComparison:ImageDenoising}
% \end{table}

\if\hidefigures0

\begin{figure*}
\begin{center}  

% Original:
\subfigure[Original image]{\includegraphics[width=0.24\textwidth]{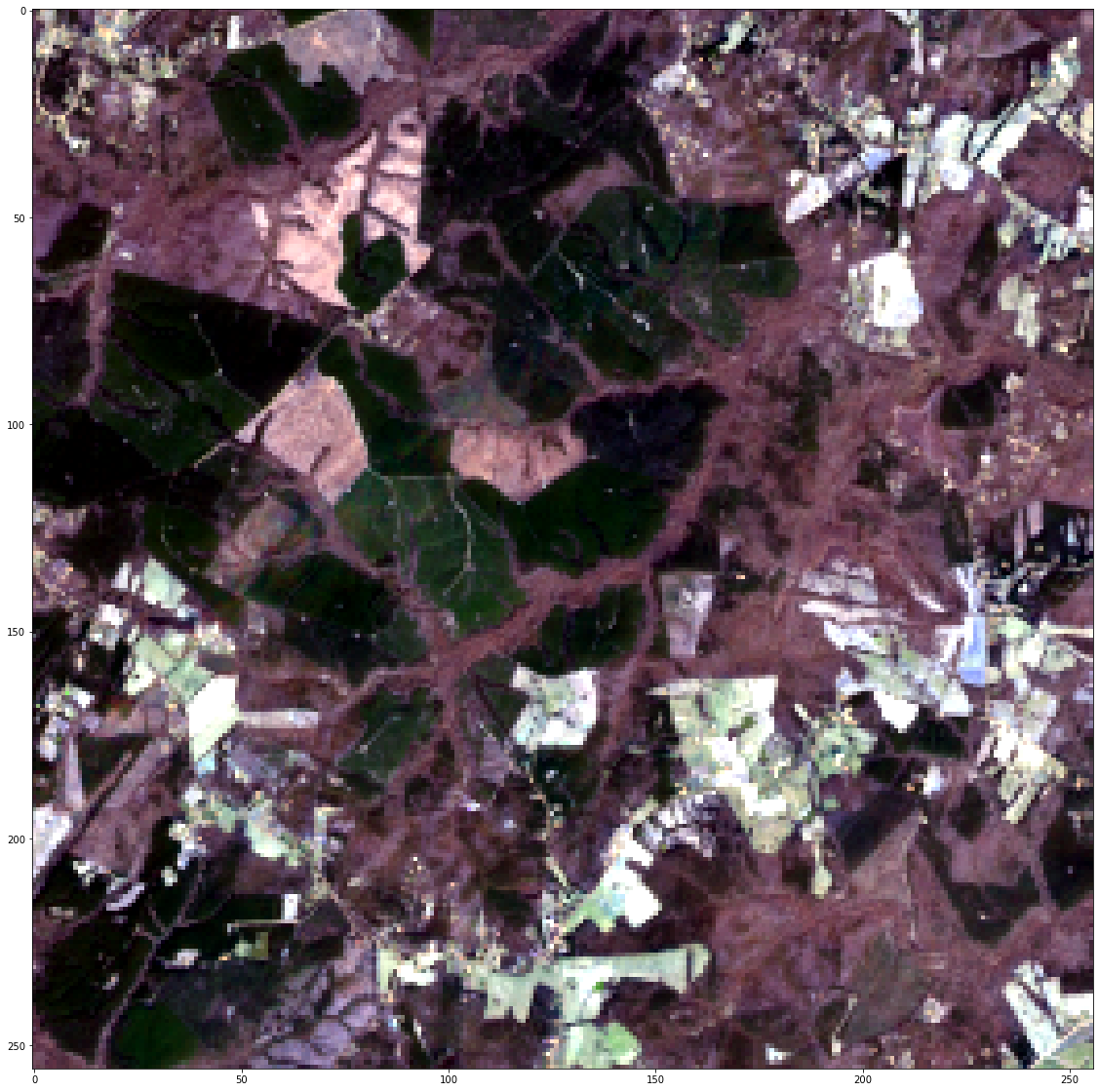}\label{subfig:original_image}}
% Noisy image:
\subfigure[Noisy image, $\sigma = 0.04$]{\includegraphics[width=0.24\textwidth]{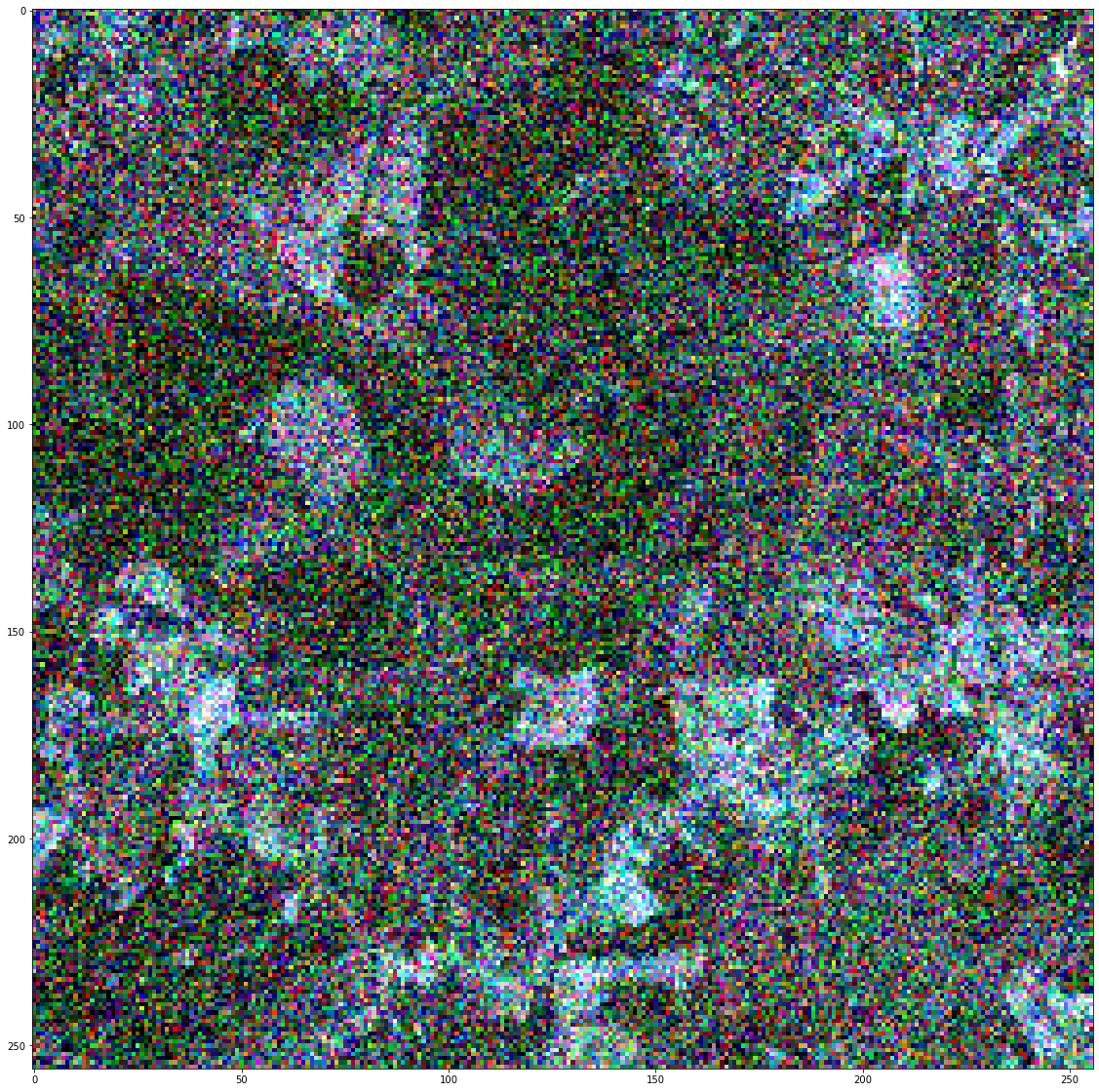}\label{subfig:original_image_with_noise}}

% RQUnet-VAE:
\subfigure[RQUnet-VAE, $(\text{PSNR}, \text{SSIM}) = (38.502, 0.966)$]{\includegraphics[width=0.24\textwidth]{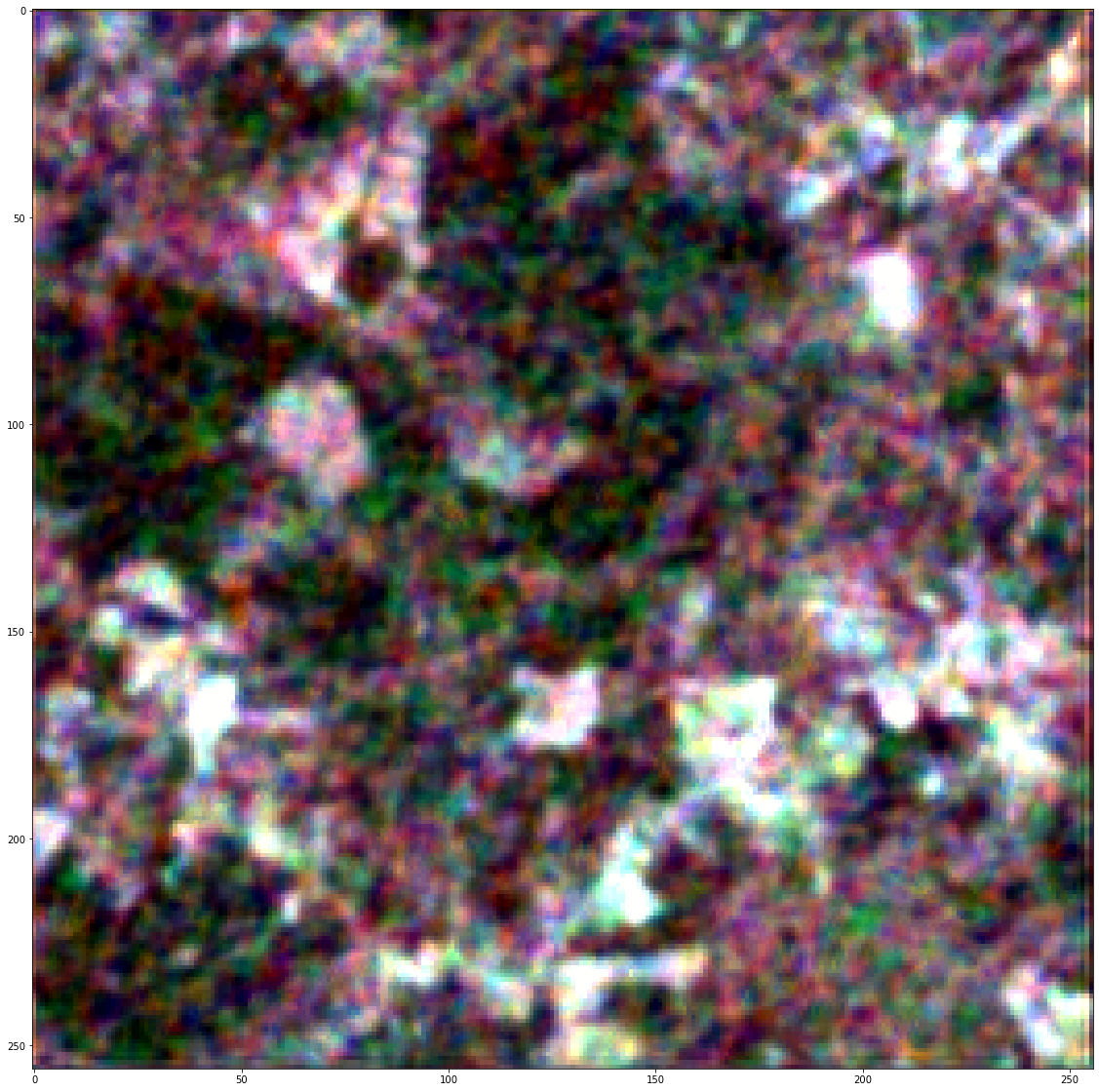}\label{subfig:original_image_scheme1}}
% GIAF-RieszDyadic:
\subfigure[Riesz Dyadic, $(\text{PSNR}, \text{SSIM}) = (37.555, 0.96)$]{\includegraphics[width=0.24\textwidth]{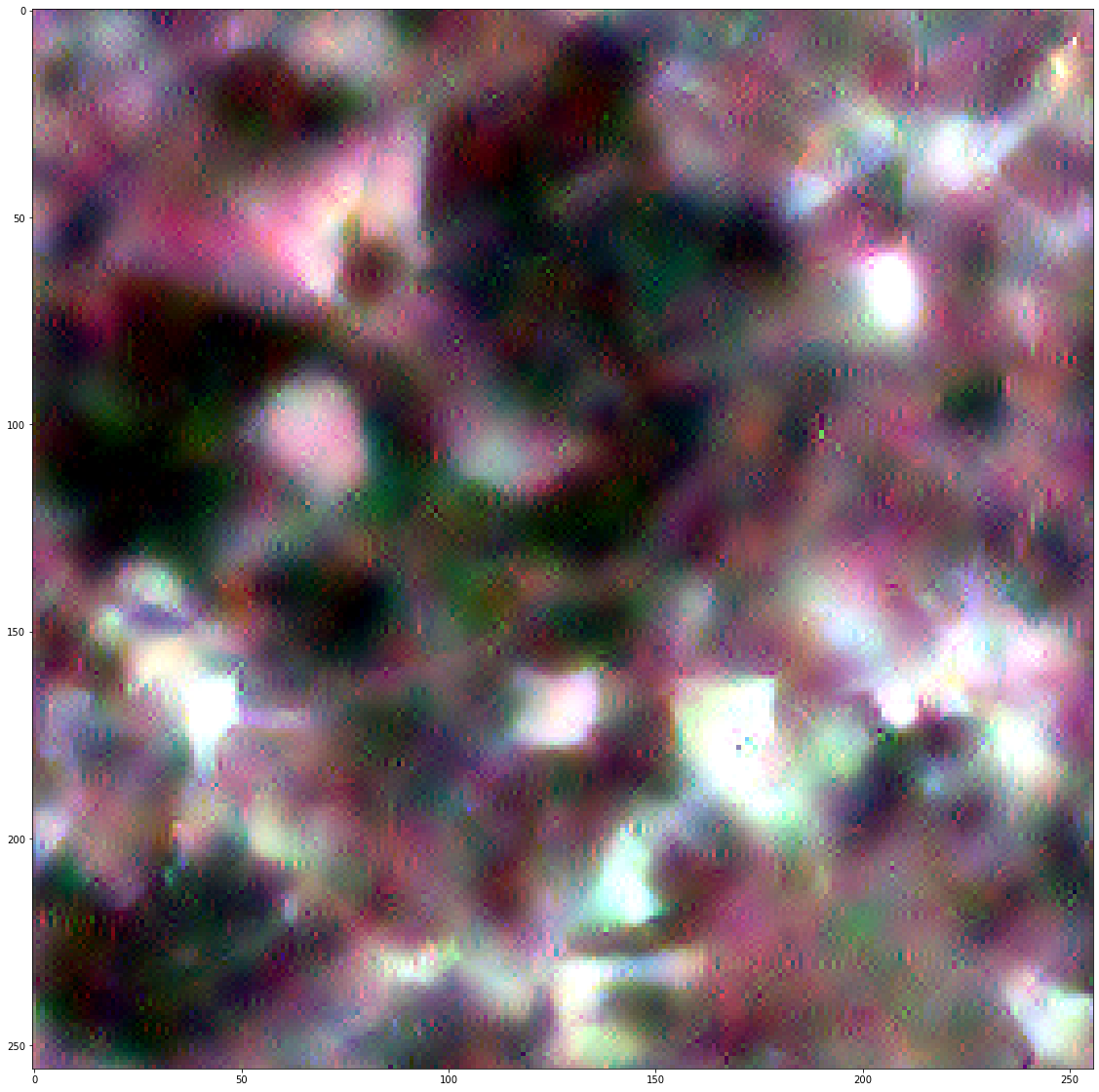}\label{subfig:original_image_riesz}}
% Curvelet-Smoothing:
\subfigure[curvelet, $(\text{PSNR}, \text{SSIM}) = (36.494, 0.944)$]{\includegraphics[width=0.24\textwidth]{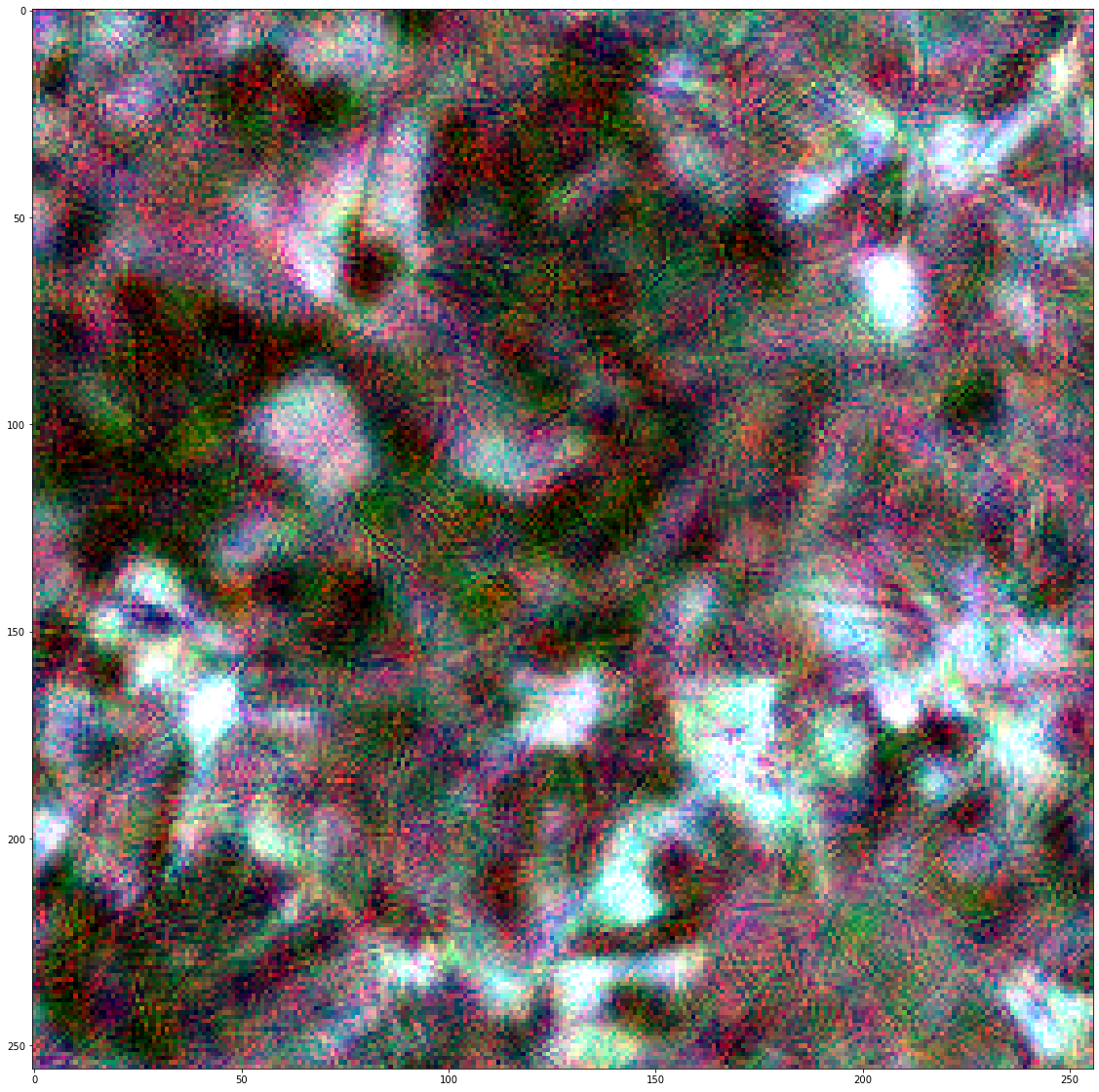}\label{subfig:original_image_curvelet}}
% Wavelet-Smoothing:
\subfigure[wavelet CDF 9/7, $(\text{PSNR}, \text{SSIM}) = (35.646, 0.936)$]{\includegraphics[width=0.24\textwidth]{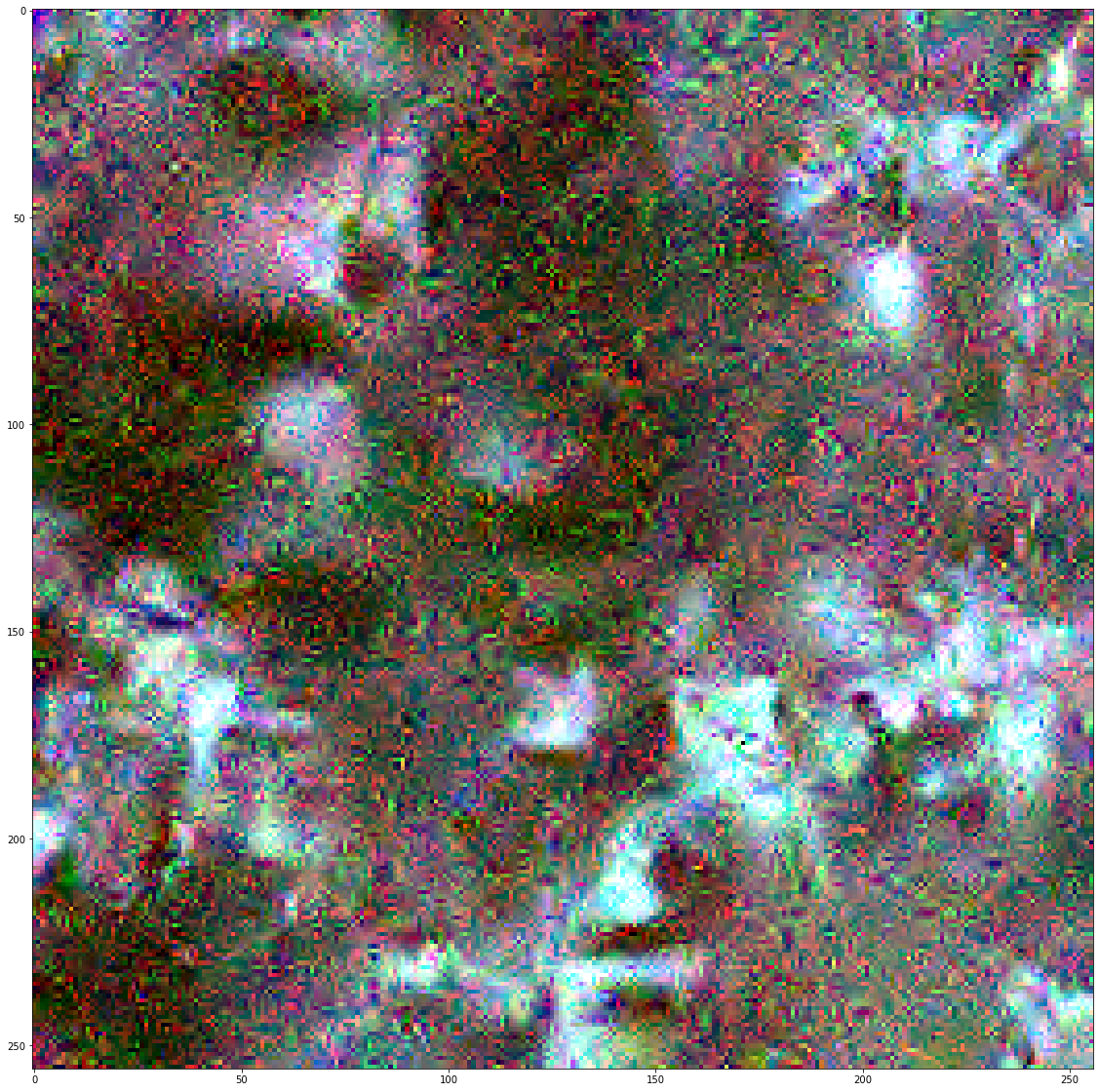}\label{subfig:original_image_with_noise_wavelet_cdf}}

% Iterative Methods:
% RQUnet-VAE-scheme3:
\subfigure[RQUnet-VAE scheme 3, $(\text{PSNR}, \text{SSIM}) = (38.822, 0.97)$]{\includegraphics[width=0.24\textwidth]{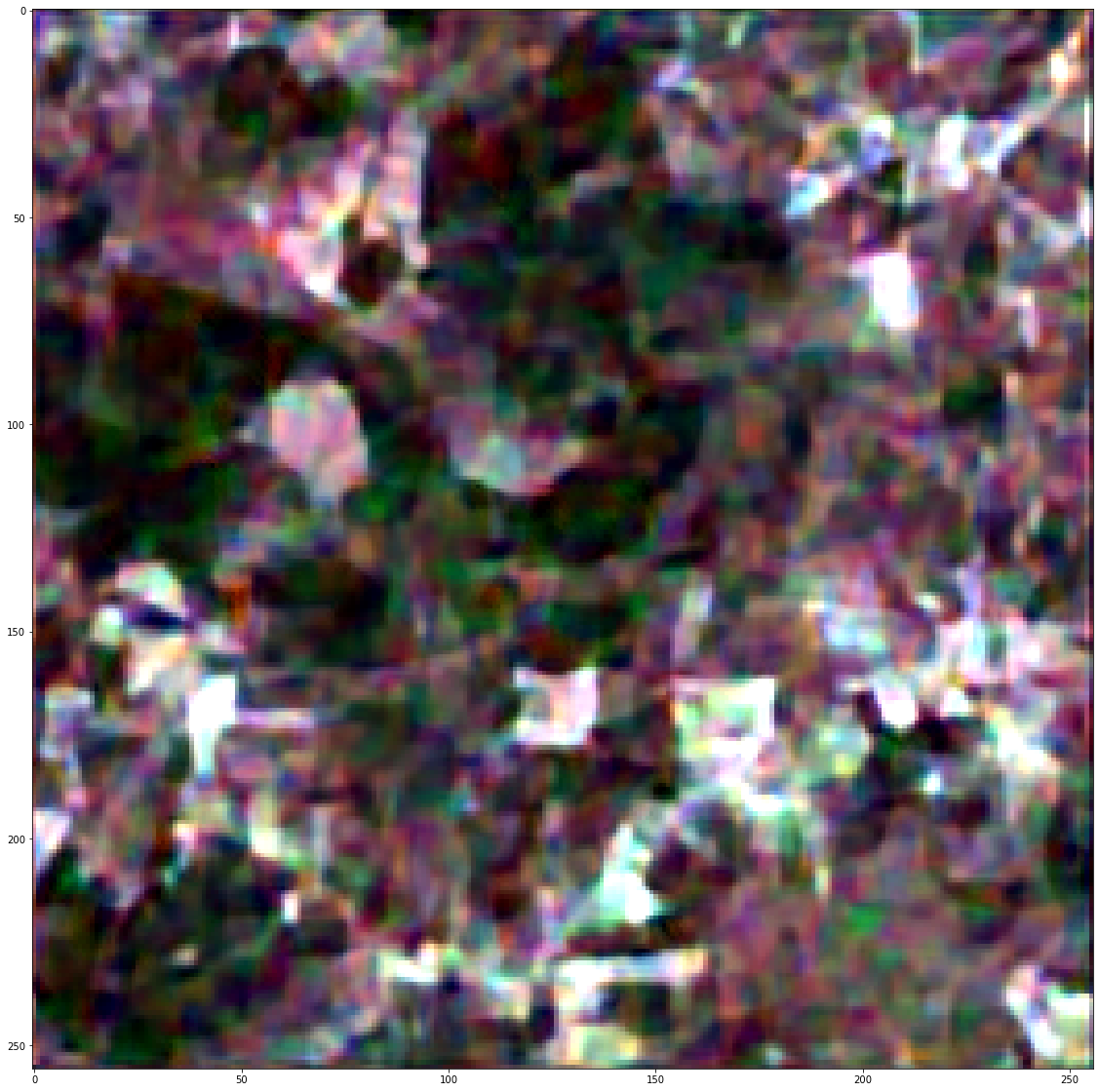}\label{subfig:original_image_scheme2}}
% GIAF-RieszDyadic:
\subfigure[GIAF-Riesz Dyadic, $(\text{PSNR}, \text{SSIM}) = (38.936, 0.97)$]{\includegraphics[width=0.24\textwidth]{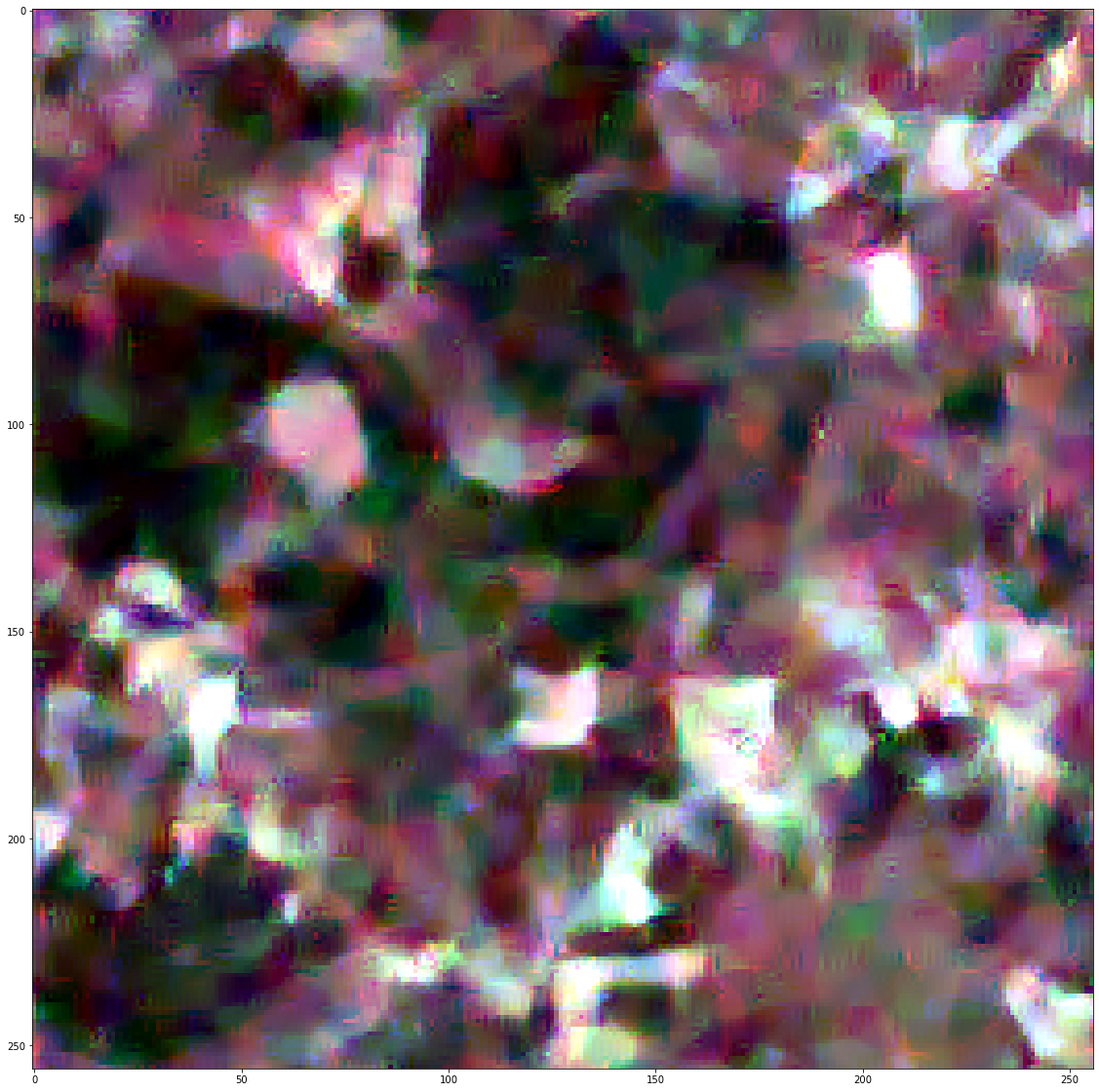}\label{subfig:original_image_giaf}}
% TVL2 L=2:
\subfigure[TV-L2 $L=2$, $(\text{PSNR}, \text{SSIM}) = (38.513, 0.967)$]{\includegraphics[width=0.24\textwidth]{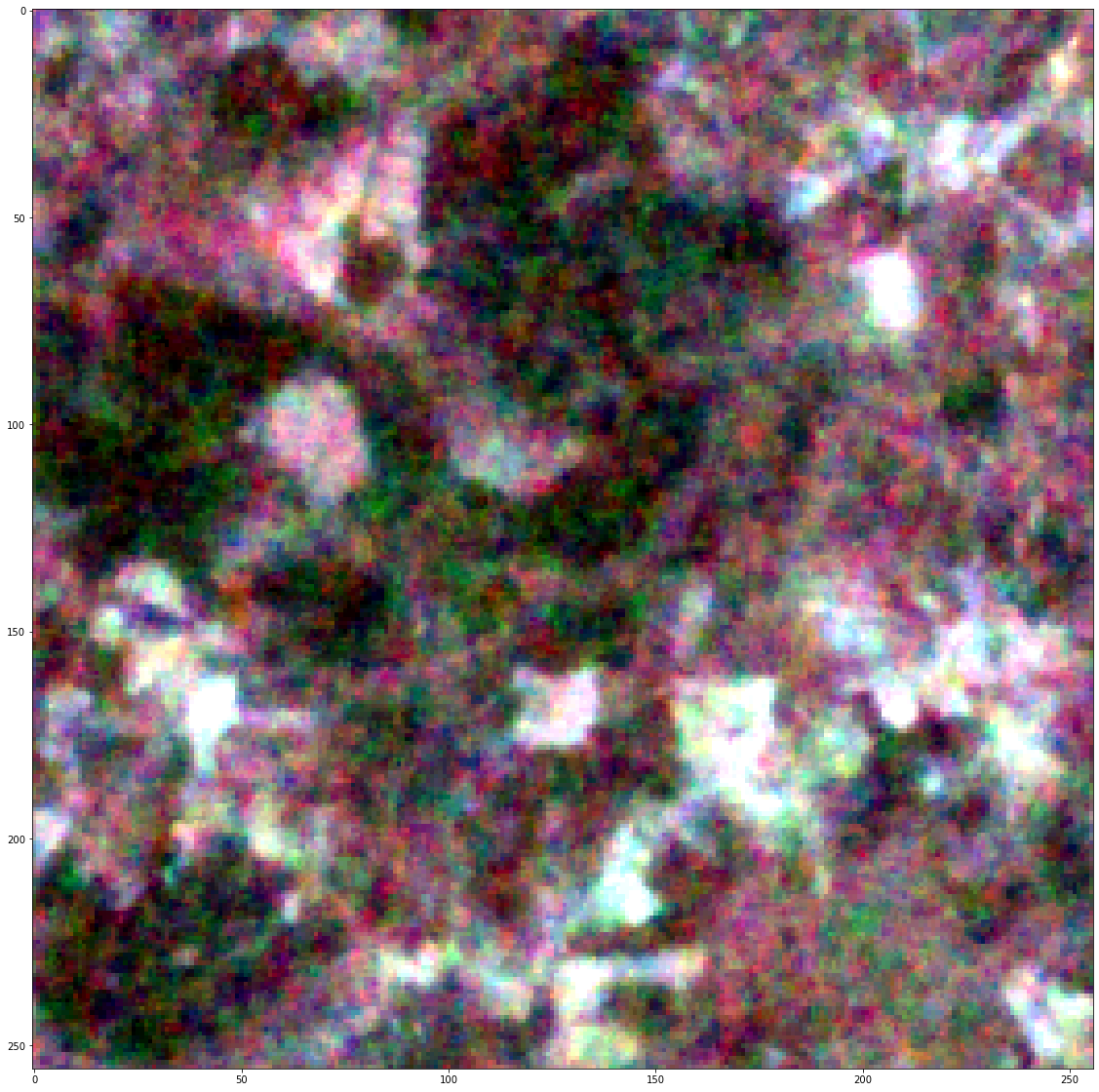}\label{subfig:original_image_tvl22}}
% TVL2 L=2:
\subfigure[TV-L2 $L=9$, $(\text{PSNR}, \text{SSIM}) = (38.535, 0.966)$]{\includegraphics[width=0.24\textwidth]{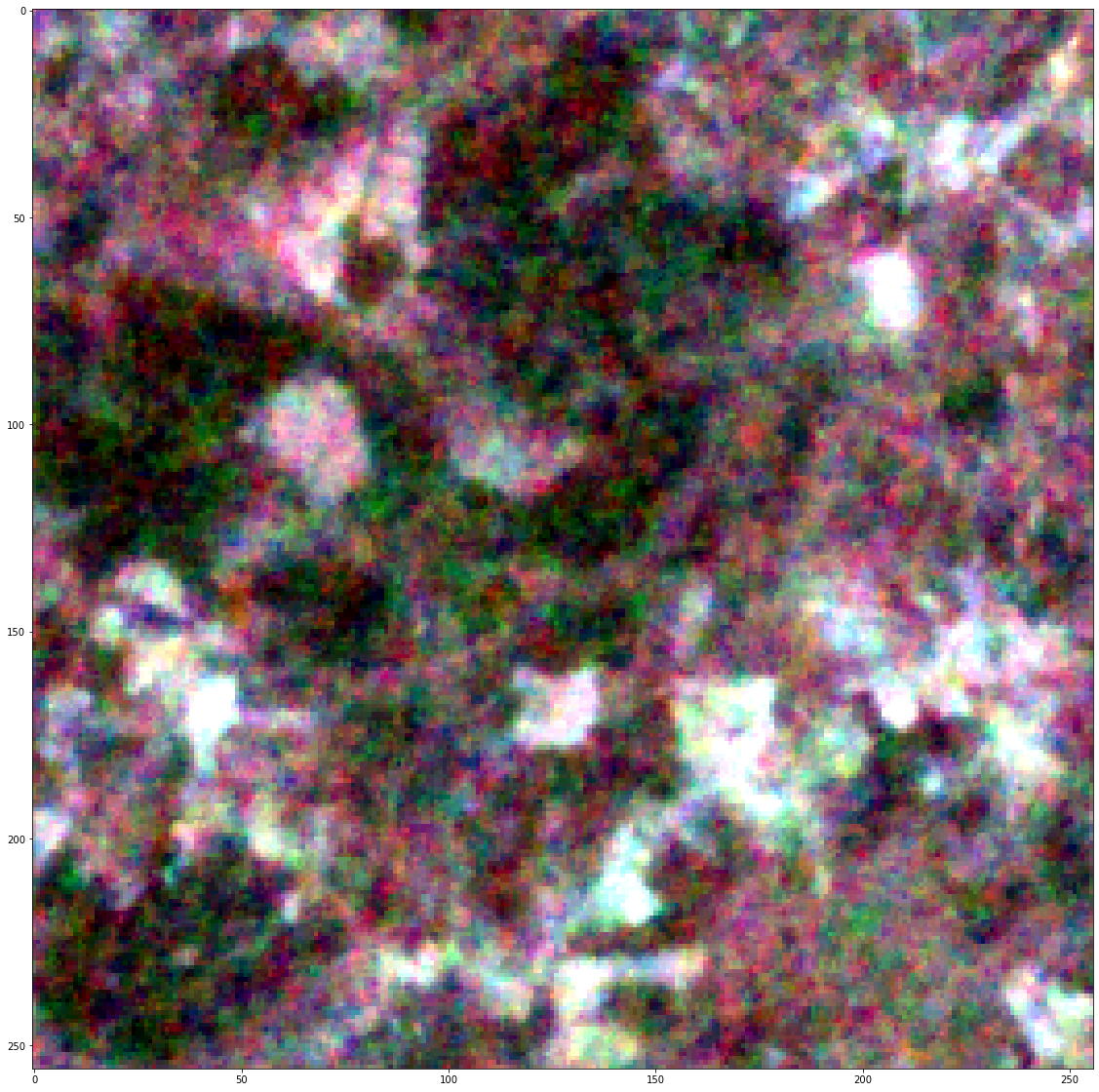}\label{subfig:original_image_tvl29}}

\caption{Qualitative image denoising results for RQUnet-VAE and competitor approaches. 
Image denoising by the RQUnet-VAE with standard deviation $\sigma = 0.04$ and optimal $\alpha^*$.}\vspace{-0.8cm} 

%The RQUnet-VAE provides robustness to the quantile $\alpha$, see MSE (f). Note that in general, we don not know how to compute the optimal $\alpha^*$ for large noise; then, for our method we simply choose some large enough $\alpha$ in practice. }
%\color{red} {Compare with TV-L2 (L=2, 10), TVG, GIAF, non-local mean}

\label{fig:ImageDenoisingScheme1:2019059hdf}
\end{center}
\end{figure*}

\fi

\if\hidefigures0

\fi

\subsection{Denoising Experiments}\label{subsec:denoising_experiments}\vspace{-0.3cm}
To impose artificial noise on images for evaluation, Gaussian noise was added to each band of the Sentinel-2 images with a standard deviation of $\text{std} =  0.04$.  All images were normalized to interval $[0 \,, 1]$ using the image minimum and maximum. To give some examples, original image shown in Figure~\ref{subfig:original_image} and the corresponding image with added noise in Figure~\ref{subfig:original_image_with_noise}. (All ten test images with added noise are visualized in SM).

% Below we present both quantitative (based on metrics in Section~\ref{exp:Quant}) and qualitative results that suggest that (i) RQUNet-VAE has an improved ability to remove noise from images compared to the baseline solutions described in Section~\ref{subsec:competitors}, and (ii) we also provide qualitative results by visually assessing the noise-reduced images resulting from the RQUNet-VAE compared to baseline approaches.

\subsubsection{Quantitative Results}\label{exp:Quant}
Table \ref{MeanVarianceComparison:ImageDenoising} provides the results of the comparison of the RQUNet-VAE Schemes 1 and 2 with the state-of-the-art methods described in Section~\ref{subsec:competitors} using the evaluation metrics described in Section~\ref{subsec:evaluationmetrics}. 
For the ten test images described in Section~\ref{subsec:datasets} the proposed RQUNet-VAE Scheme 2 yields the highest Peak-Signal-to-Noise-Ratio (PSNR) and the highest Structural Similarity Index (SSIM) of all competitors (Table 1). Among all the approaches based on harmonic analysis, RQUNet-VAE Scheme 1 yields the best results. 
Note GIAF also  iteratively computes scaling and wavelet coefficients, similar to our RQUNet-VAE Scheme 2. The main difference, however, is that GIAF employs Riesz Dyadic wavelet kernel for scaling and wavelet functions whereas RQUNet-VAE Scheme 2 uses our adapted RQUNet-VAE. Using this scheme, allows more redundant parameters to better learn specific features from the data and thus, better adapt to specific images. (Denoising results for all other test images are visualized in SM).

\begin{figure*}[tbh]
\begin{center}  

% Image 1:
% Original
\subfigure[Original image%, $\text{MSE} = 2.83 \, 10^{-6}$
]{\includegraphics[width=0.24\textwidth]{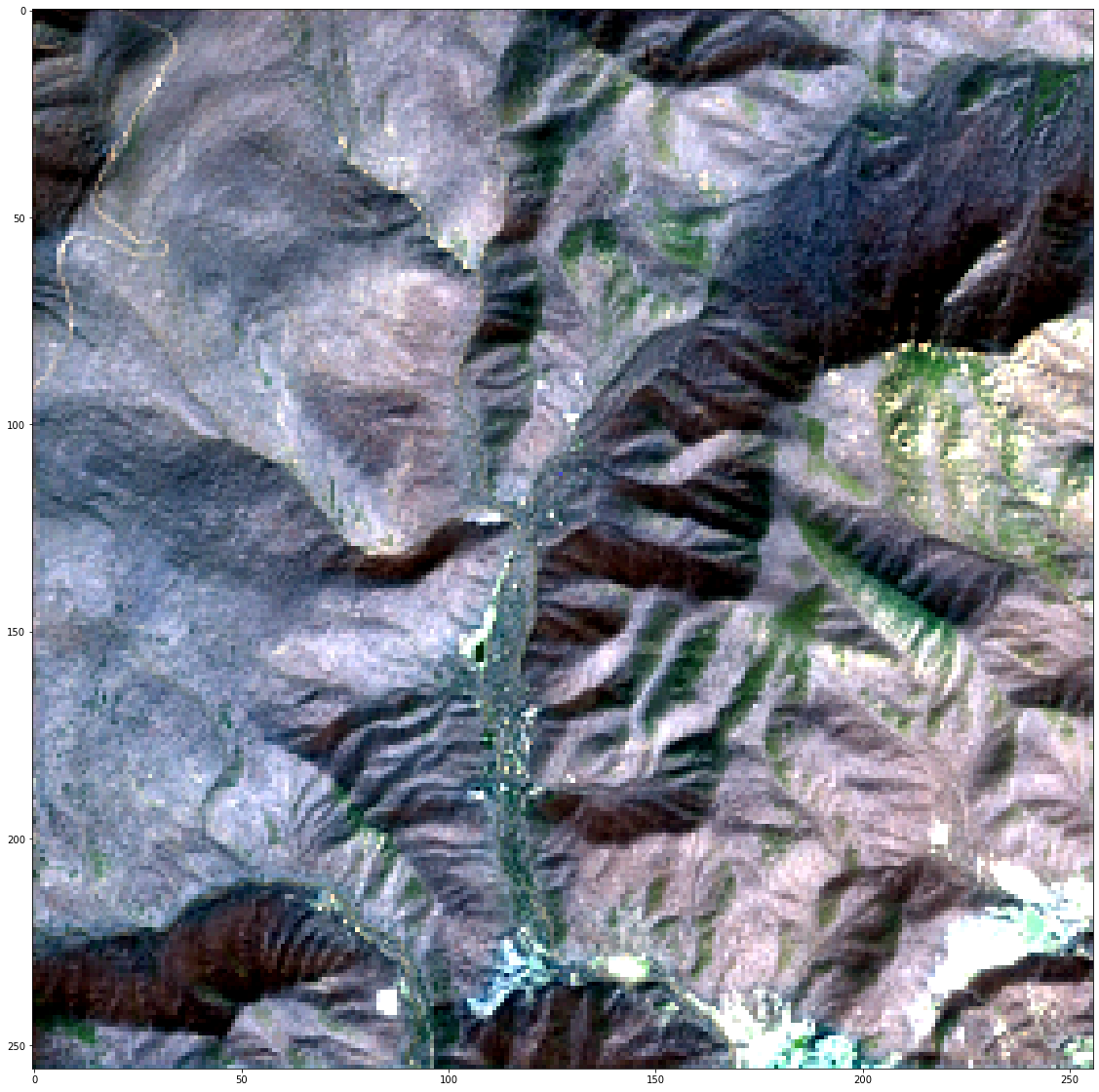}}
% Lowpass
\subfigure[Lowpass image]{\includegraphics[width=0.24\textwidth]{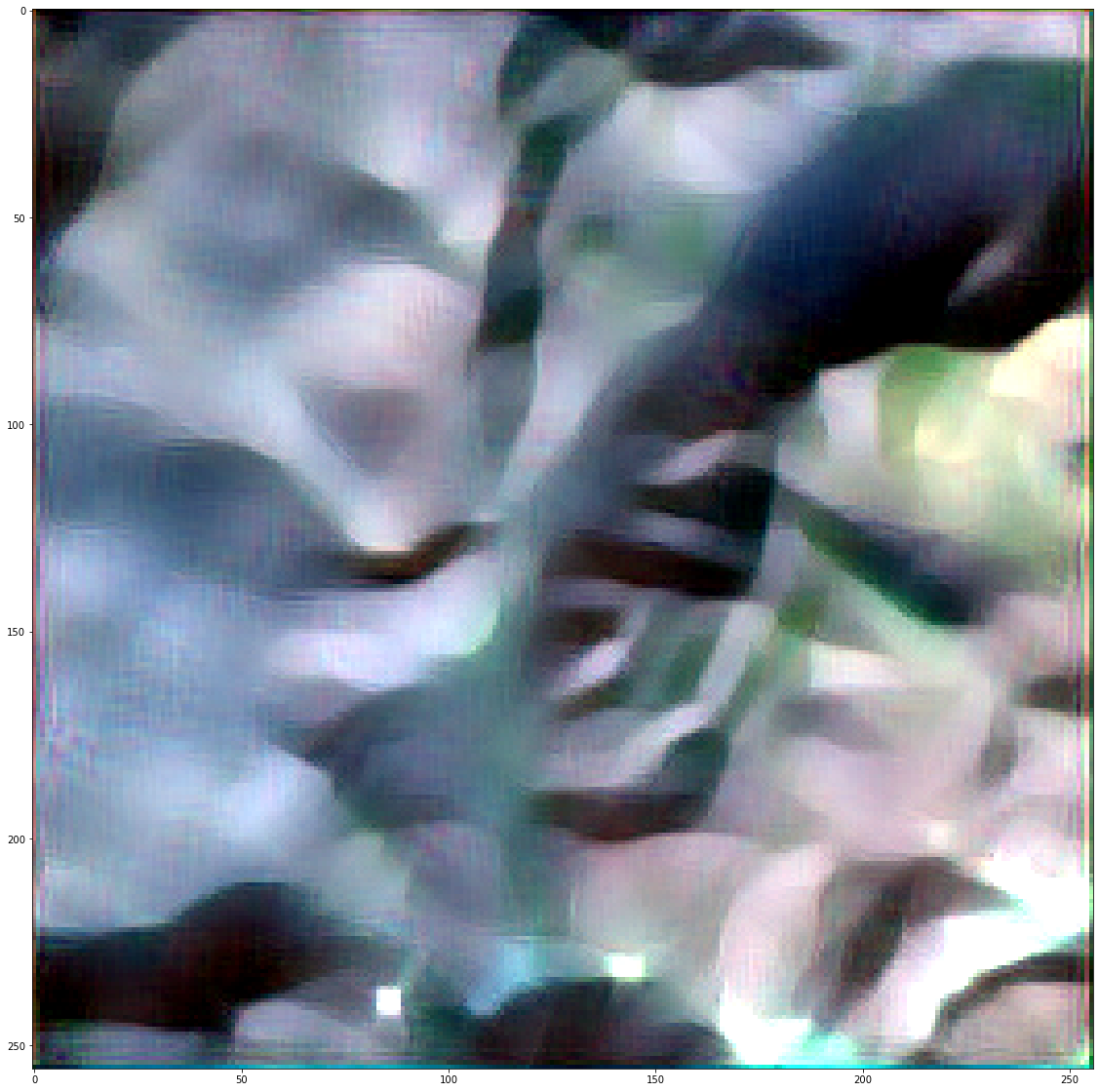}}
% Bandpass
\subfigure[Bandpass image]{\includegraphics[width=0.24\textwidth]{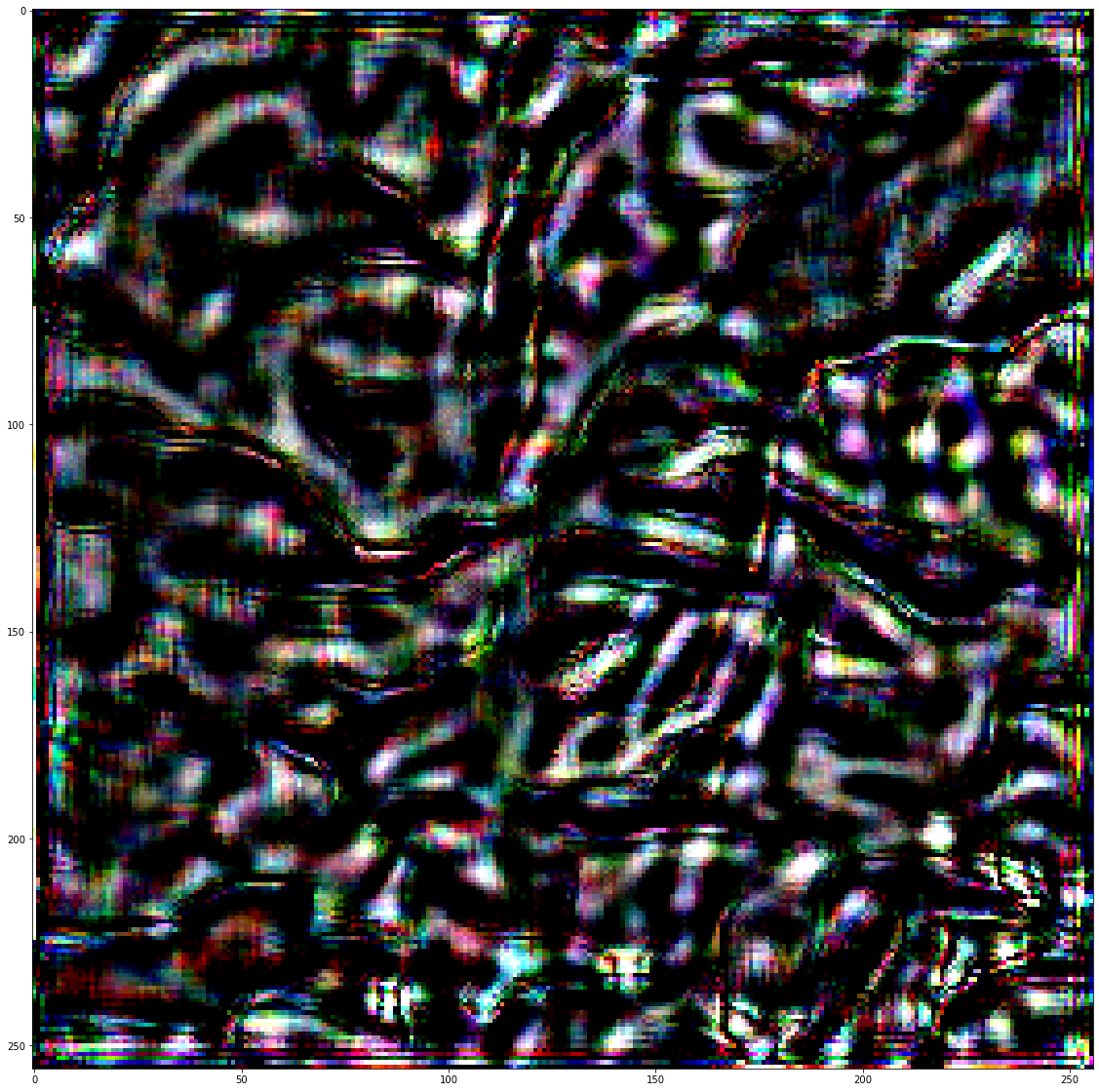}}
% Highpass
\subfigure[Highpass image]{\includegraphics[width=0.24\textwidth]{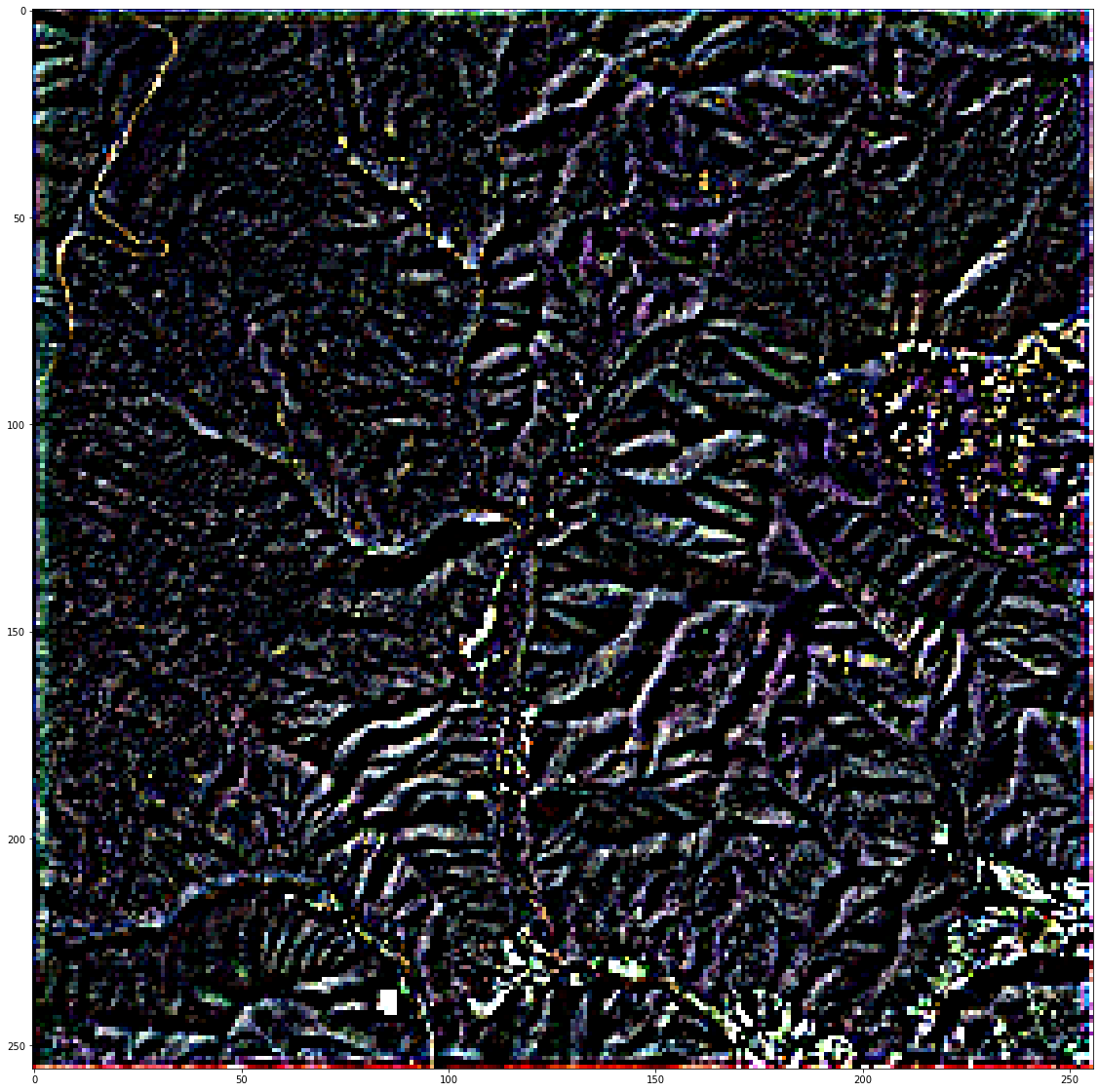}}

\vspace{-0.2cm}
% Image 2:
% Original
\subfigure[Original image%, $\text{MSE} = 4.65 \, 10^{-6}$
]{\includegraphics[width=0.24\textwidth]{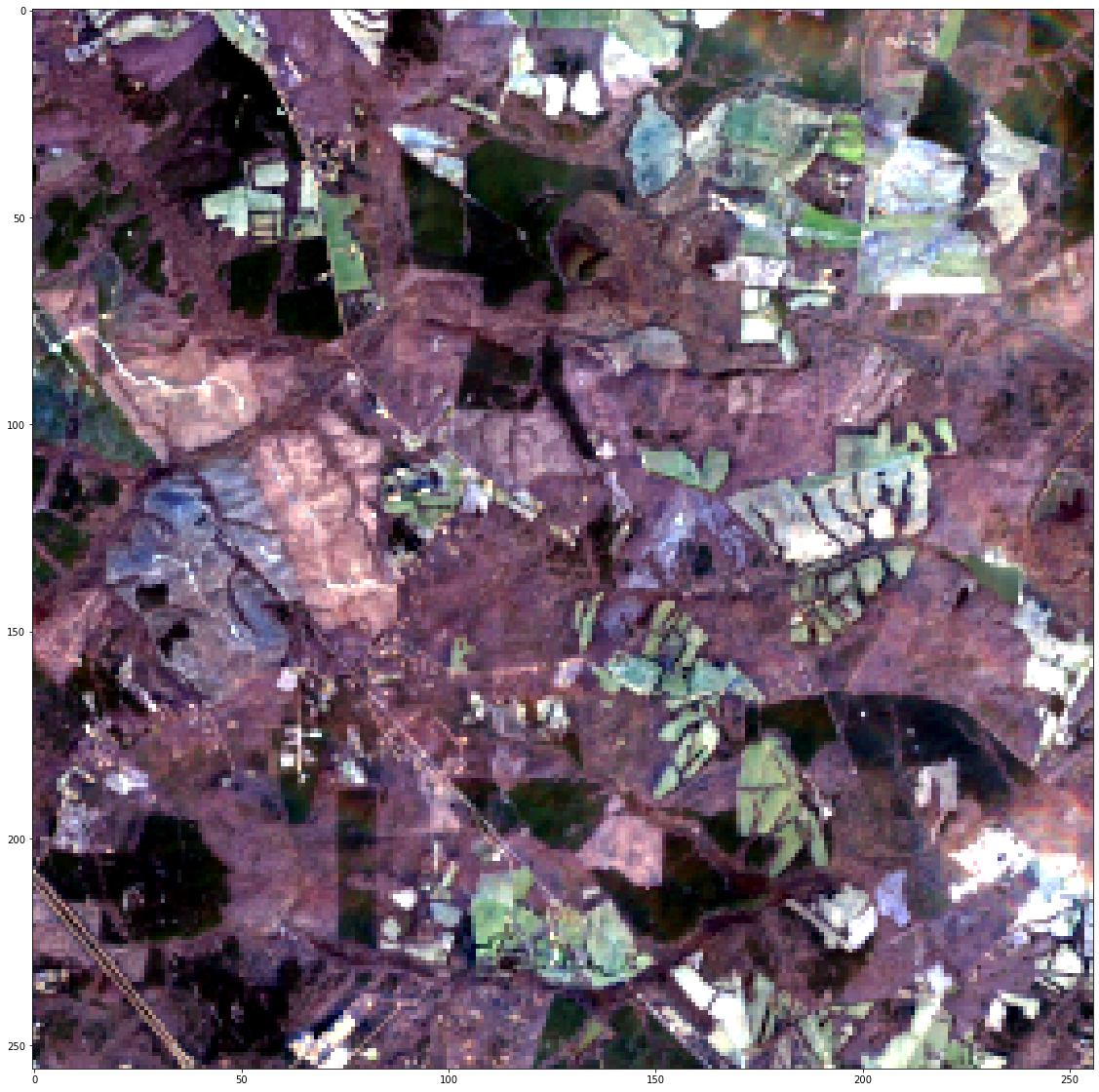}}
% Lowpass
\subfigure[Lowpass image]{\includegraphics[width=0.24\textwidth]{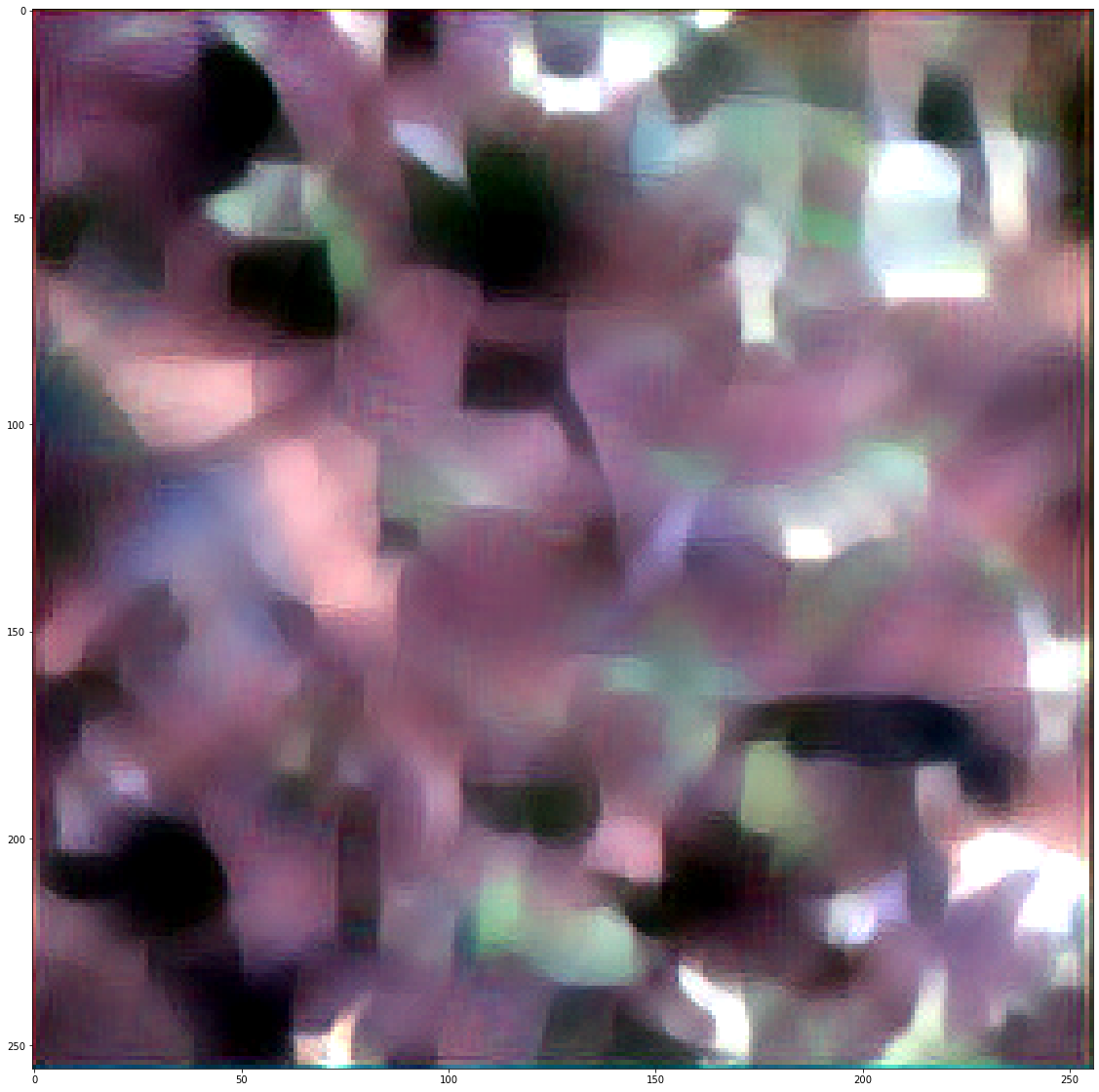}}
% Bandpass
\subfigure[Bandpass image]{\includegraphics[width=0.24\textwidth]{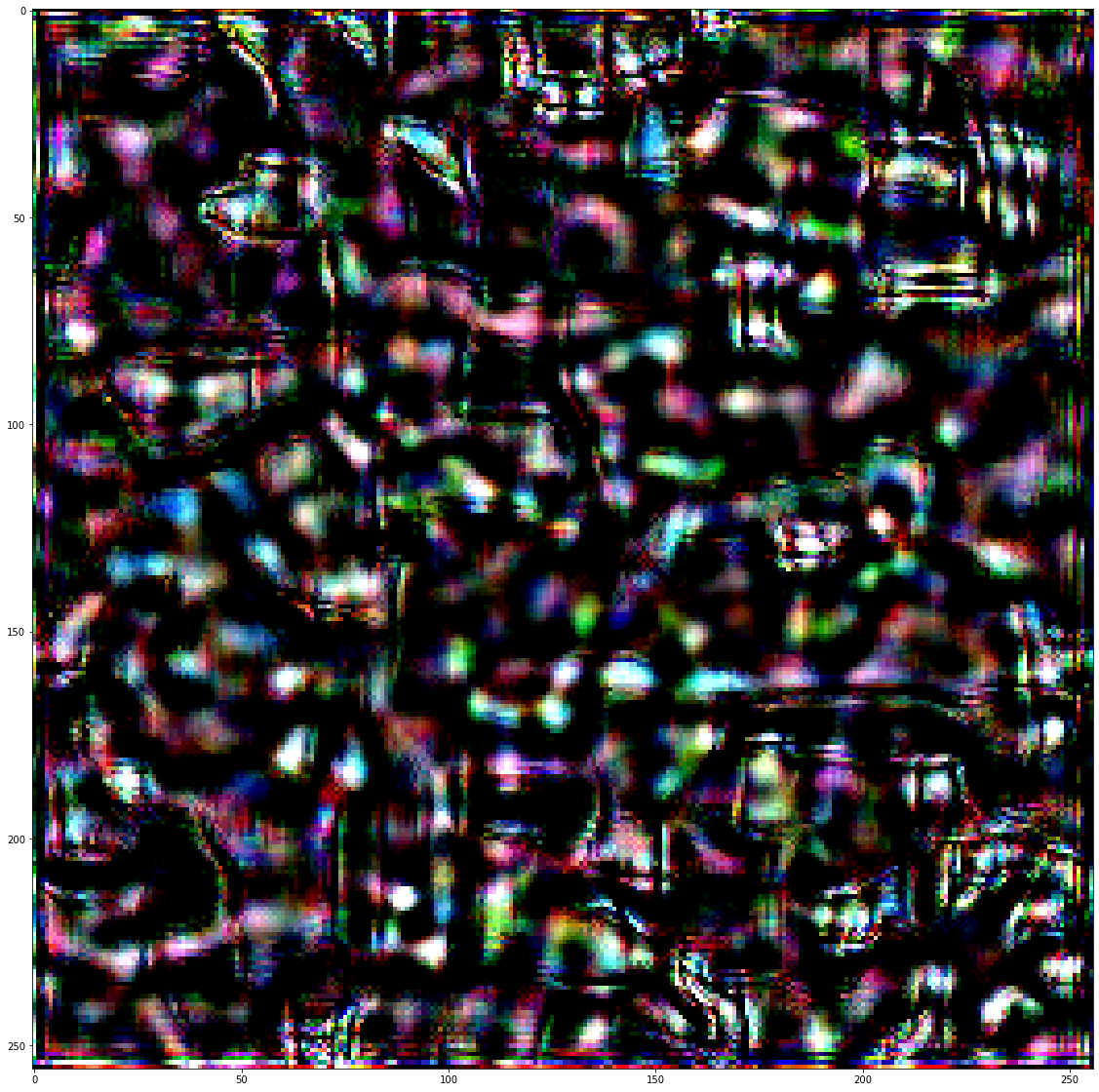}}
% Highpass
\subfigure[Highpass image]{\includegraphics[width=0.24\textwidth]{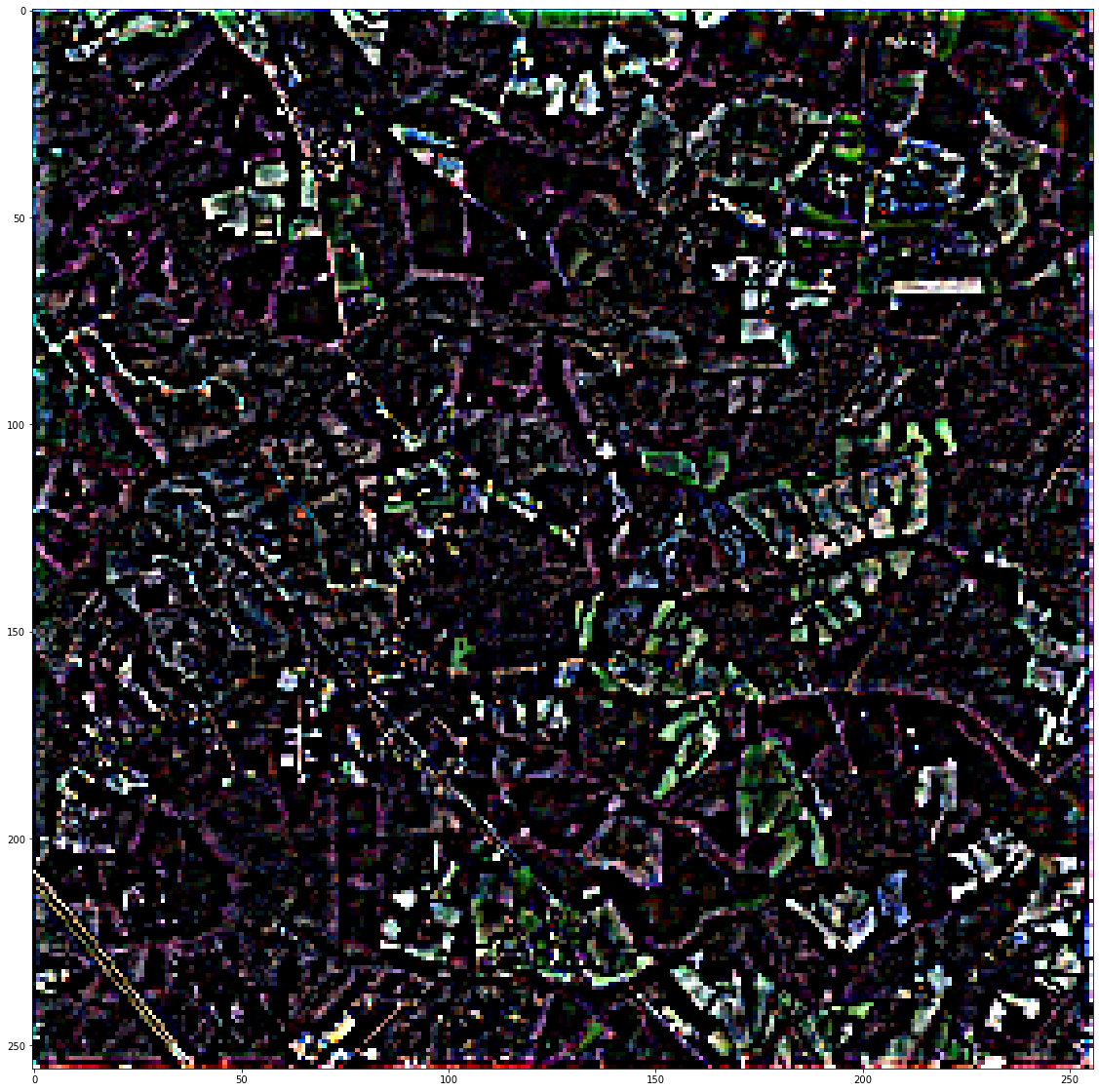}}

%\vspace{-0.2cm}
%
% Image 3:
% Original
%\subfigure[Original image%, $\text{MSE} = 6.53 \, 10^{-6}$]{\includegraphics[width=0.24\textwidth]{figures/ImageDecomposition3_ReconstructedImage.png}}
% Lowpass
%\subfigure[Lowpass image]{\includegraphics[width=0.24\textwidth]{figures/ImageDecomposition3_LowpassSignal.png}}
% Bandpass
%\subfigure[Bandpass image]{\includegraphics[width=0.24\textwidth]{figures/ImageDecomposition3_BandpassSignal.png}}
% Highpass
%\subfigure[Highpass image]{\includegraphics[width=0.24\textwidth]{figures/ImageDecomposition3_HighpassSignal.png}}

%\vspace{-0.2cm}
%
% Image 4:
% Original
%\subfigure[Original image%, $\text{MSE} = 6.53 \, 10^{-6}$]{\includegraphics[width=0.24\textwidth]{figures/ImageDecomposition4_ReconstructedImage.png}}
% Lowpass
%\subfigure[Lowpass image]{\includegraphics[width=0.24\textwidth]{figures/ImageDecomposition4_LowpassSignal.png}}
% Bandpass
%\subfigure[Bandpass image]{\includegraphics[width=0.24\textwidth]{figures/ImageDecomposition4_BandpassSignal.png}}
% Highpass
%\subfigure[Highpass image]{\includegraphics[width=0.24\textwidth]{figures/ImageDecomposition4_HighpassSignal.png}}
\vspace{-0.3cm}
\caption{Image decomposition by RQUNet-VAE by spectral analysis with a threshold $T = 3$ for spectral histogram. Parameters: number of iteration $\text{Iter} = 30$, a smoothing parameter $\alpha = 0.4$. %Mean-square-error of reconstructed images versus their original images are reported in a figure.
}
\label{fig:ImageDecomposition}\vspace{-0.7cm}
\end{center}
\end{figure*}

\subsubsection{Qualitative Results}

Figures 1, 2 and 3 in SM %\ref{fig:OriginalTestImages:19059hdf}, \ref{fig:NoisyTestImagesvar0_04:19059hdf} and \ref{fig:DenoisedTestImagesRQUnetVAEScheme3var0_04:19059hdf}
are original, noisy and denoised images with the 2nd scheme of our RQUNet-VAE. Figures~\ref{subfig:original_image_scheme1}-\ref{subfig:original_image_tvl29} provide qualitative results to assess how well the RQUNet-VAE schemes are able to reduce artificially added noise individual Sentinel2 images.  Figure~\ref{subfig:original_image_scheme1} shows the result of denoising the noisy image of Figure~\ref{subfig:original_image_with_noise} using our RQUNet-VAE Scheme~1. Since Scheme~1 is based on harmonic analysis, we compare it to noise reduction using a Riesz Dyadic wavelet kernel~\cite{RichterThaiHuckemann2020} in Figure~\ref{subfig:original_image_riesz}, using curvelets~\cite{CandesDonoho2004} in Figure~\ref{subfig:original_image_curvelet}, and using wavelet CDF 9/7~\cite{UnserBlu2003, Sweldens1997, daubechies1998factoring} in Figure~\ref{subfig:original_image_with_noise_wavelet_cdf}. The RQUNet-VAE Scheme~1 best preserves the edges of objects of the original image. The Riez Dyadic wavelet kernel in Figure~\ref{subfig:original_image_riesz} yields a blurred denoised image. In contrast, wavelets (Figure~\ref{subfig:original_image_with_noise_wavelet_cdf}) better retain edges, but much of the noise remains. Across all these figures, it is clearly discernible that the proposed RQUNet-VAE Scheme~1 provides the best combination of noise reduction and delineation of edges and objects.

Figure~\ref{subfig:original_image_scheme2} shows the denoised image using the RQUNet-VAE Scheme~2, compared to other iterative methods including GIAF~\cite{RichterThaiHuckemann2020} in Figure~\ref{subfig:original_image_giaf} and  directional TV-L2~\cite{RudinOsherFatemi1992, GoldsteinOsher2009,ThaiGottschlich2016DG3PD} using a number of directions $L=2$ in Figure~\ref{subfig:original_image_tvl22} and $L=9$ in Figure~\ref{subfig:original_image_tvl29}. %It is evident that RQUNet-VAE Scheme~2 again provides the best trade-off between reduction of noise and delineation of edges and objects, where GIAF (Figure~\ref{subfig:original_image_giaf}) oversmoothes edges between objects while TV-L2 (Figure~\ref{subfig:original_image_tvl22} and Figure~\ref{subfig:original_image_tvl29}) still retain clearly visible noise.
Figure \ref{fig:ImageDenoisingScheme1:2019059hdf} (g) illustrates that RQUNet-VAE significantly reduces artefacts while preserving texture pattern, contrast and sharp edges of objects in the reconstructed image, such that it is most similar to the original image. 
%\textcolor{red}{I am not sure one can claim that that texture pattern is preserved looking at these images. At this scale it is hard to see any texture within objects in the original image.}
%
%\textcolor{red}{Duy: no method can perfectly preserve small scale texture mixed with noise, but some large scale texture can be visible (= "preserved") while smoothing out noise!}
%
Note that the GIAF-Riesz dyadic method in Figure \ref{fig:ImageDenoisingScheme1:2019059hdf} (h) also performs very well, but a reconstructed images also contained some artefacts.  The RQUNet-VAE Scheme~2 again provides the best trade-off between reduction of noise and delineation of edges and objects, where GIAF (Figure~\ref{subfig:original_image_giaf}) oversmoothes edges between objects while TV-L2 (Figure~\ref{subfig:original_image_tvl22} and Figure~\ref{subfig:original_image_tvl29}) still retain clearly visible noise.
%\textcolor{red}{what do we define as artifacts here? compared to the RQUNET-VAE this method preserved the the edges and geometery of  objects, but lost most of the texture in the inside objects. I'm trying to make sure we describe the visual assessment objectively using the correct terminology}
%
%\textcolor{red}{Duy: no method can pareserve texture mixed with noise! but better methon perform better for some certain scale textures!}
%
This demonstrates that the RQUNet-VAE with learned and deterministic frames increases sparsity in a generalized Besov space. 
Section~\ref{sec:DenoisingSegmentation} will demonstrate how this balance of noise reduction and discrimination of edges improves machine learning results when applying the RQUNet-VAE to segmentation of high resolution images for land cover mapping.

% ========================

%\begin{figure*}
%\begin{center}  
 
% Original:
%\subfigure[Original image]{\includegraphics[width=0.24\textwidth]{figures/Image1_Diffusion_ori.png}}
% Iteration:
%\subfigure[Iteration 4]{\includegraphics[width=0.24\textwidth]{figures/Image1_Diffusion_4.png}}
% Iteration:
%\subfigure[Iteration 8]{\includegraphics[width=0.24\textwidth]{figures/Image1_Diffusion_8.png}}
% Iteration:
%\subfigure[Iteration 20]{\includegraphics[width=0.24\textwidth]{figures/Image1_Diffusion_20.png}}
%%
%
% Iteration:
%\subfigure[Iteration 32]{\includegraphics[width=0.24\textwidth]{figures/Image1_Diffusion_32.png}}
%
%\subfigure[Iteration 52]{\includegraphics[width=0.24\textwidth]{figures/Image1_Diffusion_52.png}}
% Iteration:
%\subfigure[Iteration 72]{\includegraphics[width=0.24\textwidth]{figures/Image1_Diffusion_72.png}}
% Iteration:
%\subfigure[Iteration 84]{\includegraphics[width=0.24\textwidth]{figures/Image1_Diffusion_84.png}}

%
%\caption{Diffusion process by RQUNet-VAE ($\alpha = 0.03$).
%A process eliminates small scale objects over iterations and converges to a mean of the original image, see (h). 
%}
%\label{fig:ImageDiffusionScheme3:2019059hdf}
%\end{center}
%\end{figure*}

% =======================================================

% -----
%

\vspace{-0.4cm}
\subsection{Image Decomposition Experiments}  \label{sec:ImageDecomposition}\vspace{-0.2cm}
%
%\paragraph{a. Image decomposition:}
%

This section evaluates how \RQUnetVAE Scheme 2 decomposes and subsequently denoises images. \RQUnetVAE Scheme 2 uses an iterative scheme to decompose an image into (i) a lowpass image, (ii) a bandpass image, and (iii) a highpass image. Denoising can be achieved by truncating highpass features. The RQUNet-VAE was applied to the Sentinel-2 image dataset described in Section~\ref{subsec:datasets}.
Figure~\ref{fig:ImageDecomposition} shows examples of the decomposition for four images. Figures~\ref{fig:ImageDecomposition}(a,e) shows the original Sentinel2 images of a typical rural landscape with forest cover, cultivated fields and small, distributed structures. The corresponding lowpass images in Figures~\ref{fig:ImageDecomposition}(b,f) captures the larger land cover parcels, without the fine-scale texture which has been removed. The corresponding bandpass in Figures~\ref{fig:ImageDecomposition}(c,g) captures most of the signal of texture, whereas the highpass image in Figure Figures~\ref{fig:ImageDecomposition}(d,h) captures very fine-scale texture, oscillating patterns, along with noise.

\begin{figure*}
\begin{center}  

% Frame 1:
% Original
\subfigure[Original frame 1]{\includegraphics[width=0.24\textwidth ]{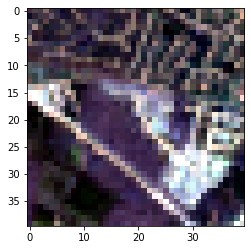}\vspace{-0.5cm}}
% Lowpass
\subfigure[Lowpass]{\includegraphics[width=0.24\textwidth]{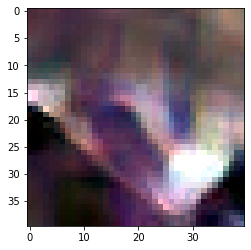}}
% Bandpass
\subfigure[Bandpass]{\includegraphics[width=0.24\textwidth]{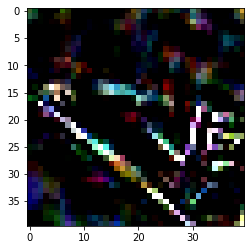}}
% Highpass
\subfigure[Highpass]{\includegraphics[width=0.24\textwidth]{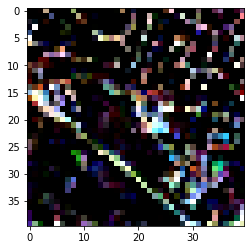}}

\vspace{-0.2cm}
% Frame 30:
% Original
\subfigure[Original frame 30]{\includegraphics[width=0.24\textwidth]{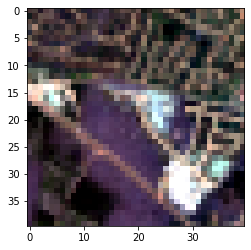}}
% Lowpass
\subfigure[Lowpass]{\includegraphics[width=0.24\textwidth]{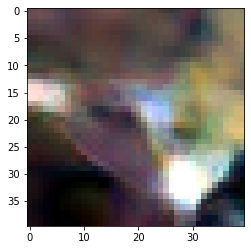}}
% Bandpass
\subfigure[Bandpass]{\includegraphics[width=0.24\textwidth]{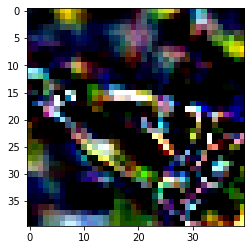}}
% Highpass
\subfigure[Highpass]{\includegraphics[width=0.24\textwidth]{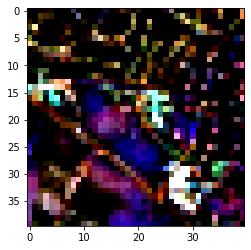}}

% % Frame 60:
% % Original
% \subfigure[Original frame 30]{\includegraphics[width=0.24\textwidth]{site296_frame60_1.png}}
% % Lowpass
% \subfigure[Lowpass]{\includegraphics[width=0.24\textwidth]{site296_frame60_4.png}}
% % Bandpass
% \subfigure[Bandpass]{\includegraphics[width=0.24\textwidth]{site296_frame60_5.png}}
% % Highpass
% \subfigure[Highpass]{\includegraphics[width=0.24\textwidth]{site296_frame60_6.png}}

% Frame 90:
% Original
\vspace{-0.2cm}
\subfigure[Original frame 60]{\includegraphics[width=0.24\textwidth]{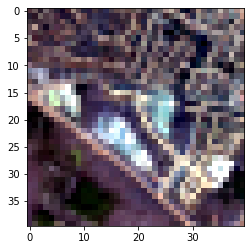}}
% Lowpass
\subfigure[Lowpass]{\includegraphics[width=0.24\textwidth]{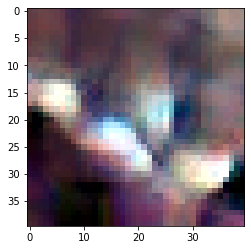}}
% Bandpass
\subfigure[Bandpass]{\includegraphics[width=0.24\textwidth]{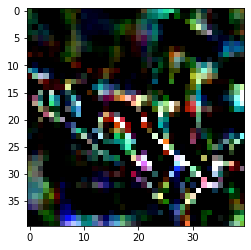}}
% Highpass
\subfigure[Highpass]{\includegraphics[width=0.24\textwidth]{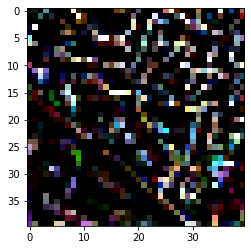}}

%\vspace{-0.2cm}
% Lowpass + Bandpass
% Frame 1:
%\subfigure[Lowpass + Bandpass for frame 1]{\includegraphics[width=0.25\textwidth]{figures/site296_frame1_2.png}}
% Frame 30:
%\subfigure[Lowpass + Bandpass for frame 30]{\includegraphics[width=0.25\textwidth]{figures/site296_frame30_2.png}}
% Frame 90:
%\subfigure[Lowpass + Bandpass for frame 90]{\includegraphics[width=0.25\textwidth]{figures/site296_frame60_2.png}}

% \subfigure[Histogram]{\includegraphics[width=0.32\textwidth]{site6616_frame1_3.png}}

% \includegraphics[width=1\textwidth]{Fig15.png}

\vspace{-0.2cm}
\caption{Time series decomposition by RQUNet-VAE with GIAF based spectral decomposition for site 296. %Threshold for histogram in spectral decomposition is 5, number of iterations is 10.
}
\label{fig:Video Decomposition:site296}\vspace{-0.7cm}
\end{center}
\end{figure*}

\paragraph{Time series decomposition}

The RQUNet-VAE was applied to a Sentinel2 time series decomposition by using Haar wavelet smoothing in time as a diffusion process, (Section~\ref{theory:shrinkage}). The proposed smoothing technique is therefore simultaneously applied in spatial domain (by RQUNet-VAE) and temporal domain (by Haar wavelet in time), following Algorithm 2 in SM. The smoothing parameter was set at $\alpha = 0.03$ for all decompositions.

The Sentinel2 time series was comprised as follows: $\left\{ \underline{\underline{\underline{f_i}}} \right\}_{i=1}^{80}$, $\underline{\underline{\underline{f_i}}} \in \mathbb R^{40 \times 40 \times 3 \times 99}$, where each of the 99 images have  three channels of image size $40 \times 40$.
The time series is then padded to $\underline{\underline{\underline{f_i}}} \in \mathbb R^{64 \times 64 \times 3 \times 120}$. The RQUNet-VAE is trained on each image in the padded dataset to obtain all unknown model parameters. An Adam optimizer was used to train RQUNet-VAE with 200 epochs and batch-size of 16.

The trained parameters are then applied to Algorithm 2 with 10 iterations for spectral decomposition. Note that Algorithm 2 is an extension of generalized intersection algorithms with fixpoints \cite{RichterThaiHuckemann2020} for video processing via spectral decomposition into lowpass, bandpass and highpass videos.

Figure~\ref{fig:Video Decomposition:site296} illustrates the time series decomposition with the RQUNet-VAE and diffusion process. The smoothing parameter was set at $\alpha = 0.03$ for all decompositions. Similar to image decomposition in the previous sections, this time series decomposition extracts a  homogeneous time series and a residual time series which contains small objects (e.g. roads), noise, and texture.

%Note that an environment change in these data is a serious problem in imaging. Indeed, a correct video has a very slow motion changes. Then, common time series techniques, e.g. time delay embedding, can remove these slow changes. However, satellite video does not have this property because each video frame is discretely captured during experiment years; then, all these frames are attached into a video. This causes significant changes in background and objects, i.e. noise, in video. 

% ======================================

% ------

\begin{figure*}[tbh]
\begin{center}  

% Frame 1:
% Original:
\subfigure[Noisy image 1]{\includegraphics[width=0.24\textwidth]{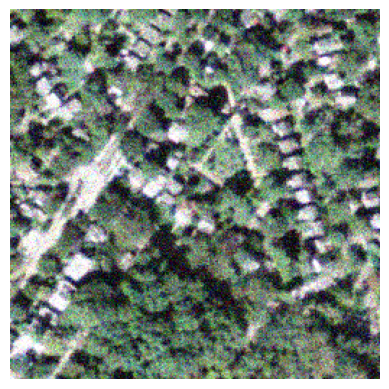}\label{fig:seg_noise_1}}
% groundtruth:
\subfigure[Ground truth]{\includegraphics[width=0.24\textwidth]{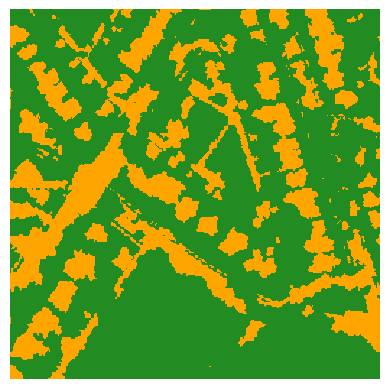}\label{fig:seg_gt_1}}
% RQUnet-VAE:
\subfigure[RQUNet-VAE]{\includegraphics[width=0.24\textwidth]{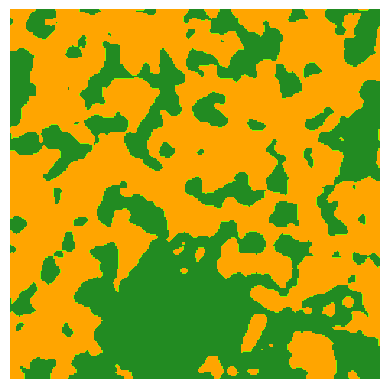}\label{fig:seg_rq_1}}
% Typical Unet:
\subfigure[Typical UNet]{\includegraphics[width=0.24\textwidth]{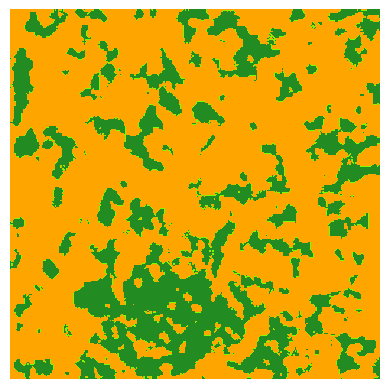}\label{fig:seg_unet_1}}
% % Confusiom matrix RQUnetVAE:
% \subfigure[Confusion matrix of RQUnet-VAE]{\includegraphics[width=0.24\textwidth]{Segmentation_conf_mat_rqunet_vae_2.png}}
% % Confusiom matrix typical Unet:
% \subfigure[Confusion matrix of typical Unet]{\includegraphics[width=0.24\textwidth]{Segmentation_conf_mat_base_2.png}}
%\color{red}{legend? of colors, prediction has light green where ground truth does not}
% -------

% Original:
\subfigure[Noisy image 2]{\includegraphics[width=0.24\textwidth]{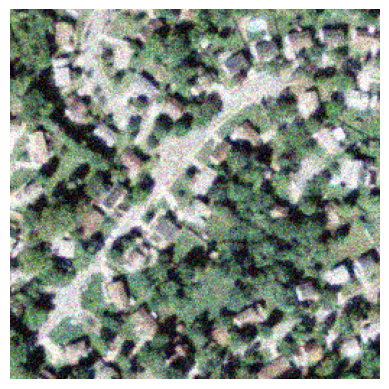}\label{fig:seg_noise_2}}
% groundtruth:
\subfigure[Ground truth]{\includegraphics[width=0.24\textwidth]{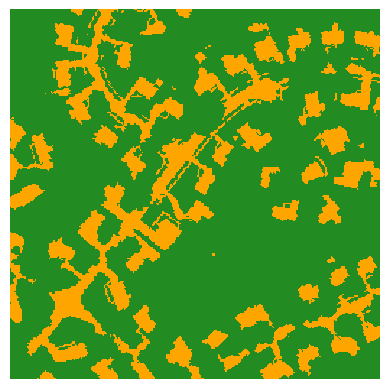}\label{fig:seg_gt_2}}
% RQUnet-VAE:
\subfigure[RQUNet-VAE]{\includegraphics[width=0.24\textwidth]{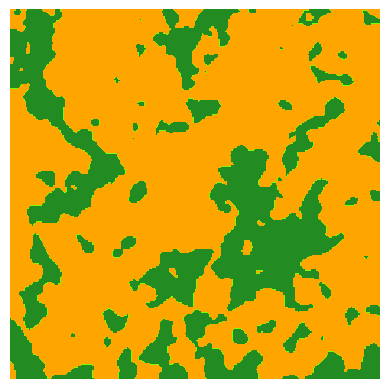}\label{fig:seg_rq_2}}
% Typical Unet:
\subfigure[Typical UNet]{\includegraphics[width=0.24\textwidth]{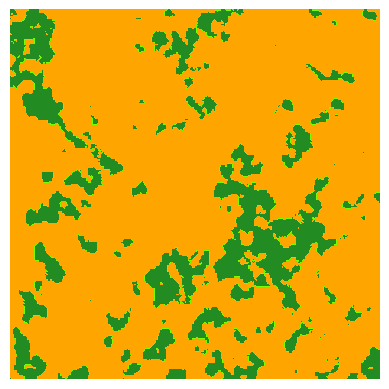}\label{fig:seg_unet_2}}
% % Confusiom matrix RQUnetVAE:
% \subfigure[Confusion matrix of RQUnet-VAE]{\includegraphics[width=0.24\textwidth]{Segmentation_conf_mat_rqunet_vae_3.png}}
% % Confusiom matrix typical Unet:
% \subfigure[Confusion matrix of typical Unet]{\includegraphics[width=0.24\textwidth]{Segmentation_conf_mat_base_3.png}}

% -------

\caption{Segmentation result with RQUNet-VAE as a smoothing term for noisy images. a,e) NAIP images with artificial noise added  ($\text{std} = 0.08$) as input to segmentation; b,f) ground-truth segmentation masks; (c,g) segmentation masks returned by our RQUNet-VAE with a smoothing parameter $\alpha = 0.5$; (d,h) segmentation mask returned by the typical UNet architecture \cite{ronneberger2015u}.}\vspace{-0.8cm}
%Our proposed method provides smooth segmented image over a typical UNet. 
%
%
%See figures \ref{fig:ImageSegmentation:moreimages} for more test images.
%\color{red}{The accuracy map mostly shows 0 or 1 and illustrates that green class / forest is completely underestimated - there is a problem with this image - try to find better example. describe confusion matrix! try with $\alpha = 0.1$}
\label{fig:ImageSegmentation:moreimages}
\end{center}
\end{figure*}

\vspace{-0.6cm}
\subsection{Image Segmentation by \RQUnetVAE  \label{sec:DenoisingSegmentation}}\vspace{-0.2cm}
The \RQUnetVAE was applied to the task of high resolution aerial image segmentation. The goal is to demonstrate that the \RQUnetVAE makes the image segmentation more robust to noise compared to the existing U-Net architecture. We first provide a brief background to the problem of image segmentation in Section~\ref{exp:image_seg_overview}. Then, in Section~\ref{exp:image_seg_experiments} we apply 1) the traditional U-Net architecture and 2) the \RQUnetVAE to the problem of image segmentation and show that \RQUnetVAE yields better segmentation results where artificial noise is added to images. This shows that our proposed \RQUnetVAE makes the existing U-Net architecture more robust to noise.

\vspace{-0.2cm}
\subsubsection{Background on Image Segmentation}\label{exp:image_seg_overview}
The conventional image segmentation procedure has the following steps,  (i) pre-processing to remove noise and unwanted small objects, (ii) segmentation, (iii) post-processing with morphological operators. This 3-step process requires a large number of parameters that need to be selected by an expert. Therefore, attempts have been made to incorporate a smoothing term into various models. The Mumford-Shah model \cite{MumfordShah1989} was proposed to introduce a smoothing term in a minimization approach for segmentation. However, computation of this smoothing term is an NP-hard problem and therefore not feasible \cite{MumfordShah1989}. Later, the Chan-Vese segmentation model \cite{ChanVese2001} was proposed that was solvable by relying on the level set method \cite{OsherSethian1988}. By adding a smoothing term in a variational formulae, this model successfully segments objects in a noisy background. 
Inspired by Chan-Vese model \cite{ChanVese2001}, we use RQUNet-VAE for segmentation by introducing a smoothing term, i.e. Riesz-Quincunx wavelet truncation, directly into the UNet-VAE. The advantage is that the \RQUnetVAE includes the smoothing term inside the Unet and thus combines (1) the ability of a Unet to learn features from an image, with (2) denoising capabilities enabled by the Riesz-Quincunx wavelet truncation. This smoothing term should eliminate small scale objects in a segmented image, e.g. texture and noise. Then, there is no need for separate pre- or post-processing steps, as these are all performed by the RQUNet-VAE segmentation.

\begin{table*}[t]
\centering
\caption{Quantitative Comparison between \RQUnetVAE and the traditional UNet architecture. Since \RQUnetVAE uses a variational (non-deterministic) approach, we provide mean and standard deviation aggregated over 20 runs. Precision, recall, and F1-Score are provided for both classes [vegetated, impervious].
\label{tab:segmentation}}
\begin{tabular}{|l|c|c|c|c|} 
 \hline
\multirow{2}{*}{Algorithm}  & Accuracy & Precision & Recall & F1-Score \\
& (overall) & (per class) & (per class) & (per class) \\
 \hline
 \RQUnetVAE Mean (Image 1) & 0.7057 & [0.9047 0.4498] & [0.5457 0.8660] & [0.6807  0.5921]  \\
 \RQUnetVAE Stdev (Image 1)   & $\pm$0.0013 & $\pm$[0.0008 0.0014] & $\pm$[0.0030 0.0016] & $\pm$[0.0024 0.0012] \\\hline
 Traditional UNet (Image 1) & 0.6640 & [0.9460 0.3960] & [0.3780 0.9500] & [0.5400  0.5590] \\\hline 
 \RQUnetVAE Mean (Image 2) & 0.6734 & [0.9837  0.3258] & [0.3662  0.9806] & [0.5337 0.4890] \\
 \RQUnetVAE Stdev (Image 2) & $\pm$0.0012 & $\pm$[0.0006 0.0009] & $\pm$[0.0027 0.0007] & $\pm$[0.0028 0.0010] \\\hline
 Traditional UNet (Image 2) & 0.6170 & [0.9860 0.2900] & [0.2440 0.9890] & [0.3920 0.4490] \\\hline 
\end{tabular}
\vspace{-0.4cm}
\label{table:1}
\end{table*}

\begin{figure*}[tbh]
\begin{center}  

% Frame 1:
% Experiment 1:
{\includegraphics[width=0.6\textwidth]{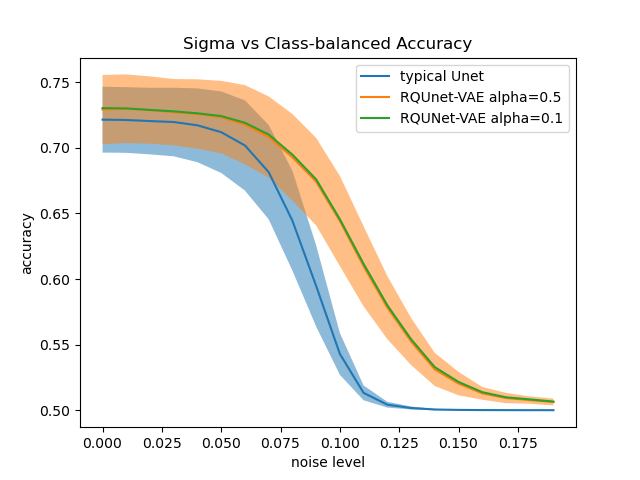}}\vspace{-0.2cm}

\caption{Outcome of segmentation experiments using the RQUNet-VAE and traditional UNet applied to 20 NAIP images with increasing added noise (0-0.20). The confidence intervals (blue and orange bands) are calculated from the mean and standard deviation of accuracy of 50 prediction iterations for each noise level and alpha value (0.5, 0.1) in the RQ scheme.  This illustrates the results of variational terms of RQUNet-VAE.\label{fig:exp:segm_accuracy}}\vspace{-0.7cm}
%\textcolor{red}{expand the description of a STD anf mean - green on top of orange}
\end{center}
\end{figure*}

\vspace{-0.2cm}
\subsubsection{Experimental Evaluation of \RQUnetVAE for Image Segmentation}\label{exp:image_seg_experiments}
We apply both 1) the traditional U-Net~\cite{ronneberger2015u} and 2) \RQUnetVAE to the problem of segmentation using the 20 NAIP images described in Section~\ref{subsec:datasets} using two classes of land cover, impervious and vegetated. Figure~\ref{fig:ImageSegmentation:moreimages} shows the segmentation results for two sample images after noise has been added (Section~\ref{subsec:denoising_experiments}, standard deviation of $0.08$). The corresponding ground truth land cover masks are shown in Figures~\ref{fig:seg_gt_1} and~\ref{fig:seg_gt_2}. 
% For both images, the resulting land cover segmentation using the \RQUnetVAE (Figures~\ref{fig:seg_rq_1} and~\ref{fig:seg_rq_2}) and, for comparison, using the traditional U-Net (Figures~\ref{fig:seg_unet_1} and~\ref{fig:seg_unet_2}). 
Table~\ref{tab:segmentation} provides quantitative segmentation results, including overall accuracy, precision, recall, and F1-score for both of the two classes. Due to the added noise, the segmentation accuracy of the U-Net is very low at $0.6640$ and $0.6170$ for the two images, respectively Table~\ref{tab:segmentation}. The precision and recall for the two classes (non-impervious and impervious), indicate that this low accuracy is attributed to a bias towards the impervious class, as suggested by the low recall for the non-impervious class (0.3780 for the traditional UNet). This implies that more than two out of five impervious pixels are incorrectly classified as vegetated. %We also observe a low precision in the impervious class (0.4498 for the traditional UNet), implying that more than half of the pixel that are classified as impervious are actually non-impervious. On the flipside, we observe a high recall for the impervious class (0.9500), showing that most impervious pixels are correctly classified as impervious. 

The performance of the traditional UNet is much lower with the added noise than on the original, which indicates that the addition of the noise substantially confuses the U-Net. In contrast, the \RQUnetVAE coped much better with the noise, yielding higher accuracy values of $0.7057 \pm 0.0013$ and $0.6634\pm 0.0012$ for the two images, respectively. The negative impact of noise on the segmentation accuracy is substantially reduced when using the \RQUnetVAE which makes the U-Net more robust to noise. Furthermore, the \RQUnetVAE is better able to discriminate the vegetated class, shown by the higher F1-scores Table~\ref{tab:segmentation}.

 %Figure~\ref{fig:ImageSegmentation:moreimages} however only shows two example images and uses a specific standard deviation of noise of $\sigma=0.08$. 
 A more comprehensive evaluation of \RQUnetVAE vs U-Net on 20 images with various levels of noise is given in Figure~\ref{fig:exp:segm_accuracy}. This figure shows the mean accuracy (solid line) of the U-Net and the \RQUnetVAE across all 20 images. 
 %In addition, minimum and maximum accuracy are shown ?? bands.
 When using the original image with no noise ($\sigma=0$), the \RQUnetVAE only provides a marginal improvement over U-Net. As the noise level $\sigma$ is increased, both \RQUnetVAE and U-Net exhibit a drop in accuracy. However, that the drop in accuracy is substantially slower for \RQUnetVAE. 
 In particular, for noise values around $\sigma=0.1$, the \RQUnetVAE has up to 10\% higher accuracy than U-Net.
 This improvement stems from the variational auto-encoder approach, which, through variational changes to the latent representation of an image, allows the identification of pixels with a high probability of belonging to one class while being assigned to another class by the deterministic U-Net.
 Figure~\ref{fig:exp:segm_accuracy} gives the results using two parameters setting of \RQUnetVAE, $\alpha=0.5$ and $\alpha=1.0$. In both cases the resulting accuracy is nearly identical (the two lines are almost perfectly on top of each other), showing that the the performance is not very sensitive to the choice of $\alpha$ and thus that \RQUnetVAE is robust to non-optimal choices of $\alpha$.
In summary, Figure~\ref{fig:exp:segm_accuracy} shows that our proposed \RQUnetVAE is much more robust to Gaussian noise added to an image, as the reduction of segmentation accuracy with higher noise levels, is less pronounced.

\vspace{-0.5cm}
\section{Conclusion}  \label{sec:Conclusion}\vspace{-0.2cm}
In this work, we introduced the \RQUnetVAE which augments the existing UNet architecture with a generalized wavelet expansion approach that we extended to a diffusion process to enable spectral decomposition. To the best of our knowledge this is the first approach that enables image decomposition and image denoising using a variational variant of the UNet architecture. 
An important application of this decomposition is denoising, achieved by truncating highpass features, that is, discarding information of decomposed images having the highest variance. We apply our proposed \RQUnetVAE to image denoising and segmentation of multi-band satellite images and their time-series which often contain noise due to multiple causes. 
During quantitative comparisons of noise reduction the RQUNet-VAE yields the highest PSNR and SSIM of all competitor methods. Among all the approaches based on harmonic analysis, RQUNet-VAE Scheme 2 yields the best results. For the application of satellite image denoising our proposed \RQUnetVAE provides superior qualitative performance compared to other competitors. The denoising by RQUNet-VAE Scheme~1 was visually compared to the noise reduction using a Riesz Dyadic wavelet and curvelets and using wavelet CDF and it provided the best combination of noise reduction and delineation of edges and objects. Furthermore the propose RQUNet-VAE Scheme~2, was compared to other iterative methods including GIAF and  directional TV-L2 and it resulted in the best trade-off between reduction of noise and delineation of edges and objects, whereas the GIAF oversmoothed edges between objects while TV-L2  clearly retained visible noise.

To quantitatively measure the improvement of the \RQUnetVAE against the traditional UNet, we applied it to high resolution aerial image segmentation. Our experiments show only a slight improvement over the traditional UNet for segmantation of the original images with little noise. However, as artificial noise is added to images, we observe that the segmentation quality of the UNet decreases more rapidly than when using the \RQUnetVAE. The superior performance of \RQUnetVAE is due to a neural network in UNet-VAE with a learned dictionary obtained from the training procedure to increase the level of sparsity for an input image in a generalized Besov space. This property is due to a fundamental concept in signal processing, that the signal is sparse in some transformed domain. This demonstrates that the \RQUnetVAE architecture is able to substantially increase the robustness to noise compared to the traditional UNet for the application of satellite image segmentation.

% ========================================

%
\vspace{-0.5cm}
\section{Acknowledgments}  \label{sec:Acknowledgment}\vspace{-0.2cm}
This research is based upon work supported in part by the Office of the Director of National Intelligence (ODNI), Intelligence Advanced Research Projects Activity (IARPA), Space-based Machine Automated Recognition Technique (SMART) program, via contract 2021-20111000003. The views and conclusions contained herein are those of the authors and should not be interpreted as necessarily representing the official policies, either expressed or implied, of ODNI, IARPA, or the U.S. Government. The U.S. Government is authorized to reproduce and distribute reprints for governmental purposes notwithstanding any copyright annotation therein. The research was conducted in collaboration with BlackSky. We thank Diego Torrejon and Isaac Corley of BlackSky for valuable comments on the manuscript. 
\vspace{-0.4cm}

% Bibliography: ------------------

{\scriptsize
\bibliography{Current}

\begin{thebibliography}{10}

\bibitem{allenby2013implementing}
J.~Allenby and C.~Phelan.
\newblock Implementing technology and precision conservation in the
  {Chesapeake} bay.
\newblock {\em Chesapeake Conservancy}, 2013.

\bibitem{CaiChanShen2008}
J.-F. Cai, R.~H. Chan, and Z.~Shen.
\newblock A framelet-based image inpainting algorithm.
\newblock {\em Applied and Computational Harmonic Analysis}, 24(2):131--149,
  2008.

\bibitem{CandesDonoho2004}
E.~J. Cand{\`e}s and D.~L. Donoho.
\newblock New tight frames of curvelets and optimal representations of objects
  with piecewise c2 singularities.
\newblock {\em Communications on Pure and Applied Mathematics: A Journal Issued
  by the Courant Institute of Mathematical Sciences}, 57(2):219--266, 2004.

\bibitem{ChanVese2001}
T.~F. Chan and L.~A. Vese.
\newblock Active contours without edges.
\newblock {\em IEEE Transactions on image processing}, 10(2):266--277, 2001.

\bibitem{claverie2018harmonized}
M.~Claverie, J.~Ju, J.~G. Masek, J.~L. Dungan, E.~F. Vermote, J.-C. Roger,
  S.~V. Skakun, and C.~Justice.
\newblock The harmonized {Landsat} and {Sentinel-2} surface reflectance data
  set.
\newblock {\em Remote sensing of environment}, 219:145--161, 2018.

\bibitem{daubechies1998factoring}
I.~Daubechies and W.~Sweldens.
\newblock Factoring wavelet transforms into lifting steps.
\newblock {\em Journal of Fourier analysis and applications}, 4(3):247--269,
  1998.

\bibitem{esser2018variational}
P.~Esser, E.~Sutter, and B.~Ommer.
\newblock A variational {U-Net} for conditional appearance and shape
  generation.
\newblock In {\em Proceedings of the IEEE conference on computer vision and
  pattern recognition}, pages 8857--8866, 2018.

\bibitem{FeilnerVilleUnser2005}
M.~Feilner, D.~Van De~Ville, and M.~Unser.
\newblock An orthogonal family of {Quincunx} wavelets with continuously
  adjustable order.
\newblock {\em IEEE Transactions on Image Processing}, 14(4):499--510, 2005.

\bibitem{Gilboa2014}
G.~Gilboa.
\newblock A total variation spectral framework for scale and texture analysis.
\newblock {\em SIAM journal on Imaging Sciences}, 7(4):1937--1961, 2014.

\bibitem{GoldsteinOsher2009}
T.~Goldstein and S.~Osher.
\newblock The split {Bregman} method for $\text{L}_1$-regularized problems.
\newblock {\em SIAM journal on imaging sciences}, 2(2):323--343, 2009.

\bibitem{IoffeSzegedy2015}
S.~Ioffe and C.~Szegedy.
\newblock Batch normalization: Accelerating deep network training by reducing
  internal covariate shift.
\newblock In {\em International conference on machine learning}, pages
  448--456. PMLR, 2015.

\bibitem{kingma2014adam}
D.~P. Kingma and J.~Ba.
\newblock Adam: A method for stochastic optimization.
\newblock {\em arXiv preprint arXiv:1412.6980}, 2014.

\bibitem{kingma2013auto}
D.~P. Kingma and M.~Welling.
\newblock Auto-encoding variational bayes.
\newblock {\em arXiv preprint arXiv:1312.6114}, 2013.

\bibitem{kleynhans2012land}
W.~Kleynhans, B.~P. Salmon, J.~C. Olivier, F.~Van~den Bergh, K.~J. Wessels,
  T.~L. Grobler, and K.~C. Steenkamp.
\newblock Land cover change detection using autocorrelation analysis on {MODIS}
  time-series data: Detection of new human settlements in the gauteng province
  of south africa.
\newblock {\em IEEE Journal of selected topics in applied earth observations
  and remote sensing}, 5(3):777--783, 2012.

\bibitem{kohl2018probabilistic}
S.~Kohl, B.~Romera-Paredes, C.~Meyer, J.~De~Fauw, J.~R. Ledsam, K.~Maier-Hein,
  S.~Eslami, D.~Jimenez~Rezende, and O.~Ronneberger.
\newblock A probabilistic {U-Net} for segmentation of ambiguous images.
\newblock {\em Advances in neural information processing systems}, 31, 2018.

\bibitem{li2017global}
J.~Li and D.~P. Roy.
\newblock {A global analysis of Sentinel-2A, Sentinel-2B and Landsat-8 data
  revisit intervals and implications for terrestrial monitoring}.
\newblock {\em Remote Sensing}, 9(9):902, 2017.

\bibitem{MumfordShah1989}
D.~B. Mumford and J.~Shah.
\newblock Optimal approximations by piecewise smooth functions and associated
  variational problems.
\newblock {\em Communications on pure and applied mathematics}, 1989.

\bibitem{national2012national}
NAIP.
\newblock National agriculture imagery program ({NAIP}) information sheet,
  2015.
\newblock
  https://www.fsa.usda.gov/Internet/FSA\_File/naip\_info\_sheet\_2015.pdf.

\bibitem{OsherSethian1988}
S.~Osher and J.~A. Sethian.
\newblock Fronts propagating with curvature-dependent speed: Algorithms based
  on hamilton-jacobi formulations.
\newblock {\em Journal of computational physics}, 79(1):12--49, 1988.

\bibitem{partington1988introduction}
J.~R. Partington, J.~R. Partington, et~al.
\newblock {\em An introduction to Hankel operators}.
\newblock Cambridge University Press, 1988.

\bibitem{PolsonScottWillard2015}
N.~G. Polson, J.~G. Scott, and B.~T. Willard.
\newblock Proximal algorithms in statistics and machine learning.
\newblock {\em Statistical Science}, 30(4):559--581, 2015.

\bibitem{pu2016variational}
Y.~Pu, Z.~Gan, R.~Henao, X.~Yuan, C.~Li, A.~Stevens, and L.~Carin.
\newblock Variational autoencoder for deep learning of images, labels and
  captions.
\newblock {\em Advances in neural information processing systems}, 29, 2016.

\bibitem{RichterThaiHuckemann2020}
R.~Richter, D.~H. Thai, and S.~F. Huckemann.
\newblock Generalized intersection algorithms with fixpoints for image
  decomposition learning.
\newblock {\em CoRR}, abs/2010.08661, 2020.

\bibitem{ronneberger2015u}
O.~Ronneberger, P.~Fischer, and T.~Brox.
\newblock U-net: Convolutional networks for biomedical image segmentation.
\newblock In {\em International Conference on Medical image computing and
  computer-assisted intervention}, pages 234--241. Springer, 2015.

\bibitem{roy2019landsat}
D.~P. Roy, H.~Huang, L.~Boschetti, L.~Giglio, L.~Yan, H.~H. Zhang, and Z.~Li.
\newblock {Landsat-8 and Sentinel-2 burned area mapping-A combined sensor
  multi-temporal change detection approach}.
\newblock {\em Remote Sensing of Environment}, 231:111254, 2019.

\bibitem{roy2005prototyping}
D.~P. Roy, Y.~Jin, P.~Lewis, and C.~Justice.
\newblock Prototyping a global algorithm for systematic fire-affected area
  mapping using {MODIS} time series data.
\newblock {\em Remote sensing of environment}, 97(2):137--162, 2005.

\bibitem{RudinOsherFatemi1992}
L.~I. Rudin, S.~Osher, and E.~Fatemi.
\newblock Nonlinear total variation based noise removal algorithms.
\newblock {\em Physica D: nonlinear phenomena}, 60(1-4):259--268, 1992.

\bibitem{salmon2013land}
B.~P. Salmon, W.~Kleynhans, F.~Van~den Bergh, J.~C. Olivier, T.~L. Grobler, and
  K.~J. Wessels.
\newblock Land cover change detection using the internal covariance matrix of
  the extended {Kalman} filter over multiple spectral bands.
\newblock {\em IEEE Journal of selected topics in applied earth observations
  and remote sensing}, 6(3):1079--1085, 2013.

\bibitem{SrivastavaHintonKrizhevskySutskeverSalakhutdinov2014}
N.~Srivastava, G.~Hinton, A.~Krizhevsky, I.~Sutskever, and R.~Salakhutdinov.
\newblock Dropout: a simple way to prevent neural networks from overfitting.
\newblock {\em The journal of machine learning research}, 15(1):1929--1958,
  2014.

\bibitem{Sweldens1997}
W.~Sweldens.
\newblock The lifting scheme: A construction of second generation wavelets.
\newblock {\em SIAM journal on mathematical analysis}, 29(2):511--546, 1998.

\bibitem{ThaiGottschlich2016DG3PD}
D.~H. Thai and C.~Gottschlich.
\newblock Directional global three-part image decomposition.
\newblock {\em EURASIP Journal on Image and Video Processing}, 2016(1):1--20,
  2016.

\bibitem{tian2020deep}
C.~Tian, L.~Fei, W.~Zheng, Y.~Xu, W.~Zuo, and C.-W. Lin.
\newblock Deep learning on image denoising: An overview.
\newblock {\em Neural Networks}, 131:251--275, 2020.

\bibitem{UnserBlu2003}
M.~Unser and T.~Blu.
\newblock Mathematical properties of the jpeg2000 wavelet filters.
\newblock {\em IEEE transactions on image processing}, 12(9):1080--1090, 2003.

\bibitem{UnserChenouardVandeville2011}
M.~Unser, N.~Chenouard, and D.~Van De~Ville.
\newblock Steerable pyramids and tight wavelet frames in ${L}_2(\mathbb{R}^d)$.
\newblock {\em {IEEE} Transactions on Image Processing}, 20(10):2705--2721,
  Oct. 2011.

\bibitem{UnserSageVandeville2009}
M.~Unser, D.~Sage, and D.~Van De~Ville.
\newblock Multiresolution monogenic signal analysis using the {R}iesz-{L}aplace
  wavelet transform.
\newblock {\em {IEEE} Transactions on Image Processing}, 18(11):2402--2418,
  Nov. 2009.

\bibitem{UnserVandeville2010}
M.~Unser and D.~Van De~Ville.
\newblock Wavelet steerability and the higher-order {R}iesz transform.
\newblock {\em {IEEE} Transactions on Image Processing}, 19(3):636--652, Mar.
  2010.

\bibitem{van2012hitempo}
F.~Van Den~Bergh, K.~J. Wessels, S.~Miteff, T.~L. Van~Zyl, A.~D. Gazendam, and
  A.~K. Bachoo.
\newblock Hitempo: A platform for time-series analysis of remote-sensing
  satellite data in a high-performance computing environment.
\newblock {\em International Journal of Remote Sensing}, 33(15):4720--4740,
  2012.

\bibitem{van2014renyi}
T.~Van~Erven and P.~Harremos.
\newblock R{\'e}nyi divergence and {Kullback-Leibler} divergence.
\newblock {\em IEEE Transactions on Information Theory}, 60(7):3797--3820,
  2014.

\bibitem{vermote2016preliminary}
E.~Vermote, C.~Justice, M.~Claverie, and B.~Franch.
\newblock Preliminary analysis of the performance of the {Landsat 8/OLI} land
  surface reflectance product.
\newblock {\em Remote Sensing of Environment}, 185:46--56, 2016.

\bibitem{VetterliMarzilianoBlu2002}
M.~Vetterli, P.~Marziliano, and T.~Blu.
\newblock Sampling signals with finite rate of innovation.
\newblock {\em IEEE transactions on Signal Processing}, 50(6):1417--1428, 2002.

\bibitem{WangBovik2009}
Z.~Wang and A.~C. Bovik.
\newblock Mean squared error: Love it or leave it? a new look at signal
  fidelity measures.
\newblock {\em IEEE signal processing magazine}, 26(1):98--117, 2009.

\bibitem{woodcock2020transitioning}
C.~E. Woodcock, T.~R. Loveland, M.~Herold, and M.~E. Bauer.
\newblock Transitioning from change detection to monitoring with remote
  sensing: A paradigm shift.
\newblock {\em Remote Sensing of Environment}, 238:111558, 2020.

\bibitem{YeHanCha2018}
J.~C. Ye, Y.~Han, and E.~Cha.
\newblock Deep convolutional framelets: A general deep learning framework for
  inverse problems.
\newblock {\em SIAM Journal on Imaging Sciences}, 11(2):991--1048, 2018.

\bibitem{YinGaoLuDaubechies2017}
R.~Yin, T.~Gao, Y.~M. Lu, and I.~Daubechies.
\newblock A tale of two bases: Local-nonlocal regularization on image patches
  with convolution framelets.
\newblock {\em SIAM Journal on Imaging Sciences}, 10(2):711--750, 2017.

\bibitem{zhang2018characterization}
H.~K. Zhang, D.~P. Roy, L.~Yan, Z.~Li, H.~Huang, E.~Vermote, S.~Skakun, and
  J.-C. Roger.
\newblock {Characterization of Sentinel-2A and Landsat-8 top of atmosphere,
  surface, and nadir BRDF adjusted reflectance and NDVI differences}.
\newblock {\em Remote sensing of environment}, 215:482--494, 2018.

\bibitem{zhu2017change}
Z.~Zhu.
\newblock Change detection using {Landsat} time series: A review of
  frequencies, preprocessing, algorithms, and applications.
\newblock {\em ISPRS Journal of Photogrammetry and Remote Sensing},
  130:370--384, 2017.

\end{thebibliography}

\bibliographystyle{abbrv}
}

\end{document}

\maketitle

\author{Duy~H.~Thai$^1$,
        Xiqi~Fei$^1$,
        Tri~M.~Le$^1$,
        Andreas~Z\"ufle$^2$,
        Konrad Wessels$^1$\\
        
        $^1$George Mason University, Department of Geography and Geoinformation Science, USA
        
        $^2$Emory University, Department of Computer Science, USA
 }

\section{Additional Images}
This section provides details on the ten satellite images used for testing of our proposed RQUNet-VAE~approach. 
Figure~\ref{fig:OriginalTestImages:19059hdf} shows the original images, Figure~\ref{fig:NoisyTestImagesvar0_04:19059hdf} shows the same images with artificial noise added, and Figure~\ref{fig:DenoisedTestImagesRQUnetVAEScheme3var0_04:19059hdf} shows the denoised images using our proposed RQUNet-VAE.

\begin{figure}[tbh]
\begin{center}  

 {\includegraphics[width=0.24\textwidth]{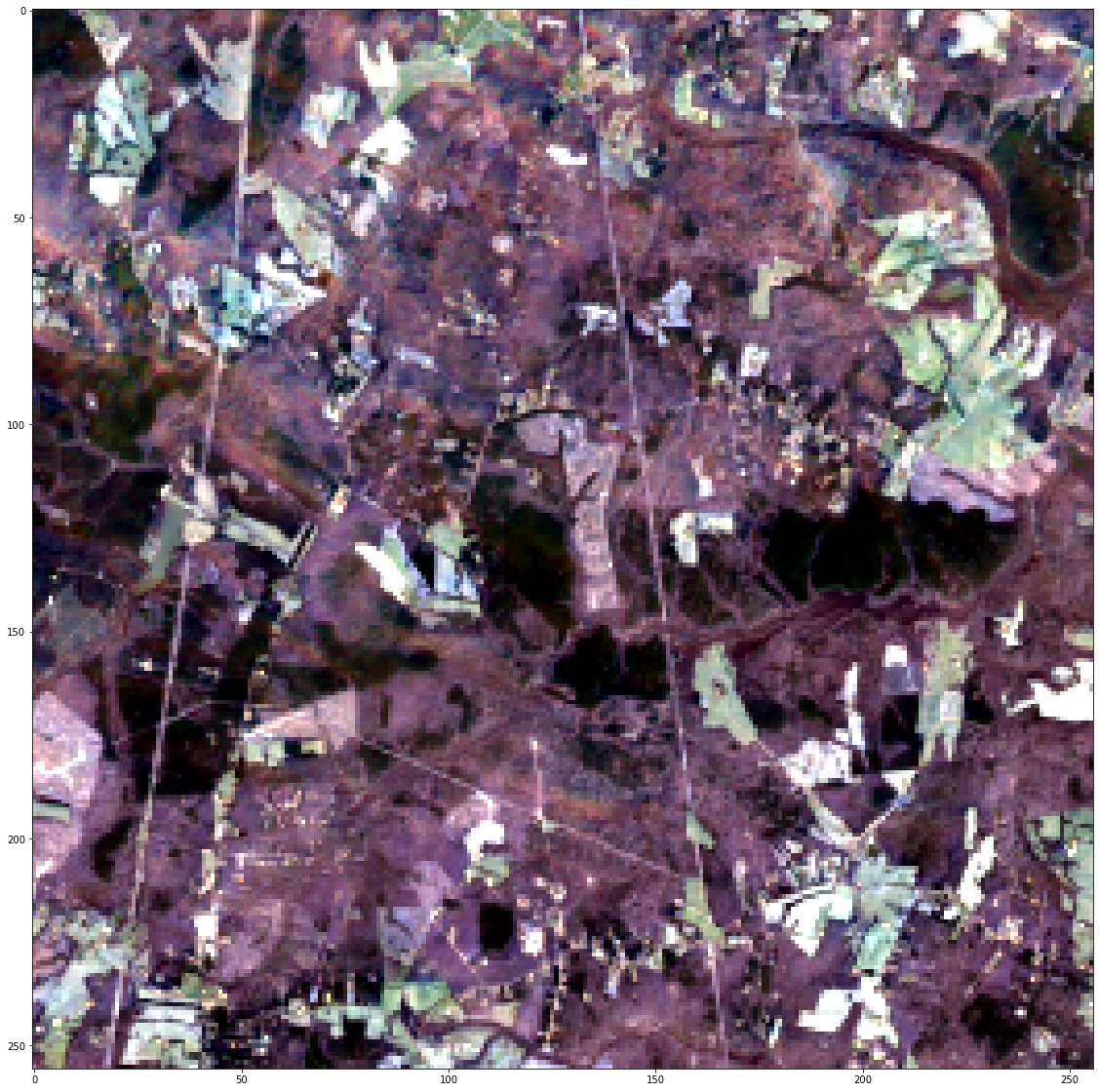}}
 {\includegraphics[width=0.24\textwidth]{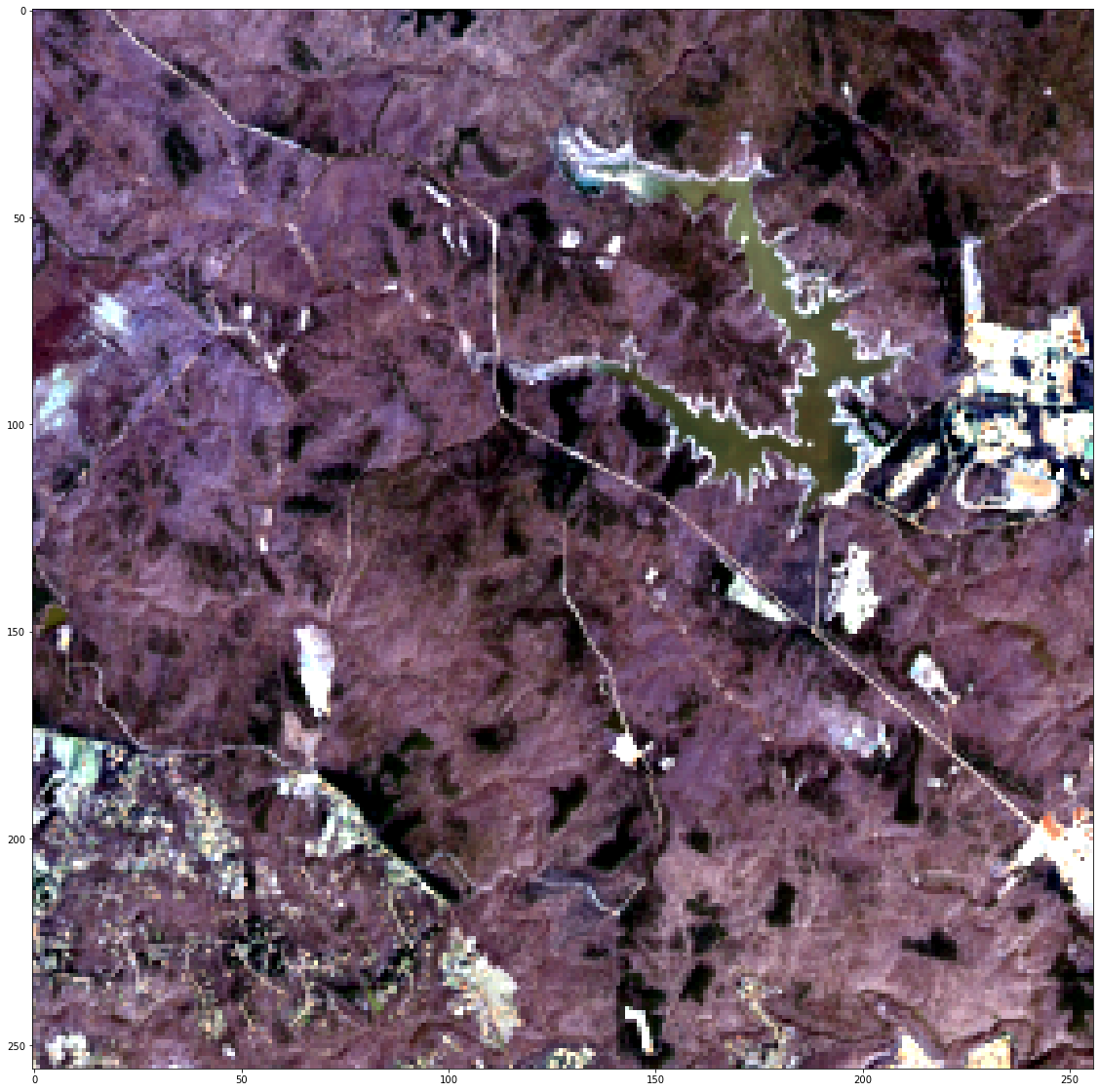}}
 {\includegraphics[width=0.24\textwidth]{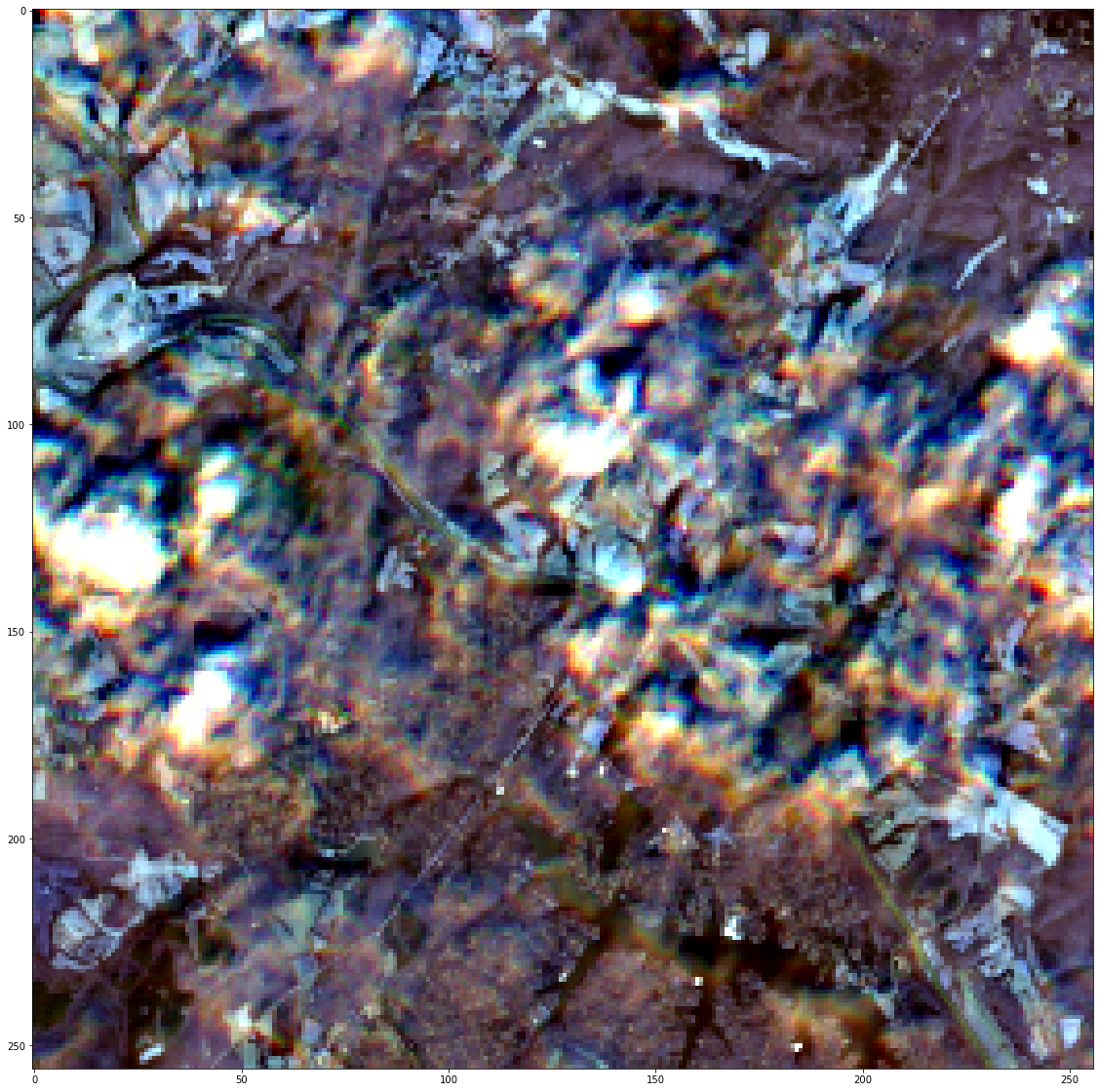}}
 {\includegraphics[width=0.24\textwidth]{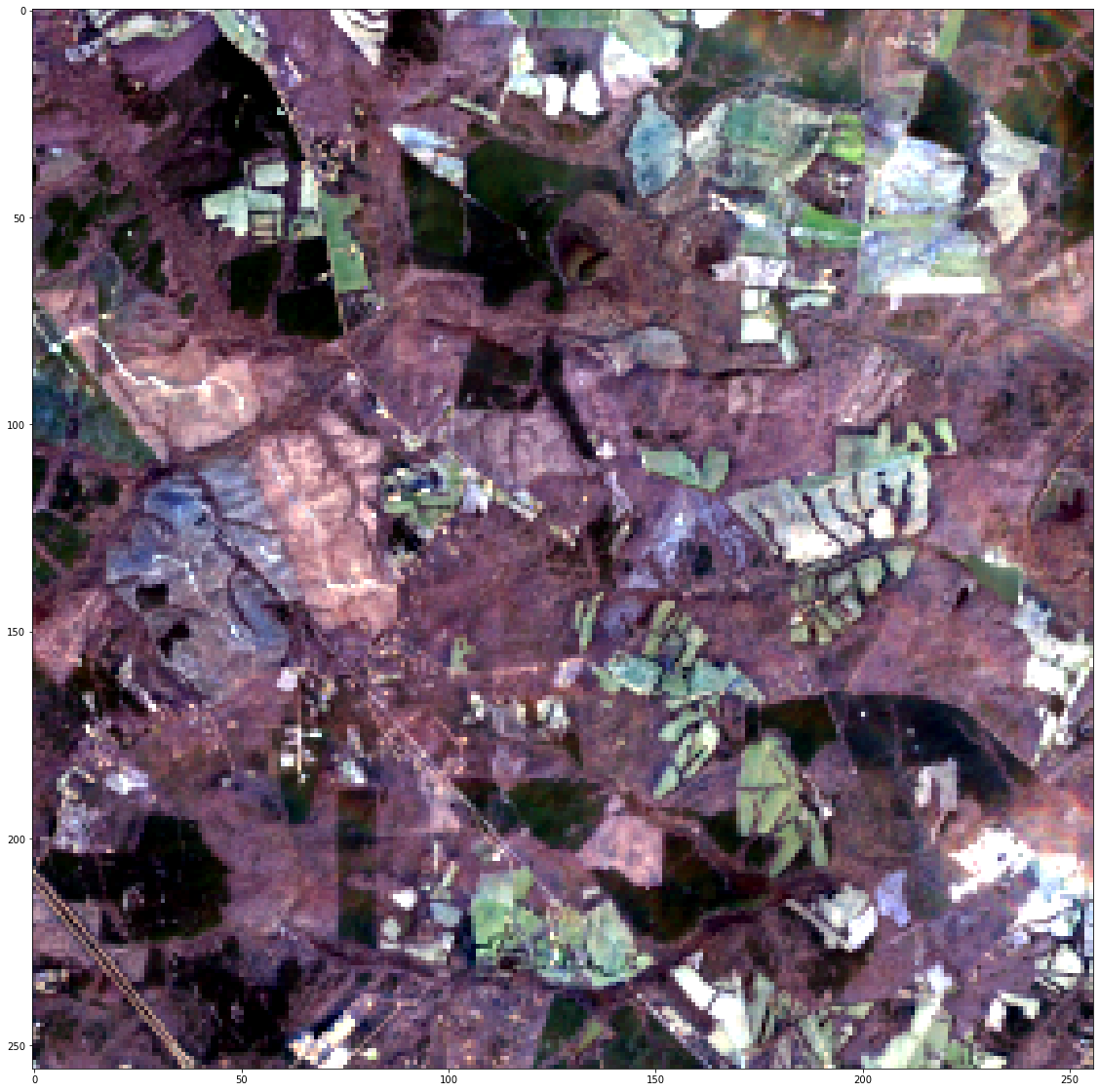}}
 {\includegraphics[width=0.24\textwidth]{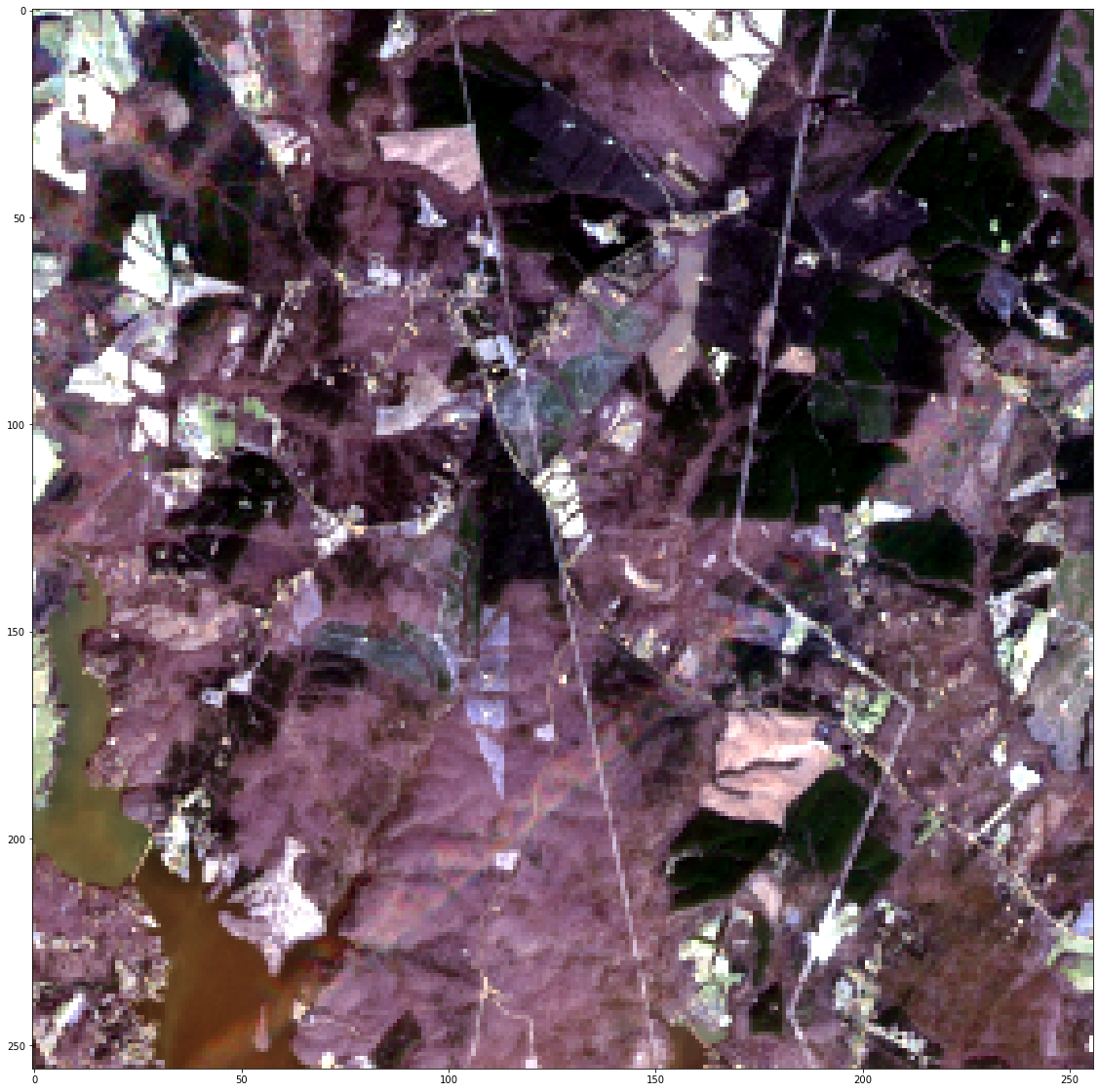}}
 {\includegraphics[width=0.24\textwidth]{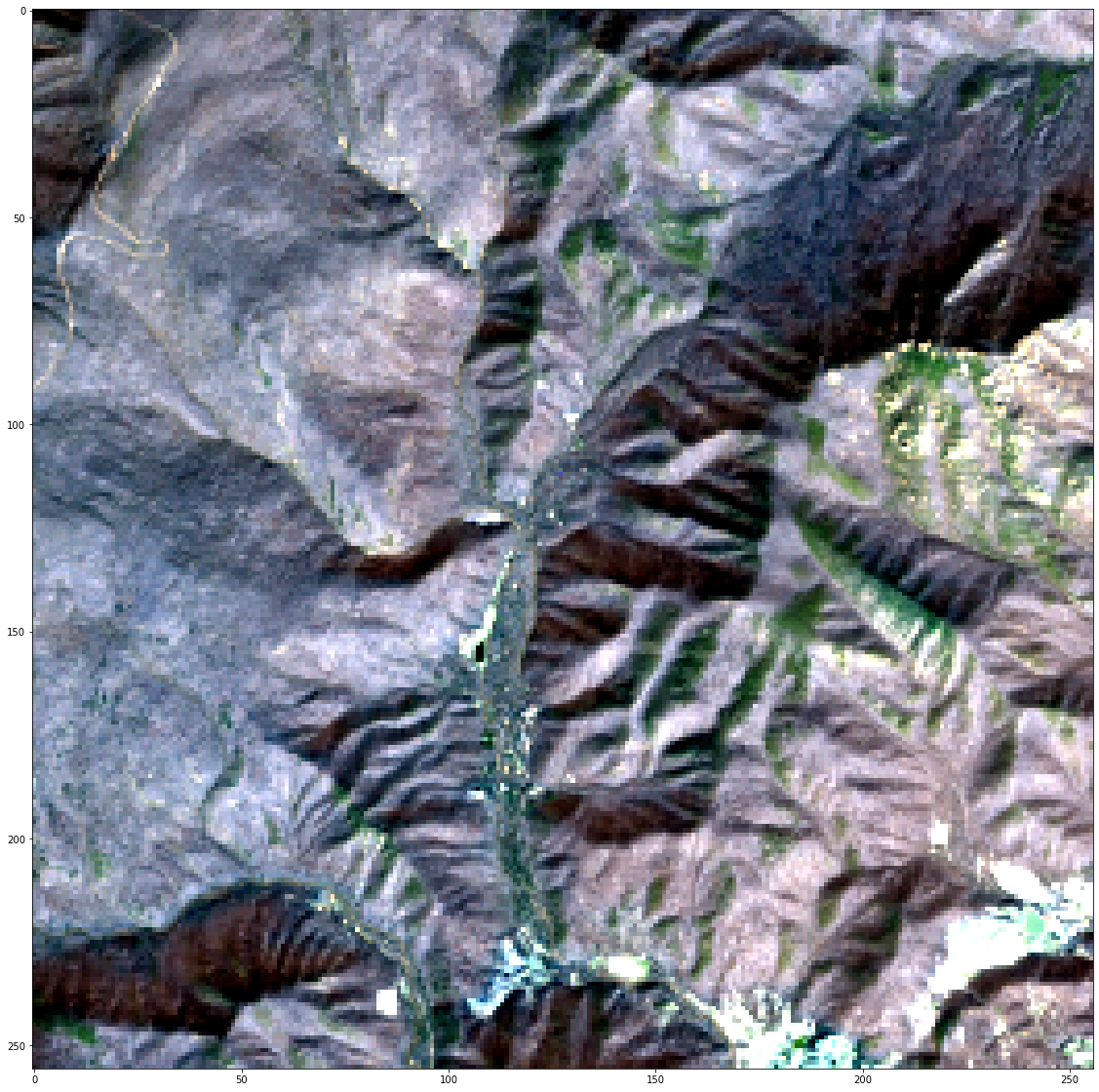}}
 {\includegraphics[width=0.24\textwidth]{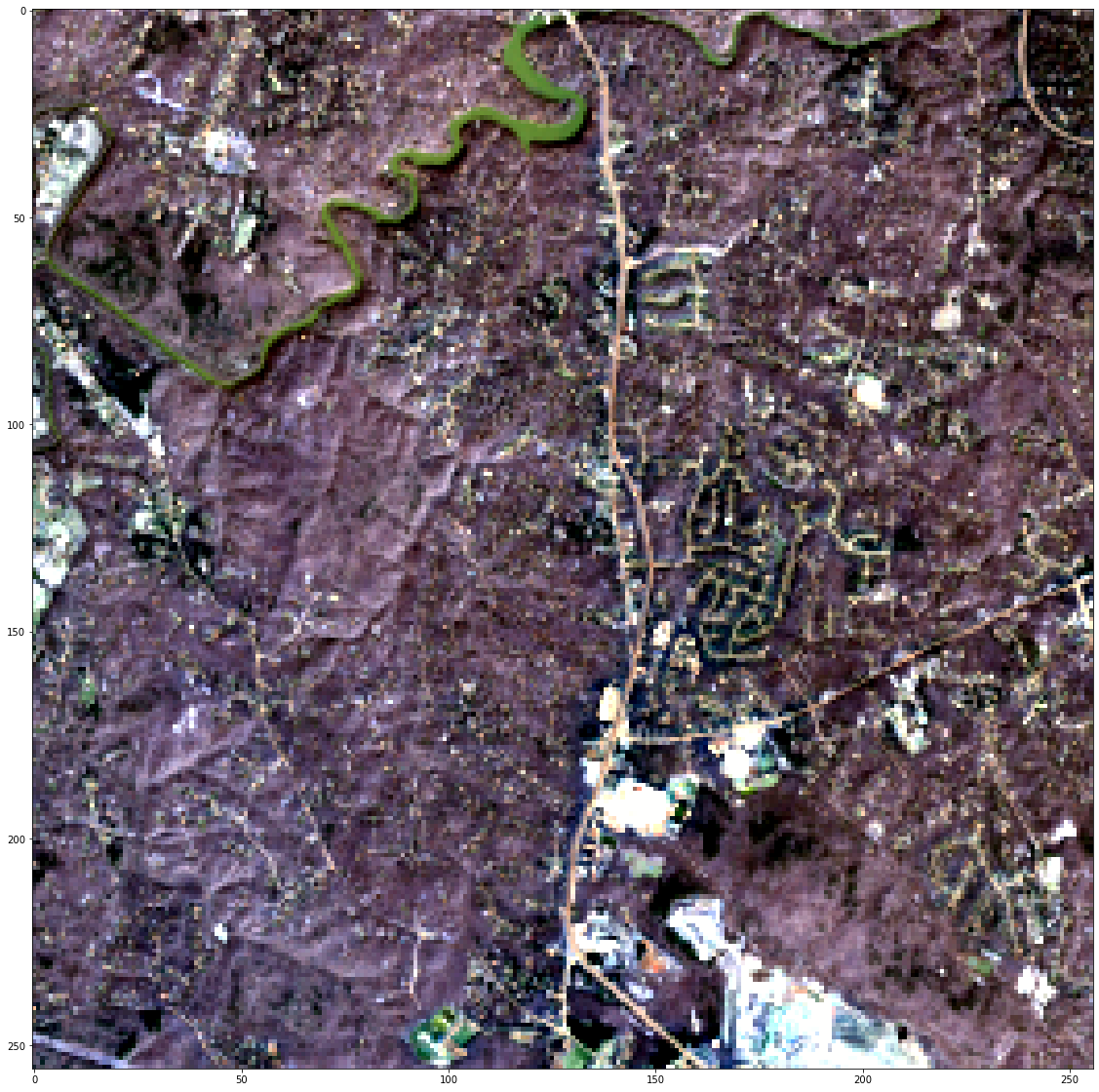}}
 {\includegraphics[width=0.24\textwidth]{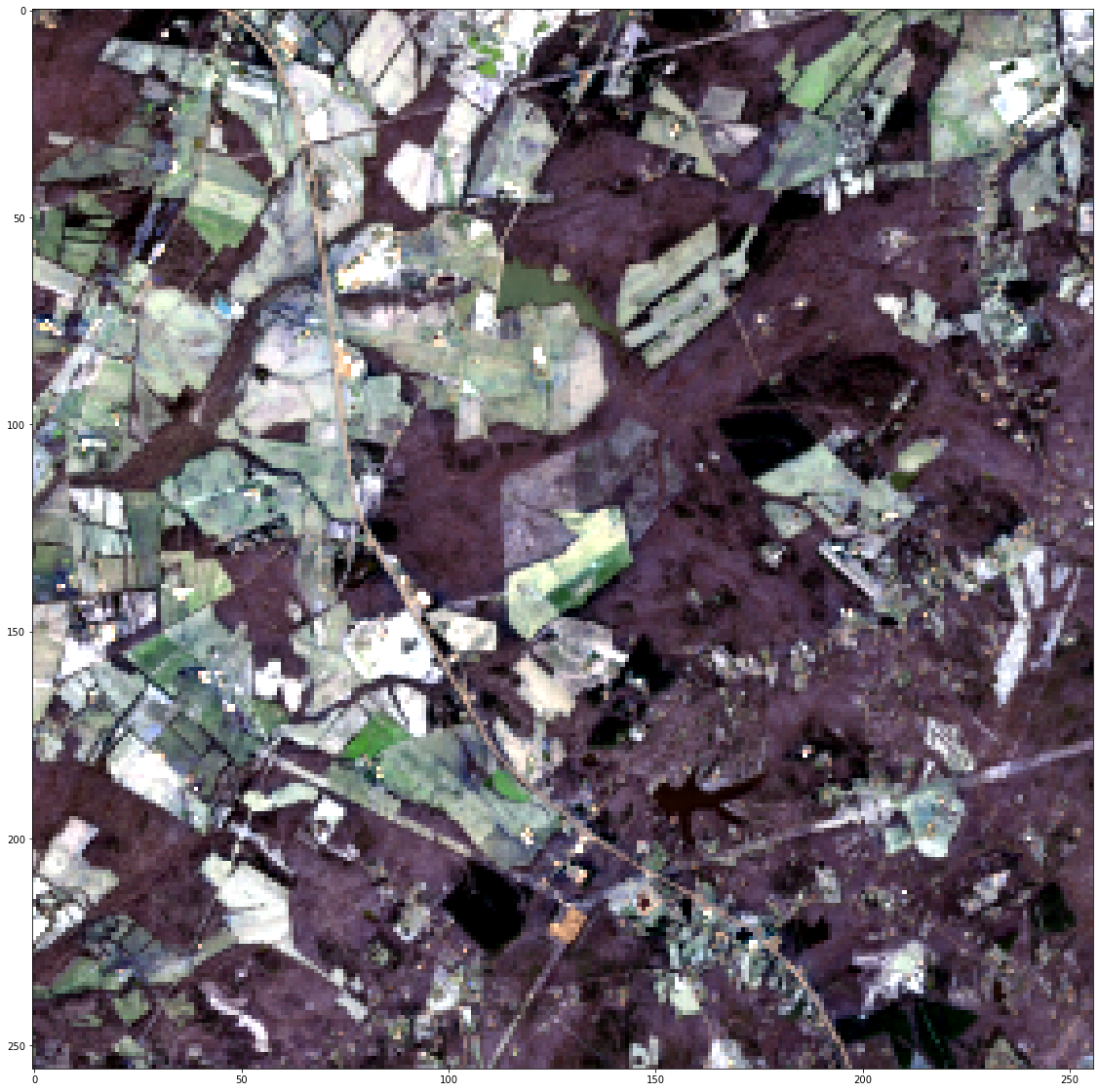}}
 {\includegraphics[width=0.24\textwidth]{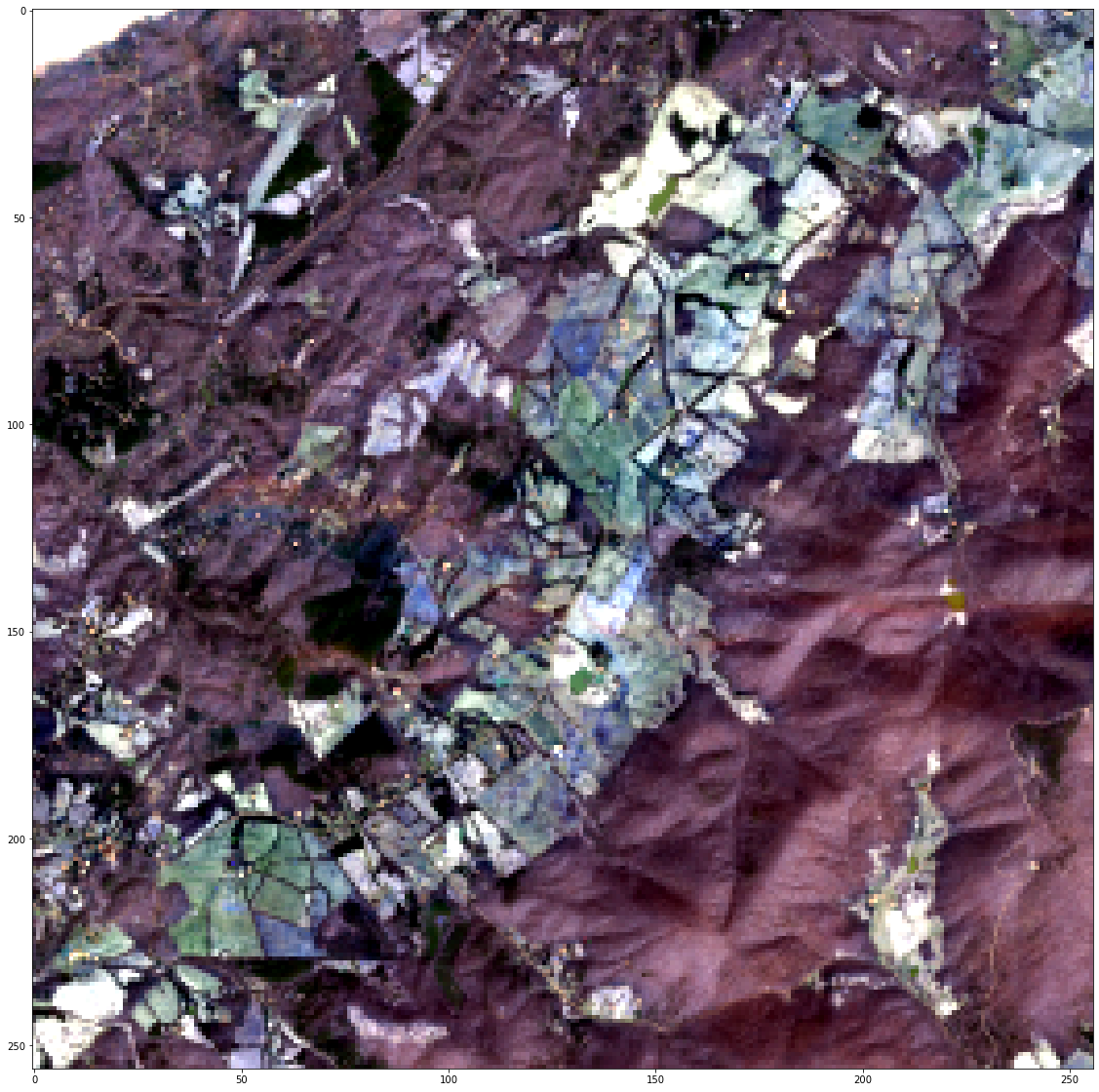}}
 {\includegraphics[width=0.24\textwidth]{figures/OriginalImage9.png}}\vspace{-0.4cm}

\caption{Original images. 
}\vspace{-0.8cm} 
\label{fig:OriginalTestImages:19059hdf}
\end{center}
\end{figure}

\begin{figure}[tbh]
\begin{center}  

 {\includegraphics[width=0.22\textwidth]{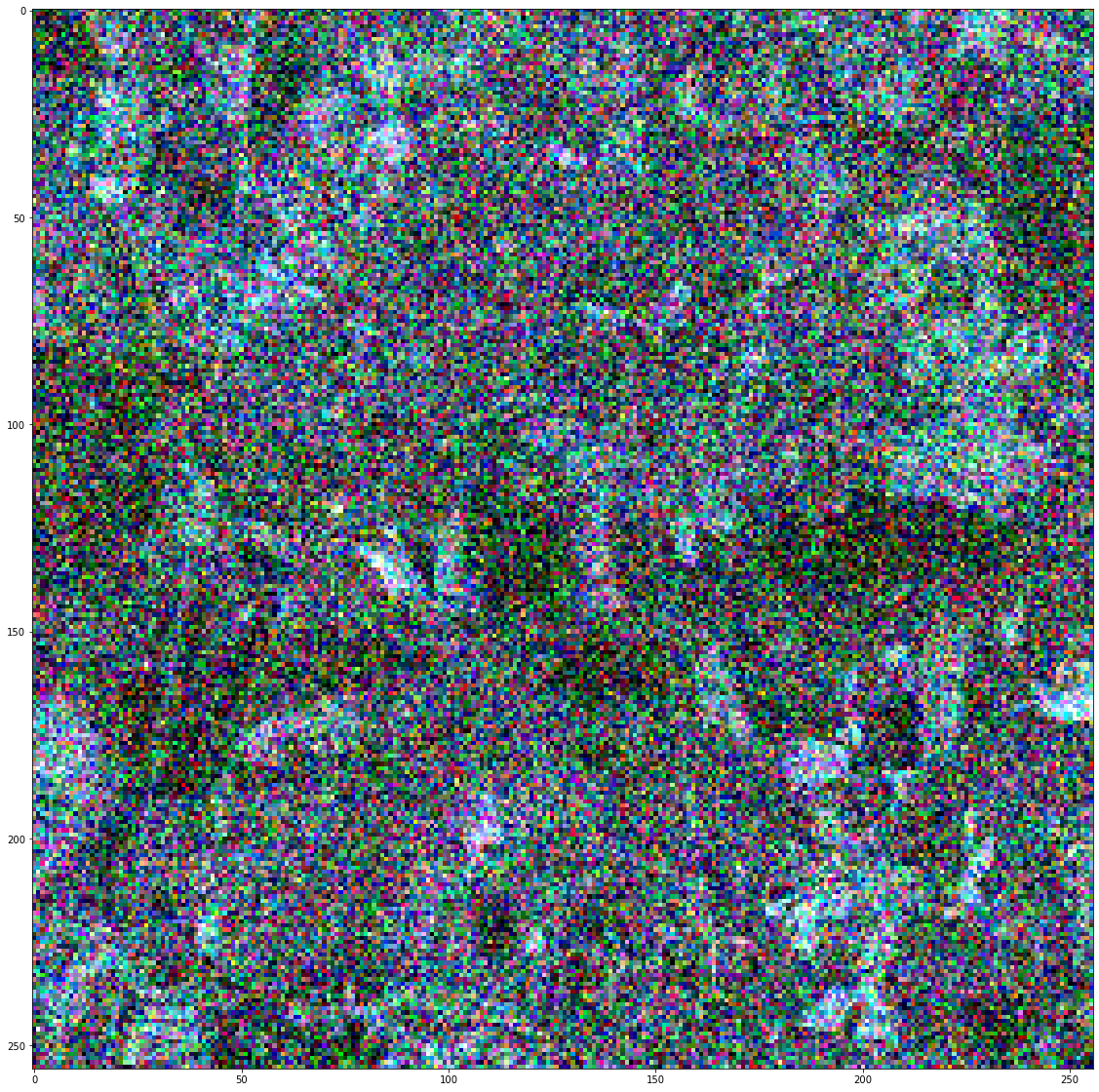}}
 {\includegraphics[width=0.22\textwidth]{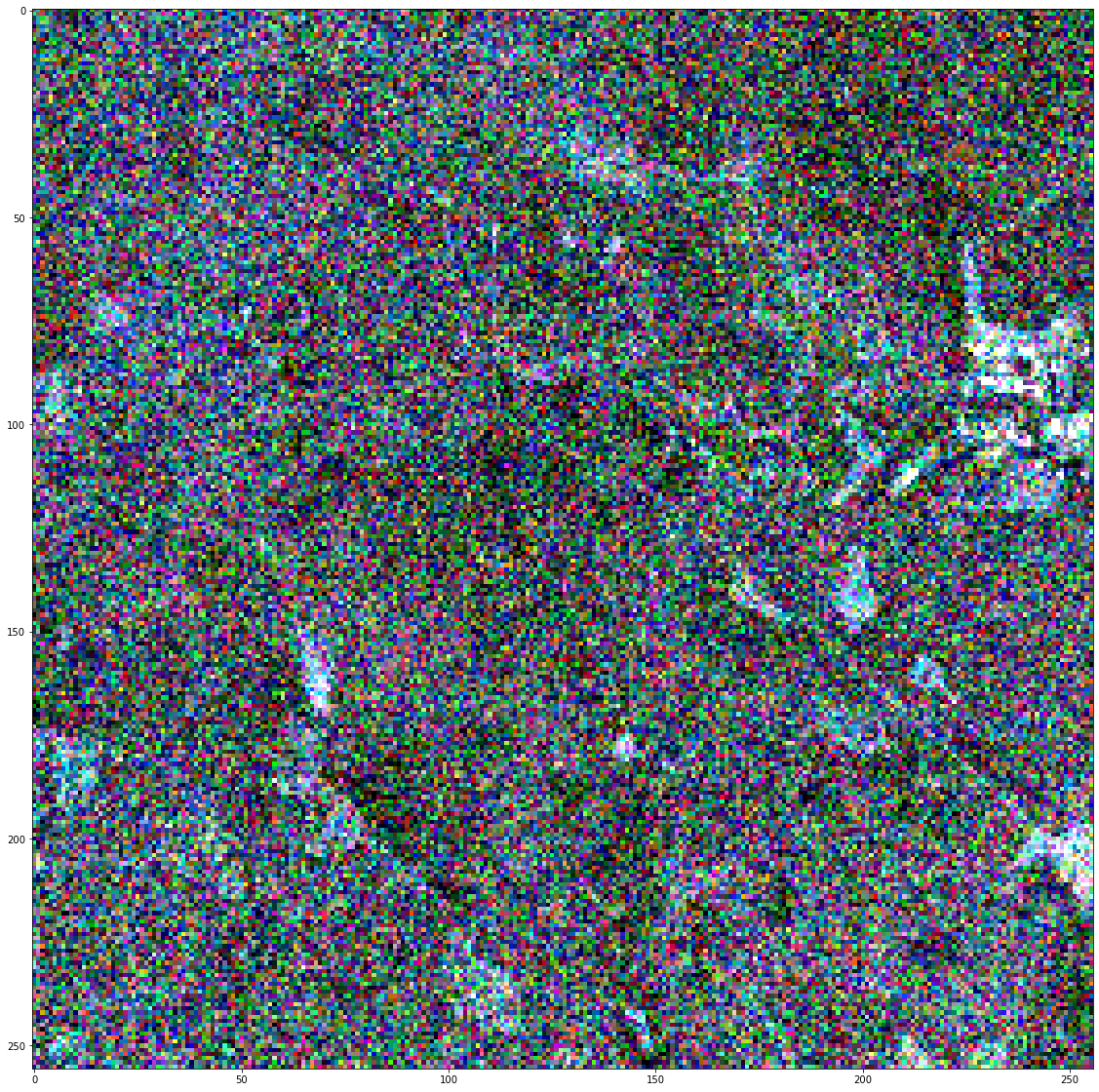}}
 {\includegraphics[width=0.22\textwidth]{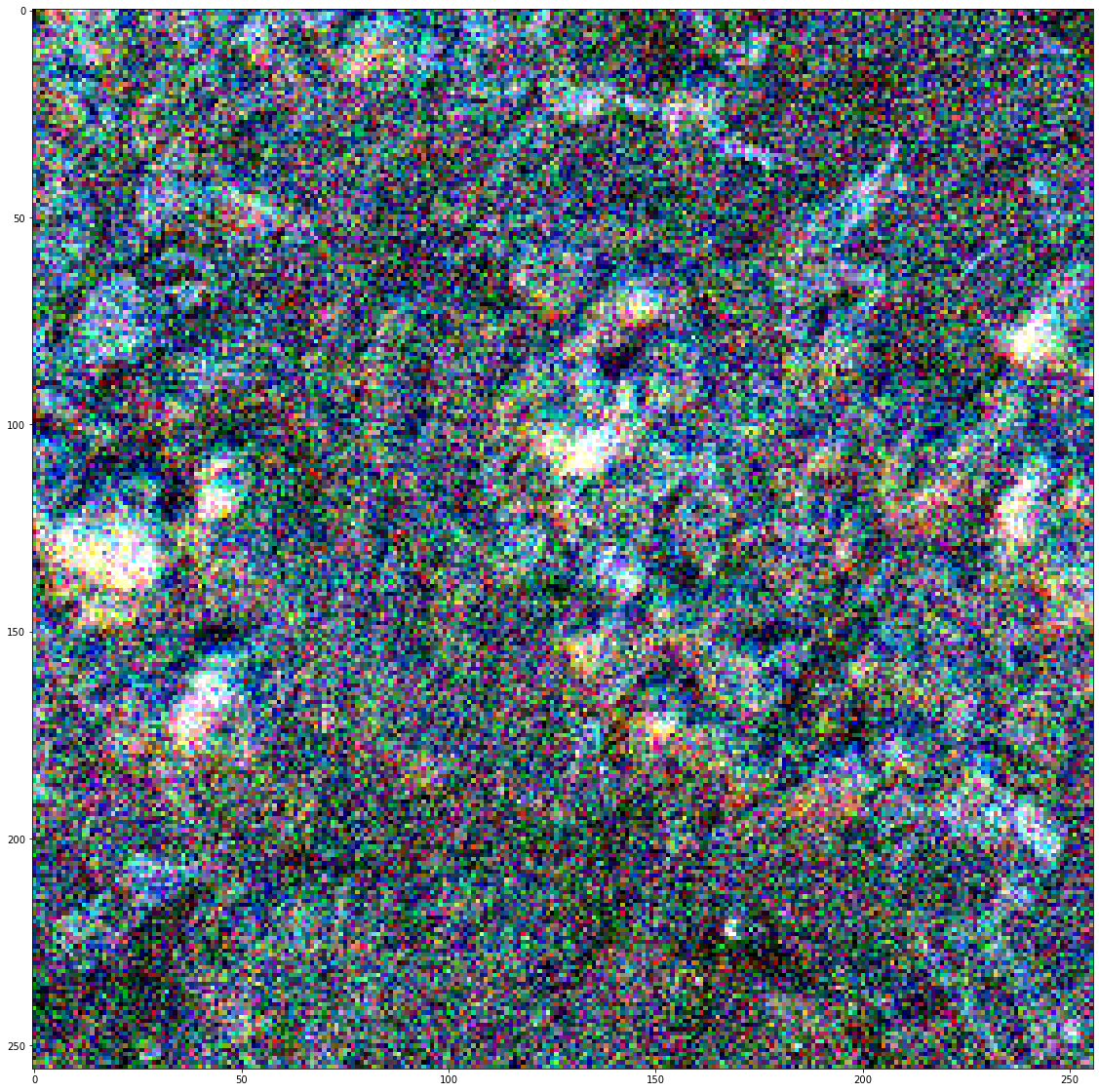}}
 {\includegraphics[width=0.22\textwidth]{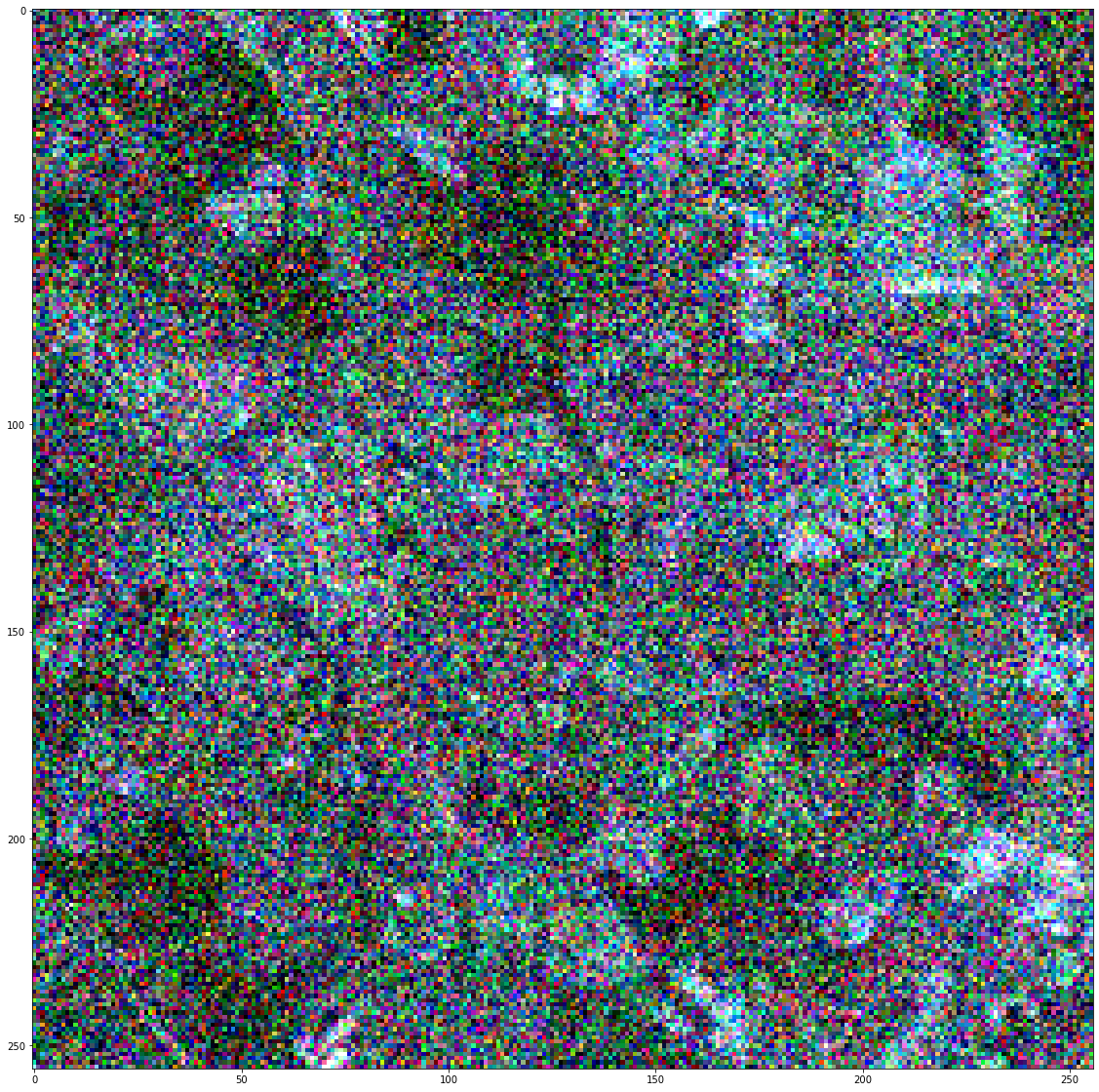}}
 {\includegraphics[width=0.22\textwidth]{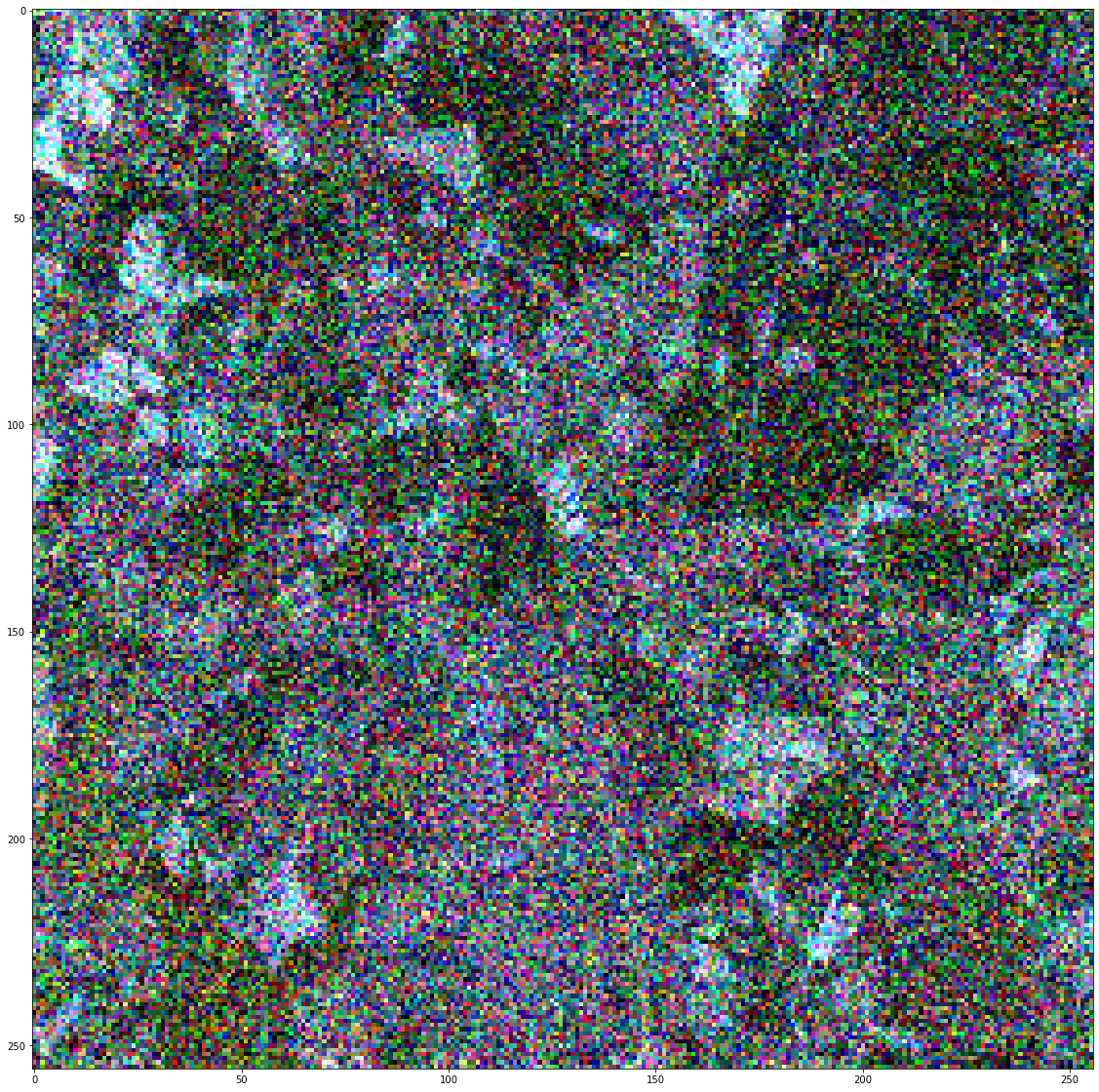}}
 {\includegraphics[width=0.22\textwidth]{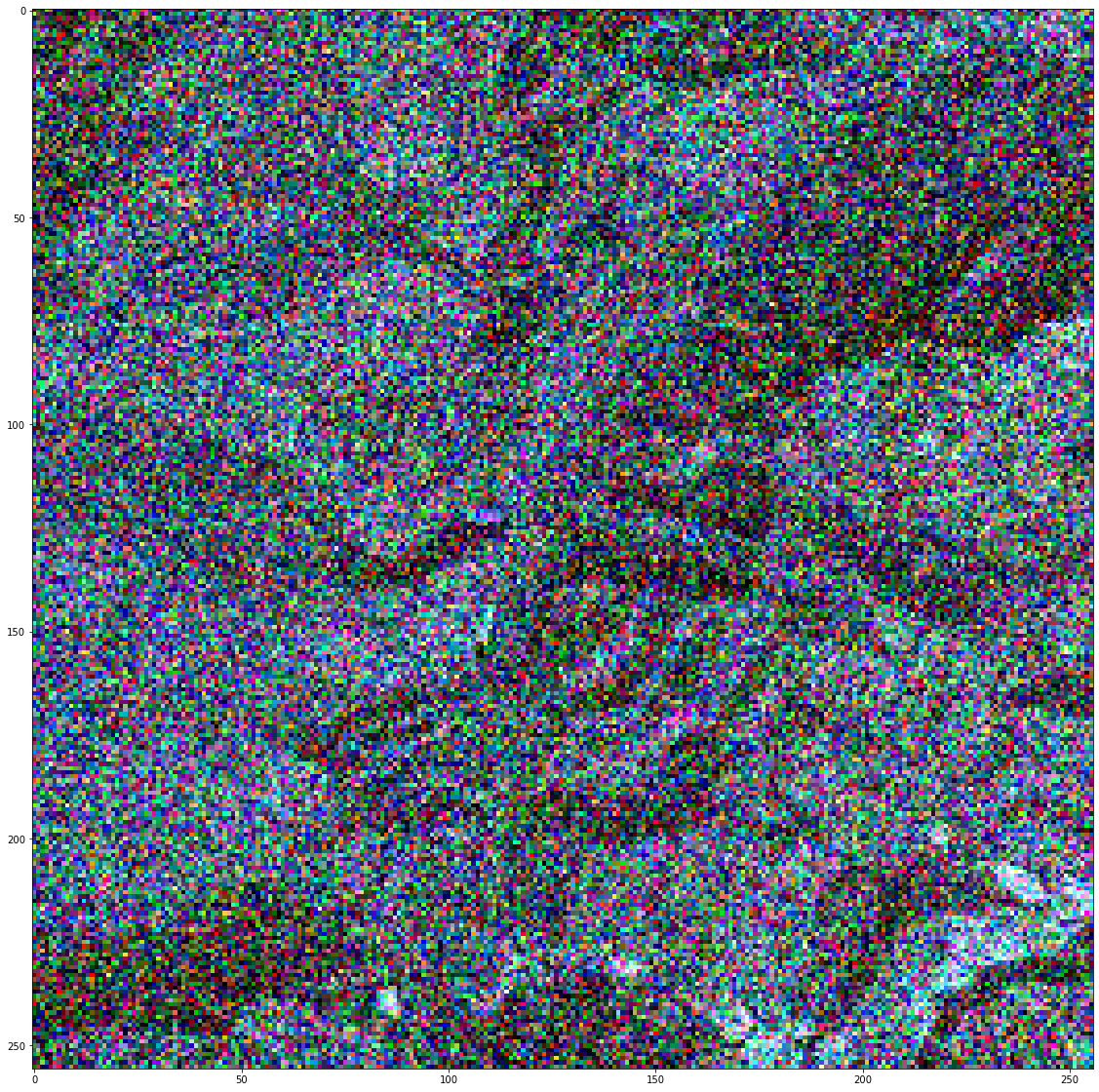}}
 {\includegraphics[width=0.22\textwidth]{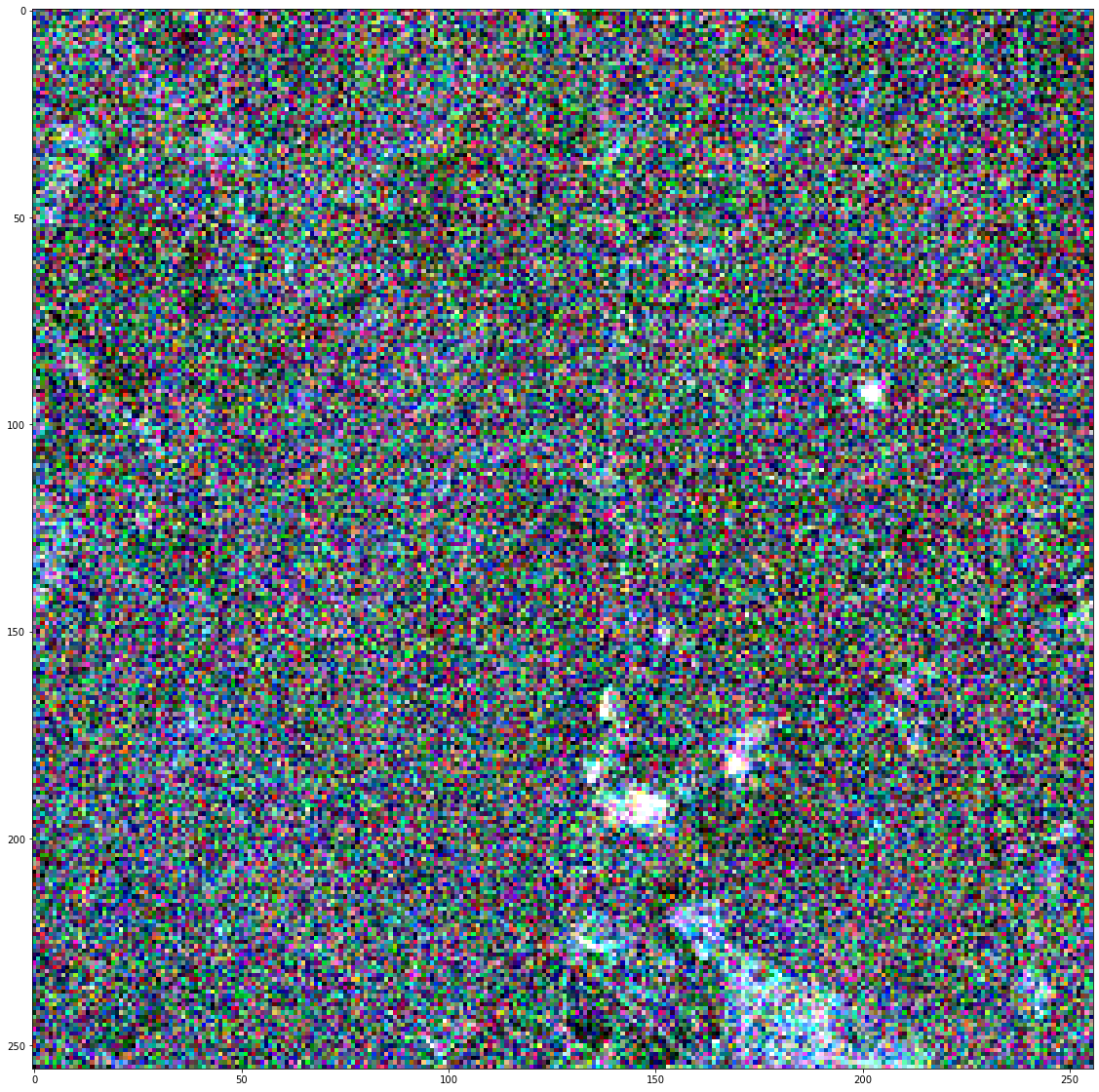}}
 {\includegraphics[width=0.22\textwidth]{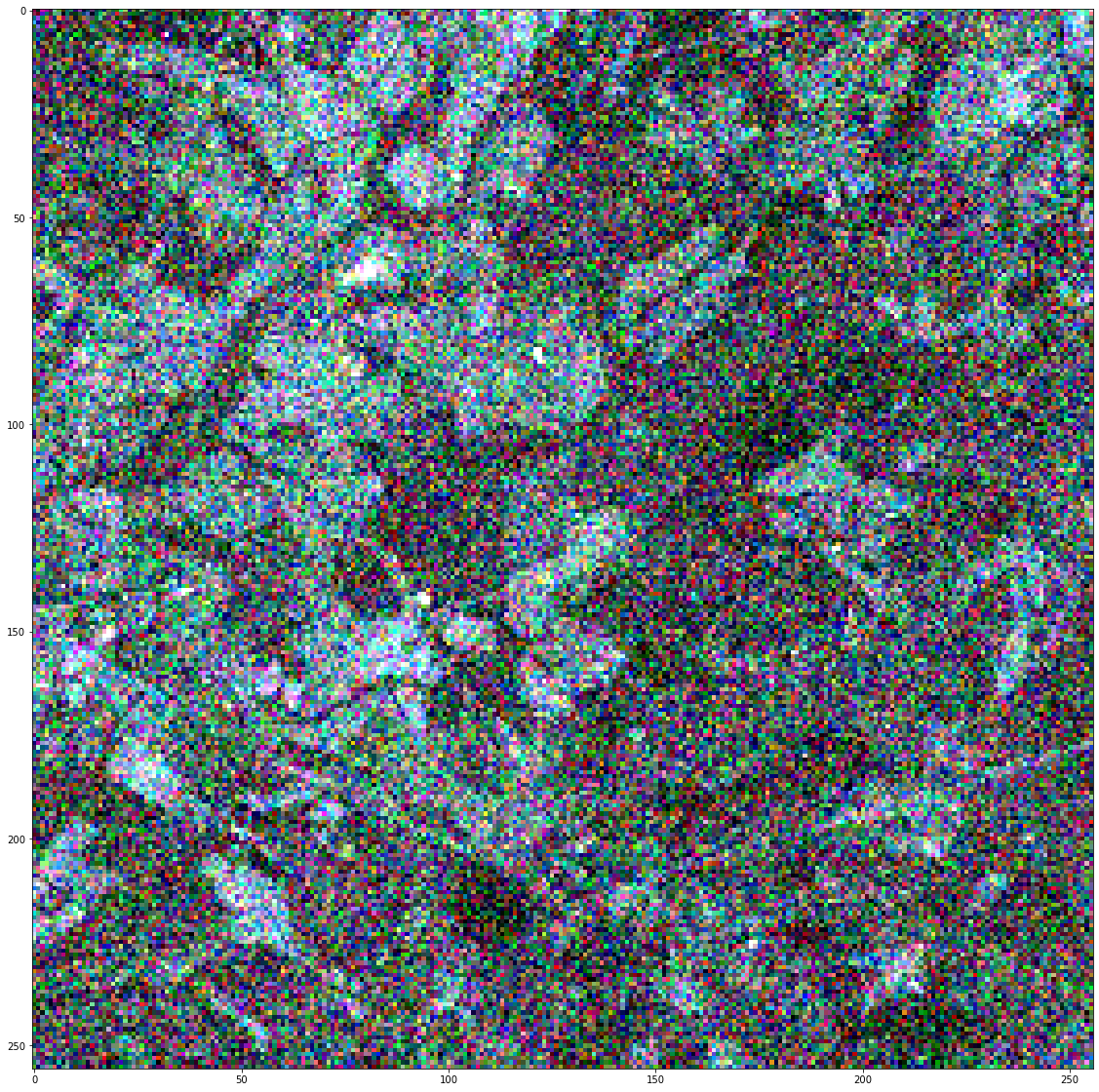}}
 {\includegraphics[width=0.22\textwidth]{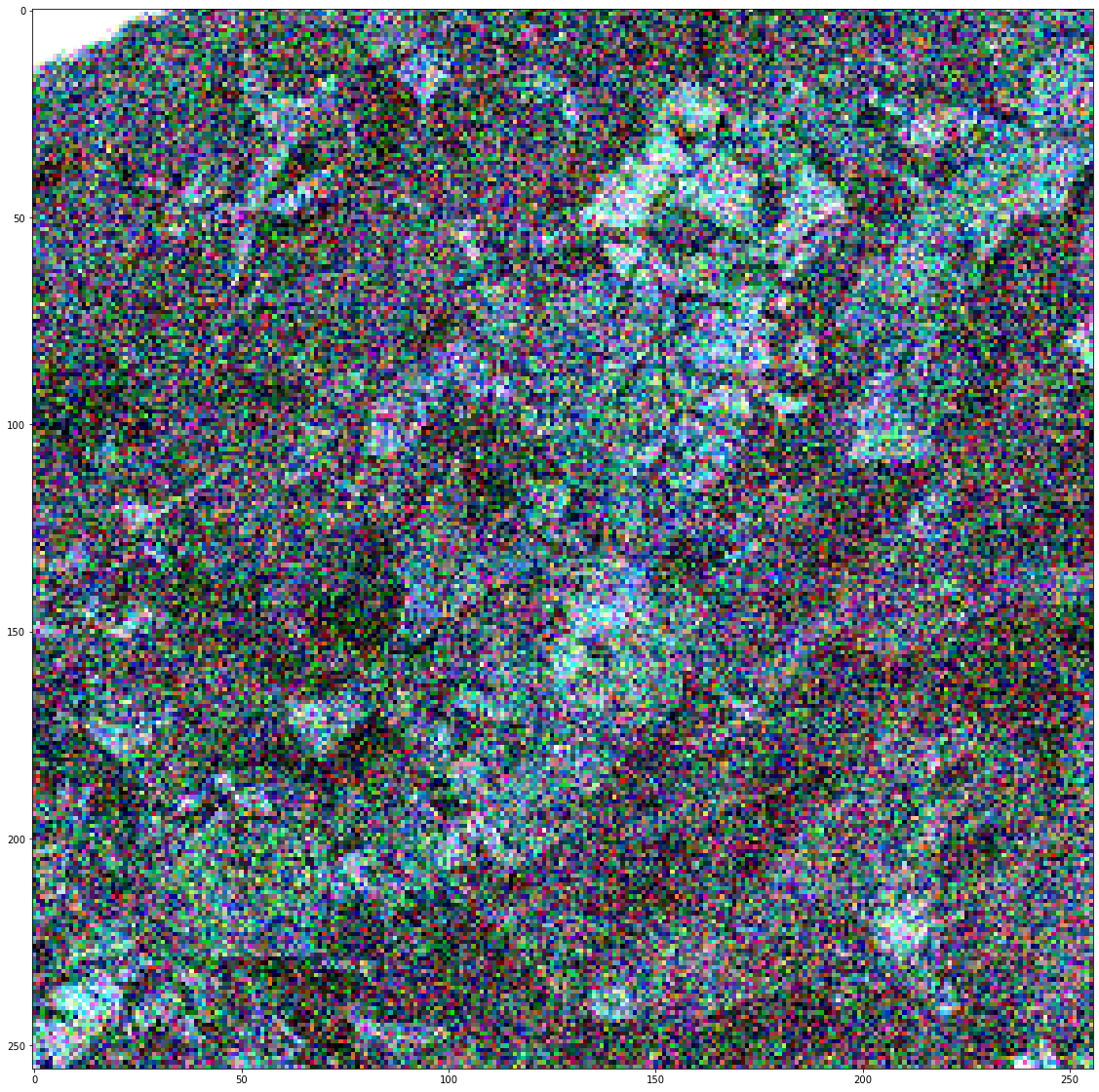}}
 {\includegraphics[width=0.22\textwidth]{figures/NoisyImage9_sigma0p04.png}}

\caption{Noisy images with variance $\sigma = 0.04$. 
}
\label{fig:NoisyTestImagesvar0_04:19059hdf}
\end{center}
\end{figure}

\begin{figure}[tbh]
\begin{center}  

 {\includegraphics[width=0.22\textwidth]{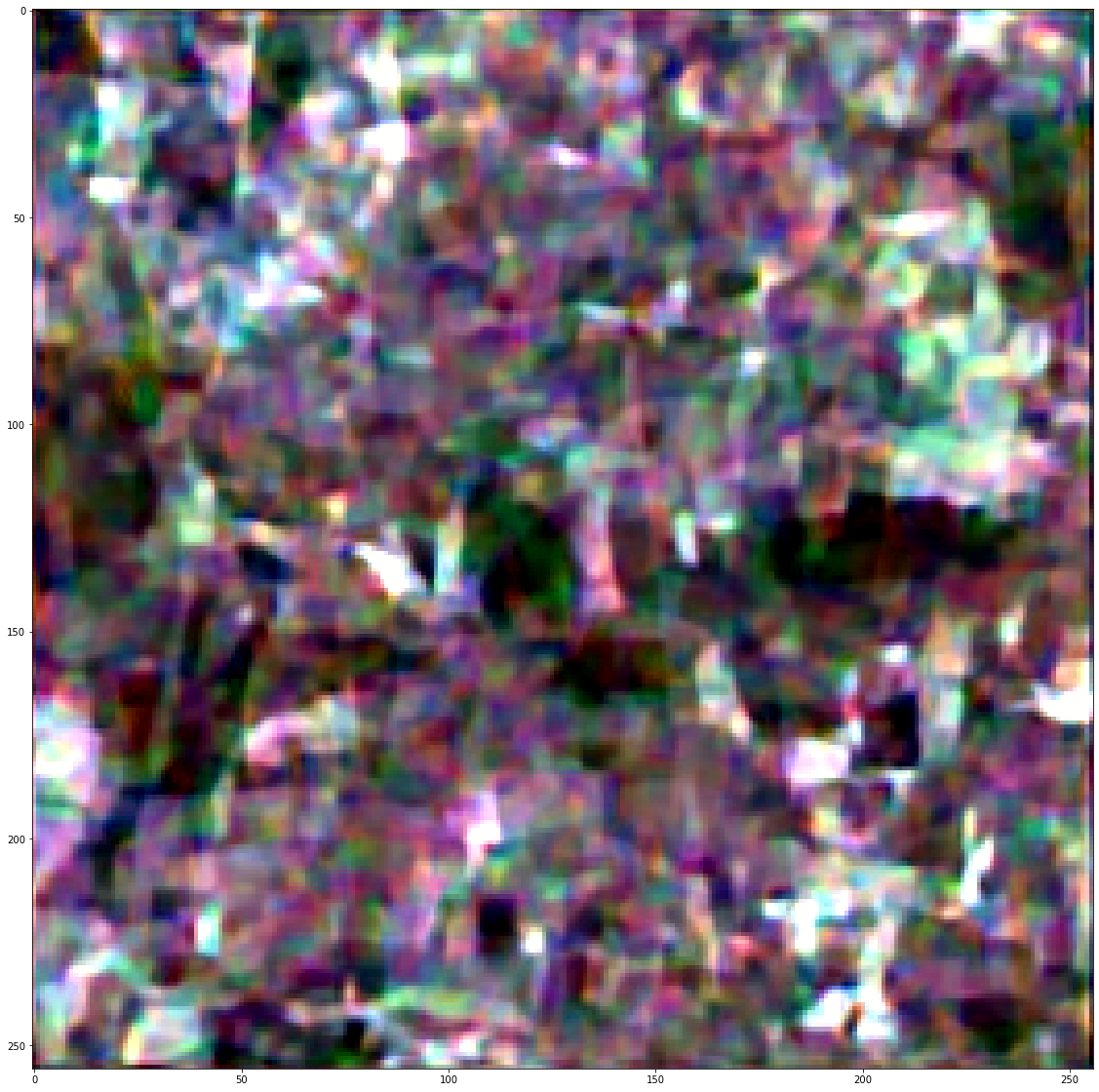}}
 {\includegraphics[width=0.22\textwidth]{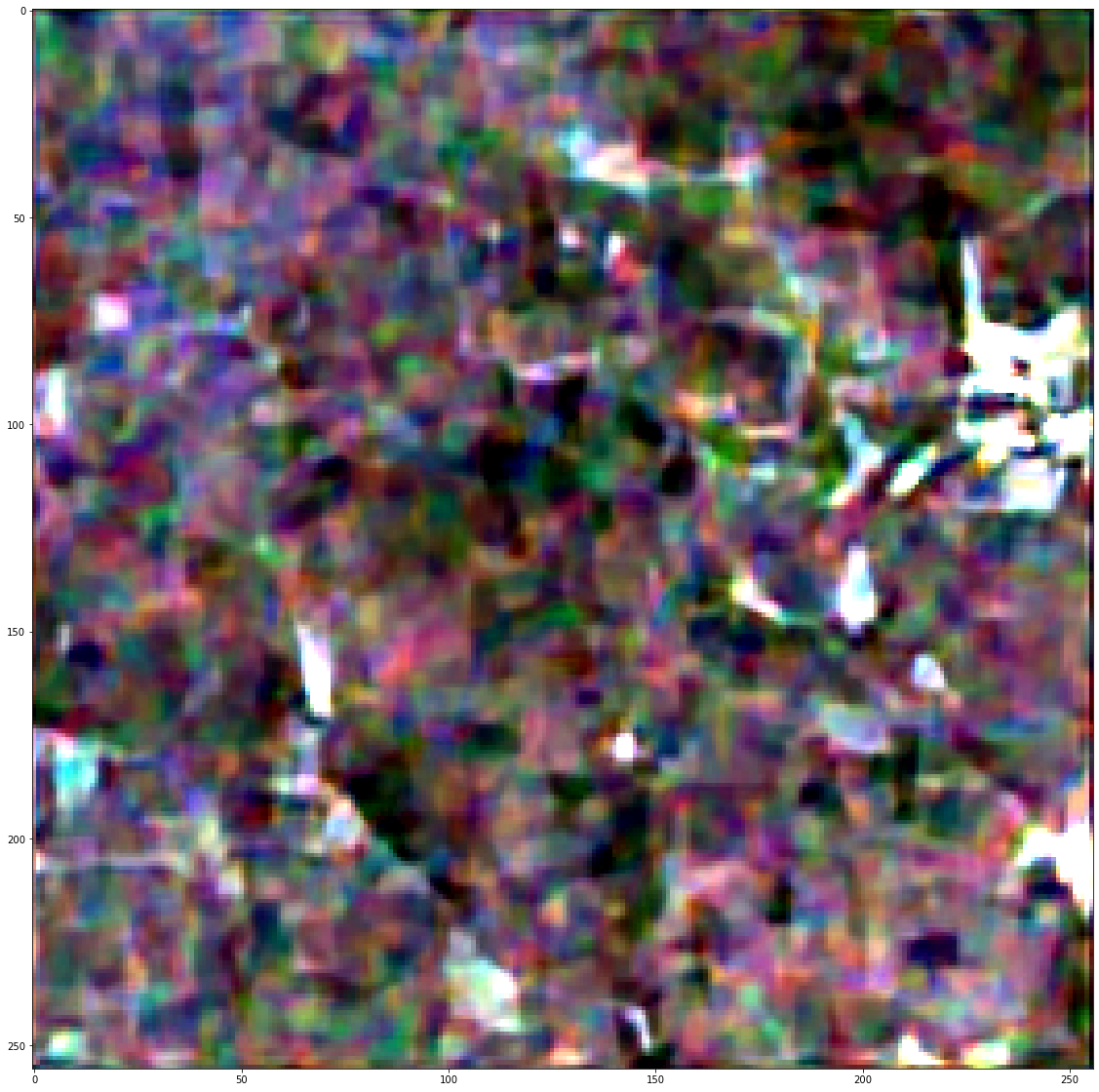}}
 {\includegraphics[width=0.22\textwidth]{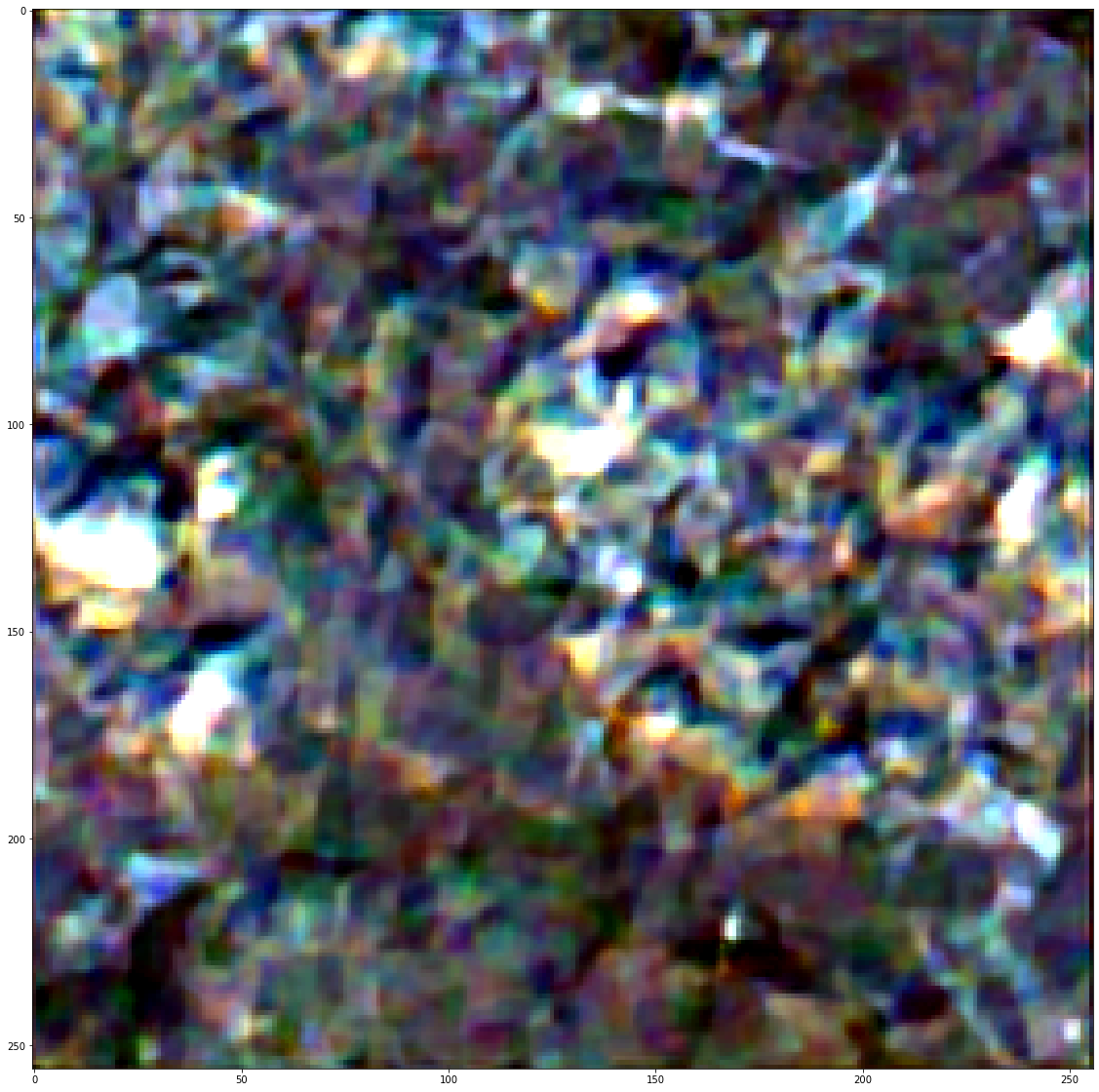}}
 {\includegraphics[width=0.22\textwidth]{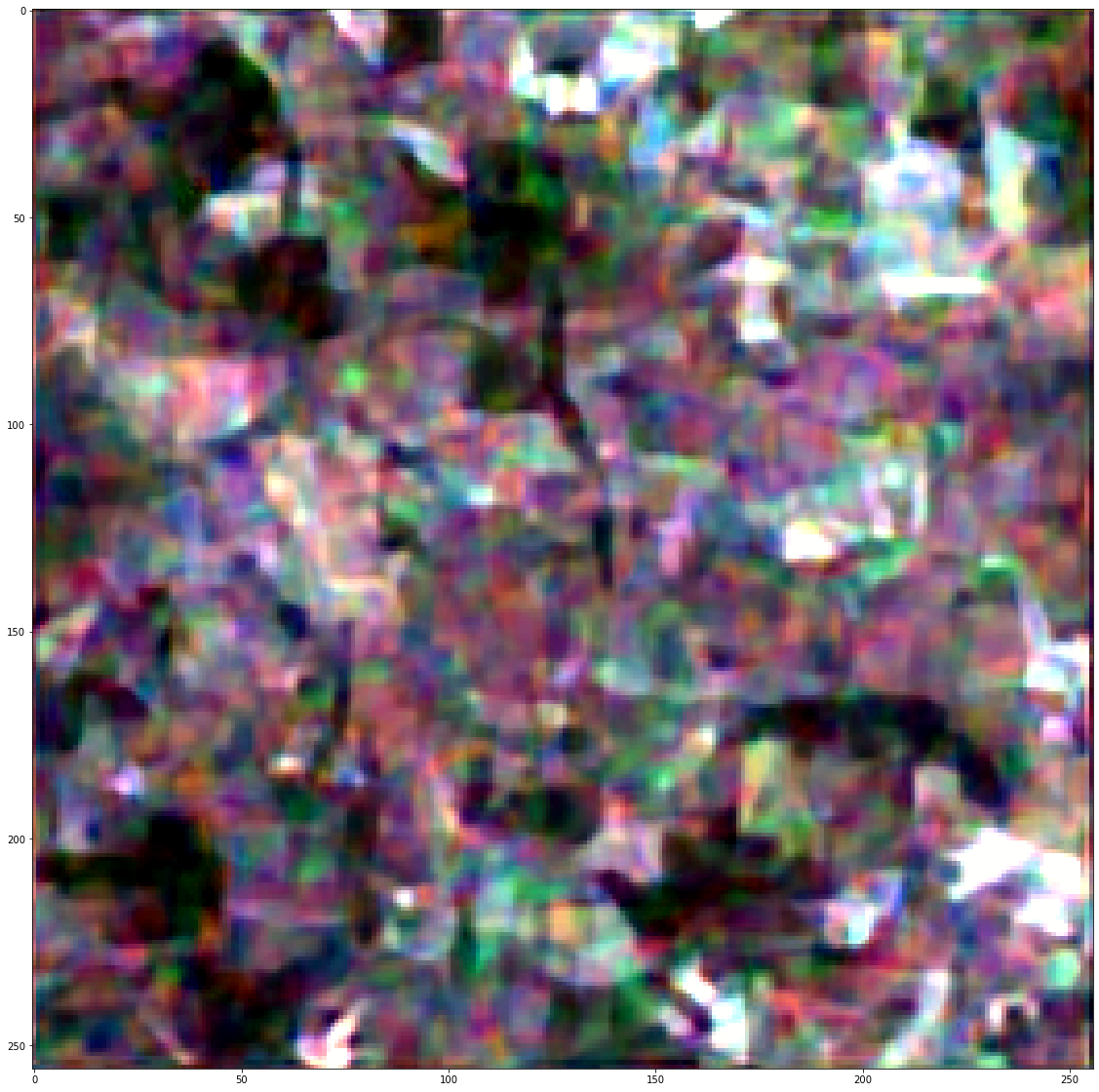}}
 {\includegraphics[width=0.22\textwidth]{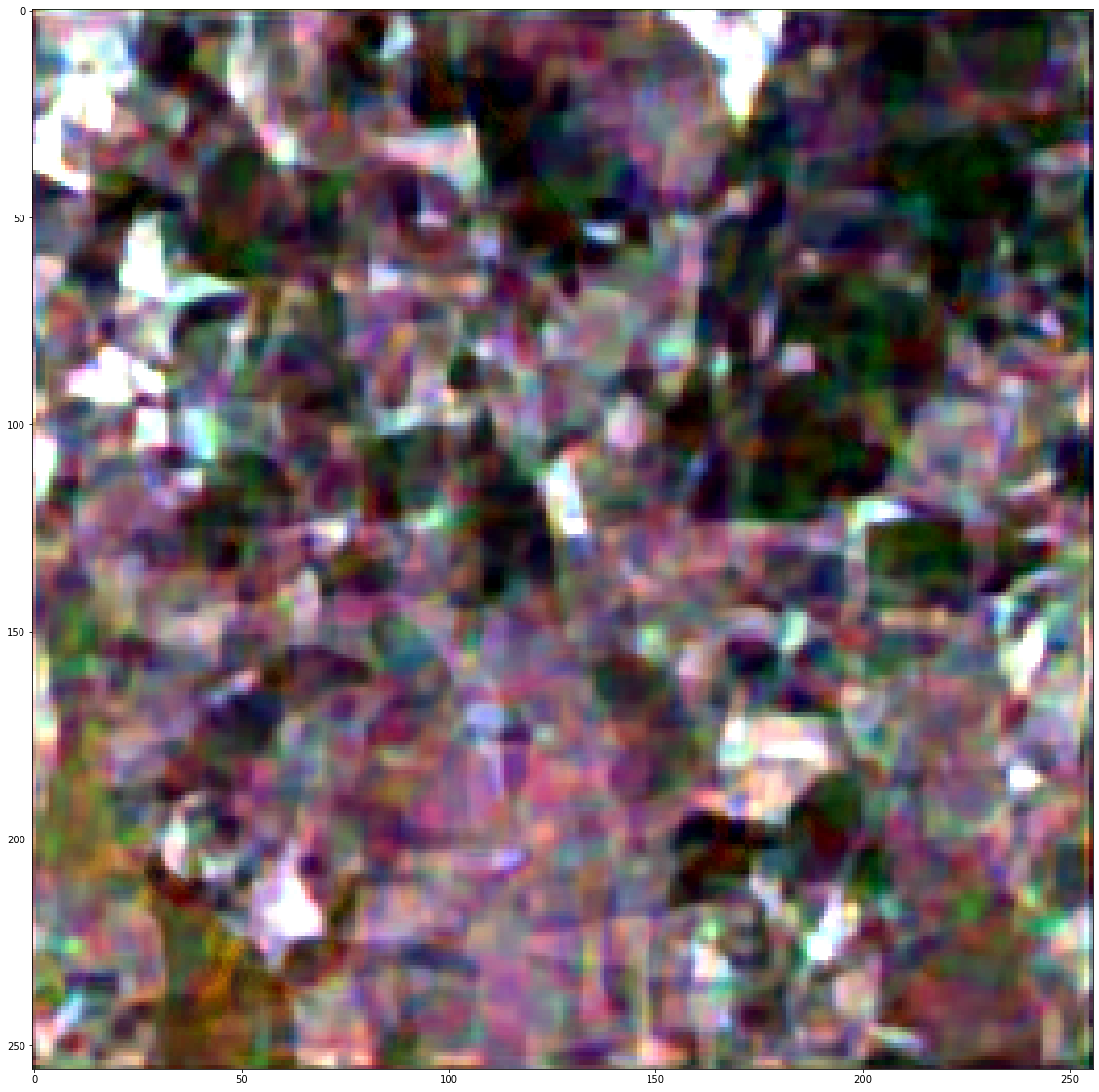}}
 {\includegraphics[width=0.22\textwidth]{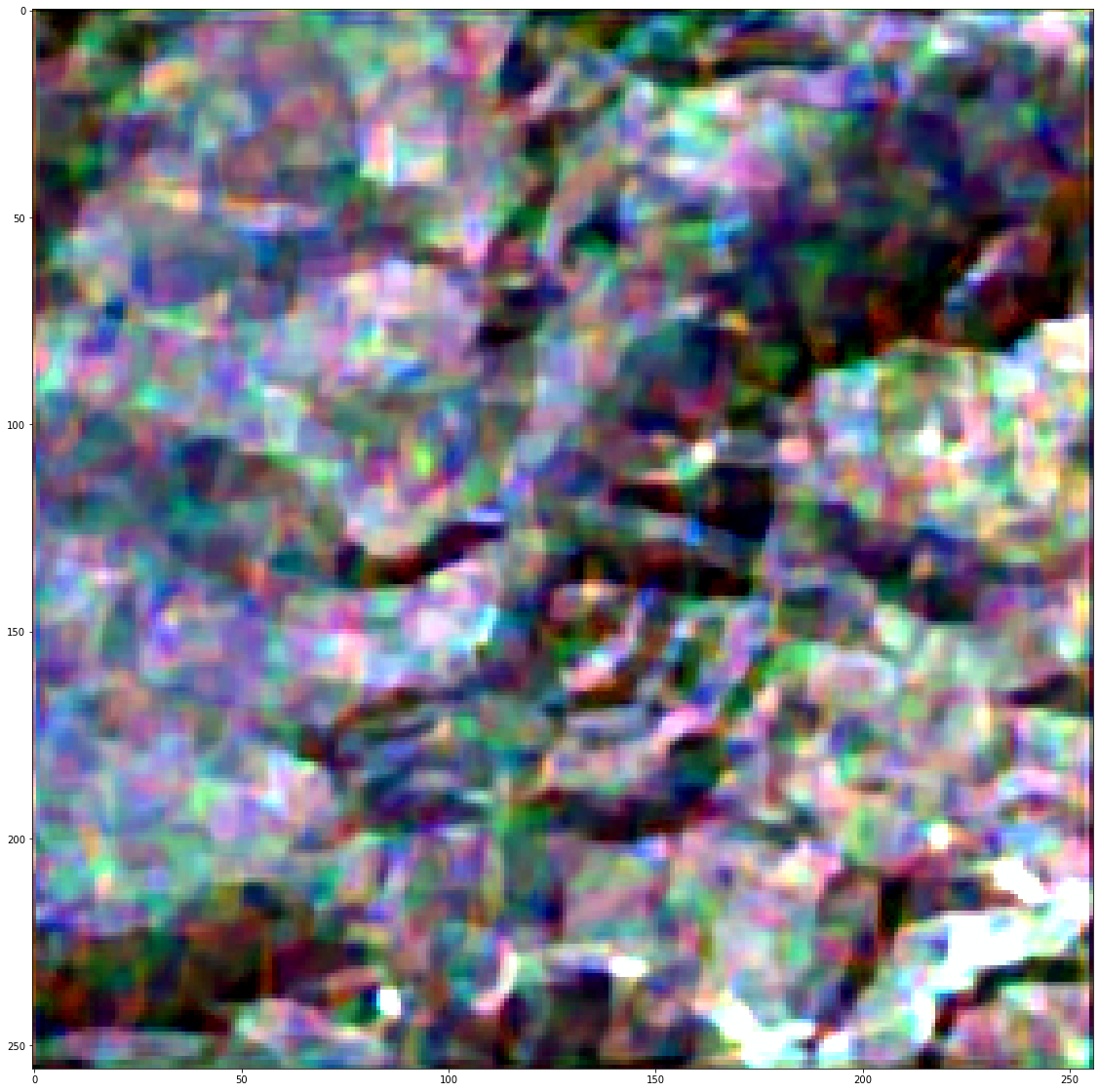}}
 {\includegraphics[width=0.22\textwidth]{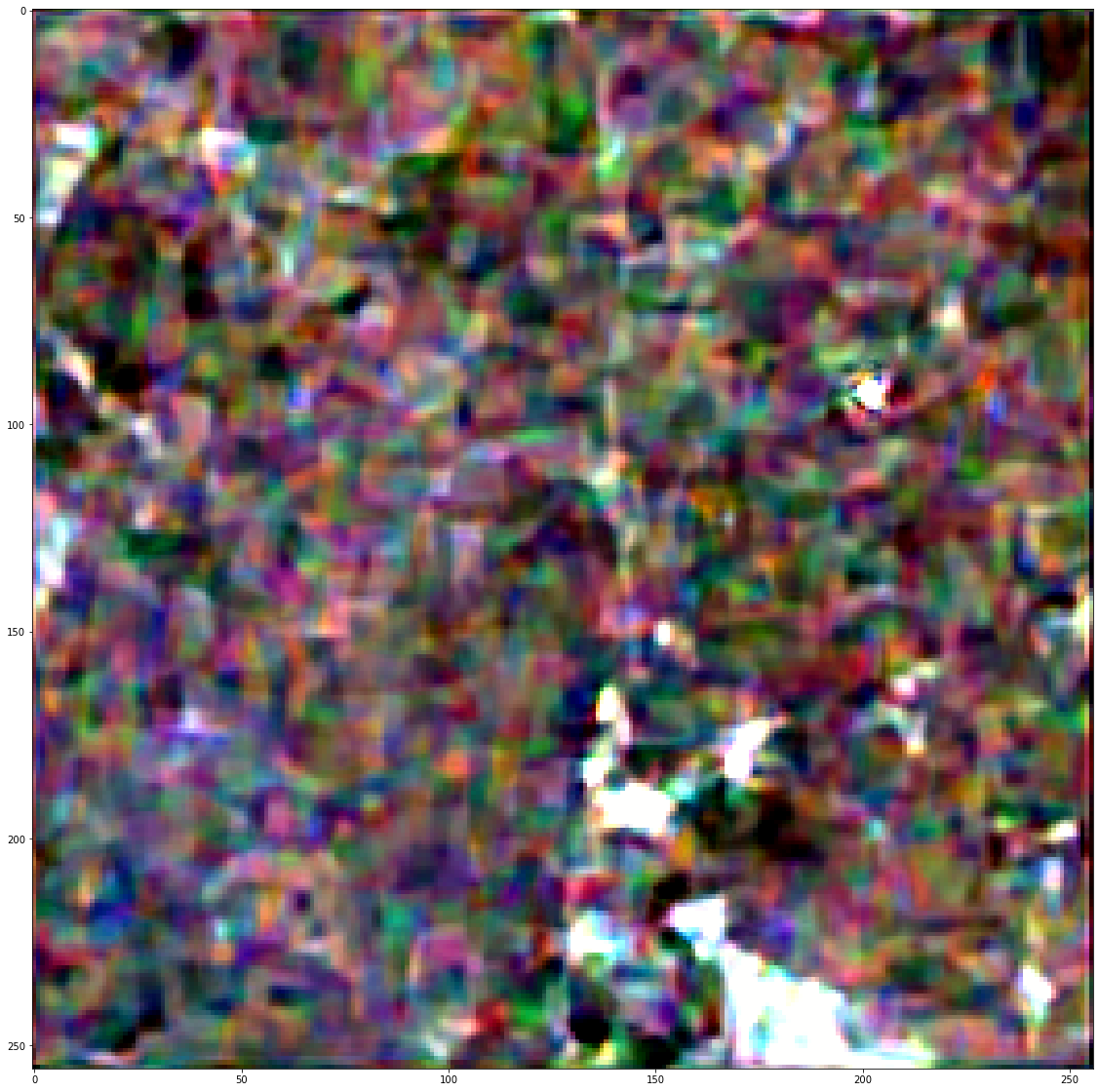}}
 {\includegraphics[width=0.22\textwidth]{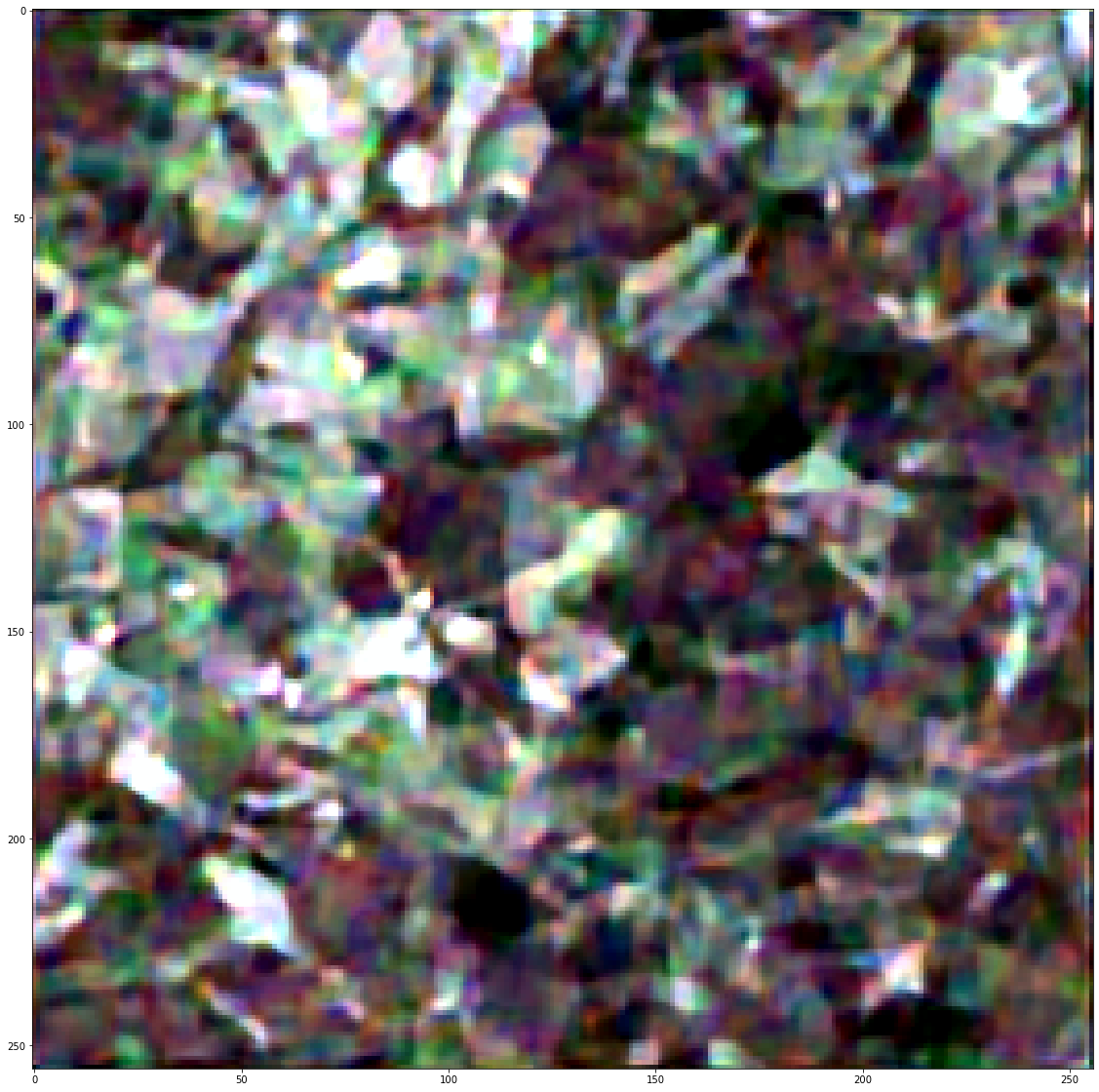}}
 {\includegraphics[width=0.22\textwidth]{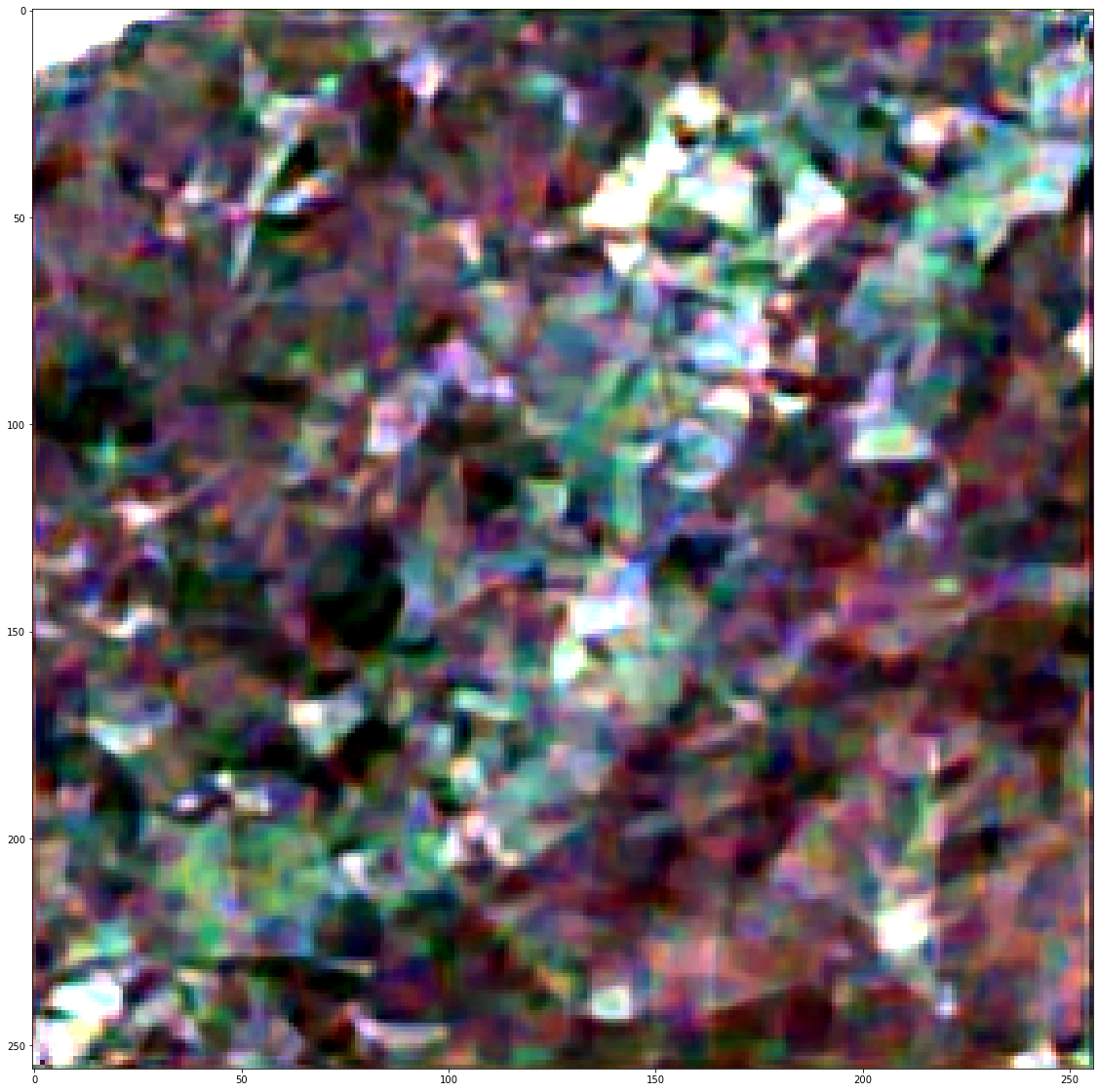}}
 {\includegraphics[width=0.22\textwidth]{figures/RQUnetVAE_smoothing_scheme3_image9_sigma0p04.png}}

\caption{Denoised images for $\sigma = 0.04$ by RQUnet-VAE smoothing. %(scheme 3) with the optimal oracle $\beta^* = 601$. 
}
\label{fig:DenoisedTestImagesRQUnetVAEScheme3var0_04:19059hdf}
\end{center}
\end{figure}

\clearpage

\section{Riesz Quincunx Filter Banks}
\begin{figure*}[!h]
\begin{center}  

% Scaling base:
\subfigure[Scaling base]{\includegraphics[width=0.20\textwidth]{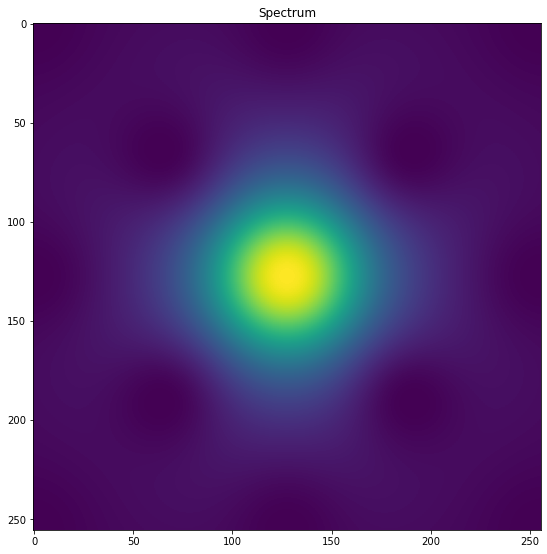}}

\vspace{-0.4cm}
% Wavelet bases at scale 0:
\subfigure[Bases at  Scale~0]{\includegraphics[width=0.20\textwidth]{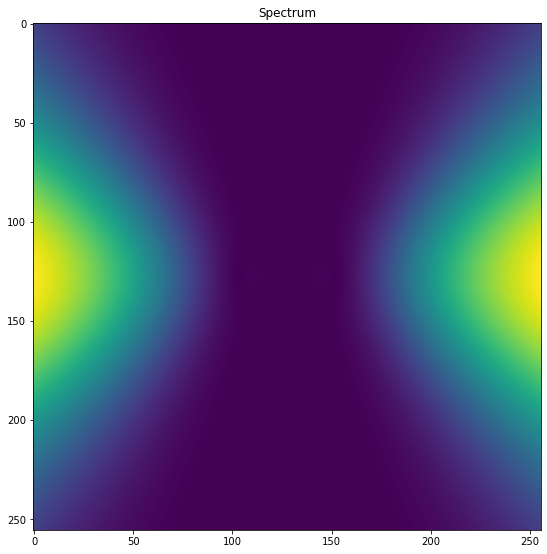}}
\subfigure[]{\includegraphics[width=0.20\textwidth]{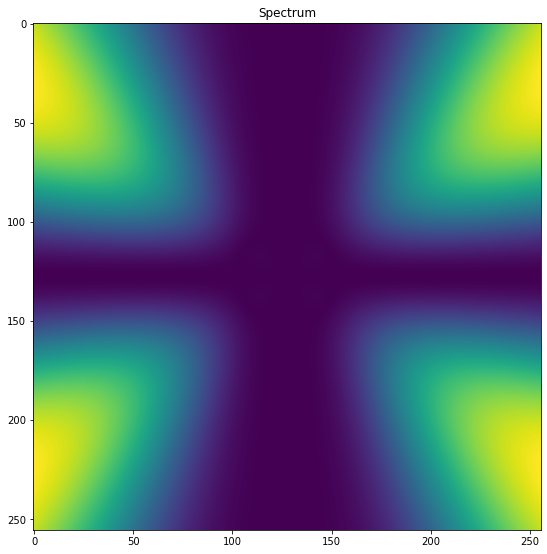}}
\subfigure[]{\includegraphics[width=0.20\textwidth]{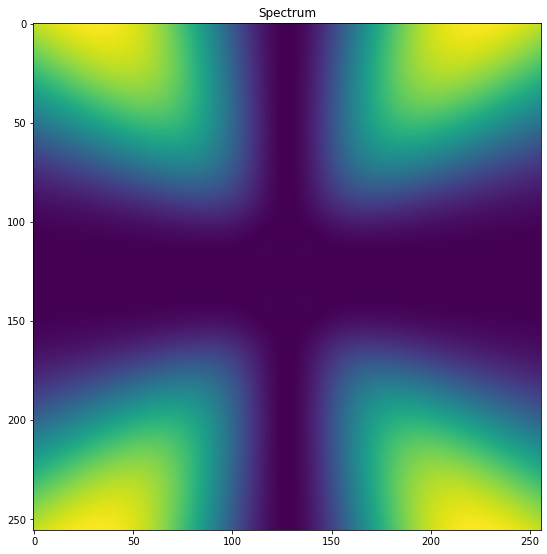}}
\subfigure[]{\includegraphics[width=0.20\textwidth]{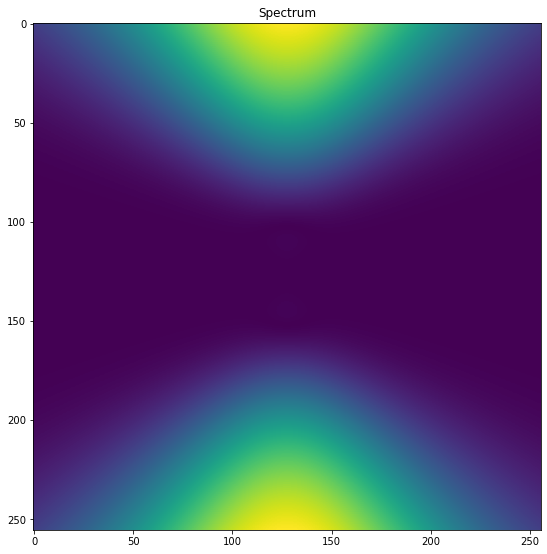}}

\vspace{-0.4cm}
% Wavelet bases at scale 1:
\subfigure[Bases at  Scale~1]{\includegraphics[width=0.20\textwidth]{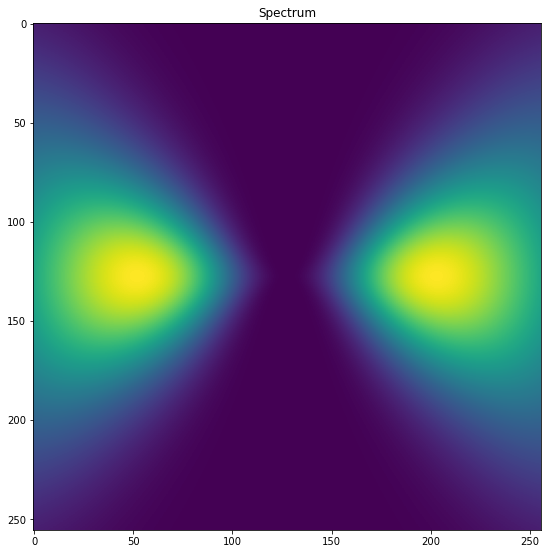}}
\subfigure[]{\includegraphics[width=0.20\textwidth]{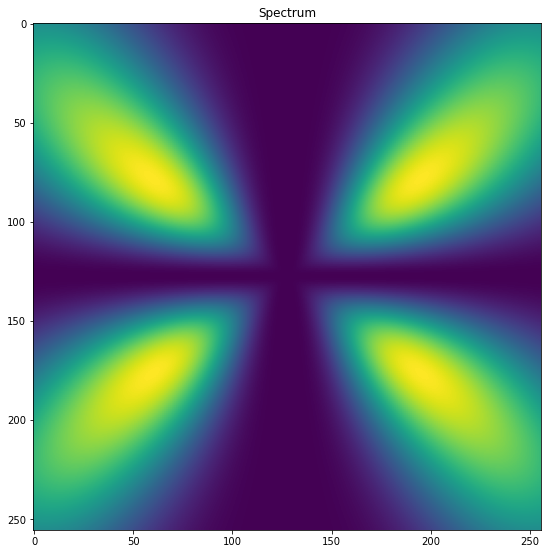}}
\subfigure[]{\includegraphics[width=0.20\textwidth]{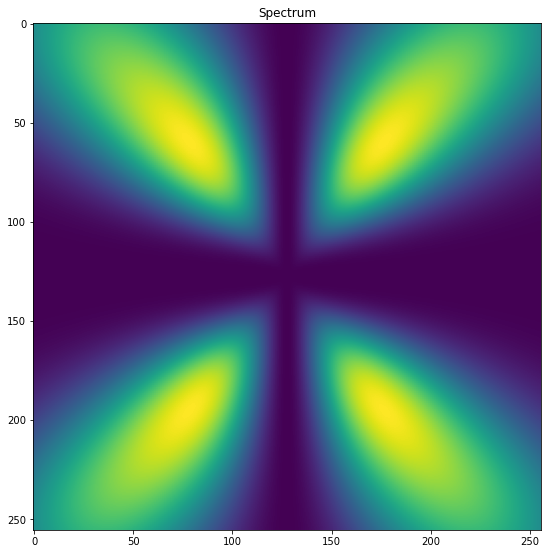}}
\subfigure[]{\includegraphics[width=0.20\textwidth]{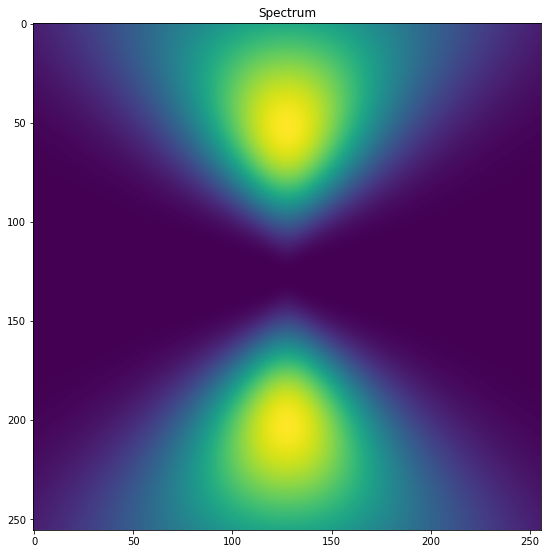}}

\vspace{-0.4cm}
% Wavelet bases at scale 2:
\subfigure[Bases at  Scale~2]{\includegraphics[width=0.20\textwidth]{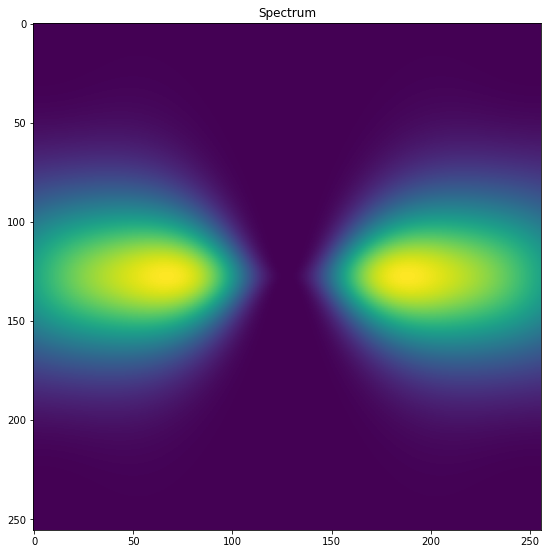}}
\subfigure[]{\includegraphics[width=0.20\textwidth]{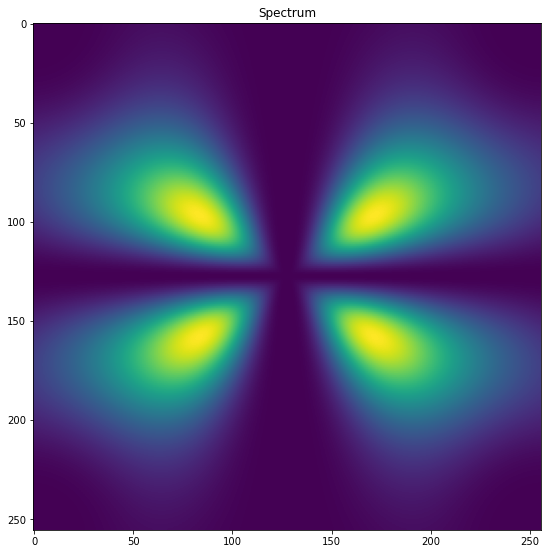}}
\subfigure[]{\includegraphics[width=0.20\textwidth]{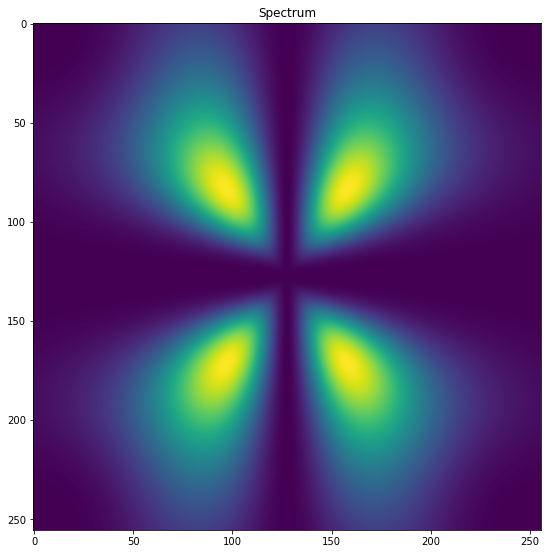}}
\subfigure[]{\includegraphics[width=0.20\textwidth]{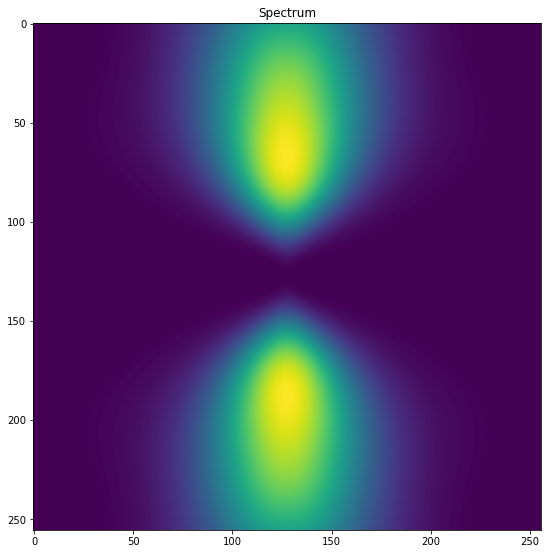}}

\vspace{-0.4cm}
% Wavelet bases at scale 3:
\subfigure[Bases at  Scale~3]{\includegraphics[width=0.20\textwidth]{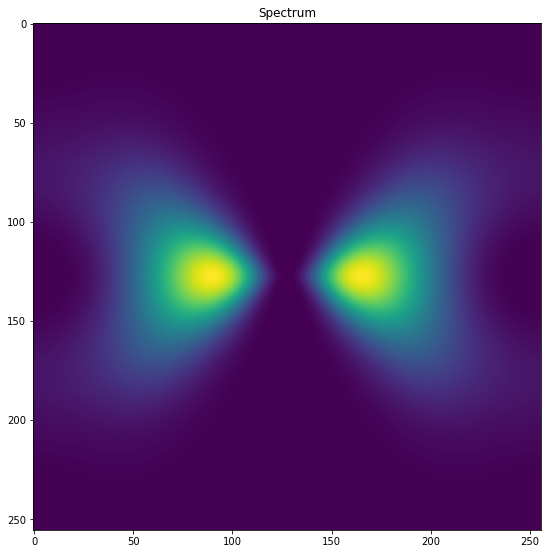}}
\subfigure[]{\includegraphics[width=0.20\textwidth]{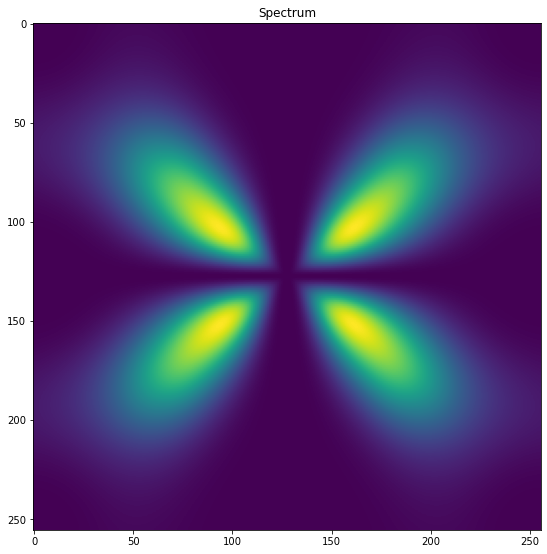}}
\subfigure[]{\includegraphics[width=0.20\textwidth]{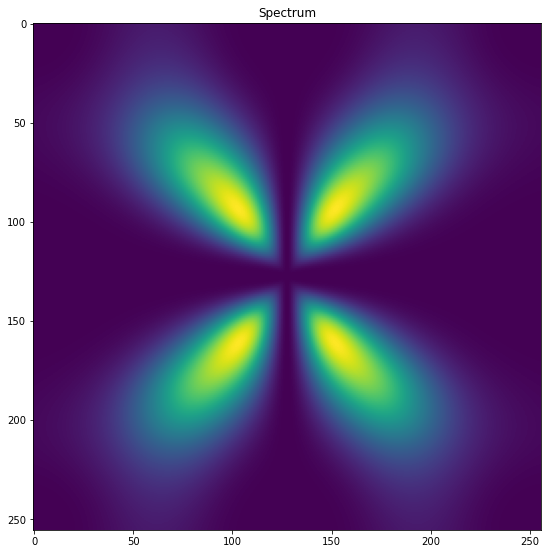}}
\subfigure[]{\includegraphics[width=0.20\textwidth]{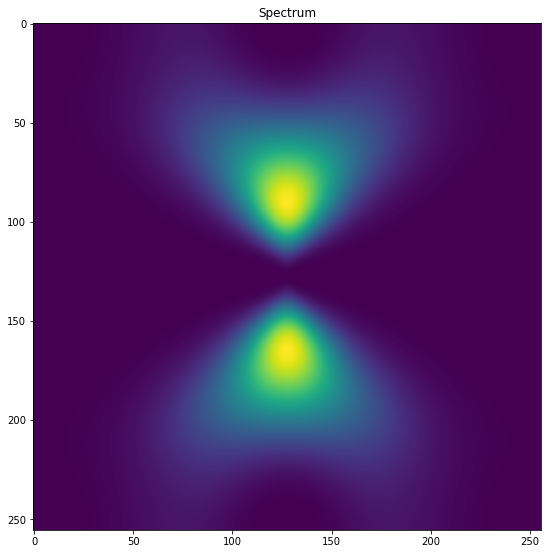}}

\vspace{-0.4cm}
\caption{This figure illustrates $N$-th order Riesz Quincunx Filter banks in the Fourier domain with 3 scales, a fractional order of Bspline $\gamma = 1.2$ and order of Riesz transform $N = 3$.
}
\label{fig:HighOrderRieszQuincunxFilterBanks}
\end{center}
\end{figure*}
Figure~\ref{fig:HighOrderRieszQuincunxFilterBanks} illustrates the N-th order Riesz Quincunx filter banks in the Fourier domain: rows are for 4 wavelet scales and columns are 4 directions for each scale corresponding to $N=0,...,3$ orders in Riesz transform.

\newpage
\clearpage

\section{Detailed Encoder and Decoder of \RQUnetVAE}
Figure~\ref{fig:RQUnetVAEencoder} shows the detailed architecture of the \RQUnetVAE encoder (left) and decoder (right). 

\begin{figure}
\begin{center}  

\includegraphics[width=0.48\textwidth]{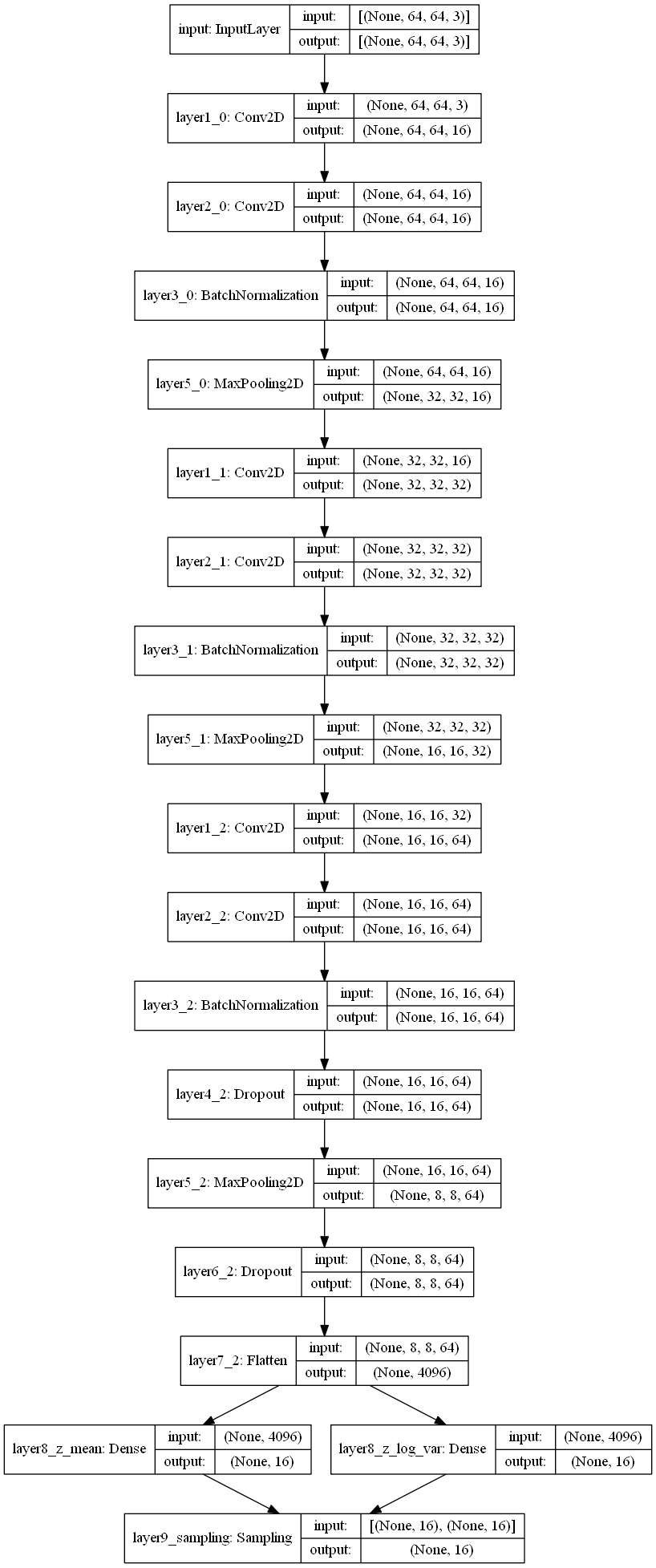}
\includegraphics[width=0.32\textwidth]{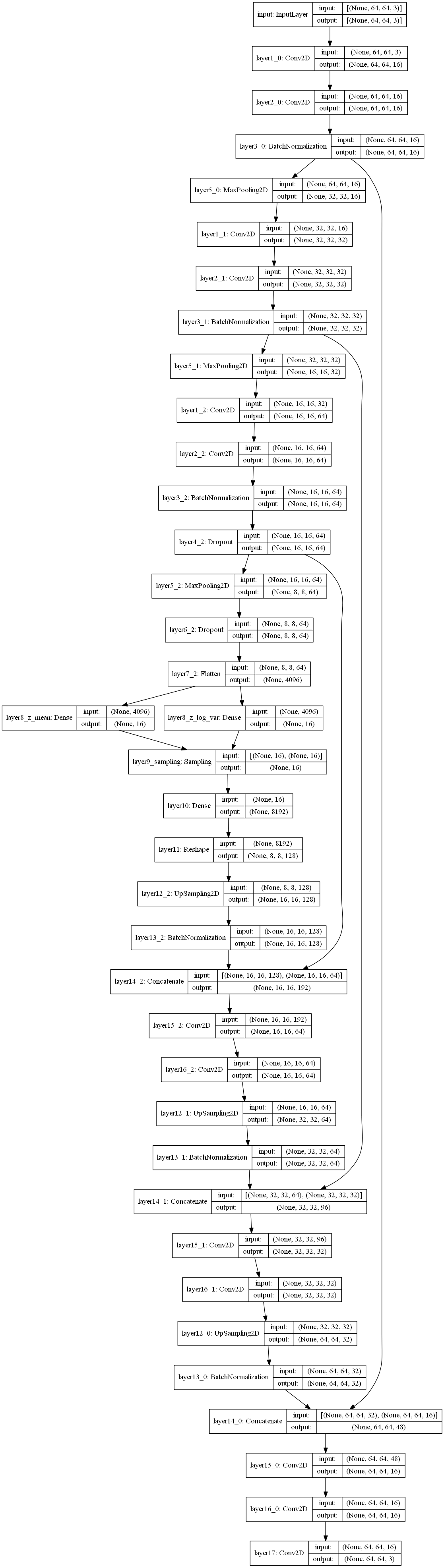}
\caption{RQUnet-VAE Encoder (left) and Decoder (right)}
\label{fig:RQUnetVAEencoder}
\end{center}
\end{figure}

\section{Additional Background and definitions}
\subsection{Hankel matrices and convolutional operations}
The following list provides definition of Hankel matrix and its inverse and a set of convolution operations used in this work. 
\begin{itemize}
    
\item Extended Hankel matrix of an image $\underline{f} \in \mathbb R^{\abs{\Omega}}$:
\begin{align*}
 \mathscr{H}_{d_1 \mid n_2}(\underline{f}) &= \begin{pmatrix} \mathscr{H}_{d_1}(f_1) & \hdots & \mathscr{H}_{d_1}(f_{n_2}) \end{pmatrix}
 \in \mathbb R^{n_1 \times n_2 d_1} \,,
\end{align*}
with a Hankel matrix:
\begin{align*}
 \mathscr{H}_{d_1}(f_i) &= \begin{pmatrix} f_i[1] & \ldots & f_i[d_1] 
  \\ \vdots & \ddots & \vdots \\
  f_i[n_1] & \ldots & f_i[n_1+d_1-1]
 \end{pmatrix} \in \mathbb R^{n_1 \times d_1} \,.
\end{align*}
\item The extended Hankel matrix of a multi-channel image $\underline{\underline{f}} \in \mathbb R^{\abs{\Omega} \times P}$: 
\begin{align*}
 \mathscr{H}_{n_1 \mid d_2 \mid P} \left( \underline{\underline{f}} \right) = \begin{pmatrix} \mathscr{H}_{n_1 \mid d_2}(\underline{f_1}) & \hdots & \mathscr{H}_{n_1 \mid d_2}(\underline{f_{P}}) \end{pmatrix} \,.
\end{align*}
Given $\{ e_k \}_{k=1}^{n_1}$ as an orthonormal basis of $\mathbb R^{n_1}$,
$\left\{ \underline{\tilde{e}_k} = \frac{1}{\sqrt{d_1}} \mathscr{H}_{d_1}(e_k) \right\}_{k=1}^{n_1} \subset \mathbb R^{n_1 \times d_1}$ is the orthonormal basis of $\mathcal{H}(n_1 \,, d_1) = \left\{ h \,:\, h \in \mathbb R^{n_1 \times d_1} \right\}$, i.e.
$\langle \underline{\tilde{e}_k} \,, \underline{\tilde{e}_l} \rangle_\text{F} = \delta_{kl}$ and an orthonormal expansion:
\begin{align*}
 \mathscr{H}_{d_1}(f_i) = \sum_{k=1}^{n_1} \underbrace{ \left\langle \underline{\tilde{e}_k} \,, \mathscr{H}_{d_1}(f_i) \right\rangle_\text{F} }_{ = \sqrt{d_1} f_i[k] } \underline{\tilde{e}_k} \,. 
\end{align*}
\item The inverse of an extended Hankel matrices for a matrix $\underline{g} = \begin{pmatrix} \underline{g_1} & \hdots & \underline{g_{n_2}} \end{pmatrix} \in \mathbb R^{n_1 \times d_1 n_2} \,, \underline{g_i} \in \mathbb R^{n_1 \times d_1}$ as:
\begin{align}
 \label{eq:InverseExtendedHankelMatrix}
 \mathscr{H}^\dagger_{d_1 \mid n_2}\left(\underline{g}\right) &= \begin{pmatrix} \mathscr{H}^\dagger_{d_1}\left(\underline{g_1}\right) & \ldots & \mathscr{H}^\dagger_{d_1}\left(\underline{g_{n_2}}\right) \end{pmatrix} \in \mathbb R^{\abs{\Omega}} \,,
\end{align}
having an inverse Hankel matrix:
\begin{align}  \label{eq:InverseHankelMatrix}
 \mathscr{H}^\dagger_{d_1}\left(\underline{g_i}\right) &= \frac{1}{\sqrt{d_1}} \begin{pmatrix} \left\langle \underline{\tilde{e}_1} \,, \underline{g_i} \right\rangle_\text{F} \\ \vdots \\ \left\langle \underline{\tilde{e}_{n_1}} \,, \underline{g_i} \right\rangle_\text{F} \end{pmatrix} 
 \in \mathbb R^{n_1} \,.
\end{align}

 \item Filter banks and their family of matrices for the Riesz-Quincunx wavelet: 
 \begin{itemize}
     \item scaling filter bank 
  $\underline{\phi} \in \mathbb R^{d_1 \times d_2}$,
  \item primal wavelet filter bank 
  $\underline{\underline{\psi}} = \left\{ \underline{\psi_1} \,, \hdots \,, \underline{\psi_{P}} \right\} \in \mathbb R^{d_1 \times d_2 \times P} \,, \underline{\psi_i} \in \mathbb R^{d_1 \times d_2}$,
  \item dual wavelet filter bank  $\underline{\underline{\tilde\psi}} = \left\{ \underline{\tilde\psi_1} \,, \hdots \,, \underline{\tilde\psi_{P}} \right\} \in \mathbb R^{d_1 \times d_2 \times P} \,, \underline{\tilde\psi_i} \in \mathbb R^{d_1 \times d_2}$,
  \item having matrix forms

  $\underline{\underline{\Psi}}
  = \left\{ \underline{\Psi_1} \,, \hdots \,, \underline{\Psi_{P}} \right\} \,,~ 
  \underline{\underline{\tilde \Psi}}
  = \left\{ \underline{\tilde\Psi_1} \,, \hdots \,, \underline{\tilde\Psi_{P}} \right\}$ such that for $p\in 1,...,P$:
  \begin{align}  \label{eq:RieszQuincunxFilter}
   \underline{\Psi_p} &= \begin{pmatrix} \underline{\Psi^1_p} \\ \vdots \\ \underline{\Psi^{d_2}_p} \end{pmatrix}
   = \begin{pmatrix} \psi_{p,1}^{1} & \hdots &  \psi_{p,d_2}^{1}
    \\ \vdots & \ddots & \vdots \\
    \psi_{p,1}^{d_2} & \hdots & \psi_{p,d_2}^{d_2}
   \end{pmatrix} \,,
   % ---
   \underline{\tilde\Psi_p} &= \begin{pmatrix} \underline{\tilde\Psi^1_p} \\ \vdots \\ \underline{\tilde\Psi^{d_2}_p} \end{pmatrix}
   = \begin{pmatrix} \tilde{\psi}_{p,1}^{1} & \hdots & \tilde{\psi}_{p,d_2}^{1}
    \\ \vdots & \ddots & \vdots \\
    \tilde{\psi}_{p,1}^{d_2} & \hdots & \tilde{\psi}_{p,d_2}^{d_2}
   \end{pmatrix}
   \in \mathbb R^{d_2 d_1 \times d_2}% \,,~
  % p = 1 \,, \ldots \,, P \,,
  \end{align}
  and 
  \begin{align}  \label{eq:ScalingFilterBanks}
   &\underline{\Phi} = \begin{pmatrix}  \underline{\Phi^1} \\ \vdots \\ \underline{\Phi^{d_2}} \end{pmatrix}
   = \begin{pmatrix} \phi_{1}^{1} & \hdots & \phi_{d_2}^{1}
    \\ \vdots & \ddots & \vdots \\
    \phi_{1}^{d_2} & \hdots & \phi_{d_2}^{d_2}
   \end{pmatrix} \,,
   % ----  
   \underline{\tilde\Phi} = \begin{pmatrix}  \underline{\tilde\Phi^1} \\ \vdots \\ \underline{\tilde\Phi^{d_2}} \end{pmatrix}
   = \begin{pmatrix} \tilde\phi_{1}^{1} & \hdots & \tilde\phi_{d_2}^{1}
    \\ \vdots & \ddots & \vdots \\
    \tilde\phi_{1}^{d_2} & \hdots & \tilde\phi_{d_2}^{d_2}
   \end{pmatrix}
   \in \mathbb R^{d_2 d_1 \times d_2}
  \end{align}
  with 
  $\phi_{i_2}^{i_1}[i_3] := \phi[i_3, i_2-i_1] \,, %   
  \psi_{p,i_2}^{i_1}[l] := \psi_p[i_3, i_2-i_1] \,,
  \tilde{\psi}_{p,i_2}^{i_1}[i_3] := \tilde{\psi}_p[i_3, i_2-i_1]$ and
  $\phi_{i_2}^{i_1} \,, \psi_{p,i_2}^{i_1} \,, \tilde{\psi}_{p,i_2}^{i_1} \in \mathbb R^{d_1} 
   \,, i_1 = 1 \,, \ldots \,, d_2$.

  \item Local bases are defined as: $\xi_i \,, \tilde{\xi}_i \in \mathbb R^{n_1}$,
  \begin{align}  \label{eq:LocalBasis}
   \underline{\Xi} &= \begin{pmatrix} \xi_1 & \hdots & \xi_d \end{pmatrix} \,,~
   \underline{\tilde{\Xi}} = \begin{pmatrix} \tilde{\xi}_1 & \hdots & \tilde{\xi}_d \end{pmatrix}
   \in \mathbb R^{n_1 \times d} \,,
  \end{align}
  
   \end{itemize}
 \item Filter banks and their family of matrices for Unet-VAE: 
  \begin{align} % \label{eq:FilterBanks:UnetVAE}
   \underline{\underline{\underline{\theta}}} &= \left\{
   \underline{\underline{\theta_1}} \,, \hdots \,, \underline{\underline{\theta_{Q}}} \right\} 
   \in \mathbb R^{d_1 \times d_2 \times P \times Q} \,,
   \notag
   % -----
   \underline{\underline{\theta_q}} = \left\{ \underline{\theta_{q,1}} \,, \hdots \,, \underline{\theta_{q,P}} \right\} \in \mathbb R^{d_1 \times d_2 \times P} \,,~
   \underline{\theta_{q,p}} \in \mathbb R^{d_1 \times d_2} \,,  
  \end{align}
  whose matrix form is: 
  \begin{align}
   \label{eq:FilterBanks:UnetVAE:matrixform}
   \underline{\underline{\Theta}} &= \left\{ \underline{\Theta_{1}}
   \,, \hdots \,, 
   \underline{\Theta_{Q}}\right\}
   \in \mathbb R^{d_1 d_2 P \times d_2 \times Q} \,,
   % ----
   \underline{\Theta_{q}} = \begin{pmatrix} \underline{\Theta_{q,1}} \\ \vdots \\ \underline{\Theta_{q,P}} \end{pmatrix} 
   \in \mathbb R^{d_1 d_2 P \times d_2} \,, 
   \\  % ----
   \underline{\Theta_{q,p}} &= \begin{pmatrix} \underline{\Theta_{q,p}^{1}} \\ \vdots \\ \underline{\Theta_{q,p}^{d_2}} \end{pmatrix}
   = \begin{pmatrix} \theta_{q,p,1}^{1} & \hdots & \theta_{q,p,d_2}^{1}
    \\ \vdots & \ddots & \vdots \\
    \theta^{n_2}_{q,p,1} & \hdots & \theta_{q,p,d_2}^{d_2}
   \end{pmatrix}
   = \begin{pmatrix} \theta_{q,p,1} & \hdots & \theta_{q,p,d_2} \end{pmatrix}
   \in \mathbb R^{d_1 d_2 \times d_2} \,,
   \notag
   \\ % ---
   \theta_{q,p,i_2}^{i_1}[i_3] &:= \theta_{q,p}[i_3, i_2-i_1] \,,~ i_1 = 1 \,, \ldots \,, d_2 \,,~  \theta_{q,p,i_2}^{i_1} \in \mathbb R^{d_1} \,.
   \notag
  \end{align}

\item Convolution Operations:
Given filter bank matrices $\left( \underline{\phi} \,, \underline{\underline{\psi}} \,, \underline{\underline{\underline{\theta}}} \right)$
with their matrix forms $\left( \underline{\Phi} \,, \underline{\underline{\Psi}} \,, 
\underline{\underline{\Theta}} \right)$
in (Equations~\ref{eq:RieszQuincunxFilter}-\ref{eq:FilterBanks:UnetVAE:matrixform}) respectively, we have convolution operations for a matrix $\underline{f} \in \mathbb R^{n_1 \times n_2}$ and a multi-band image $\underline{\underline{f}} = \{ \underline{f_1} \,, \ldots \,, \underline{f_{P}} \} \in \mathbb R^{\abs{\Omega} \times P}$:
\begin{itemize}
 \item {\bfseries 1D convolution:}
  \begin{align}  \label{eq:1DConvolution}
   \mathfrak{C}_{\phi_{i_1}^{i_2}} &\,:\, \mathbb R^{n_1} \rightarrow \mathbb R^{n_1} \,;~
   \mathfrak{C}_{\phi_{i_1}^{i_2}}\left(f_{i_2}\right) = \left( \sum_{k_1=1}^{d_1} f_{i_2}[k_1] \check{\phi}_{i_1}^{i_2}[k_2-k_1] \right)_{k_2=1}^{n_1}   
   = \mathscr{H}_{d_1} \left( f_{i_2} \right) \phi_{i_1}^{i_2} \,,
   \end{align}
   \item{\bfseries 2D convolution:}
   \begin{align}
   \label{eq:2DConvolution}
   \underline{\mathfrak{C}}_{\underline{\phi}} &\,:\, \mathbb R^{n_1 \times n_2} \rightarrow \mathbb R^{n_1 \times d_2} \,;~
   \underline{\mathfrak{C}}_{\underline{\phi}}(\underline{f}) 
   = \left( \sum_{k_1=1}^{d_1} \sum_{k_2=1}^{d_2} f[k_1,k_2] \check{\phi}[r_1-k_1, r_2-k_2] \right)_{r_1=1, \ldots, n_1}^{r_2=1, \ldots, d_2} \,,
  \end{align} 

 % ---
 \item {\bfseries Matrix-family convolution:}
  \begin{align} \label{eq:MIMOConvolution}
   \underline{\underline{\mathfrak{C}}}_{\underline{\underline{\psi}}} &\,:\, \mathbb R^{n_1 \times n_2} \rightarrow \mathbb R^{n_1 \times d_2 \times P} \,;~
   % ---
   \underline{\underline{\mathfrak{C}}}_{\underline{\underline{\psi}}} ( \underline{f} )
   = \left\{ \underline{\mathfrak{C}}_{\underline{\psi_1}}(\underline{f}) \,, \hdots \,,  \underline{\mathfrak{C}}_{\underline{\psi_{P}}}(\underline{f}) \right\} 
   = \mathscr{H}_{n_1 \mid d_2}(\underline{f}) \underline{\underline{\Psi}} \,,
  \end{align}

  % ----
  \item {\bfseries Anisotropic matrix-family convolution:}
   \begin{align*} 
    \underline{\underline{\underline{\mathfrak{C}}}}^\text{ani}_{\underline{\underline{\psi}}} &\,:\, \mathbb R^{n_1 \times n_2 \times P} \rightarrow \mathbb R^{n_1 \times d_2 \times P} \,;~
    % ---
    \end{align*}
    
    \begin{align*}
    \underline{\underline{\underline{\mathfrak{C}}}}^\text{ani}_{\underline{\underline{\psi}}} \left( \underline{\underline{f}} \right)
    = \left\{ \underline{\mathfrak{C}}_{\underline{\psi_1}}(\underline{f_1}) \,, \hdots \,,  \underline{\mathfrak{C}}_{\underline{\psi_{P}}}(\underline{f_{P}}) \right\} 
    = \left\{ \mathscr{H}_{n_1 \mid d_2}(\underline{f_1}) \underline{\Psi_1} \,, \hdots \,,  \mathscr{H}_{n_1 \mid d_2}(\underline{f_{P}}) \underline{\Psi_{P}} \right\} \,.
   \end{align*}

  % ----
  \item {\bfseries Isotropic matrix-family convolution:}
   \begin{align*}
    \underline{\underline{\underline{\mathfrak{C}^\text{iso}}}}_{\underline{\underline{\underline{\theta}}}} &\,:\, \mathbb R^{n_1 \times n_2 \times P} \rightarrow \mathbb R^{n_1 \times d_2 \times Q} \,, 
    % ----
    \underline{\underline{\tilde f}} = \left\{ \underline{\tilde f_1} \,, \ldots \,, \underline{\tilde f_{Q}} \right\} 
    % ---
    = \underline{\underline{\underline{\mathfrak{C}^\text{iso}}}}_{\underline{\underline{\underline{\theta}}}} \left( \underline{\underline{f}} \right)
    = \mathscr{H}_{n_1 \mid d_2 \mid P} \left( \underline{\underline{f}} \right)
    \underline{\underline{\Theta}}
   \end{align*}
   with:
   \begin{align*}
    \underline{\tilde f_q} &= \sum_{p=1}^{P} \underline{\mathfrak{C}}_{\underline{\theta_{q,p}}}(\underline{f_p}) 
    = \sum_{p=1}^{P} \mathscr{H}_{n_1 \mid d_2}(\underline{f_p}) \underline{\Theta_{q,p}}
    = \mathscr{H}_{n_1 \mid d_2 \mid P} \left( \underline{\underline{f}} \right) \underline{\Theta_{q}} \,. 
   \end{align*}
 
\end{itemize}
  
\end{itemize}

Then, we have the following relation for 2D convolution operation:
\begin{prop}  \label{prop:2DconvolutionHankelmatrix}
 A 2D convolution is defined by Hankel matrix by describing it via a 1D convolution:
 \begin{align*} 
  \underline{\mathfrak{C}}_{\underline{\phi}}(\underline{f})[k_1,k_2] 
  &= \sum_{i=1}^{d_2} \mathfrak{C}_{\phi_{k_2}^{i}}(f_{i})[k_1] 
  = \sum_{i=1}^{d_2} \left( \mathscr{H}_{d_1}(f_i) \phi_{k_2}^{i} \right)[k_1] \,,
  k_1=1, \ldots, n_1 \,, k_2=1, \ldots, d_2 
 \end{align*} 
 which is written in a matrix form as:
 \begin{align}
  \label{eq:2DSISOConvolution}
  \underline{\mathfrak{C}}_{\underline{\phi}} \left( \underline{f} \right) 
  &= \mathscr{H}_{n_1 \mid d_2} \left( \underline{f} \right) \underline{\Phi} \,.
 \end{align}
 
\end{prop}
\begin{proof}
We provide a proof of Proposition~\ref{prop:2DconvolutionHankelmatrix} in Section \ref{proof:prop2DconvolutionHankel}.
\end{proof}
Moreover, note that we have an adjoint operator 
$\underline{\underline{\mathfrak{C}}}^*_{\underline{\underline{\psi}}} = \sum_{p=1}^P \mathfrak{C}^*_{\psi_p} \,:\, \mathbb R^{n_1 \times d_2 \times P} \rightarrow \mathbb R^{n_1 \times n_2}$.
And, for $d_1 = n_1 \,, d_2 = n_2$, we have
$\underline{\mathfrak{C}}_{\underline{\phi}}(\underline{f}) = \mathcal{F}^{-1} \left( \underline{\widehat{P}} \odot \underline{\widehat{F}} \right)$, $\underline{\widehat{P}} = \left[ \widehat{P} \left( e^{j \omega} \right) \right]_{\omega \in [-\pi, \pi]^d}$ 
and 
$\underline{\widehat{F}} = \left[ \widehat{F} \left( e^{j \omega} \right) \right]_{\omega \in [-\pi, \pi]^d}$.
Its adjoint operator 
$\underline{\mathfrak{C}}^*_{\underline{\phi}} = \underline{\mathfrak{C}}_{\underline{\check{\phi}}} \,:\, \mathbb R^{n \times m} \rightarrow \mathbb R^{n \times m}$
is defined with a discrete time-reversed kernel 
\begin{align*}
 \check{\phi}(x_1,x_2) = \phi(-x_1,-x_2) \stackrel{\cF}{\longleftrightarrow} \widehat{\check{\phi}}(\omega_1, \omega_2) = \widehat{\phi}^*(\omega_1, \omega_2)
\end{align*}
whose discretized version is:
\begin{align*}
 \check{\phi}[k_1,k_2] = \phi[-k_1,-k_2] \stackrel{\cF}{\longleftrightarrow} \widehat{P}^*\left(e^{j \omega_1}, e^{j \omega_2}\right) = \widehat{P}\left(e^{-j \omega_1}, e^{-j \omega_2} \right) \,.
\end{align*}

% ------------------

%
\subsection{Proximal Operators and Moreau-Yosida envelope}
\begin{def}
\label{def:MoreauYosidaEnvelope}
 Given a non-smooth function
$\mathscr{P}(\cdot)$ which can be convex or non-convex,
 its Moreau-Yosida envelope ($\frac{1}{\mu} > 0$-Lipschitz differentiable) is 
 \begin{align}  \label{eq:MoreauYosidaEnvelope}
  \mathscr{P}_\mu(\cdot) := \inf_{\underline{u} \in \mathbb R^{n_1 \times n_2}} \left\{ \mathscr{P} (\underline{u}) + \frac{1}{2\mu} \norm{ \underline{u} - \cdot }^2_\text{F} \right\} \,,
 \end{align}
 with $\lim_{\mu \rightarrow 0} \mathscr{P}_\mu(\theta) = \mathscr{P}(\theta)$ and $\norm{\cdot}_\text{F}$ is Frobenius norm.
 Its gradient 
 \begin{align}  \label{eq:gradientMYenvelope}
  \nabla \mathscr{P}_{\mu}(\cdot) = \frac{1}{\mu} \left( \cdot - \text{prox}_{\mu \mathscr{P}} (\cdot) \right) 
  \,:\, \mathbb R^{n_1 \times n_2} \rightarrow \mathbb R^{n_1 \times n_2}
 \end{align}
 is defined by a proximity operator 
\begin{align}  \label{eq:GradientMoreauEnvelope}
 \text{prox}_{\mu \mathscr{P}} (\cdot) &:= \argmin_{\underline{u} \in \mathbb R^{n_1 \times n_2}} \left\{
 \mathscr{P} (\underline{u}) + \frac{1}{2\mu} \norm{ \underline{u} - \cdot }^2_\text{F} \right\} \,.
\end{align}
\end{def}
%
%For more proximal operators and Moreau-Yosida envelope, we refer the readers to \cite{PolsonScottWillard2015, Beck2017}. 
%
We also note that some activation functions in neural network can be well written as a proximal operator associated with its abstract function, e.g. Rectified Linear Unit (ReLU) activation function:
\begin{align*}
 \text{prox}_\text{ReLU} \left( \underline{f} \right) &= \argmin_{\underline{u} \in \mathbb R^{n_1 \times n_2}} \left\{
 \mathscr{P}_\text{ReLU} (\underline{u}) + \frac{1}{2} \norm{ \underline{u} - \underline{f} }^2_{\text{F}} \right\}
 \\
 &= \left[ \max \left( 0 \,, f_{i_1,i_2} \right) \right]_{i_1,i_2=0}^{n_1-1,n_2-1} \,,~ 
 \underline{f} = \left[ f_{i_1,i_2} \right]_{i_1,i_2=0}^{n_1-1,n_2-1} \,,
\end{align*}
where $\mathscr{P}_\text{ReLU} (\cdot)$ is a non-defined function.

\subsection{Kullback Leibler divergence:}
Given distributions $\mathtt{F}(\text{d}z) = \mathtt{f}(z)\text{d}z$ and $\mathtt{G}(\text{d}z) = \mathtt{g}(z)\text{d}z$ on a domain $\mathbb R^d$, their KL-divergence is:
\begin{align}  \label{eq:KLdivergence}
 \text{KL}\left( \mathtt{F} \mid \mid \mathtt{G} \right) &= \mathbb E_{Z \sim \mathtt{F}} \left[ \log \frac{\mathtt{f}(Z)}{\mathtt{g}(Z)} \right]
 = \int_{\mathbb R^d} \mathtt{f}(z) \log \frac{\mathtt{f}(z)}{\mathtt{g}(z)} \text{d}z
 \notag
 \\ 
 &\geq 0 \,. 
\end{align}

\clearpage
\section{Proofs}\label{sec:proofs}

\subsection{Proof of Proposition II.1}
\label{proof:FrameletDecomposition}

Due to the unity condition (Equation 1) and 
$\mathscr{H}^\dagger_{d_1 \mid n_2} \circ \mathscr{H}_{d_1 \mid n_2} = \text{Id}$, convolutional framelet decomposition is:
\begin{align*}
 \underline{f} &= \mathscr{H}^\dagger_{d_1 \mid n_2} \bigg( \underline{\tilde{\Xi}} 
 \underline{\Xi}^\text{T} \mathscr{H}_{d_1 \mid n_2}\left( \underline{f} \right) \underline{\Phi}
 \underline{\tilde{\Phi}}^\text{T} \bigg)
\end{align*}
where: 
\begin{align*}
 \underline{c_f} &= \underline{\Xi}^\text{T} \mathscr{H}_{d_1 \mid n_2} \left( \underline{f} \right) \underline{\Phi} 
 = \begin{pmatrix} \xi_1^\text{T} \\ \vdots \\ \xi_d^\text{T} \end{pmatrix}
 \mathscr{H}_{d_1 \mid n_2}\left( \underline{f} \right) 
 \begin{pmatrix} \phi_1 & \hdots & \phi_{d_2} \end{pmatrix}
 = \begin{pmatrix} \xi_1^\text{T} \mathscr{H}_{d_1 \mid n_2}\left( \underline{f} \right) \phi_1 & \hdots & \xi_1^\text{T} \mathscr{H}_{d_1 \mid n_2} \left( \underline{f} \right) \phi_{d_2}
  \\ \vdots & \ddots & \vdots \\ 
  \xi_d^\text{T} \mathscr{H}_{d_1 \mid n_2}\left( \underline{f} \right) \phi_1 & \hdots & \xi_d^\text{T} \mathscr{H}_{d_1 \mid n_2}\left( \underline{f} \right) \phi_{d_2} 
 \end{pmatrix}
 \\ % ---
 &= \begin{pmatrix} c_{f,1} & \hdots & c_{f,d_2} \end{pmatrix}
 = \left[ c_{f}^{l,s} \right]_{l = 1, \ldots, d}^{s = 1, \ldots, d_2}
\end{align*}
where: 
\begin{align*}
 c_f^{l,s} &= \xi_l^\text{T} \mathscr{H}_{d_1 \mid n_2}\left( \underline{f} \right) \phi_s
 = \xi_l^\text{T} 
 \begin{pmatrix} \mathscr{H}_{d_1}\left(f_1\right) & \hdots & \mathscr{H}_{d_1}(f_{n_2}) \end{pmatrix} 
 \begin{pmatrix} \phi_s^1 \\ \vdots \\ \phi_s^{n_2} \end{pmatrix}
 = \sum_{i=1}^{n_2} \xi_l^\text{T} \mathscr{H}_{d_1}(f_i) \phi_s^i
 \\ 
 &= \sum_{i=1}^{n_2} \left\langle f_i \,, \mathfrak{C}_{\phi_s^i} \left( \xi_l \right) \right\rangle_{\ell_2} \,.
\end{align*}
The last equality is due to
$u^\text{T} \mathscr{H}_{d_1}(a) v = u^\text{T} \mathfrak{C}_v(a) = \langle a \,, \mathfrak{C}_v(u) \rangle_{\ell_2}$ for $u \,, v \,, a \in \mathbb R^d$.

Now, we expand $\underline{f}$:
\begin{align*}
 \underline{f} &= \mathscr{H}^\dagger_{d_1 \mid n_2} \left( \underline{\tilde{\Xi}} \underline{c_f} \underline{\tilde{\Phi}^\text{T}} \right)
 = \begin{pmatrix} \mathscr{H}^\dagger_{d_1} \left( \underline{\tilde{\Xi}} \underline{c_f}  \underline{\tilde{\Phi}^{1,\text{T}}} \right)
 & \hdots & 
 \mathscr{H}^\dagger_{d_1} \left( \underline{\tilde{\Xi}} \underline{c_f} \underline{\tilde{\Phi}^{n_2,\text{T}}} \right) \end{pmatrix}
 \\  % ----
 &= \begin{pmatrix} \mathscr{H}^\dagger_{d_1} \left( \sum_{s=1}^{d_2} \sum_{l=1}^d c_{f}^{l,s} \tilde{\xi}_l \tilde{\phi}^{1,\text{T}}_{s} \right)
 & \hdots & 
 \mathscr{H}^\dagger_{d_1} \left( \sum_{s=1}^{d_2} \sum_{l=1}^d c_{f}^{l,s} \tilde{\xi}_l \tilde{\phi}^{n_2,\text{T}}_{s} \right) \end{pmatrix}
 \\  % ----
 &= \begin{pmatrix} \frac{1}{d_1} \sum_{s=1}^{d_2} \sum_{l=1}^d c^{l,s}_f 
 \mathfrak{C}_{\tilde{\phi}^{1}_s} \left( \tilde{\xi}_l \right)
 & \hdots & 
 \frac{1}{d_1} \sum_{s=1}^{d_2} \sum_{l=1}^d c^{l,s}_f
 \mathfrak{C}_{\tilde{\phi}^{n_2}_s} \left( \tilde{\xi}_l \right) \end{pmatrix}
 \\  % ----
 &= \frac{1}{d_1} \sum_{s=1}^{d_2} \begin{pmatrix}  
 \mathfrak{C}_{\tilde{\phi}^{1}_s} \left( \underline{\tilde{\Xi}} c_{f,s} \right)
 & \hdots & 
 \mathfrak{C}_{\tilde{\phi}^{n_2}_s} \left( \underline{\tilde{\Xi}} c_{f,s} \right) \end{pmatrix} \,.
\end{align*}
This concludes the proof. Note that the 3rd equality is due to:
\begin{align*}
 \underline{\tilde{\Xi}} \underline{c_f} \underline{\tilde{\Phi}^{i,\text{T}}} 
 &= \begin{pmatrix} \tilde{\xi}_1 & \hdots & \tilde{\xi}_m \end{pmatrix}
 \begin{pmatrix} c_{f}^{1,1} & \hdots & c_{f}^{1,d_2} 
  \\ \vdots & \ddots & \vdots \\
  c_{f}^{d,1} & \hdots & c_{f}^{d,d_2}
 \end{pmatrix}
 \begin{pmatrix} \tilde{\phi}^{i}_{1,\text{T}} \\ \vdots \\ \tilde{\phi}^{i,\text{T}}_{d_2}
 \end{pmatrix}
 % ----
 = \sum_{s=1}^{d_2} \sum_{l=1}^d \tilde{\xi}_l c_{f}^{l,s} \tilde{\phi}^{i,\text{T}}_{s} \,;
\end{align*}
and, the 4th equality is due to:
\begin{align*}
 \mathscr{H}^\dagger_{d_1} \left( \sum_{s=1}^{d_2} \sum_{l=1}^d c_f^{l,s} \tilde{\xi}_l \tilde{\phi}^{n_2,\text{T}}_s \right) 
 &= \frac{1}{\sqrt{d_1}} \sum_{s=1}^{d_2} \sum_{l=1}^d c_f^{l,s}
 \begin{pmatrix}
  \left\langle \frac{1}{\sqrt{d_1}} \mathscr{H}_{d_1}\left( \underline{e_1} \right) \,, \tilde{\xi}_l \tilde{\phi}^{n_2,\text{T}}_s \right\rangle_\text{F}
  \\ \vdots \\
  \left\langle \frac{1}{\sqrt{d_1}} \mathscr{H}_{d_1} \left( \underline{e_{n_1}} \right) \,, \tilde{\xi}_l \tilde{\phi}^{n_2,\text{T}}_s \right\rangle_\text{F}
 \end{pmatrix}
 \\  % ----
 &= \frac{1}{d_1} \sum_{s=1}^{d_2} \sum_{l=1}^d c^{l,s}_f
 \begin{pmatrix}
  \left\langle \underline{e_1} \,, \mathfrak{C}_{\tilde{\phi}^{n_2}_s}\left(\tilde{\xi}_l\right) \right\rangle_\text{F}
  \\ \vdots \\
  \left\langle \underline{e_{n_1}} \,, \mathfrak{C}_{\tilde{\phi}^{n_2}_s}
  \left( \tilde{\xi}_l \right) \right\rangle_\text{F}
 \end{pmatrix}
 % ----
 = \frac{1}{d_1} \sum_{s=1}^{d_2} \sum_{l=1}^d c^{l,s}_f 
 \mathfrak{C}_{\tilde{\phi}^{n_2}_s} \left( \tilde{\xi}_l \right) \,.
\end{align*}

\subsection{Proof of Proposition II.2}
\label{proof:UnityCondition}
Given a wavelet expansion acting on a discrete function $f \in \ell_2 \left( \mathbb Z^2 \right)$:
\begin{align*} 
 f[k] &= \sum_{m \in \mathbb Z^2}
 \left\langle f \,, \tilde \varphi_I (\cdot - m) \right\rangle_{\ell_2} \varphi_I (k - m)
 + \sum_{i=0}^I \sum_{l=0}^L \sum_{m \in \mathbb Z^2}
 \left\langle f \,, \tilde \psi_{il} (\cdot - m) \right\rangle_{\ell_2} \psi_{il} (k - m) \,,
\end{align*}
we compute its scaling and wavelet coefficients in the Fourier domain:
\begin{align}  \label{eq:ScalingCoefs}
 c_I[m] &= \left\langle f \,, \tilde{\varphi}_I \left( \cdot - m \right) \right\rangle_{\ell_2}
 \stackrel{\cF}{\longleftrightarrow} 
 \widehat{C}_I \left( e^{j \omega} \right) = \sum_{m \in \mathbb Z^2} c_I[m] e^{-j \langle m \,, \omega \rangle_{\ell_2} }
 = \sum_{k \in \mathbb Z^2} f[k] \sum_{m \in \mathbb Z^2} 
 \check{\tilde{\varphi}}_I \left(m - k \right) e^{-j \langle m \,, \omega \rangle_{\ell_2} } \,.
\end{align}
To compute a Fourier transform of $\check{\tilde{\varphi}}_I \left(m - k \right)$,
we find a Fourier transform of its continuous version for $x \in \mathbb R^d$:
\begin{align*}
 &\check{\tilde{\varphi}}_I \left(x - k \right) 
 ~\stackrel{\cF}{\longleftrightarrow}~
 e^{ - j \langle k \,, \omega \rangle_{\ell_2} }
 \int_{\mathbb R^2} \check{\tilde{\varphi}}_I (x) e^{ - j \langle x \,, \omega \rangle_{\ell_2} } 
 \text{d} x
 = \widehat{\tilde{\varphi}}^*_I (\omega)
 e^{ - j \langle k \,, \omega \rangle_{\ell_2} } \,.
\end{align*}
By Poisson summation formulae, we have the following identity:
\begin{align*}
 \sum_{m \in \mathbb Z^d} 
 \check{\tilde{\varphi}}_I \left(m - k \right) e^{-j \langle m \,, \omega \rangle_{\ell_2} }
 = \sum_{m \in \mathbb Z^d} \widehat{\tilde{\varphi}}^*_I \left( 2\pi m + \omega \right)
 e^{ - j \langle k \,, 2\pi m + \omega \rangle_{\ell_2} }
\end{align*}
Note $e^{j 2\pi m} = 1 \,, \forall m \in \mathbb Z^d$;
then, scaling coefficient 
$c_I[m] \stackrel{\cF}{\longleftrightarrow} \widehat{C}_I \left( e^{j \omega} \right)$ in (\ref{eq:ScalingCoefs}) is
\begin{align}  \label{eq:ScalingCoefs:Fourier}
 \widehat{C}_I \left( e^{j \omega} \right) &= \sum_{k \in \mathbb Z^2} f[k] \sum_{m \in \mathbb Z^2} \widehat{\tilde{\varphi}}^*_I \left( 2\pi m + \omega \right)
 e^{ - j \langle k \,, 2\pi m + \omega \rangle_{\ell_2} }
 = \sum_{m \in \mathbb Z^2} \widehat{\tilde{\varphi}}^*_I \left( 2\pi m + \omega \right)
 \widehat{F} \left( e^{j (2\pi m + \omega) } \right)
 \notag
 \\  % ----
 &= \widehat{F}\left( e^{j \omega} \right) \left[ \widehat{\tilde{\varphi}}^*_I (\omega)
 + \sum_{m \in \mathbb Z^2 \backslash \{0\}} \widehat{\tilde{\varphi}}^*_I \left( 2\pi m + \omega \right)
 \right] \,,
\end{align}
where $f[k] \stackrel{\cF}{\longleftrightarrow} \widehat{F}\left( e^{j \omega} \right)$.
Similarly, wavelet coefficient $d_{il}[m] \stackrel{\cF}{\longleftrightarrow} \widehat{D}_{il}\left( e^{j \omega} \right)$ is:
\begin{align}  \label{eq:WaveletCoefs:Fourier}
 \widehat{D}_{il} \left( e^{j \omega} \right) 
 &= \widehat{F}\left( e^{j \omega} \right) \left[ \widehat{\tilde{\psi}}^*_{il} (\omega)
 + \sum_{m \in \mathbb Z^2 \backslash \{0\}} \widehat{\tilde{\psi}}^*_{il} \left( 2\pi m + \omega \right)
 \right] \,.
\end{align}
Now, we compute wavelet expansion in the Fourier domain as: 
\begin{align*}  
 \widehat{F}\left( e^{j \omega} \right) &= \sum_{k \in \mathbb Z^2} f[k] e^{-j \langle k \,, \omega \rangle_{\ell_2} }
 \\ % -----
 &= \sum_{m \in \mathbb Z^2} c_I[m] \sum_{k \in \mathbb Z^2} \varphi_I \left( k - m \right) e^{-j \langle k \,, \omega \rangle_{\ell_2}}
 + \sum_{i=0}^I \sum_{l=0}^L \sum_{m \in \mathbb Z^2} d_{il}[m]
 \sum_{k \in \mathbb Z^2} \psi_{il} \left( k - m \right) e^{-j \langle k \,, \omega \rangle_{\ell_2}} 
 \\ % -----
 &= \sum_{k \in \mathbb Z^2} \widehat{\varphi}_I \left( 2\pi k + \omega \right)
 \sum_{m \in \mathbb Z^2} c_I[m]
 e^{ - j \langle m \,, 2\pi k + \omega \rangle_{\ell_2} }
 + \sum_{i=0}^I \sum_{l=0}^L 
 \sum_{k \in \mathbb Z^2} \widehat{\psi}_{il} \left( 2\pi k + \omega \right)
 \sum_{m \in \mathbb Z^2} d_{il}[m]
 e^{ - j \langle m \,, 2\pi k + \omega \rangle_{\ell_2} }
 \\ % -----
 &= \widehat{C}_I \left( e^{j \omega} \right)
 \left( \widehat{\varphi}_I (\omega)
 + \sum_{k \in \mathbb Z^2 \backslash \{ 0 \}} \widehat{\varphi}_I \left( 2\pi k + \omega \right) \right)
 + \sum_{i=0}^I \sum_{l=0}^L 
 \widehat{D}_{il}\left( e^{j \omega} \right)
 \left( 
 \widehat{\psi}_{il}(\omega)
 + \sum_{k \in \mathbb Z^2 \backslash \{ 0 \}} \widehat{\psi}_{il} \left( 2\pi k + \omega \right) \right) \,.
\end{align*}
The 3rd equality is from Poisson summation formulae.
We obtain the unity condition (Equation 6) from Equation 16 and Equation 17.
To compensate the error in Equation 6, we modify a primal wavelet frame at scale 0 as:
\begin{align*}  
 &1 = \widehat{\tilde{\varphi}}^*_I (\omega) \widehat{\varphi}_I (\omega)
 + \sum_{i=1}^I \sum_{l=0}^L 
 \widehat{\tilde{\psi}}^*_{il} (\omega) \widehat{\psi}_{il}(\omega)
 + \sum_{l=0}^L 
 \widehat{\tilde{\psi}}^*_{0l} (\omega) \widehat{\psi}_{0l}(\omega)
 + \widehat{e}(\omega) 
 \\ % ----
 \Leftrightarrow~&
 \widehat{\bar \psi}_{0}(\omega) := \widehat{\psi}_{0}(\omega)
 + \frac{\widehat{e}(\omega)}{\widehat{\tilde{\psi}}^*_{0} (\omega)} \,. 
\end{align*}
This is due to a unity property of $L$-th Riesz transform $\sum_{l=0}^L \abs{\widehat{\mathscr{R}^l}(\omega)}^2 = 1$.
Finally, the unity condition (Equation 6) becomes:
\begin{align*}
 \widehat{\tilde{\varphi}}^*_I (\omega) \widehat{\varphi}_I (\omega)
 + \widehat{\tilde{\psi}}^*_{0} (\omega) \widehat{\bar \psi}_{0}(\omega) 
 + \sum_{i=1}^I \sum_{l=0}^L 
 \widehat{\tilde{\psi}}^*_{il} (\omega) \widehat{\psi}_{il}(\omega)
 = 1 \,.
\end{align*}
Then, multiply 2 sides with $F\left( e^{j \omega} \right)$ and take an inverse Fourier transform, we obtain a wavelet expansion:
\begin{align*}
 &\widehat{F}\left( e^{j \omega} \right) = 
 \widehat{F}\left( e^{j \omega} \right) \widehat{\tilde{\varphi}}^*_I (\omega) \widehat{\varphi}_I (\omega)
 + \sum_{i=0}^I \sum_{l=0}^L 
 \widehat{F}\left( e^{j \omega} \right) \widehat{\tilde{\psi}}^*_{il} (\omega) \widehat{\psi}_{il}(\omega)
 \\ % ----
 \stackrel{\cF}{\longleftrightarrow}~&
 f[k] = \sum_{m \in \mathbb Z^2}
 \left\langle f \,, \tilde \varphi_I (\cdot - m) \right\rangle_{\ell_2} \varphi_I (k - m)
 + \sum_{i=0}^I \sum_{l=0}^L \sum_{m \in \mathbb Z^2} 
 \left\langle f \,, \tilde \psi_{il} (\cdot - m) \right\rangle_{\ell_2} \psi_{il} (k - m)
\end{align*}
with $\widehat{\psi}_{0n}(\omega) = \widehat{\bar \psi}_{0n}(\omega)$.

% =================

%
\subsection{Proof of Proposition II.3}
\label{proof:RieszQuincunxWavelet:FrameletDecomposition}
Denote $\underline{\psi_0} := \underline{\varphi_I}$ and $\underline{\tilde\psi_0} := \underline{\tilde\varphi_I}$ whose matrix forms are $\underline{\Psi_0} = \underline{\Phi_I} \,, \underline{\tilde\Psi_0} = \underline{\tilde\Phi_I} \in \mathbb R^{n_1 n_2 \times n_2}$ as in (\ref{eq:RieszQuincunxFilter}), respectively.
Note that $\underline{\mathfrak{C}}_{\underline{\phi}} := \underline{\mathfrak{C}}^*_{\underline{\varphi_I}} \underline{\mathfrak{C}}_{\underline{\tilde\varphi_I}}$.
A convolutional form of an expansion (\ref{eq:RieszQuincunxWavelet}) is recast with Hankel matrix as:
\begin{align*} % \label{eq:RieszQuincunxWavelet}  
 \underline{f} &= \underline{\mathfrak{C}}^*_{\underline{\varphi_I}} \underline{\mathfrak{C}}_{\underline{\tilde\varphi_I}}\left(\underline{f}\right) 
 + \sum_{p=1}^P \underline{\mathfrak{C}}^*_{\underline{\psi_p}}
 \underline{\mathfrak{C}}_{\underline{\tilde\psi_p}}\left(\underline{f}\right) 
 = \sum_{p=0}^P \underline{\mathfrak{C}}^*_{\underline{\psi_p}}
 \underline{\mathfrak{C}}_{\underline{\tilde\psi_p}}\left(\underline{f}\right)
 % ----
 = \sum_{p=0}^P \mathscr{H}_{n_1 \mid n_2} \left( \underline{w_{f,p}} \right) \underline{\check \Psi_p}
 \\ % ----
 &= \sum_{p=0}^P \begin{pmatrix} \mathscr{H}_{n_1} \left( w_{f,p,1} \right) & \hdots & \mathscr{H}_{n_1} \left( w_{f,p,n_2} \right) \end{pmatrix} 
 \begin{pmatrix} \underline{\check\Psi^1_p} \\ \vdots \\ \underline{\check\Psi^{n_2}_p} \end{pmatrix} 
 % ----
 = \sum_{p=0}^P \sum_{i=1}^{n_2} \mathscr{H}_{n_1} \left( w_{f,p,i} \right)
 \underbrace{ \begin{pmatrix} 
  \check{\psi}_{p,1}^{i} & \hdots & \check{\psi}_{p,d_2}^{i}
 \end{pmatrix} }_{ = \underline{\check\Psi^{i}_p} }
 \\  % ----
 &= n_1 \sum_{p=0}^P \frac{1}{n_1} \sum_{i=1}^{n_2} 
 \begin{pmatrix} 
  \mathfrak{C}_{\check{\psi}_{p,1}^{i}} \left( \underline{\tilde\Xi} w_{f,p,i} \right)
  & \hdots & 
  \mathfrak{C}_{\check{\psi}_{p,n_2}^{i}} \left( \underline{\tilde\Xi} w_{f,p,i} \right) 
 \end{pmatrix}
 \\
 &= n_1 \sum_{p=0}^P \mathscr{H}^\dagger_{n_1 \mid n_2} \left( \underline{w_{f,p}} \, \underline{\check{\Psi}_p^\text{T}} \right) 
 \\ % ---
 &= n_1 \mathscr{H}^\dagger_{n_1 \mid n_2} \left( \sum_{p=0}^P \underline{w_{f,p}} \, \underline{\check{\Psi}_p^\text{T}} \right)
 = n_1 \mathscr{H}^\dagger_{n_1 \mid n_2} \left( \underline{w_{f}} \underline{\check{\Psi}^\text{T}} \right) \,,
\end{align*}
where wavelet coefficient is
\begin{align*}
 \underline{w_{f,p}} &:= \underline{\mathfrak{C}}_{\underline{\tilde\psi_p}}\left(\underline{f}\right)
 = \mathscr{H}_{n_1 \mid n_2}\left(\underline{f}\right) \underline{\tilde\Psi_p}
 \\  % ---
 &= \begin{pmatrix} w_{f,p,1} & \hdots & w_{f,p,n_2} \end{pmatrix} \in \mathbb R^{n_1 \times n_2} \,,~
 p = 0 \,, \ldots \,, P 
 \\  % ---
 \Leftrightarrow~ 
 \underline{w_{f}} &:= \begin{pmatrix} \underline{w_{f,0}} & \hdots & \underline{w_{f,P}} \end{pmatrix}
 = \mathscr{H}_{n_1 \mid n_2}\left(\underline{f}\right)
 \underline{\tilde{\Psi}} \,.
\end{align*}
Note that the last 3rd equality is due to proposition \ref{prop:FrameletDecomposition} with an identity local basis $\underline{\tilde\Xi} = \text{Id}$. 
The 2nd last equality is because $\mathscr{H}^\dagger_{n_1 \mid n_2}$ is a linear operator: 
Given a matrix 
$\underline{g_p} = \begin{pmatrix} \underline{g_{p,1}} & \hdots & \underline{g_{p,n_2}} \end{pmatrix}$ and constant $\{a_p\}_{p=0}^P \subset \mathbb R$, we have:
\begin{align*}    
 \sum_{p=0}^P a_p \mathscr{H}^\dagger_{n_1 \mid n_2} \left( \underline{g_p} \right)
 = \begin{pmatrix}
  \sum_{p=0}^P a_p \mathscr{H}^\dagger_{n_1} \left( \underline{g_{p,1}} \right) %
  & \hdots &
  \sum_{p=0}^P a_p \mathscr{H}^\dagger_{n_1} \left( \underline{g_{p,n_2}} \right)
 \end{pmatrix} \,.
\end{align*}
And, by a definition of an inverse Hankel matrix (\ref{eq:InverseHankelMatrix}), we have:
\begin{align*}
 \sum_{p=0}^P a_p \mathscr{H}^\dagger_{n_1} \left( \underline{g_{p,i}} \right)
 = \frac{1}{\sqrt{n_1}} 
 \begin{pmatrix} \left\langle \underline{\tilde{e}_1} \,, \sum_{p=0}^P a_p \underline{g_{p,i}} \right\rangle_\text{F} 
 \\ \vdots \\ 
 \left\langle \underline{\tilde{e}_{n_1}} \,, \sum_{p=0}^P a_p \underline{g_{p,i}} \right\rangle_\text{F} \end{pmatrix} 
 = \mathscr{H}^\dagger_{n_1} \left( \sum_{p=0}^P a_p \underline{g_{p,i}} \right) \,.
\end{align*}
Thus, we have a linear property of an extended Hankel matrix:
\begin{align*}    
 \sum_{p=0}^P a_p \mathscr{H}^\dagger_{n_1 \mid n_2} \left( \underline{g_p} \right)
 &= \begin{pmatrix}
  \mathscr{H}^\dagger_{n_1} \left( \sum_{p=0}^P a_p \underline{g_{p,1}} \right) %
  & \hdots &
  \mathscr{H}^\dagger_{n_1} \left( \sum_{p=0}^P a_p \underline{g_{p,n_2}} \right)
 \end{pmatrix} 
 \\ % ---
 &\stackrel{(\ref{eq:InverseExtendedHankelMatrix})}{=} \mathscr{H}^\dagger_{n_1 \mid n_2} \left(
  \begin{pmatrix} \sum_{p=0}^P a_p \underline{g_{p,1}} & \hdots & \sum_{p=0}^P a_p \underline{g_{p,n_2}}
 \end{pmatrix}
 \right) 
 \\ % ---
 &= \mathscr{H}^\dagger_{n_1 \mid n_2} \left( \sum_{p=0}^P a_p \underline{g_p} \right) \,.
\end{align*}
In the end, we have the proposed non-subsampled Riesz-Quincunx wavelet is a framelet decomposition with an identity local basis:
\begin{align*} % \label{eq:RieszQuincunxWavelet}  
 \underline{f} &= n_1 \mathscr{H}^\dagger_{n_1 \mid n_2} \left( \mathscr{H}_{n_1 \mid n_2}\left(\underline{f}\right)
 \underline{\tilde{\Psi}} \, \underline{\check{\Psi}^\text{T}} \right) \,.
\end{align*}
Since $\mathscr{H}^\dagger_{n_1 \mid n_2} \circ \mathscr{H}_{n_1 \mid n_2} = \text{Id}$, this implies the unity condition:
\begin{align*}
 \frac{1}{n_1} \text{Id}_{n_1 n_2 \times n_1 n_2}
 &= \underline{\tilde{\Psi}} \, \underline{\check{\Psi}^\text{T}} 
 = \underline{\tilde\Phi_I} \, \underline{\check\Phi_I^\text{T}}
 + \sum_{p=1}^P \underline{\tilde\Psi_p} \, \underline{\check{\Psi}_p^\text{T}} \,.
\end{align*}

% =================

%
\subsection{Proof of Proposition II.4}
\label{proof:SkipConnectingMapping}
The mapping $\mathscr{T}^{(i)}$ in  Equation 13 is defined with 
$\underline{\underline{\underline{\mathfrak{C}}}}^\text{iso}_{\underline{\underline{\underline{\theta^{1(i)}}}}} \,:\, \mathbb R^{n_1 \times n_2 \times P} \rightarrow \mathbb R^{n_1 \times d_2 \times 2^i L}$ and
$\underline{\underline{\underline{\mathfrak{C}}}}^\text{iso}_{\underline{\underline{\underline{\theta^{2(i)}}}}} \,:\, \mathbb R^{n_1 \times d_2 \times 2^i L} \rightarrow \mathbb R^{n_1 \times n_2 \times 2^i L}$.
Given a local basis $\underline{\Xi^{(i)}_{\text{aug}}} = \begin{pmatrix} \text{Id} & \underline{\Xi^{(i)}} \end{pmatrix}$, an encoder is defined with a lowpass signal and skip-connecting signal as outputs:
\begin{align}  %\label{eq:EncoderUnet:1}
 &\underline{\underline{c_\text{aug}^{(i)}}} = \underline{\Xi^{(i)  \text{T}}_{\text{aug}}} \mathscr{T}^{(i)} \left( \underline{\underline{s^{(i-1)}}} \right) 
 = \begin{pmatrix} \mathscr{T}^{(i)} \left( \underline{\underline{s^{(i-1)}}} \right) \\ \underline{\Xi^{(i),\text{T}}} \mathscr{T}^{(i)} \left( \underline{\underline{s^{(i-1)}}} \right) \end{pmatrix}
 = \begin{pmatrix} \underline{\underline{c^{(i)}}} \\ \underline{\underline{s^{(i)}}} \end{pmatrix} \,,~
 i = 1 \,, \ldots \,, I \,,~
 \underline{\underline{s^{(0)}}} = \underline{\underline{f}} \,,~
 \underline{\underline{s^{(I)}}} = \mathscr{R}_p \left( \underline{\underline{s^{(I)}}} \right)
 % ----
 \\  \label{eq:EncoderUnet:2}
 \Leftrightarrow~&
 \begin{cases} \displaystyle
  \underline{\underline{s^{(I)}}} = \mathscr{R}_p \circ \underline{\Xi^{(I),\text{T}}} \mathscr{T}^{(I)} \circ \ldots \circ \underline{\Xi^{(1),\text{T}}} \mathscr{T}^{(1)} \left( \underline{\underline{f}} \right) \,,
  % ---
  \\ \displaystyle
  \underline{\underline{c^{(i)}}} = \mathscr{T}^{(i)} \circ \underline{\Xi^{(i-1),\text{T}}} \mathscr{T}^{(i-1)} 
  \circ \ldots \circ 
  \underline{\Xi^{(1),\text{T}}} \mathscr{T}^{(1)} \left( \underline{\underline{f}} \right) \,,~ i = 1 \,, \ldots \,, I \,,
 \end{cases}
\end{align}
where
\begin{align*}
 \underline{\underline{c^{(i)}}} \in \mathbb R^{2^{-(i-1)} n_1 \times 2^{-(i-1)} n_2 \times 2^{(i-1)} L} \,,~
 \underline{\underline{s^{(i)}}} \in \mathbb R^{2^{-(i-1)-1} n_1 \times 2^{-(i-1)-1} n_2 \times 2^{(i-1)} L} \,.
\end{align*}
Each subband of the skip connection
$\underline{\underline{c^{(i)}}} = \left\{ \underline{c^{(i)}_{0}} \,, \ldots \,, \underline{c^{(i)}_{2^i L-1}} \right\} \,,
\underline{c^{(i)}_{l}} \in \mathbb R^{2^{-(i-1)} n_1 \times 2^{-(i-1)} n_2}$,
is passed to the $N$-th order Riesz Quincunx wavelet expansion (Equation 11) at scale $I'$:
$i= 0 \,, \ldots \,, I-1 \,, l = 0 \,, \ldots \,, 2^i L-1$, $k \in \mathbb Z^2$:
\begin{align*}    
 \underline{c^{(i)}_l} &= \underline{\mathfrak{C}}_{\underline{\phi}}\left(\underline{c^{(i)}_l}\right)  
 + \sum_{p=1}^P \underline{\mathfrak{C}}^*_{\underline{\psi_p}}
 \underline{\mathfrak{C}}_{\underline{\tilde\psi_p}}\left(\underline{c^{(i)}_l}\right) 
 = \underline{\mathfrak{C}}_{\underline{\phi}}(\underline{c^{(i)}_l})  
 + \underline{\underline{\mathfrak{C}}}^*_{\underline{\underline{\psi}}}
 \underline{\underline{\mathfrak{C}}}_{\underline{\underline{\tilde\psi}}}(\underline{c^{(i)}_l})
 \\  % ----
 &= \left( \mathscr{H}_{2^{-(i-1)}n_1 \mid 2^{-(i-1)}n_2}(\underline{c^{(i)}_l}) \underline{\Phi} \right)  
 + \sum_{p=1}^P \left( \mathscr{H}_{2^{-(i-1)}n_1 \mid 2^{-(i-1)}n_2} \left( \mathscr{H}_{2^{-(i-1)}n_1 \mid 2^{-(i-1)}n_2}\left(\underline{c^{(i)}_l}\right) \underline{\tilde\Psi_p} \right) \underline{\check \Psi_p} \right) 
 % ----
 \\ \Leftrightarrow
 \underline{c^{(i)}_l} &= \mathfrak{W} \, \mathfrak{W}^{-1} \left\{ \underline{c^{(i)}_l} \right\} \,, 
\end{align*}
which satisfies the condition of perfect reconstruction. Scaling and wavelet functions are defined in the previous section.

% =================

%
\subsection{Proof of Proposition II.5}
\label{proof:RQUnetVAE:VariationalTerm}
The proof is straight-forward by noting that a lowpass signal as an encoder's output in (\ref{eq:EncoderUnet:2})
\begin{align*}
 \underline{\underline{s^{(I)}}} = \mathscr{R}_p \circ \underline{\Xi^{(I),\text{T}}} \mathscr{T}^{(I)} \circ \ldots \circ \underline{\Xi^{(1),\text{T}}} \mathfrak{T}^{(1)} \left( \underline{\underline{f}} \right) \,.    
\end{align*}
And, unknown parameters are:
$\gamma_m := \left\{ \gamma_c \,, \underline{W^\mu} \,, b^\mu \right\}$ and
$\gamma_s := \left\{ \gamma_c \,, \underline{W^\sigma} \,, b^\sigma \right\}$.

% =================

%
\subsection{Proof of Proposition II.6}
\label{proof:RQUnetVAE:decoder}
Now, we define an ``inverse" mapping of (Equation 15) as a linear perceptron network and a reshape operation mapping:
\begin{align}  \label{eq:InverseMappingLatentVariable}
 \mathscr{F}^s(\cdot) = \text{uvec} \left( \underline{W^s} \cdot + b^s \right) 
 \,:\, \mathbb R^d \rightarrow
 \mathbb R^{2^{-i-1} n_1 \times 2^{-i-1} n_2 \times 2^i L}
\end{align}
with $\underline{W^s} \in \mathbb R^{2^{-(I+1)} n_1 n_2 L \times d}$ and $b^s \in \mathbb R^{2^{-(I+1)} n_1 n_2 L}$.

Then, combining Equations 16 and 20, a reconstructed lowpass image generated from the latent variable $y$ in Equation 16:
\begin{align*}
 \underline{\underline{\widehat{s}^{(I)}}} &=
 \mathscr{F}^s(y) 
 = \text{uvec} \left( \underline{W^s} \mathscr{M}_{\gamma_m} \left( \underline{\underline{f}} \right) 
 + b^s \right)
 + \text{uvec} \left( \left( \underline{W^s} \mathscr{S}_{\gamma_s}^\frac{1}{2} \left( \underline{\underline{f}} \right) \right) \odot \epsilon \right)
 \\  % ---
 &\in \mathbb R^{2^{-i-1} n_1 \times 2^{-i-1} n_2 \times 2^i L} \,,~
 \epsilon \stackrel{\text{i.i.d.}}{\sim} \mathcal{N}_d \left( \mathbf{0}_d \,, \text{Id}_d \right) \,.
\end{align*}
The last equality is because $\text{uvec}(\cdot)$ is a linear operator.

Given an augmented local basis $\underline{\tilde{\Xi}^{(i)}_\text{aug}} = \begin{pmatrix} \text{Id}_N & \mathscr{B} \circ \underline{\tilde{\Xi}^{(i)}} \end{pmatrix}$, 
a decoder at scale $i \in \{ I - 1 \,, \ldots \,, 0 \}$ is defined by an output encoder and a skip connection $\underline{\underline{\widehat{c}_\text{aug}^{(i)}}} = \begin{pmatrix} \underline{\underline{c^{(i)}}} & \underline{\underline{\widehat{s}^{(i+1)}}} \end{pmatrix}$ as:
$\underline{\underline{\widehat{s}^{(I)}}} = \mathfrak{F}^s(y)$,
\begin{align}  \label{eq:DecoderScale:i}
 &\begin{cases}
  \underline{\underline{\widehat{c}_\text{aug}^{(i-1)}}} = \tilde{\mathscr{T}}^{(i)} \left( \underline{\underline{\widehat{c}_\text{aug}^{(i)}}} \right) = \begin{pmatrix} \underline{\underline{c^{(i-1)}}} & \underline{\underline{\widehat{s}^{(i)}}} \end{pmatrix} \,,~
  i = I \,, \ldots \,, 1 \,,
  \\  % ----
  \underline{\underline{\widehat{s}^{(0)}}} = \tilde{\mathscr{T}}^{(0)} \left( \underline{\underline{\widehat{c}_\text{aug}^{(0)}}} \right) \,,~ i = 0 \,,
 \end{cases}
 % ---
 \\  \label{eq:DecoderScale:iterated_i}
 \Leftrightarrow~&
 \underline{\underline{\widehat{s}^{(0)}}} = \tilde{\mathscr{T}}^{(0)} \circ \tilde{\mathscr{T}}^{(1)} \circ \ldots \circ \tilde{\mathscr{T}}^{(I-1)} \begin{pmatrix} \underline{\underline{c^{(I-1)}}} & \underline{\underline{\widehat{s}^{(I)}}} \end{pmatrix} 
 := \tilde{\mathscr{T}}_I \left( \underline{\underline{\underline{c}}} \,, \mathfrak{F}^s(y) \right) \,.
\end{align}
Then, we have Proposition II.6.

% =================

%
\subsection{Proof of Proposition II.7}
\label{proof:RQUnetVAEparameters}
Since $\mathtt{p}\left(y \mid \underline{\underline{f}}\right) = \mathtt{k}_\alpha\left(y \mid \underline{\underline{f}}\right)$, a likelihood of the latent variable in an encoder (Equation 26) is computed by a marginal likelihood:
\begin{align*}
 \mathtt{p}(z) &= \int_{\mathcal{F}} \mathtt{p}\left(z \mid \underline{\underline{f}}\right) \mathtt{p}\left(\underline{\underline{f}}\right) \text{d}\underline{\underline{f}}
 % ---
 = \mathbb E_{\underline{\underline{f}} \sim \textfrak{F}} \left[ \mathtt{k}_\alpha \left(z \mid \underline{\underline{f}}\right) \right] \,.
\end{align*}
Assume we have a bunch of realization of latent variable $z$ as
$\textfrak{Z} = \left\{ z_i \right\}_{i=1}^{n_z} \subset \mathbb R^d$; then, we have a maximum of an expected log-marginal-likelihood $\mathtt{p}(z)$ as:
\begin{align*}  
 \alpha^\dagger &= \argmax_{\alpha} \mathbb E_{\underline{\underline{f}} \sim \textfrak{F}} \left[ \log \mathtt{k}_\alpha \left(\textfrak{Z} \mid \underline{\underline{f}}\right) \right]
 = \argmax_{\alpha} \mathbb E_{\underline{\underline{f}} \sim \textfrak{F}} \left[ \frac{1}{n_y} \sum_{i=1}^{n_y} \log \mathtt{k}_\alpha \left( z_i \mid \underline{\underline{f}} \right) \right] \,,~
 z_i \stackrel{\text{i.i.d.}}{\sim} \mathbb{K}_\alpha\left(z \mid \underline{\underline{f}}\right)
\end{align*}
% ---
which is
\begin{align}  \label{eq:VAE:likelihood:latent} 
 \left( \gamma_m^\dagger \,, \gamma_s^\dagger \,, \alpha^\dagger \right) &= \argmax_{\gamma_m, \gamma_s, \alpha} \mathbb E_{\underline{\underline{f}} \sim \textfrak{F}} \left[ \frac{1}{n_z} \sum_{i=1}^{n_z} \log \mathtt{k}_\alpha(z_i \mid \underline{\underline{f}}) \right] \,,~
 z_i \stackrel{\text{i.i.d.}}{\sim} \mathbb{G}_{\gamma_m, \gamma_s}(z \mid \underline{\underline{f}})
 = \mathtt{g}_{\gamma_m, \gamma_s}\left(z \mid \underline{\underline{f}}\right)\text{d}z
 \notag
 \\ % ---
 &\approx \argmax_{\gamma_m, \gamma_s, \alpha} \mathbb E_{\underline{\underline{f}} \sim \textfrak{F}} \left[ \mathbb E_{z \sim \mathbb{G}_{\gamma_m, \gamma_s}(\cdot \mid \underline{\underline{f}})} \left[ \log \mathtt{k}_\alpha\left(z \mid \underline{\underline{f}}\right) \right] \right]
 \notag
 \\ % ---
 &= \argmax_{\gamma_m, \gamma_s, \alpha} \left\{ \mathbb E_{\underline{\underline{f}} \sim \textfrak{F}} \left[  \mathbb E_{z \sim \mathbb{G}_{\gamma_m, \gamma_s}\left(\cdot \mid \underline{\underline{f}}\right)} \left[ \log \mathtt{k}_\alpha(z \mid \underline{\underline{f}}) \right]
 - \mathbb E_{z \sim \mathbb{G}_{\gamma_m, \gamma_s}(\cdot \mid \underline{\underline{f}})} \left[ \log \mathtt{g}_{\gamma_m,\gamma_s}(z \mid \underline{\underline{f}}) \right] \right] \right\}
 \notag %
 \\ % ---
 &= \argmax_{\gamma_m, \gamma_s, \alpha} \left\{ 
 - \mathbb E_{\underline{\underline{f}} \sim \textfrak{F}} \left[ \mathbb E_{z \sim \mathbb{G}_{\gamma_m, \gamma_s}\left(\cdot \mid \underline{\underline{f}}\right)} \left[ \log \frac{\mathtt{g}_{\gamma_m,\gamma_s}\left(z \mid \underline{\underline{f}}\right)}{\mathtt{k}_\alpha\left(z \mid \underline{\underline{f}}\right)} \right] \right]
 \right\}
 \notag
 \\  % ---
 &= \argmin_{\gamma_m, \gamma_s, \alpha} \mathbb E_{\underline{\underline{f}} \sim \textfrak{F}} \left[ \text{KL} \left( \mathbb{G}_{\gamma_m, \gamma_s}\left(z \mid \underline{\underline{f}}\right) \mid \mid \mathbb{K}_\alpha\left(z \mid \underline{\underline{f}}\right) \right) \right] \,.
\end{align}
In the end, we need to find $(\gamma_m \,, \gamma_s \,, \theta)$ that minimize the KL-distance between the true distribution and its approximated version in Equation 27 over all data $\textfrak{F}$ as:
\begin{align}  \label{eq:KLdivergece:approximation}
 \min_{\gamma_m, \gamma_s, \alpha} \mathbb{E}_{\underline{\underline{f}} \sim \textfrak{F}} \left[ \text{KL} \left( \mathbb{G}_{\gamma_m, \gamma_s} (z \mid \underline{\underline{f}}) \mid \mid \mathbb{K}_\alpha(z \mid \underline{\underline{f}}) \right) \right] \,.
\end{align}
From Bayes' rule (Equation 27),
we derive an evidence lower bound of the KL-divergence in (\ref{eq:KLdivergece:approximation}):
\begin{align*}
 &\text{KL} \left( \mathbb{G}_{\gamma_m,\gamma_s}(z \mid \underline{\underline{f}}) \mid \mid \mathbb{K}_\alpha(z \mid \underline{\underline{f}}) \right)
 \\
 =& \mathbb E_{z \sim \mathbb{G}_{\gamma_m, \gamma_s}(\cdot \mid \underline{\underline{f}})} \left[ \log \frac{\mathtt{g}_{\gamma_m,\gamma_s}(z \mid \underline{\underline{f}})}{\mathtt{k}_\alpha(z \mid \underline{\underline{f}})} \right] 
 % ---
 \stackrel{(27)}{=} 
 \mathbb E_{z \sim \mathbb{G}_{\gamma_m,\gamma_s}(\cdot \mid \underline{\underline{f}})} \left[ \log \frac{\mathtt{p}(\underline{\underline{f}}) \mathtt{g}_{\gamma_m,\gamma_s}(z \mid \underline{\underline{f}})}{\mathtt{h}_\alpha\left(\underline{\underline{f}} \mid z\right) \mathtt{p}(z)} \right]
 \\  % -----
 =& - \mathbb E_{z \sim \mathbb{G}_{\gamma_m,\gamma_s}\left(\cdot \mid \underline{\underline{f}}\right)} \left[ \log \mathtt{h}_\alpha\left(\underline{\underline{f}} \mid z\right) \right]
 + \mathbb E_{z \sim \mathbb{G}_{\gamma_m,\gamma_s}\left(\cdot \mid \underline{\underline{f}}\right)} \left[ \log \frac{\mathtt{g}_{\gamma_m,\gamma_s}\left(z \mid \underline{\underline{f}}\right)}{\mathtt{p}(z)} \right]
 + \log \mathtt{p}\left(\underline{\underline{f}}\right)
\end{align*}
 which is equivalent to:
\begin{align*} 
 \mathcal{L}\left(\mathbb{G}_{\gamma_m,\gamma_s} \,, \alpha \,; \underline{\underline{f}}\right) &:=
 \mathbb E_{z \sim \mathbb{G}_{\gamma_m,\gamma_s}\left(\cdot \mid \underline{\underline{f}}\right)} \left[ \log \mathtt{h}_\alpha\left(\underline{\underline{f}} \mid z\right) \right]
 - \text{KL} \left[ \mathbb{G}_{\gamma_m,\gamma_s}\left(z \mid \underline{\underline{f}}\right) \mid \mid \mathbb{P}(z) \right]
 \\ %
 &= \log \mathtt{p}\left(\underline{\underline{f}}\right)
 - \text{KL} \left( \mathbb{G}_{\gamma_m,\gamma_s}\left(z \mid \underline{\underline{f}}\right) \mid \mid \mathbb{K}_\alpha\left(z \mid \underline{\underline{f}}\right) \right)
 \\  %----
 &\leq \log \mathtt{p}\left(\underline{\underline{f}}\right) \,.
\end{align*}
The last inequality is due to $\text{KL}(\cdot \mid \mid \cdot) \geq 0$.
Note that, via a reparamization trick from Equation 20 
\begin{align*}  
 z &\stackrel{\text{i.i.d.}}{\sim} \mathcal{N}_d \left( \mathscr{M}_{\gamma_m} \left( \underline{\underline{f}} \right) \,, \text{diag} \left\{ \mathscr{S}_{\gamma_s} \left( \underline{\underline{f}} \right) \right\} \right)
 = \mathbb{G}_{\gamma_m,\gamma_s}\left(z \mid \underline{\underline{f}} \right) \,,~
 \\ \Leftrightarrow~
 z &= \mathscr{M}_{\gamma_m}\left(\underline{\underline{f}}\right) + \mathscr{S}_{\gamma_s}^\frac{1}{2}\left(\underline{\underline{f}}\right) \odot \epsilon \,,~ 
 \epsilon \stackrel{\text{i.i.d.}}{\sim} \mathcal{N}(0 \,, \text{Id}) 
\end{align*}
an evidence bound is recast as:
\begin{align*}
 \mathcal{L} \left(\mathbb{G}_{\gamma_m,\gamma_s} \,, \alpha \,; \underline{\underline{f}}\right) &=
 \mathbb E_{\epsilon \sim \mathcal{N}(0 \,, \text{Id})} \left[ \log \mathtt{h}_\alpha\left(\underline{\underline{f}} \mid z = \mathscr{M}_{\gamma_m}\left(\underline{\underline{f}}\right) + \mathscr{S}_{\gamma_s}^\frac{1}{2} \left(\underline{\underline{f}}\right) \odot \epsilon\right) \right]
 - \text{KL} \left[ \mathbb{G}_{\gamma_m,\gamma_s}\left(z \mid \underline{\underline{f}}\right) \mid \mid \mathbb{P}(z) \right] \,.
\end{align*}
Then, a minimization (\ref{eq:KLdivergece:approximation}) is equivalent to maximize an expected  evidence bound over all observation $\textfrak{F}$:
\begin{align}  \label{eq:ExpectedEvidenceBound}  
 \left( \gamma_m^\dagger \,, \gamma_s^\dagger \,, \alpha^\dagger \right) &= \argmax_{\gamma_m, \gamma_s, \alpha} 
 \mathbb{E}_{\underline{\underline{f}} \sim \textfrak{F}} \left[ \log \mathtt{p}\left(\underline{\underline{f}}\right) - \text{KL} \left( \mathbb{G}_{\gamma_m, \gamma_s}\left(z \mid \underline{\underline{f}}\right) \mid \mid \mathbb{K}_\alpha\left(z \mid \underline{\underline{f}}\right) \right) \right]
 \notag
 \\  % ----
 &= \argmax_{\gamma_m, \gamma_s, \alpha} \mathbb{E}_{\underline{\underline{f}} \sim \textfrak{F}} \left[ 
 \mathcal{L}\left(\mathbb{G}_{\gamma_m,\gamma_s} \,, \alpha \,; \underline{\underline{f}}\right) \right] 
 \notag
 \\  % ----
 &= \argmax_{\gamma_m, \gamma_s, \alpha}
 \frac{1}{T} \sum_{i=1}^{T}
 \mathcal{L}\left(\mathbb{G}_{\gamma_m, \gamma_s} \,, \alpha \,; \underline{\underline{f_i}}\right) \,.
\end{align}
This is because $\mathtt{p}\left(\underline{\underline{f}}\right)$ is independent to $(\gamma_m \,, \gamma_s \,, \alpha)$.

For choosing $\mathbb{P}(z) = \mathcal{N}_d \left( \mathbf{0}_d \,, \text{Id}_d \right)$ and $\mathbb{G}_{\gamma_m, \gamma_s} \left(z \mid \underline{\underline{f}}\right) = \mathcal{N}_d \left( \mathscr{M}_{\gamma_m}\left(\underline{\underline{f}}\right) \,,
\text{diag} \left\{ \mathscr{S}_{\gamma_s} \left( \underline{\underline{f}} \right) \right\} \right)$, the KL-divergence term is defined above.
From Equation 25, we have a likelihood:
\begin{align}  
 \mathtt{h}_\alpha \left( \underline{\underline{f}} \mid z \right) =
 \mathcal{N}_D \left( \underline{\underline{f}} \,; \mathscr{D}_\alpha \left( \underline{\underline{\underline{c}}} \,, z \right) \,, \sigma^2 \text{Id} \right) \propto 
 \exp \left( - \frac{1}{2 \sigma^2} \norm{\underline{\underline{f}} - \mathscr{D}_\alpha \left( \underline{\underline{\underline{c}}} \,, z \right)}^2_{\ell_2} \right)
 \,.
\end{align}
Then, a minimization (\ref{eq:ExpectedEvidenceBound}) becomes the following minimization problem:
\begin{align}  \label{eq:RQUnetVAELossFunc}
 \left( \gamma_m^\dagger \,, \gamma_s^\dagger \,, \alpha^\dagger \right) 
 &= \argmax_{\gamma_m \,, \gamma_s, \alpha} \left\{ \mathscr{L}(\gamma_m,\gamma_s \,, \alpha) = \sum_{i=1}^{T}
 \mathcal{L} \left( \mathbb{G}_{\gamma_m,\gamma_s} \,, \alpha \,; \underline{\underline{f_i}} \right) \right\} \,.
\end{align}

% ===============

%
\subsection{Proof of Proposition II.9}
\label{proof:RQUnetVAEsegmentation}
 We firstly define a vector-valued function for $K$-clusters:
 \begin{align*}
  u \,:\, \mathbb R^P \rightarrow \mathbb R^K \,;~ 
  u \left( f_{i,l} \right) = \begin{pmatrix} u_1 \left( f_{i,l} \right) \\ \vdots \\ u_K \left( f_{i,l} \right) \end{pmatrix} \,,~
  u_k \left( f_{i,l} \right) \in \mathbb R \,,
 \end{align*}
 whose definition is defined in a matrix form by our proposed RQUnet-VAE as:
\begin{align}  \label{eq:RQUnetVAEsegmentation:map}
  u \left( \underline{\underline{f_i}} \right) 
  &:= \left[ u \left( f_{i,l} \right) \right]_{l \in \Omega}
  \notag
  \\  % ---
  &= \left\{ u_1 \left( \underline{\underline{f_i}} \right) \,, \hdots \,, 
  u_K \left( \underline{\underline{f_i}} \right) \right\}
  \in \mathbb R^{n_1 \times n_2 \times K}
  \,,~ u_k \left( \underline{\underline{f_i}} \right) \in \mathbb R^{n_1 \times n_2}
  \notag
  \\  % ----
  &= \underline{\underline{\underline{\mathfrak{C}^\text{iso}}}}_{\underline{\underline{\underline{\theta}}}}
  \circ 
  \mathscr{D}_\alpha \left( \mathscr{C}_{\gamma_c} \left(\underline{\underline{f_i}} \right) \,, \mathscr{M}_{\gamma_m} \left( \underline{\underline{f_i}} \right) + \mathscr{S}_{\gamma_s}^\frac{1}{2} \left( \underline{\underline{f_i}} \right) \odot \epsilon \right) \,,
 \end{align}
 where an isotropic matrix-family convolution is for the $K$-clusters:
 \begin{align*}
  \underline{\underline{\underline{\mathfrak{C}^\text{iso}}}}_{\underline{\underline{\underline{\theta}}}} &\,:\, \mathbb R^{\abs{\Omega} \times P} \rightarrow \mathbb R^{\abs{\Omega} \times K} \,.
 \end{align*}
 We define distributions for the set $\left( \textfrak{F} \,, \textfrak{F}^\text{gt} \right)$ as
 $\mathbb{Q}\left( f_{i,l} \right) = q \left( f_{i,l} \right) \text{d}f_{i,l} \,, 
 \mathbb{P}\left(f_{i,l}\right) = p \left( f_{i,l} \right) \text{d}f_{i,l}$ and their densities are:
 \begin{align*}
  q \left( f_{i,l} \right) &=
  \left[ q_k \left( f_{i,l} \right) \right]_{k=1}^K
  \in \mathbb R^K
  \,,~
  q_k \left( f_{i,l} \right) = \text{softmax} \left( u \left( f_{i,l} \right) \right)_k
  = \frac{ \exp \left( u_k \left( f_{i,l} \right) \right) }{ \sum_{h=1}^K \exp \left( u_h \left( f_{i,l} \right) \right) } \,,
  \\
  p \left( f_{i,l} \right) &= \left[ p_k \left( f_{i,l} \right) \right]_{k=1}^K \in \mathbb R^K 
  \,,~
  p_k \left( f_{i,l} \right) = f^\text{gt}_{i,l,k} \in \{ 0 \,, 1 \} \,,~
  \sum_{k=1}^K p_k \left( f_{i,l} \right) = 1 \,;
 \end{align*}
 then, cross-entropy loss between distributions $\mathbb{P}$ and $\mathbb{Q}$ is:
 \begin{align*}
  \mathcal{H}(\mathbb{P} \,, \mathbb{Q}) &= 
  - \mathbb E_\mathbb{P} \left[ \log \mathbb{Q} \right]
  := - \sum_{k=1}^K \mathbb E_{f \sim \textfrak{F}} \left[  p_k(f) \log q_k(f) \right]
  \\  % ---
  &= \lim_{n \to \infty} 
  - \frac{1}{n} \sum_{i=1}^n \mathscr{H} \left( \underline{\underline{f_i^\text{gt}}} \,, u \left( \underline{\underline{f_i}} \right) \right) \,, 
  \\ % 
  \mathscr{H} \left( \underline{\underline{f_i^\text{gt}}} \,, u \left( \underline{\underline{f_i}} \right) \right)
  &= \sum_{l \in \Omega} \sum_{k=1}^K f^\text{gt}_{i,l,k} \log \text{softmax} \left( u \left( f_{i,l} \right) \right)_k \,.
 \end{align*}
 The last equality is due to Monte Carlo approximation method.

\subsection{Proof of proposition \ref{prop:2DconvolutionHankelmatrix}}
\label{proof:prop2DconvolutionHankel}
We start with a 2D convolution (\ref{eq:2DConvolution})
\begin{align*}  
 \underline{\tilde{f}} &= \underline{\mathfrak{C}}_{\underline{\phi}}(\underline{f}) 
 = \begin{pmatrix} \tilde{f}_1 & \hdots & \tilde{f}_{d_2} \end{pmatrix}
 \in \mathbb R^{n_1 \times d_2} \,,~ \tilde{f}_i \in \mathbb R^{n_1} \,.
\end{align*} 
Then, we rewrite an element of the above 2D convolution as a 1D convolution for a 2D image 
$\underline{f} = \begin{pmatrix} f_{1} & \hdots & f_{n_2} \end{pmatrix}
\in \mathbb R^{\abs{\Omega}} 
\,,~ f_i \in \mathbb R^{n_1}$:
$r_1=1, \ldots, n_1 \,, r_2=1, \ldots, d_2$,
\begin{align*}
 \tilde{f}[r_1,r_2] &= \underline{\mathfrak{C}}_{\underline{\phi}}(\underline{f})[r_1,r_2] 
 = \sum_{k_1=1}^{d_1} \sum_{k_2=1}^{d_2} f[k_1,k_2] \check{\phi}[r_1-k_1, r_2-k_2]
 \\ % ----
 &= \sum_{k_2=1}^{d_2} \left( \sum_{k_1=1}^{d_1} f_{k_2}[k_1] \check{\phi}_{r_2}^{k_2}[r_1-k_1] \right)
 \\ % ----
 &= \sum_{k_2=1}^{d_2} \mathfrak{C}_{\phi_{r_2}^{k_2}}(f_{k_2})[r_1] 
 := \tilde{f}_{r_2}[r_1] \,,
\end{align*}
where column vector notations are
\begin{align*}
 f_{k_2}[k_1] := f[k_1,k_2] \,,~ 
 \phi_{r_2}^{k_2}[r_1-k_1] := \phi_{k_2}[r_1-k_1,r_2] 
 := \phi[r_1-k_1, r_2-k_2] \,.
\end{align*}
In conclusion, we have a useful relation between 2D and 1D convolution operations: $k_1=1, \ldots, n_1 \,, k_2=1, \ldots, d_2$,
\begin{align}  \label{eq:Relation2D1DConvolution:1}
 \tilde{f}_{k_2}[k_1] = \underline{\mathfrak{C}}_{\underline{\phi}}(\underline{f})[k_1,k_2] 
 = \sum_{i=1}^{d_2} \mathfrak{C}_{\phi_{k_2}^{i}}(f_{i})[k_1] 
 = \sum_{i=1}^{d_2} \left( \mathscr{H}_{d_1}(f_i) \phi_{k_2}^{i} \right)[k_1] \,,
\end{align}
where $f_{k_2}[k_1] := f[k_1,k_2]$.

Then, a multi-input-numlti-output convolution is:
\begin{align*}
 \underline{\tilde{f}} &= \begin{pmatrix} \tilde{f}_1 \,, & \hdots \,, & \tilde{f}_{d_2} \end{pmatrix} \in \mathbb R^{n_1 \times d_2} \,,~
 \tilde{f}_{i} \in \mathbb R^{n_1} \,,
 \\  % ----
 &= 
 \sum_{i=1}^{d_2} \begin{pmatrix} \mathfrak{C}_{\phi_{1}^{i}}(f_{i}) \,, & \hdots \,, & \mathfrak{C}_{\phi_{d_2}^{i}}(f_{i}) \end{pmatrix}
 % ----
 \stackrel{(\ref{eq:1DConvolution})}{=} \sum_{i=1}^{d_2} \mathscr{H}_{n_1}(f_i) 
 \underbrace{ \begin{pmatrix} \phi_{1}^{i} \,, & \hdots \,, & \phi_{d_2}^{i} \end{pmatrix} }_{ = \underline{\Phi^i} }
 \\  % ----
 &= \mathscr{H}_{n_1 \mid d_2}(\underline{f}) \underline{\Phi} \,.
\end{align*}
This concludes a relation (\ref{eq:2DConvolution}).

\clearpage
\newpage

\begin{figure*}[tbh]
\begin{center}  
\includegraphics[width=0.98\textwidth]{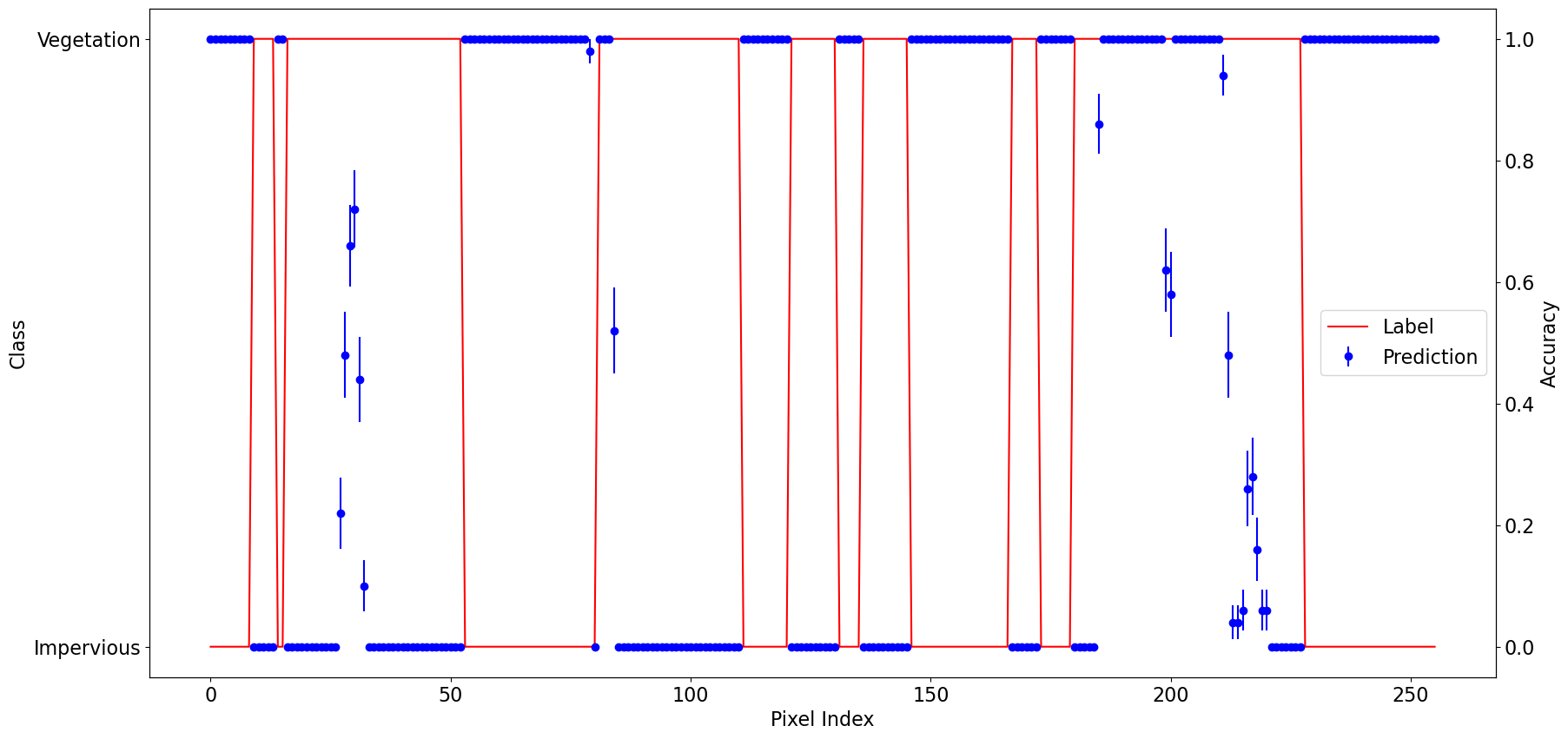}\label{fig:pixels}

\caption{The probability that individual  pixels are assigned to the impervious class along single horizontal line of pixels in an image during 50 RQUNet-VAE predictions compared to ground truth.   Blue dots correspond to predicted accuracy (right y-axis) whereas the red line corresponds to the ground truth land cover class (left y-axis).  This illustrates the results of variational terms of RQUNet-VAE.}
%\textcolor{red}{expand the description of a STD anf mean - green on top of orange}
\end{center}
\end{figure*}

\section{Additional Experiments}
While the main paper describes precision, recall, and F1-scores for entire images, this experiment looks at the accuracy of individual pixels an image to understand where our \RQUnetVAE incurs errors and where the variational approach is able to identify pixels having low confidence. For this purpose, we look at one horizontal line of pixels taken from the vertical center of the first noisy image shown in Figure~\ref{fig:NoisyTestImagesvar0_04:19059hdf}. For each pixel, from the left to the right of the image, Figure~\ref{fig:pixels} shows the ground truth label and the prediction accuracy. We see that for this line of pixels, the ground truth changes frequently between the impervious and the vegetation (non-impervious) class. We observe that there are many pixels where \RQUnetVAE provides a prediction accuracy of 1, implying that the class was predicted correctly in each of 50 runs of \RQUnetVAE (which is not deterministic due to the variational auto-encoder module). For example, on the very left of Figure~\ref{fig:pixels} we observe that numerous vegetation pixels were classified correctly in all cases. Then however, the ground truth has a few pixels classified as impervious, which our \RQUnetVAE is not able to capture. For these pixels, we have an accuracy of $0$, meaning that even the repeated iterations of our variational approach are not able to correctly classify these pixels. Across this horizontal line of pixels we also have numerous pixel where the variational approach provides a non-binary classification. For these cases, we provide the accuracy (as the fraction of correct classifications among the 50 runs) as well as the standard deviation of this accuracy denoted by the whiskers around the point estimate in the figure. We observe that in these cases, the variational approach allows
our \RQUnetVAE to possibly make better decisions by allowing to base a decision on multiple runs and take the consensus of all runs.

While these results appear weak due to many miss-classifications, we again reiterate that the underlying image has been heavily obfuscated with noise (compare the original images in Figure~\ref{fig:OriginalTestImages:19059hdf} with the noised images in Figure~\ref{fig:NoisyTestImagesvar0_04:19059hdf}) and that the experiments in the main paper show that the classic U-Net yields worse results, thus showing that our proposed \RQUnetVAE augments the traditional U-Net architecture by making it more robust to noise. 
\clearpage
\newpage

\bibliography{Current}

\bibliographystyle{abbrv}